\documentclass[12pt,oneside]{amsart}
\usepackage[margin=1in]{geometry}
\usepackage{abstract}
\usepackage{enumitem,verbatim}
\usepackage{multirow}
\usepackage{microtype}
\usepackage[all,cmtip]{xy}
\usepackage{physics}
\usepackage{hyperref}
\usepackage{color}
\hypersetup{colorlinks,linkcolor=blue,citecolor=cyan}
\usepackage{amsthm}
\usepackage[nameinlink]{cleveref}
\usepackage{amsmath,amssymb,stackrel,nicefrac}
\usepackage{amsaddr}
\usepackage{mathtools}
\usepackage[dvipsnames]{xcolor}
\usepackage{graphicx}
\usepackage{tikz, tikz-cd}

\usetikzlibrary{decorations.pathmorphing,shapes}

\newtheorem{lem}{Lemma}
\newtheorem{prop}{Proposition}
\newtheorem{thm}{Theorem}
\newtheorem{conj}{Conjecture}
\theoremstyle{definition}

\newtheorem{eg}{Example}

\newtheorem{qn}{Question}
\newtheorem{rmk}{Remark}
\definecolor{purple}{rgb}{0.7,0,0.7}

\newcommand{\overcross}{
 {\mathchoice
  {\includegraphics[height=1.6ex]{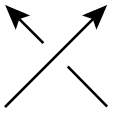}}
  {\includegraphics[height=1.6ex]{overcrossing.png}}
  {\includegraphics[height=1.2ex]{overcrossing.png}}
  {\includegraphics[height=0.9ex]{overcrossing.png}}
 }
}
\newcommand{\undercross}{
 {\mathchoice
  {\includegraphics[height=1.6ex]{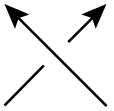}}
  {\includegraphics[height=1.6ex]{undercrossing.png}}
  {\includegraphics[height=1.2ex]{undercrossing.png}}
  {\includegraphics[height=0.9ex]{undercrossing.png}}
 }
}
\newcommand{\nocross}{
 {\mathchoice
  {\includegraphics[height=1.6ex]{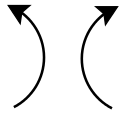}}
  {\includegraphics[height=1.6ex]{no_crossing.png}}
  {\includegraphics[height=1.2ex]{no_crossing.png}}
  {\includegraphics[height=0.9ex]{no_crossing.png}}
 }
}

\newcommand{\R}{{\mathbb{R}}}
\newcommand{\C}{{\mathbb{C}}}
\newcommand{\ind}{\operatorname{index}}
\newcommand{\m}[0]{\mathbb}

\newcommand{\G}[0]{\mathfrak}
\newcommand{\Bf}[0]{\boldsymbol}

\def\CN{\mathcal{N}}
\def\Z{{\mathbb{Z}}}
\def\be{\begin{equation}}
\def\ee{\end{equation}}
\newcommand{\bea}{\begin{eqnarray}}
\newcommand{\eea}{\end{eqnarray}}

\newcommand{\qbin}[2]{\begin{bmatrix}{#1}\\ {#2}\end{bmatrix}}
\title{Branches, quivers, and ideals for knot complements}

\date{October 2021}
\author{Tobias Ekholm$^{1,2}$, Angus Gruen$^3$, Sergei Gukov$^3$, Piotr Kucharski$^{3,4,5}$, Sunghyuk Park$^3$, Marko Sto\v{s}ić$^{6,7}$ and Piotr Su{\l}kowski$^{3,4}$}

\address{$^1$Department of Mathematics, Uppsala University, \\ L{\"a}gerhyddsv{\"a}gen 1, 752 37 Uppsala, Sweden}
\address{$^2$Institut Mittag-Leffler, Aurav{\"a}gen 17, 182 60 Djursholm, Sweden}
\address{$^3$Division of Physics, Mathematics and Astronomy, California Institute of Technology, 1200~E.~California Blvd., Pasadena, CA 91125, USA}
\address{$^4$Faculty of Physics, University of Warsaw, ul. Pasteura 5, 02-093 Warsaw, Poland}
\address{$^5$Institute of Physics \& Korteweg-de Vries Institute for Mathematics, \\ University of Amsterdam, Science Park 904 \& 105-107, Amsterdam, The~Netherlands}
\address{$^6$CAMGSD, Department of Mathematics, Instituto Superior T\'ecnico, Av. Rovisco Pais, 1049-001 Lisbon, Portugal}
\address{$^7$Mathematical Institute SANU, Knez Mihajlova 36, 11000 Beograd, Serbia}

\email{tobias.ekholm@math.uu.se}
\email{agruen@caltech.edu}
\email{p.j.kucharski@uva.nl}
\email{spark3@caltech.edu}
\email{psulkows@fuw.edu.pl}
\email{mstosic@isr.ist.utl.pt}

\begin{document}

\maketitle
\begin{abstract}
We generalize the $F_K$ invariant, i.e. $\widehat Z$ for the complement of a knot $K$ in the 3-sphere, the~knots-quivers correspondence, and $A$-polynomials of knots, and find several interconnections between them. We associate an~$F_K$ invariant to any branch of the~$A$-polynomial of $K$ and we work out explicit expressions for several simple knots. We show that these~$F_K$ invariants can be written in the form of a quiver generating series, in analogy with the knots-quivers correspondence. We discuss various methods to obtain such quiver representations, among others using $R$-matrices. We generalize the~quantum $a$-deformed $A$-polynomial to an~ideal that contains the recursion relation in the~group rank, i.e. in the parameter $a$, and describe its classical limit in terms of the Coulomb branch of a 3d-5d theory. We also provide $t$-deformed versions. Furthermore, we study how the~quiver formulation for closed 3-manifolds obtained by surgery leads to the~superpotential of 3d $\mathcal{N}=2$ theory $T[M_3]$ and to the~data of the~associated modular tensor category $\text{MTC} [M_3]$.
\end{abstract}

\newpage
\tableofcontents
\newpage

\section{Introduction}
The~interface of high energy theoretical physics and knot theory (and, more generally, low dimensional topology) is interesting from several points of view. On the~one hand, it yields exact and non-perturbative results that illustrate important physical concepts; on the~other, it leads to definitions of new invariants of knots (and three-manifolds) or reveals relations between such invariants originally defined in other ways. In recent years, several interesting results along these lines have been found. The following examples are relevant to this paper: a~3d-3d duality that relates three-manifolds to quantum field theories with extended supersymmetry~\cite{DGG}, $\widehat{Z}$ invariants of three-manifolds~\cite{GPV,GPPV},  and the~knots-quivers correspondence~\cite{KRSS1,KRSS2}. Our~aim here is to generalize and unify these concepts.

The~$A$-polynomial of a~knot $K$~\cite{CCGLS,Gar,Guk} $A_K(x,y)$ is an~important starting point for the~study undertaken in this paper. The~zero set of the~$A$-polynomial is an~algebraic curve in $(\C^\ast)^2$ which can be interpreted as the~moduli space of vacua of the~Chern-Simons theory with $SL(2,\mathbb{C})$ gauge group on the~complement of $K$ in $S^3$. It has been known for some time that the~Chern-Simons theory associates a~perturbative series in $\hbar$ (i.e., roughly, the~inverse of the~level $k$) to various branches of the~$A$-polynomial ~\cite{DGLZ}. All such series are annihilated by the~quantum counterpart of the~$A$-polynomial, i.e. an~operator $\widehat{A}_K(\hat x, \hat y,q)$, with $\hat y\hat x = q\hat x\hat y$, and with $q\to 1$ limit equal to the~original $A$-polynomial. This quantum $A$-polynomial also gives the~recursion relations (in colour) for coloured Jones polynomials. Furthermore, quite recently it was shown that the~perturbative series associated to the~abelian branch can be resummed into a~power series $F_K(x,q)$ in variables $q=e^\hbar$ and $x$ with integer coefficients~\cite{GM}. The~series $F_K(x,q)$, which is also referred to simply as the~$F_K$ invariant, can be interpreted as the~$\widehat{Z}$ invariant of the~three-manifold that is the complement of $K$ in $S^3$; $\widehat{Z}$ invariants were originally introduced for closed three-manifolds in~\cite{GPV,GPPV}, and then generalized to the~open version just mentioned.

The~$A$-polynomial and associated invariants admit $a$- and $t$-deformations. Here, the~$a$-deformation amounts to generalizing the~$SL(2,\mathbb{C})$ gauge group of the~Chern-Simons theory to $SL(N,\mathbb{C})$ and then replacing $N$-dependence of various quantities by a~uniform dependence on $a=q^N$, while $t$-deformation makes contact with homological knot invariants. The~$a$-deformed $A$-polynomials introduced in~\cite{AV} are closely related to augmentation polynomials of knot contact homology~\cite{AENV}, and their further $t$-deformation gives rise to super-$A$-polynomials~\cite{FGSA,FGS,FGSS}. Recently, an $SL(N,\mathbb{C})$ generalization and subsequent $a$-deformation of $F_K$ invariants were proposed respectively in~\cite{Park1} and~\cite{EGGKPS}.

In this paper we generalize the~($a$-deformed) $F_K$ invariants to other branches $\alpha$ of $A$-polynomials: we show that $\hbar$ expansions associated to all (not only abelian) branches $\alpha$ of $A$-polynomial can be also resummed into a~series $F^{(\alpha)}_K(x,q)$ with integer coefficients in $q$ and $x$ (and $a$, in the~$a$-deformed case; in the~rest of this section we often suppress $a$ from the~notation for $F_K$-invariants and $A$-polynomials). All these series are also annihilated by the~same quantum $A$-polynomial $\widehat{A}_K(\hat x, \hat y,q)$; in other words, they can be determined by the~same recursion relation encoded in $\widehat{A}_K(\hat x, \hat y,q)$, but with different initial condition. Furthermore, we show that the~initial conditions of $F^{(\alpha)}_K(x,q)$ for various branches of $A$-polynomial are encoded in slopes of the~boundaries of the~Newton polygon associated to the~$A$-polynomial. 

We further show that these different invariants $F^{(\alpha)}_K(x,q)$ can be expressed in the~form of quiver generating series. This statement can be regarded as an~extension of the~knots-quivers correspondence, introduced in~\cite{KRSS1,KRSS2} and further discussed in~\cite{PSS,PS,EKL1,EKL2,EKL3,Kimura:2020qns,Stosic:2017wno,SW2004,Larraguivel:2020sxk,Jankowski:2021flt}. The~first step towards such an~extension, for the~$F_K$ invariants associated to the~abelian branch for $(2,2p+1)$ torus knots, was taken in~\cite{Kuch}. In the~original formulation, the~knots-quivers correspondence is the~statement that a~generating function of coloured HOMFLY-PT polynomials can be written in the~form of a~quiver generating series; among others, this implies that LMOV invariants \cite{OV,LMV} are expressed in terms of integer motivic Donaldson-Thomas invariants of a~quiver~\cite{KS2,Efi,MR,FR}, which proves integrality of the~former invariants. A natural setup for the~original formulation of the~knots-quivers correspondence is the~Ooguri-Vafa configuration~\cite{OV,LMV}, which involves branes on the~conormal of a~knot~$K$ -- for this reason, we refer to quivers that arise in the~original formulation as knot conormal quivers. Analogously, we refer to quivers that encode series $F^{(\alpha)}_K(x,q)$ as knot complement quivers. For both conormals and complements, quiver can be interpreted in enumerative geometry with nodes corresponding to certain basic holomorphic curves and arrows as their boundary linkings. Here we will also discuss the~relation between the~knot conormal and knot complement, as well as their corresponding quivers. We point out that writing $F^{(\alpha)}_K(x,q)$ as a~quiver generating series shows that the~whole information of $F^{(\alpha)}_K(x,q)$ is encoded in a~finite set of data: the~quiver matrix, slopes of boundaries of a~Newton polygon, and a~few other parameters that enter a~quiver change of variables. Thus, the~seemingly complicated, quantum and non-perturbative information about Chern-Simons theory for various branches is actually captured by a~finite number of integer coefficients.

There is also an~interesting interplay between classical branches of the~$A$-polynomial and quivers encoding the~corresponding $F^{(\alpha)}_K(x,q)$: once we assume that a~quiver form of $F^{(\alpha)}_K(x,q)$ exists, we can determine it from the~form of the~classical branch (up to a~finite number of terms that can be fixed from the~knowledge of the first few terms of $F^{(\alpha)}_K(x,q)$). In this sense, the~$F_K$ series associated to various branches and (knot complement) quivers are intimately related. 

We summarize the~above two points in the~following conjecture.

\begin{conj}
Given a~knot $K$, let $y^{(\alpha)}(x,a)$ be a~branch of $y$ near $x=0$ (or $x=\infty$) of the~$a$-deformed $A$-polynomial of $K$, $A_K(x,y,a)$. 
\begin{enumerate}
    \item There exists a~wave function $F_K^{(\alpha)}(x,a,q)$ associated to this branch in a~sense that
    \[
    \langle \hat{y}\rangle := \lim_{q\rightarrow 1}\frac{F_K^{(\alpha)}(q x,a,q)}{F_K^{(\alpha)}(x,a,q)} = y^{(\alpha)}(x,a)
    \]
    and this wave function is annihilated by the~quantum $a$-deformed $A$-polynomial $\widehat{A}_K(\hat x,\hat y,a,q)$ (which is the~same for all branches $y^{(\alpha)}(x,a)$)
    \[
    \widehat{A}_K(\hat x,\hat y,a,q) F_K^{(\alpha)}(x,a,q) = 0. 
    \]
    \item The~wave function $F_K^{(\alpha)}(x,a,q)$ has a~quiver form, in the~sense of knots-quivers correspondence~\cite{KRSS1,KRSS2,Kuch}. 
\end{enumerate}

\end{conj}

In view of importance of the~above extension of the~knots-quivers correspondence, we develop several techniques to determine knot complement quivers corresponding to various branches of $A$-polynomial. One straightforward approach amounts to determining explicitly the~quiver generating series from the~first several terms in the $F_K$ series, taking advantage of the~recursion relation encoded in quantum $A$-polynomial (applied to relevant initial terms). Another approach that we propose is based on a~redefinition of knot conormal quivers that involves a~change of framing. Furthermore, taking advantage of the~results of~\cite{Park2,Park3}, we show that $F_K$ in quiver form can be reconstructed from $R$-matrices, as well as inverted Habiro series; in some specific examples we show that this approach works also for the $SL(N,\mathbb{C})$ case and then  generalized to include $a$-dependence. In general, it is not guaranteed that $F_K$ for a~given branch can be determined by all these approaches; however, combining them, we can determine $F_K$ for various branches of $A$-polynomial for many knots. 

We also extend the~$A$-polynomial. As mentioned above,  we consider $F_K$ invariants for various branches of $A$-polynomial, which gives the algebraic curve $A_K(x,y)=0$ with $x,y\in\mathbb{C}^*$. The~quantum $A$-polynomial $\widehat{A}_K(\hat x, \hat y)$ with $\hat y \hat x = q \hat x \hat y$ gives the~recursion relations for the~series $F^{(\alpha)}_K(x,q)$. In the~$a$-deformed case, there are two variables: $x$ and $a$, and degenerating the~brane system of the~knot complement to a~multiple of the~unknot conormal, we see that the~variables $x$ and $xa$ play almost identical roles. It is then natural to try to promote also the~variable $a$ to an~operator $\hat a$. This is indeed possible, and after introducing a~dual variable $b$, we find another algebraic curve $\{B_K(a,b)=0\}$, which is the zero set of what we call the~$B$-polynomial. Also the~$B$-polynomial has a~quantum counterpart $\widehat{B}_K(\hat a,\hat b)$, with $\hat{b}\hat{a}=q\hat{a}\hat{b}$, which again annihilates $F^{(\alpha)}_K(x,a,q)$. Thus, similarly to the~quantum $A$-polynomial, the~quantum $B$-polynomial captures information about $F^{(\alpha)}_K(x,a,q)$ that can be extracted recursively. Note also that $a$ is the~variable associated to the~closed holomorphic central $\mathbb{CP}^1$ in the~resolved conifold and the~$B$-polynomial indeed encodes information about closed string contributions (which are not detected by the~original $A$-polynomial).  

Note that, in case of $SL(N,\mathbb{C})$ theory, the authors of~\cite{GS} introduced an~algebraic curve and an~associated polynomial called the~$C$-polynomial, as well as an~operator that shifts $N$ and the~corresponding recursion relation in $N$ for coloured polynomials implemented by the~quantum $C$-polynomial. The~dependence on $N$ was subsequently generalized to $a$-dependence in~\cite{MM21}. Our $B$-polynomial is similar to, but different from the~$C$-polynomial (the definition of the $B$-polynomial involves the full coloured HOMFLY-PT polynomial, while that of the $C$-polynomial involves coefficients of a cyclotomic expansion), so our claims are accordingly independent of those in~\cite{GS,MM21}.

Physically, $x$ represents an~open string modulus that can be identified with a~Coulomb branch parameter of a~3d $\mathcal{N}=2$ theory associated to a~Lagrangian brane and $a$ is the~K\"ahler parameter of the~resolved conifold, which can be identified with a~Coulomb branch parameter of a~5d gauge theory.  Therefore, the~dependence on both $a$ and $x$, as well as their conjugate variables $b$ and $y$, should arise in the~description of a~coupled 3d-5d system, which means that the~$A$- and $B$-polynomials are not independent but should be viewed as specializations of a~higher dimensional variety associated to what we call the~$AB$-ideal. Likewise, the~quantum $A$- and $B$-polynomials arise as specializations of a~$D$-module that we call the~quantum $AB$-ideal, and which quantizes the~classical $AB$-ideal. This is summarized in the~following conjecture.

\begin{conj} \label{conj2}
Let us endow $(\mathbb{C}^*)^4$ with the~holomorphic symplectic form
\[
\Omega := d\log x \wedge d\log  y + d\log a \wedge d\log b,\quad x,y,a,b\in\mathbb{C}^*.
\]
For every knot $K$, there is a~holomorphic Lagrangian subvariety $\Gamma_K \subset (\mathbb{C}^*)^4$ with the~following properties:
\begin{enumerate}
    \item This holomorphic Lagrangian is preserved under the~Weyl symmetry
\[
x \mapsto a^{-1}x^{-1},\quad y \mapsto y^{-1},\quad a\mapsto a,\quad b\mapsto y^{-1}b.
\]
    \item The~projection of $\Gamma_K$ on $(\mathbb{C}^*)^3_{x,y,a}$ is the~zero set of the~$a$-deformed $A$-polynomial of $K$.
    \item Moreover, if $\hat{x}, \hat{y}, \hat{a}, \hat{b}$ are operators such that
\[
\hat{y}\hat{x} = q\hat{x}\hat{y},\quad \hat{b}\hat{a} = q\hat{a}\hat{b},
\]
and all the~other pairs commute, then the~ideal defining $\Gamma_K$ can be quantized to a~left ideal $\hat{\Gamma}_K\subset \mathbb{C}[\hat{x}^{\pm 1},\hat{y}^{\pm 1},\hat{a}^{\pm 1},\hat{b}^{\pm 1}]$ that annihilates $F_K(x,a,q)$. 
\end{enumerate}
\end{conj}

It is natural to extend the~above conjecture to the refined case. First, in analogy with (coloured) superpolynomials, we consider a~$t$-deformation of $F_K$~\cite{EGGKPS}. Then we introduce invariants $F^{(\alpha)}_K(x,a,q,t)$ for various branches. We can observe that in various expressions the~variable $t$ is mixed with $x$ or $a$, again corresponding to a~geometric degeneration to a~multiple of the~unknot conormal,  in a~way which suggests that there is a~natural conjugate variable $u$ of $t$. We conjecture that for a~given knot there exists a~holomorphic Lagrangian subvariety in $(\mathbb{C}^*)^6$ that captures the~semiclassical properties of the~$t$-deformation of $F_K$. This variety generalizes a~Lagrangian variety in $(\mathbb{C}^*)^4$ from Conjecture \ref{conj2} and its quantization is a~$D$-module in variables  $(\hat x, \hat y, \hat a,\hat b,\hat t,\hat u)$  that annihilates $F_K(x,a,q,t)$.
These statements are summarized the~following conjecture.

\begin{conj}\label{conj:refholLag}
Let us endow $(\mathbb{C}^*)^6$ with the~holomorphic symplectic form
\[
\Omega := d\log x \wedge d\log  y + d\log a \wedge d\log b + d\log t \wedge d\log u,\quad x,y,a,b,t,u\in\mathbb{C}^*.
\]
For every knot $K$, there is a~holomorphic Lagrangian subvariety $\Gamma_K \subset (\mathbb{C}^*)^6$ with the~following properties:
\begin{enumerate}
    \item This holomorphic Lagrangian is preserved under the~Weyl symmetry
    \[
    x\mapsto (-t)^{3}a^{-1}x^{-1},\quad y \mapsto t^s y^{-1},\quad a\mapsto a,\quad b\mapsto (-t)^{\frac{s}{2}}y^{-1}b,\quad t\mapsto t,\quad u \mapsto x^{-s}y^{-3}a^{-\frac{s}{2}}u,
    \]
    where $s$ is a~version of $s$-invariant of the~knot $K$. 
    \item The~projection of $\Gamma_K$ on $(\mathbb{C}^*)^4_{x,y,a,t}$ is the~zero set of the~super-$A$-polynomial of $K$. 
    \item Moreover, if $\hat{x}, \hat{y}, \hat{a}, \hat{b}, \hat{t}, \hat{u}$ are operators such that
\[
\hat{y}\hat{x} = q\hat{x}\hat{y},\quad \hat{b}\hat{a} = q\hat{a}\hat{b},\quad \hat{u}\hat{t} = q\hat{t}\hat{u},
\]
and all the~other pairs commute, then the~ideal defining $\Gamma_K$ can be quantized to a~left ideal $\hat{\Gamma}_K\subset \mathbb{C}[\hat{x}^{\pm 1},\hat{y}^{\pm 1},\hat{a}^{\pm 1},\hat{b}^{\pm 1},\hat{t}^{\pm 1},\hat{u}^{\pm 1}]$ that annihilates $F_K(x,a,q,t)$. 
\end{enumerate}
\end{conj}
\noindent We confirm the~above conjectures in a~number of nontrivial examples throughout the~paper.

The~rest of the~paper is organized as follows. Section \ref{sec-pre} summarizes basic information on knot invariants, brane systems, $\hat{Z}$ and $F_K$ invariants, and the~knots-quivers correspondence. In section \ref{sec:FK different branches} we introduce $F_K$ invariants for various branches of the $A$-polynomial and provide various examples of such objects. In section \ref{sec:FK invariants and quivers} we show that $F_K$ invariants can be written in the~form of quiver generating series, discuss some properties of such series, identify corresponding quivers in various examples, and explain how they can be reconstructed from the~$A$-polynomial and its branches. Section \ref{sec:Quiver from R-matrices} presents how to reconstruct quivers using $R$-matrices or inverted Habiro series. In section \ref{sec:Bpoly} we discuss the~quantized Coulomb branch from a more general 3d-5d physical perspective, introduce the $B$-polynomial, and show that both the $A$- and $B$-polynomials are specializations of a~higher dimensional complex Lagrangian variety, which has a~quantum counterpart that we call the~quantum $AB$-ideal. Section \ref{sec-3-logsft} contains a~discussion of other properties of various objects introduced earlier: consistency with constraints imposed by anomalies, subtle properties of vacua of supersymmetric theories corresponding to $F_K$, surgeries and relation to closed 3-manifolds, as well as relations to modular tensor categories and logarithmic vertex operator algebras. We conclude with a~brief summary of possible future directions of research in section \ref{sec-future}.

\section{Preliminaries}  \label{sec-pre}

In this section we present background material on knot polynomials, brane configurations and enumerative geometry, $\widehat Z$ and $F_K$ invariants, as well as the~knots-quivers correspondence.

\subsection{Knot polynomials}

\subsubsection{HOMFLY-PT polynomials}\label{sec:HOMFLY}
Let $K\subset S^{3}$ be a~knot. Its \emph{HOMFLY-PT polynomial} $P(K;a,q)$  is a~topological invariant~\cite{HOMFLY,PT} which can be calculated via the~skein relation:
\begin{equation*}
    a^{1/2}P(\overcross)-a^{-1/2}P(\undercross)=(q^{1/2}-q^{-1/2})P(\nocross)
\end{equation*}
with a~normalisation condition $P(0_{1};a,q)=1$. This is called the~{\it reduced} normalisation and corresponds to dividing by the~full natural HOMFLY-PT polynomial for the~unknot (here denoted by bar):
\begin{equation*}
    P(K;a,q)=\frac{\bar{P}(K;a,q)}{\bar{P}(0_{1};a,q)},
    \qquad
    \bar{P}(0_{1};a,q)=\frac{a^{1/2}-a^{-1/2}}{q^{1/2}-q^{-1/2}}.
\end{equation*}
For $a=q^2$, $P(K;a,q)$ reduces to the~Jones polynomial $J(K;q)$~\cite{Jon85}, whereas the~substitution $a=q^N$ leads to the~$\mathfrak{sl}_N$ Jones polynomial $J^{\mathfrak{sl}_N}(K;q)$~\cite{RT90}.

More generally, the~\emph{coloured HOMFLY-PT polynomials} $P_{R}(K;a,q)$ are similar polynomial knot invariants defined as the~expectation value of the~knot viewed as a~Wilson line in the~Chern-Simons gauge theory on $S^{3}$~\cite{witten1989}, which depends also on a~representation $R$ of $\mathfrak{sl}_N$. In this setting, the~original HOMFLY-PT corresponds to the~defining representation. The~substitutions $a=q^2$ and $a=q^N$ lead to the~coloured Jones polynomial $J_{R}(K;q)$ and coloured $\mathfrak{sl}_N$ Jones polynomials $J^{\mathfrak{sl}_N}_{R}(K;q)$ respectively. We will be interested mainly in the~HOMFLY-PT polynomials coloured by the~totally symmetric representations $R=S^{r}$ with $r$~boxes in one row of the~Young diagram. In order to simplify the~notation, we will denote them by $P_{r}(K;a,q)$ and call them simply the
HOMFLY-PT polynomials.

There is also a~$t$-deformation of the~HOMFLY-PT polynomials~\cite{DGR,GS1}. The~\emph{superpolynomial} $\mathcal{P}_r(K,a,q,t)$ is defined as the~Poincar\'e polynomial of the~triply-graded homology that categorifies the~HOMFLY-PT polynomial:
\begin{equation}
\begin{split}
    P_{r}(K;a,q)&=\sum_{i,j,k}(-1)^{k}a^{i}q^{j}\dim \mathcal{H}^{S^{r}}_{i,j,k}(K),\\   
   \mathcal{P}_{r}(K;a,q,t)&=\sum_{i,j,k}a^{i}q^{j}t^{k}\dim \mathcal{H}^{S^{r}}_{i,j,k}(K).
\end{split}   \label{Pr}
\end{equation}
From above equations one can immediately see that the~superpolynomial reduces to HOMFLY-PT for $t=-1$:
\[
    \mathcal{P}_{r}(K;a,q,-1)=P_{r}(K;a,q).
\]

\subsubsection{\texorpdfstring{$A$}{A}-polynomials}

The~\emph{$A$-polynomial} is a~knot invariant coming from the~character variety of the~complement of a~given knot $K$ in $S^3$~\cite{CCGLS}. It takes the~form of an~algebraic curve $A_K(x,y)=0$, for $x,y\in\mathbb{C}^*$. According to the~volume conjecture, it also captures the~asymptotics of the~coloured Jones polynomials $J_r(K;q)$ for large colours $r$. The~quantisation of the~$A$-polynomial encodes information about all colours, not only large ones. Namely, it gives the~recurrence relations satisfied by $J_r(K;q)$, which can be written in the~form
 \begin{equation*}
     \hat{A}_K(\hat x,\hat y) J_{*}(K;q) = 0,
 \end{equation*}
where $\hat x$ and $\hat y$ act by
\begin{equation*}
    \hat{x} J_r = q^r J_r, \quad \quad \hat{y}J_r = J_{r+1},
\end{equation*}
and satisfy the~relation $\hat{y}\hat{x}=q\hat{x}\hat{y}$. The~above conjecture was proposed independently in the~context of quantisation of the~Chern-Simons theory~\cite{Guk} and in parallel mathematics developments~\cite{Gar}. The~operator $\hat{A}_K(\hat x,\hat y)$ is referred to as the~\emph{quantum $A$-polynomial}; in the~classical limit $q=1$ it reduces to the~$A$-polynomial.

The~above conjectures were generalized to coloured HOMFLY-PT polynomials~\cite{AV} and coloured superpolynomials~\cite{FGSA,FGS}, which we briefly introduced in \eqref{Pr}. In these cases the~objects mentioned in the~previous paragraph become $a$- and $t$-dependent. In particular, the~asymptotics of coloured superpolynomials $\mathcal{P}_{r}(K;a,q,t)$ for large $r$ is captured by an~algebraic curve called the~\emph{super-$A$-polynomial}, defined by the~equation $A_K(x,y,a,t)=0$. For $t=-1$ it reduces to $a$-deformed $A$-polynomial, and upon setting in addition $a=1$, we obtain the~original $A$-polynomial as a~factor. For brevity, all these objects are often referred to as~$A$-polynomials. The~quantisation of the~super-$A$-polynomial gives rise to quantum super-$A$-polynomial $\hat{A}_K(\hat x,\hat y,a,q,t)$, which is a~$q$-difference operator that encodes the~recurrence relations for the~coloured superpolynomials:
 \begin{equation*}
     \hat{A}_K(\hat x,\hat y,a,q,t) \mathcal{P}_{*}(K;a,q,t) = 0.
 \end{equation*}
 A~universal framework that enables us to determine a~quantum $A$-polynomial from an~underlying classical curve $A(x,y)=0$ was proposed in~\cite{GS2} (irrespective of extra parameters these curves depend on, and also beyond examples related to knots).


\subsection{Brane configurations and enumerative geometry}\label{sec:Brane constructions}

\subsubsection{Large-\texorpdfstring{$N$}{N} transition}
\label{sec:Large N transition}
The~physical background of this work can be represented by the~system of $N$ fivebranes supported on $\mathbb{R}^2 \times S^1 \times Y$, where $Y$ is embedded (as the~zero-section) inside the~Calabi-Yau space $T^* Y$ and $\mathbb{R}^2 \times S^1 \subset \mathbb{R}^4 \times S^1$:
\begin{equation}
\begin{split}\text{spacetime}:\quad & \mathbb{R}^{4}\times S^{1}\times T^* Y\\
 & \cup\phantom{\ \times S^{1}\times\ \ }\cup\\
N~\text{M5-branes}:\quad & \mathbb{R}^{2}\times S^{1}\times Y.
\end{split}
\label{MgeneralY}
\end{equation}
Finding the~large-$N$ limit of this system for general $Y$ is highly nontrivial (see~\cite[sec.7]{GPV} and~\cite[Remark 2.4]{ES}). However, when $Y$ is a~knot complement $M_{K}:=S^{3}\backslash K$, there is an~equivalent description for which the~study of large-$N$ behaviour can be reduced to the~celebrated ``large-$N$ transition''~\cite{GV,OV}. 

We consider first a~description without transition. From the~viewpoint of 3d-3d correspondence, $N$ fivebranes on $Y = M_{K}$ produce a~4d $\mathcal{N}=4$ theory -- which is a~close cousin of 4d $\mathcal{N}=4$ super-Yang-Mills but is {\it not} 4d $\mathcal{N}=4$ SYM -- on a~half-space $\mathbb{R}^3 \times \mathbb{R}_+$ coupled to 3d $\mathcal{N}=2$ theory $T[M_K]$ on the~boundary.
Indeed, near the~boundary $T^2 =\Lambda_K= \partial M_K$, the~compactification of $N$ fivebranes produces a~4d $\mathcal{N}=4$ theory which has moduli space of vacua $\text{Sym}^N (\mathbb{C}^2 \times \mathbb{C}^*)$~\cite{CGPS}.
(The~moduli space of vacua in 4d $\mathcal{N}=4$ SYM is $\text{Sym}^N (\mathbb{C}^3)$.)
The~$SU(N)$ gauge symmetry of this theory appears as a~global symmetry of the~3d boundary theory $T[M_K]$. In particular, the~variables $x_i \in \mathbb{C}^*$ 
are complexified fugacities for this global (``flavour'') symmetry.
For $G=SU(2)$, the~moduli space of vacua of the~knot complement theory $T[M_K]$ gives precisely the~$A$-polynomial of $K$. Similarly, $G_{\mathbb{C}}$ character varieties of $M_{K}$ are realised as spaces of vacua in $T[M_{K},SU(N)]$ with $G=SU(N)$~\cite{FGS,FGSS}.

We next give another equivalent description of the~physical system \eqref{MgeneralY} with $Y = M_{K}$, where the~large-$N$ behaviour is easier to analyse:
\begin{equation}
\begin{split}\text{spacetime}:\quad & \mathbb{R}^{4}\times S^{1}\times T^* S^3\\
 & \cup\phantom{\ \times S^{1}\times\ \ }\cup\\
N~\text{M5-branes}:\quad & \mathbb{R}^{2}\times S^{1}\times S^3 \\
\rho~\text{M5$'$-branes}:\quad & \mathbb{R}^{2}\times S^{1}\times L_K.
\end{split}
\label{Mdeformed}
\end{equation}
This brane configuration is basically a~variant of \eqref{MgeneralY} with $Y=S^3$ and $\rho$ extra M5-branes supported on $\mathbb{R}^2 \times S^1 \times L_K$, where $L_K \subset T^* S^3$ is the~conormal bundle of the~knot $K \subset S^3$. There is, however, a~crucial difference between fivebranes on $S^3$ and $L_K$. Since the~latter are non-compact in two directions orthogonal to $K$, they carry no dynamical degrees of freedom away from $K$. One can path integrate those degrees of freedom along $K$, which effectively removes $K$ from $S^3$ and puts the~corresponding boundary conditions on the~boundary $T^2 = \partial M_K$. The~resulting system is precisely \eqref{MgeneralY} with $Y = M_{K}$. Equivalently, one can use the~topological invariance along $S^3$ to move the~tubular neighbourhood of $K \subset S^3$ to ``infinity''. This creates a~long neck isomorphic to $\mathbb{R} \times T^2$, as in the~above discussion. Either way, we end up with a~system of $N$ fivebranes on the~knot complement and no extra branes on $L_K$, so that the~choice of $GL(\rho,\mathbb{C})$ flat connection on $L_K$ is now encoded in the~boundary condition for $SL(N,\mathbb{C})$ connection\footnote{To be more precise, it is a~$GL(N,\mathbb{C})$ connection, but the~dynamics of the~$GL(1,\mathbb{C})$ sector is different from that of the~$SL(N,\mathbb{C})$ sector and can be decoupled.} on $T^2 = \partial M_K$. In particular, the~latter has at most $\rho$ nontrivial parameters $x_i \in \mathbb{C}^*$, $i=1,\ldots, \rho$.

In this paper we consider the~simplest case of $\rho=1$. Then we can use the~geometric transition of~\cite{GV}, upon which there is one brane on $L_K$ and $N$ fivebranes on the~zero-section of $T^* S^3$ disappear. The~Calabi-Yau space $T^* S^3$ undergoes a~topology changing transition to a~new Calabi-Yau space $X$, the~so-called ``resolved conifold'', which is the~total space of $\mathcal{O} (-1) \oplus \mathcal{O} (-1) \to \mathbb{CP}^1$, and only the~Ooguri-Vafa fivebranes supported on the~conormal bundle $L_K$ remain:
\begin{equation}
\begin{split}\text{spacetime}:\quad & \mathbb{R}^{4}\times S^{1}\times X\\
 & \cup\phantom{\ \times S^{1}\times\ \ }\cup\\
\rho~\text{M5$'$-branes}:\quad & \mathbb{R}^{2}\times S^{1}\times L_{K}.
\end{split}
\label{Mresolved}
\end{equation}
Note that on the~resolved conifold side, i.e. after the~geometric transition, $\log a =\text{Vol} (\mathbb{CP}^1) + i \int B = N\hbar$ is the~complexified K\"ahler parameter which enters the~generating function of enumerative invariants. 

To summarize, a~system of $N$ fivebranes on a~knot complement \eqref{MgeneralY} is equivalent to a~brane configuration \eqref{Mresolved}, with a~suitable map that relates the~boundary conditions in the~two cases. There is another system, closely related to \eqref{Mresolved}, that one can obtain from \eqref{Mdeformed} by first reconnecting $\rho$ branes on $L_K$ with $\rho$ branes on $S^3$. This give $\rho$ branes on $M_K$ (that go off to infinity just like $L_K$ does) plus $N-\rho$ branes on $S^3$. Assuming that $\rho \sim O(1)$ as $N \to \infty$ (e.g. $\rho=1$ in the~context of this paper), after the~geometric transition we end up with a~system like \eqref{Mresolved}, except $L_K$ is replaced by $M_K$ and $\text{Vol} (\mathbb{CP}^1) + i \int B = (N-\rho) \hbar$.
Both of these systems on the~resolved side compute the~HOMFLY-PT polynomials of $K$ coloured by Young diagrams with at most $\rho$ rows.

\subsubsection{Twisted superpotential}

The~leading, genus-0 contribution to the~generating function of enumerative invariants is the~twisted superpotential. It can be computed either on the~resolved side of the~transition, where $a$ is a~K\"ahler parameter, or on the~original (``deformed'') side, for a~family of theories labelled by $N$. Either way, one finds that the~twisted superpotential is given by the~double-scaling limit that combines large-colour and semiclassical limits of the~HOMFLY-PT polynomials~\cite{FGS,FGSS}:
\begin{equation}\label{eq:W from P}
P_{r}(K;a,q)\stackrel[\hbar\rightarrow0]{r\rightarrow\infty}{\longrightarrow}\int\prod_{i}\frac{dz_{i}}{z_{i}}\exp\left(\frac{1}{\hbar}\widetilde{\mathcal{W}}_{T[M_{K}]}(z_{i},x,a)+O(\hbar^{0})\right),
\end{equation}
with $x=q^{r}$ kept fixed. We can read off the~structure of $T[M_{K}]$
from the~terms in $\widetilde{\mathcal{W}}_{T[M_{K}]}(z_{i},x,a)$:
\begin{equation}
\begin{split}\textrm{Li}_{2}\ensuremath{\left(a^{n_{Q}}x^{n_{M}}z_{i}^{n_{z_i}}\right)}\qquad & \longleftrightarrow\qquad\text{(chiral field)}\,,\\
\frac{\kappa_{ij}}{2}\log\zeta_{i}\cdot\log\zeta_{j}\qquad & \longleftrightarrow\qquad\text{(Chern-Simons coupling)}\,.
\end{split}
\label{eq:Li2 and logs dictionary}
\end{equation}
Each dilogarithm is interpreted as the~one-loop contribution of a~chiral superfield with charges $(n_{Q},n_{M},n_{z_i})$ under the~global symmetries $U(1)_{Q}$ (arising from the~internal 2-cycle in $X$) and~$U(1)_{M}$ (corresponding to the~non-dynamical gauge field on the~M5-brane), as well as the~gauge group $U(1)\times\ldots\times U(1)$. Quadratic-logarithmic terms are identified with Chern-Simons couplings among the~various $U(1)$ symmetries, with $\zeta_{i}$ denoting the~respective fugacities.

We can integrate out the~dynamical fields (whose VEVs are given by $\log z_{i}$) using the~saddle point approximation to obtain the~effective twisted superpotential:
\begin{equation}
\label{eq:Weff}
\widetilde{\mathcal{W}}^{\textrm{eff}}_{T[M_{K}]}(x,a)=\frac{\partial\widetilde{\mathcal{W}}_{T[M_{K}]}(z_{i},x,a)}{\partial\log z_{i}}.
\end{equation}
Then after introducing the~dual variable $y$ (the~effective Fayet-Iliopoulos parameter), we arrive at the~$A$-polynomial:
\begin{equation}
\label{eq:AfromW}
\log y=\frac{\partial\widetilde{\mathcal{W}}^{\textrm{eff}}_{T[M_{K}]}(x,a)}{\partial\log x}\qquad\Leftrightarrow\qquad A_K(x,y,a)=0.
\end{equation}

\subsubsection{Curve counts}
We recall a~geometric picture underlying the~quiver description of the~invariant $F_K$. Recall the~curve counting interpretation of $F_K$ from~\cite{EGGKPS}. We start from the~deformed conifold $T^{\ast}S^{3}$ and the~knot complement $M_K$, with the~Legendrian conormal $\Lambda_{K}\subset ST^{\ast} S^{3}$ as ideal boundary. In case $K$ is fibered, we shift $M_{K}$ off of $S^3$ along a~closed $1$-form $\beta$ that generates $H^{1}(M_{K})=\R$. We take this form to agree with the~form $d\mu$ that is dual to the~meridian circle on the~boundary of a~tubular neighbourhood of the~knot. If $K$ is not fibered, such a~1-form necessarily has zeroes and the~shift leaves intersection points between $M_K$ and $S^3$ that can be normalized to locally appear as cotangent fibers.    

We want to count (generalized) holomorphic curves with boundary on $M_{K}$. There are two ways to do this for fibered knots: either we consider $M_{K}$ as a~Lagrangian submanifold in the~resolved conifold $X$ or we can use a~sufficiently SFT-stretched almost complex structure on $T^{\ast}S^{3}$ for which all curves leave a~neighbourhood of the~zero section, see~\cite[Section~2.5]{ES}. The~resulting counts (and in fact the~curves) are the~same. In the~non-fibered case, the~second approach still works: after stretching $M_{K}$ intersects a~neighbourhood of $S^{3}$ in $T^{\ast} S^{3}$ in a~finite collection of cotangent fibers. Then possible curves in the~inside region (near~$S^3$) have boundaries on these fibers and positive punctures at Reeb chords corresponding to geodesics connecting them. The~dimension of such a~curve is 
\[ 
\dim = \sum_{j}(\ind(\gamma_{j})+1)\ge 2, 
\]  
where the~sum runs over positive punctures of the~curve, $\gamma_{j}$ is the~Reeb chord at the~puncture, and $\ind(\gamma_{j})$ is the~Morse index of the~corresponding geodesic.
It follows that no such curve can appear after stretching, since the~outside part would then have negative index. This means that there is a~curve count also for $M_{K}$. As we will discuss below, although this curve count is well defined and invariant, when intersections between $S^{3}$ and $M_{K}$ cannot be removed, it is only one point in a~space of curve counts that also takes into account certain punctured curves. The~present discussion applies to the~more involved curve counts of punctured curves as well, although the~geometric description is less direct.   

In this setting the~quiver picture arises from a~description of a~the~Lagrangian $M_K$ as deformed by a~collection of basic holomorphic disks. Here the~disks are embedded and deform the~standard cotangent neighbourhood of $M_K$ (where there are no holomorphic curves) to the~cotangent bundle with neighbourhoods of these basic disks attached. Such neighbourhoods support holomorphic curves that come from branched covers of the~basic disks with constant curves attached and the~total count of generalized holomorphic curves is determined by the~linking data of the~basic disk boundaries. It is given by the~quiver partition function, where the~nodes correspond to the~basic disks and the~arrows to the~linking information. Here the~data of a~node is the~homological degree of its boundary and the~number of wraps around the~central $\mathbb{CP}^1$. We point out that nodes of boundary degree zero play a~special role for non-fibered knots.

\subsection{\texorpdfstring{$\hat{Z}$}{Z} and \texorpdfstring{$F_K$}{FK} invariants}\label{sec:Zhat and FK}

\subsubsection{\texorpdfstring{$\mathfrak{sl}_2$}{sl(2)} and \texorpdfstring{$\mathfrak{sl}_N$}{\mathfrak{sl}_N} invariants}

In their study of $3$d $\mathcal{N}=2$ theories $T[Y]$ for 3-manifolds $Y$, Gukov-Putrov-Vafa~\cite{GPV} and Gukov-Pei-Putrov-Vafa~\cite{GPPV} conjectured the~existence of the~$3$-manifold invariants $\widehat{Z}(Y)$ (also known as ``homological blocks'' or ``GPPV invariants'') valued in $q$-series with integer coefficients.
These $q$-series invariants exhibit peculiar modular properties, the~exploration of which was initiated in~\cite{BMM1,CCFGH,BMM2,CFS}.

More recently, Gukov-Manolescu~\cite{GM}
introduced a~version of $\widehat{Z}$ for knot complements, which they called $F_K$. If $K\subset S^3$ is a~knot, then $F_K=\widehat{Z}(S^3\setminus K)$.
The~motivation was to study $\widehat{Z}$ more systematically using Dehn surgery. 
Recall that the~Melvin-Morton-Rozansky expansion~\cite{MM,BNG,Roz1,Roz2} (also known as ``loop expansion'' or ``large colour expansion'') of the~coloured Jones polynomials is the~asymptotic expansion near $\hbar \rightarrow 0$ while keeping $x=q^r=e^{r\hbar}$ fixed:  
\begin{equation*}
    J_r(K;q=e^\hbar) = \sum_{j\geq 0}\frac{p_j(x)}{\Delta_K(x)^{2j+1}}\frac{\hbar^j}{j!},\quad\quad p_j(x)\in \mathbb{Z}[x,x^{-1}],\quad p_0 = 1.
\end{equation*}
Here $\Delta_K(x)$ is the~Alexander polynomial of $K$. 
The~main conjecture of~\cite{GM} was then the~following. 
\begin{conj}\label{conjGukovManolescu}
    For every knot $K\subset S^3$, there exists a~two-variable series 
    \begin{equation}\label{eq:GMconj}
    F_K(x,q) = \frac{1}{2}\sum_{\substack{m\geq 1\\m\text{ odd}}} f_m(q) (x^{m/2}-x^{-m/2}),\quad\quad f_m(q)\in \mathbb{Z}[q^{-1},q]]
    \end{equation}
    such that its asymptotic expansion agrees with the~Melvin-Morton-Rozansky expansion of the~coloured Jones polynomials\footnote{\cite{GM}  uses the~unreduced normalisation. In the~reduced normalisation, used in the~major part of this paper, \eqref{eq:GM-Melvin-Morton} reads $F_K(x,q=e^{\hbar}) = \sum_{j\geq 0}\frac{p_j(x)}{\Delta_K(x)^{2j+1}}\frac{\hbar^j}{j!}.$}: 
    \begin{equation}\label{eq:GM-Melvin-Morton}
        F_K(x,q=e^{\hbar}) = (x^{1/2}-x^{-1/2})\sum_{j\geq 0}\frac{p_j(x)}{\Delta_K(x)^{2j+1}}\frac{\hbar^j}{j!}.
    \end{equation}
    Moreover, this series is annihilated by the~quantum $A$-polynomial: 
    \begin{equation*}
        \hat{A}_K(\hat{x},\hat{y},q)F_K(x,q) = 0.
    \end{equation*}
\end{conj}
Let us stress that while the~same form of quantum $A$-polynomial $\hat{A}_K(\hat x,\hat y)$ arises in the~analysis of coloured Jones polynomial and $F_K$ invariants, there is a~subtle but important difference between these two situations, which has to do with the~initial conditions that need to be imposed.

Conjecture \ref{conjGukovManolescu} concerns $\mathfrak{g}=\mathfrak{sl}_2$. An~extension 
to arbitrary $\mathfrak{g}$ was studied in~\cite{Park1}. In particular, the~existence of an~$\mathfrak{sl}_N$ generalisation of $F_K$, which we denote by~$F_K^{\mathfrak{sl}_N}$, was conjectured and it was observed that its specialisation to symmetric representations, $F_K^{\mathfrak{sl}_N,sym}$, is annihilated by the~corresponding quantum $A$-polynomial: 
\begin{equation}
    \hat{A}_K(\hat{x},\hat{y},a=q^N,q)F_K^{\mathfrak{sl}_N,sym}(x,q) = 0.
\end{equation}

\subsubsection{\texorpdfstring{$a$}{a}-deformed \texorpdfstring{$F_K$}{FK} invariants}

The~analysis of~\cite{EGGKPS} showed that $F_K^{\mathfrak{sl}_N,sym}(x,q)$ for all $N$ can be captured by $a$-deformed invariants $F_{K}(x,a,q)$:
\begin{equation}
F_{K}(x,a=q^{N},q)=F_K^{\mathfrak{sl}_N,sym}(x,q).
\end{equation}
They are annihilated by the~quantum $a$-deformed $A$-polynomials:
\begin{equation}
\hat{A}_K(\hat{x},\hat{y},a,q)F_{K}(x,a,q)=0,
\end{equation}
where the~operators act as follows:
\begin{equation}
\hat{x}F_{K}(x,a,q)=x F_{K}(x,a,q),\qquad\hat{y}F_{K}(x,a,q)=F_{K}(q x,a,q).
\end{equation}
Moreover, the~asymptotic expansion of $F_{K}(x,a,q)$ should agree with that of HOMFLY-PT polynomials. That is 
\begin{equation}\label{eq:asymptotics}
    \log F_K(e^{r\hbar},a,e^\hbar) = \log P_r(K;a,e^{\hbar})
\end{equation}
as $\hbar$-series. In certain cases (including all $(2,2p+1)$~torus knots) lifting this relation to the~resummed level leads to a~simple substitution:
\begin{equation}
F_{K}(x,a,q)=\left.P_{r}(K;a,q)\right|_{q^{r}=x}\label{eq:connection between F_K and HOMFLY-PT},
\end{equation}
but this fails for many simple knots like $4_1$.
From the~physical point of view,  $F_K^{\mathfrak{sl}_N,sym}(x,q)$ and
$F_{K}(x,a,q)$ encode BPS spectra on two sides of the~large-$N$ transition discussed in Section~\ref{sec:Large N transition}.

\subsubsection{Normalisation and Conventions} \label{sec: conventions}

Before continuing, we make a~few remarks on different conventions involving $F_K$ and $J^{\mathfrak{sl}_N}_r$.

\begin{itemize}
    \item In~\cite{GM} and Conjecture \ref{conjGukovManolescu}, $F_K$ is presented in the~\textit{balanced} expansion which involves a~summation over both positive and negative powers of $x$ with Weyl symmetry being manifest. However, for our purposes it will be more natural to work with the~\textit{positive} expansion, as in~\cite{EGGKPS}. This means that we express $F_K$ as a~power series in $x$ expanded around $0$. There is also a~closely related \textit{negative} expansion coming from expanding around $x = \infty$ or by applying Weyl symmetry to the~positive expansion. The~balanced expansion can be rederived by averaging the~positive and negative expansions.
    
    \item When working with quiver forms (Sections \ref{sec:KQ correspondence} and \ref{sec:FK invariants and quivers}), we often treat $F_K$ as an~integer power series starting with $1$, see e.g. Equation \eqref{eq:GM-Q correspondence}. We stress that this is only correct up to an~overall prefactor
    \begin{equation} \label{eq: Full Prefac FK}
       \exp(\frac{p(\log x,\log a)}{\hbar})
    \end{equation}
    where $p$ is a~polynomial of degree at most 2. These prefactors are important for some properties of $F_K$ and can be derived from the~$A$- and $B$-polynomials as in Section \ref{subsec: Computing FK on different branches}.
    
    \item In the~literature there are a~collection of different normalisations in which $F_K$ and $J^{\mathfrak{sl}_N}_r$ are presented. The~three possibilities correspond to the~different values which can be assigned to the~unknot invariant.
        \begin{itemize}
            \item The~\textit{reduced} normalisation corresponds to normalizing away the~unknot,
            \begin{equation*}
           J_r^{\mathfrak{sl}_N}(\mathbf{0}_1, q) = 1 = F_{\mathbf{0}_1}(x, a, q).
            \end{equation*}
            This is the~convention most present in the~literature on HOMFLY-PT, superpolynomials, and $A$-polynomials, e.g.~\cite{DGR,FGSA,FGS,FGSS,NRZS, EGGKPS}.
            \item The~\textit{unreduced} normalisation corresponds to normalizing away the~denominator of the~full unknot factor,
            \begin{equation*}
                J_r^{\mathfrak{sl}_N, unreduced}(\mathbf{0}_1, q) = \frac{(xq; q)_{\infty}}{(xa; q)_{\infty}} = F_{\mathbf{0}_1}^{unreduced}(x, a, q).
            \end{equation*}
            This convention is common in the~growing literature on $F_K$ invariants, e.g.~\cite{GM,Park1,Park2,GHNPPS, EGGKPS}, particularly when studying unknot invariants or working with the~balanced expansion of $F_K$.
            \item The~\textit{fully unreduced} normalisation corresponds to leaving the~full unknot factor intact,
            \begin{equation*}
                J_r^{\mathfrak{sl}_N, fully\,unreduced}(\mathbf{0}_1, q) = e^{\frac{-\log(x)\log(a)}{2\hbar}}x^{\frac{1}{2}}\frac{(a; q)_{\infty}(xq; q)_{\infty}}{(xa; q)_{\infty}(q; q)_{\infty}} = F_{\mathbf{0}_1}^{fully\,unreduced}(x, a, q).
            \end{equation*}
            This normalisation is natural in the~context of enumerative invariants and can be found in~\cite{OV,AENV,EN,EKL2,EKL1,ES,DE}. In the~literature this normalisation is usually called just ``unreduced", but since we join different perspectives, we distinguish it from the~one discussed in the~previous point.
        \end{itemize}
    We primarily work in the~reduced normalisation and each use of the~other normalisation (e.g. in the~analysis of the~unknot) is clearly stated.
\end{itemize}

\subsection{Knots-quivers correspondence}\label{sec:KQ correspondence}

\subsubsection{Quivers and their representations}

A~quiver $Q$ is an\,oriented graph, i.e.~a~pair $(Q_{0},Q_{1})$
where $Q_{0}$ is a~finite set of vertices and $Q_{1}$~is a~finite
set of arrows between them. We number the~vertices by $1,2,...,m=|Q_{0}|$.
An~adjacency matrix of~$Q$ is the~$m\times m$ integer matrix
with entries $C_{ij}$ equal to the~number of arrows from $i$ to~$j$.
If $C_{ij}=C_{ji}$, we call the~quiver symmetric.

A~quiver representation with a~dimension vector $\boldsymbol{d}=(d_{1},...,d_{m})$
is an~assignment of a~vector space of dimension $d_{i}$ to the~node
$i\in Q_{0}$ and a~linear map $\gamma_{ij}:\mathbb{C}^{d_{i}}\rightarrow\mathbb{C}^{d_{j}}$
to each arrow from vertex $i$ to vertex $j$. Quiver representation
theory studies moduli spaces of quiver representations. While explicit expressions for invariants describing those spaces are difficult to find in general, they are quite well understood in the~case of
symmetric quivers~\cite{KS1,KS2,Efi,MR,FR}. Important
information about the~moduli space of representations of a~symmetric
quiver is encoded in the~\emph{motivic generating series} defined as
\begin{equation}\label{eq:motivic generating series}
P_{Q}(\boldsymbol{x},q)=\sum_{\boldsymbol{d}\geq0}(-q^{1/2})^{\boldsymbol{d}\cdot\boldsymbol{C}\cdot\boldsymbol{d}}\frac{\boldsymbol{x^{d}}}{(q;q)_{\boldsymbol{d}}}=\sum_{d_{1},\ldots,d_{m}\geq0}(-q^{1/2})^{\sum_{i,j}C_{ij}d_{i}d_{j}}\prod_{i=1}^{m}\frac{x_{i}^{d_{i}}}{(q;q)_{d_{i}}},
\end{equation}
where the\,denominator is the~$q$-Pochhammer symbol:
\begin{equation*}
(z;q)_{n}=\prod_{k=0}^{n-1}(1-zq^{k}).
\end{equation*}
When $z=q$, we will simplify the~notation and write $(q)_n$ and $(q)_{\Bf{d}}$ instead of $(q;q)_n$, and $(q;q)_{\Bf{d}}=\prod_i (q;q)_i$ respectively.

Let us define the~\emph{plethystic exponential} of $f=\sum_{n}a_{n}t^{n}$, $a_{0}=0$ in the~following way:
\begin{equation*}
\mathrm{Exp}\bigl(f\bigr)(t)=  
\exp\left(\sum_{k}\tfrac{1}{k}f(t^{k})\right) =\prod_{n}(1-t^{n})^{a_{n}}.
\end{equation*}
Then we can write
\begin{equation}
\begin{split}
P_{Q}(\boldsymbol{x},q)&=\textrm{Exp}\left(\frac{\Omega(\boldsymbol{x},q)}{1-q}\right),\label{eq:P_Q=Exp}\\
\Omega(\boldsymbol{x},q)&=\sum_{\boldsymbol{d},s}\Omega_{\boldsymbol{d},s}\boldsymbol{x}^{\boldsymbol{d}}q^{s/2}=\sum_{\boldsymbol{d},s}\Omega_{(d_{1},...,d_{m}),s}\left(\prod_{i}x_{i}^{d_{i}}\right)q^{s/2},    
\end{split}
\end{equation}
where $\Omega_{\boldsymbol{d},s}$ are motivic Donaldson-Thomas (DT) invariants ~\cite{KS1,KS2}. The~DT invariants have two geometric interpretations, either as the~intersection homology Betti numbers of the~moduli space of all semi-simple representations
of~$Q$ of dimension vector~$\boldsymbol{d}$, or as the~Chow-Betti numbers of the~moduli space of all simple representations of~$Q$ of dimension vector~$\boldsymbol{d}$; see~\cite{MR,FR}.~\cite{Efi} provides a~proof of integrality of DT invariants for
the~symmetric quivers.

\subsubsection{Knots-quivers correspondence for knot conormals\label{subsec:Knots-quivers-correspondence}}
In the~context of the~knots-quivers correspondence, we combine $P_{r}(K;a,q)$
into the~HOMFLY-PT generating series:
\begin{equation*}
P_{K}(y,a,q)=\sum_{r=0}^{\infty}\frac{y^{-r}}{(q)_r}P_{r}(K;a,q).
\end{equation*}
Using this expression we can encode the~Labastida-Mari\~{n}o-Ooguri-Vafa (LMOV) invariants~\cite{OV,LM,LMV} in the~following way:
\begin{equation}\label{P^K=Exp}
P_{K}(y,a,q)=\mathrm{Exp}\left(\frac{N(y,a,q)}{1-q}\right),
\qquad
N(y,a,q)=\sum_{r,i,j}N_{r,i,j}y^{-r}a^{i/2}q^{j/2}.
\end{equation}
According to the~LMOV conjecture~\cite{OV,LM,LMV},  $N_{r,i,j}$ are integer numbers counting BPS states in the~effective 3d $\mathcal N=2$ theories described in Section~\ref{sec:Brane constructions}.

The~knots-quivers correspondence for the~knot conormals~\cite{KRSS1,KRSS2}
is an~assignment of a~symmetric quiver~$Q$ (with adjacency matrix $C$), vector $\Bf{n}=(n_1,\ldots, n_m)$ with integer entries, and vectors $\Bf{a}=(a_1,\ldots, a_m)$, $\Bf{l}=(l_1,\ldots, l_m)$ with half-integer entries to a~given knot~$K$ in such a~way that
\begin{equation}\label{eq:KQ correspondence knot conormal}
P_{K}(y,a,q)
=\sum_{\boldsymbol{d}\geq0}(-q^{1/2})^{\boldsymbol{d}\cdot\boldsymbol{C}\cdot\boldsymbol{d}}\frac{y^{\Bf{n}\cdot \Bf{d}}a^{\Bf{a}\cdot \Bf{d}}q^{\Bf{l}\cdot \Bf{d}}}{(q)_{\boldsymbol{d}}}
=\left.P_{Q}(\boldsymbol{x},q)\right|_{x_{i}=y^{n_{i}}a^{a_{i}}q^{l_{i}}}.
\end{equation}
The~possibility of such assignment was proven for all 2-bridge knots in~\cite{Stosic:2017wno} and for all arborescent knots in~\cite{SW2004}. Some exotic cases with $n_i<-1$ (the~simplest examples are $\mathbf{9}_{42}$ and $\mathbf{10}_{132}$) require a~generalisation of the~correspondence, for more details see~\cite{EKL3}.

Equation \eqref{eq:KQ correspondence knot conormal} can be rewritten as
\begin{equation}\label{eq:KQDT}
N(y,a,q)=\left.\Omega(\boldsymbol{x},q)\right|_{x_{i}=y^{n_i} a^{a_{i}}q^{l_{i}}}\,,
\end{equation}
which ties the~knots-quivers correspondence with LMOV conjecture
using the~fact that DT invariants of symmetric quivers are integer.

\subsubsection{Knots-quivers correspondence for knot complements}\label{sec:KQ for knot complements}

The~knots-quivers correspondence can be generalized to knot complements~\cite{Kuch}. Then it is an~assignment of a~symmetric quiver~$Q$, an~integer $n_{i}$, and half-integers, $a_{i}$, $l_{i}$, $\ensuremath{i\in Q_{0}}$ to a~given knot complement $M_{K}=S^{3}\backslash K$ in such a~way that
\begin{equation}
F_{K}(x,a,q)
=\sum_{\boldsymbol{d}\geq0}(-q^{1/2})^{\boldsymbol{d}\cdot\boldsymbol{C}\cdot\boldsymbol{d}}\frac{x^{\Bf{n}\cdot \Bf{d}}a^{\Bf{a}\cdot \Bf{d}}q^{\Bf{l}\cdot \Bf{d}}}{(q)_{\boldsymbol{d}}}
=\left.P_{Q}(\boldsymbol{x},q)\right|_{x_{i}=x^{n_{i}}a^{a_{i}}q^{l_{i}}}.\label{eq:GM-Q correspondence}
\end{equation}

Due to the~limitations of available data, the~possibility of such assignment was proven so far only for $(2,2p+1)$ torus knots. Sections~\ref{sec:FK invariants and quivers} and \ref{sec:Quiver from R-matrices} provide evidence for \eqref{eq:GM-Q correspondence} for many new cases of knot complements (however, exotic cases like $\mathbf{9}_{42}$ and $\mathbf{10}_{132}$ are still out of range).

Using \eqref{eq:P_Q=Exp}, one can define analogues of LMOV invariants for knot complements~\cite{Kuch}:
\begin{equation*}
N(x,a,q)=\left.\Omega(\boldsymbol{x},q)\right|_{x_{i}=x^{n_{i}}a^{a_{i}}q^{l_{i}}},
\end{equation*}
which can be encoded in the~$F_K$ invariants in analogy to \eqref{P^K=Exp}:
\begin{equation}
F_{K}(x,a,q)  =\textrm{Exp}\left(\frac{N(x,a,q)}{1-q}\right).
\label{eq:F_K LMOV invariants}
\end{equation}

The~integrality of DT, LMOV, and analogous invariants has a~natural interpretation in terms of counting BPS states in effective 3d $\mathcal{N}=2$ theories. $N(x,a,q)$ and $\Omega(\boldsymbol{x},q)$ count BPS states in theories $T[M_K]$ and  $T[Q]$ respectively. Note that basing on Section~\ref{sec:Zhat and FK}, we can rewrite \eqref{eq:W from P} using $F_K$ invariant:
\begin{equation}
F_{K}(x,a,q)\underset{\hbar\rightarrow0}{\longrightarrow}\int\prod_{i}\frac{dz_{i}}{z_{i}}\exp\left(\frac{1}{\hbar}\widetilde{\mathcal{W}}_{T[M_{K}]}(z_{i},x,a)+\mathcal{O}(\hbar^{0})\right).\label{eq:Semiclassical limit F_K}
\end{equation}
Similarly, the~structure of $T[Q]$ is encoded
in the~semiclassical limit of the~motivic generating series~\cite{EKL1}:
\begin{equation}
\begin{split} & P_{Q}(\boldsymbol{x},q)\stackrel[\hbar\rightarrow0]{q^{d_{i}}=y_{i}}{\longrightarrow}\int\prod_{i}\frac{dy_{i}}{y_{i}}\exp\ensuremath{\left[\frac{1}{\hbar}\widetilde{\mathcal{W}}_{T[Q]}(\boldsymbol{x},\boldsymbol{y})+\mathcal{O}(\hbar^{0})\right]},\\
 & \widetilde{\mathcal{W}}_{T[Q]}(\boldsymbol{x},\boldsymbol{y})=\sum_{i}\textrm{Li}_{2}(y_{i})+\log\left(\ensuremath{(-1)^{C_{ii}}x_{i}}\right)\,\log y_{i}+\sum_{i,j}\frac{C_{ij}}{2}\log y_{i}\,\log y_{j}.
\end{split}
\label{eq:Semiclassical limit P_Q}
\end{equation}
Using the~dictionary \eqref{eq:Li2 and logs dictionary},
we can interpret the~elements of \eqref{eq:Semiclassical limit P_Q}
in the~following way:
\begin{itemize}
\item The~integral $\int\prod_{i}\frac{dy_{i}}{y_{i}}$ corresponds to having
the~gauge group $U(1)^{(1)}\times\dots\times U(1)^{(m)}$,
\item $\textrm{Li}_{2}(y_{i})$ represents the~chiral field with charge
$1$ under $U(1)^{(i)}$,
\item $\frac{C_{ij}}{2}\log y_{i}\,\log y_{j}$ corresponds to the~gauge
Chern-Simons couplings, $\kappa_{ij}^{\textrm{eff}}=C_{ij}$,
\item $\log\left(\ensuremath{(-1)^{C_{ii}}x_{i}}\right)\,\log y_{i}$ represents
the~Chern-Simons coupling between a~gauge symmetry and its dual topological
symmetry (the~Fayet-Iliopoulos coupling).
\end{itemize}
The~saddle point of the~twisted superpotential encodes the~moduli
space of vacua of $T[Q]$ and defines the~quiver $A$-polynomials
\cite{EKL1,EKL2,PSS,PS}:
\begin{equation} \label{eq:Classical Quiver A Polynomial}
\frac{\partial\widetilde{\mathcal{W}}_{T[Q]}(\boldsymbol{x},\boldsymbol{y})}{\partial\log y_{i}}=0\qquad\Leftrightarrow\qquad A_{i}(\boldsymbol{x},\boldsymbol{y})=1-y_{i}-x_{i}(-y_{i})^{C_{ii}}\prod_{j\neq i}y_{j}^{C_{ij}}=0.
\end{equation}
$A_{i}(\boldsymbol{x},\boldsymbol{y})$ is a~classical limit of the
quantum quiver $A$-polynomial, which annihilates the~motivic generating
series:
\begin{equation*}
\begin{split}\hat{A}_{i}(\hat{\boldsymbol{x}},\hat{\boldsymbol{y}},q)P_{Q}(\boldsymbol{x},q) & =0,\\
\hat{x}_{i}P_{Q}(x_{1},\ldots,x_{i},\ldots,x_{m},q) & =x_{i}P_{Q}(x_{1},\ldots,x_{i},\ldots,x_{m},q),\\
\hat{y}_{i}P_{Q}(x_{1},\ldots,x_{i},\ldots,x_{m},q) & =P_{Q}(x_{1},\ldots,qx_{i},\ldots,x_{m},q).
\end{split}
\end{equation*}
The~general formula for the~quantum quiver $A$-polynomial corresponding
to the~quiver with adjacency matrix $C$ is given by~\cite{EKL2}
\begin{equation} \label{eq:QQP}
\hat{A}_{i}(\hat{\boldsymbol{x}},\hat{\boldsymbol{y}},q)=1-\hat{y}_{i}-\hat{x}_{i}(-q^{1/2}\hat{y}_{i})^{C_{ii}}\prod_{j\neq i}\hat{y}_{j}^{C_{ij}},
\end{equation}
and we can see that
\begin{equation*}
A_{i}(\boldsymbol{x},\boldsymbol{y})=\underset{q\rightarrow1}{\lim}\hat{A}_{i}(\hat{\boldsymbol{x}},\hat{\boldsymbol{y}},q).
\end{equation*}

\subsubsection{Quiver equivalences}

Since the formulation of the~knots-quivers correspondence it has been clear that it is not a~bijection and more than one symmetric quiver can correspond to the~same knot~\cite{KRSS2} -- such quivers are called \emph{equivalent}. Later, the~study of geometric and physical interpretations~\cite{EKL2,EKL1} lead to the~formulation of quiver transformations that preserve the~motivic generating function~\eqref{eq:motivic generating series}: unlinking, linking, and removing (or adding) a~redundant pair of nodes. 

\begin{itemize}
\item The~unlinking of nodes $a,b\in Q_{0}$ is defined as a~transformation of $Q$ leading to a~new quiver $Q'$ such that there is a~new node $n$: $Q'_{0}=Q_{0}\cup n$, $x'_n=q^{-1} x_a x_b$ ($x'_i=x_i$ for all $i\in Q_{0}$), and the~number of arrows of the~new quiver is given by
\begin{align*}
\label{eq:unlinking arrows}
C'_{ab} & =C_{ab}-1, & C'_{nn} & =C_{aa}+2C_{ab}+C_{bb}-1,\\
C'_{in} & =C_{ai}+C_{bi}-\delta_{ai}-\delta_{bi}, & C'_{ij} & =C_{ij}\quad\textrm{for all other cases,} \nonumber
\end{align*}
where $\delta_{ij}$ is the~Kronecker delta. 
\item The~linking of nodes $a,b\in Q_{0}$ is defined as a~transformation of $Q$ leading to a~new quiver $Q'$ such that there is a~new node $n$: $Q'_{0}=Q_{0}\cup n$, $x'_n= x_a x_b$ ($x'_i=x_i$ for all $i\in Q_{0}$), and the~number of arrows of the~new quiver is given by
\begin{align*}
C'_{ab} & =C_{ab}+1, & C'_{nn} & =C_{aa}+2C_{ab}+C_{bb},\\
C'_{in} & =C_{ai}+C_{bi}, & C'_{ij} & =C_{ij}\quad\textrm{for all other cases.} \nonumber
\end{align*}
\item The~pair of nodes $a,b\in Q_{0}$ is redundant if $x_a=q x_b$, $C_{aa}=C_{ab}=C_{bb}-1$, and $C_{ai}=C_{bi}$ for all $i\in Q_{0}\backslash \{a,b\}$.
\end{itemize}

A general analysis of equivalent quivers has been conducted recently in~\cite{Jankowski:2021flt}. In this work the unlinking was used to study equivalences of symmetric quivers $Q$ and $Q'$ such that $Q'_0=Q_0$ and $x'_i=x_i$ for all $i\in Q_0$; in consequence it was proved that if $Q$ and $Q'$ are related by a~sequence of disjoint transpositions, each exchanging non-diagonal elements 
    \begin{equation*}
         C_{ab}\leftrightarrow C_{cd}, \qquad C_{ba}\leftrightarrow C_{dc},
    \end{equation*}
for some pairwise different $a,b,c,d,\in  Q_{0}$, such that
    \begin{equation*}
         \lambda_{a}\lambda_{b} = \lambda_{c}\lambda_{d}
    \end{equation*}
and 
    \begin{equation*}
        C_{ab} = C_{cd}-1,\qquad\quad
        C_{ai}+C_{bi}=C_{ci}+C_{di}-\delta_{ci}-\delta_{di},\quad \forall i\in Q_{0},
    \end{equation*}
    or
    \begin{equation*}
        C_{cd} = C_{ab}-1,\qquad\quad
        C_{ci}+C_{di}=C_{ai}+C_{bi}-\delta_{ai}-\delta_{bi},\quad \forall i\in Q_{0},
    \end{equation*}
then $Q$ and $Q'$ are equivalent. Moreover, these conditions were conjectured to be necessary for $Q$ and $Q'$ to correspond to the~same knot.

It turns out that transformations presented above can be successfully applied to quivers corresponding to knot complements, which will be useful in Section~\ref{sec:FK invariants and quivers}.

\section{\texorpdfstring{$F_K$}{FK} for various branches}\label{sec:FK different branches}

The~perturbative invariants of complex Chern-Simons theory are extensively studied in~\cite{DGLZ}. While the~method of that paper allowed the~computation of perturbative invariants up to any order, for many years they were not widely used. Recently, in~\cite{GM}, it was found that they can be nicely packaged into a~two-variable series associated to each knot complement. More precisely, the~authors of~\cite{GM} conjecture that the~\emph{abelian branch}\footnote{At finite $N$, abelian branch denotes the~solution $y^{(\text{ab})}=1$ of the~classical $A$-polynomial.} perturbative invariants can be resummed nicely into a~power series in $q$ and $x$ with integer coefficients:
\[
F_K^{(\text{ab})}(x,q) \stackrel{q=e^\hbar}{=\joinrel=\joinrel=} e^{\frac{1}{\hbar}(S_0^{(\text{ab})}(x)+S_1^{(\text{ab})}(x)\hbar + S_2^{(\text{ab})}(x)\hbar^2 + \cdots)}.
\]
The~relation between the~perturbative invariants $S_j^{(\text{ab})}(x)$ and the~two-variable series  $F_K^{(\text{ab})}(x,q)$ can be thought of as the~relation between the~Gromov-Witten invariants and the~Donaldson-Thomas (or BPS) invariants, and in this sense $F_K(x,q)$ is the~non-perturbative complex Chern-Simons partition function associated to the~abelian branch. See also~\cite{EGGKPS} for an~account of this story from the~point of view of topological strings. In order to simplify notation, we will drop the~superscript $^{(\text{ab})}$ and understand that $F_K(x,q)$ corresponds to the abelian branch.

It should be noted that the~method of~\cite{DGLZ} applies to all branches, so we have perturbative invariants associated to each branch $y^{(\alpha)}(x)$. This immediately brings the~following question. 
\begin{qn}
    Are there $F_K^{(\alpha)}(x,q)$ for other branches $y^{(\alpha)}(x)$? That is, can we resum $e^{\frac{1}{\hbar}(S_0^{(\alpha)}(x)+S_1^{(\alpha)}(x)\hbar + S_2^{(\alpha)}(x)\hbar^2 + \cdots)}$ into a~two-variable series with integrality?
\end{qn}
Intriguingly, through various examples we find that the~answer to this question seems to be positive. There are indeed $F_K$ for other branches that can serve as the~non-perturbative Chern-Simons partition functions for the~knot complements.\footnote{Physically, $F_K^{(\alpha)}(x,q)$ is a~``half-index'' of a~2d/3d combined system~\cite{GGP1} for the~3d theory $T[S^3 \setminus K]$ with a~2d $(0,2)$ boundary condition labeled by $\alpha$ and a~discrete flux, whose fugacity is $x$, cf.~\cite{GPV,GM}.}
To explain how to get the~series $F_K^{(\alpha)}(x,q)$ for various branches, in the~following subsection we take a~closer look at the~Newton polygon of the~$A$-polynomial. 

\subsection{The~edges and branches of the~\texorpdfstring{$A$}{A}-polynomial}

The~Newton polygon of the~$A$-polynomial contains a~wealth of information about the~knot. For instance, as discussed in~\cite{CCGLS}, the~slope of each edge of the~Newton polygon equals the~boundary slope of an~incompressible surface of the~knot complement. For our purposes, the important aspect is a~correspondence between the~edges of the~Newton polygon and the~branches of the~$A$-polynomial, which we explain below. 
\begin{figure}[ht]
    \centering
    \includegraphics[scale=0.5]{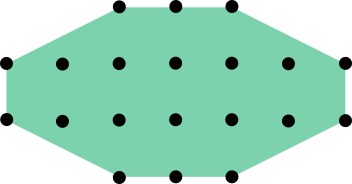}
    \caption{The~Newton polygon for $A_{\mathbf{4}_1}$}
    \label{fig:Apoly_Newton_polygon_figure_eight}
\end{figure}

Solving the~classical $A$-polynomial equation $A_K(x,y)=0$ for $y$, we get different branches $y^{(\alpha)}(x)$. We are interested in the~behaviour of $y^{(\alpha)}(x)$ near $x=0$ (or $x=\infty$, via Weyl symmetry). Consider the~Newton polygon of $A_K$. As an~example, the~Newton polygon for the~figure-eight knot is depicted in Figure \ref{fig:Apoly_Newton_polygon_figure_eight}. The~horizontal direction represents the~$x$-degree and the~vertical direction represents the~$y$-degree. When $x$ is close to $0$, the~dominant terms are the~vertices of the~Newton polygon that are left-most among the~vertices on the~same horizontal line. Let us call such vertices ``left-vertices''. The~equation $A_K(x,y^{(\alpha)}(x))=0$ requires that near $x = 0$ we should asymptotically have $y^{(\alpha)}(x)\sim x^{-\frac{n_x}{n_y}}$ for some slope~$\frac{n_y}{n_x}$ of an~edge spanned by two of the~left-vertices. Let us call such an~edge a~``left-edge''. Moreover, for any given slope of a~left-edge, we can construct a~classical solution $y^{(\alpha)}(x)$ using the~asymptotic dictated by the~slope. 
Therefore, we have just shown the~following proposition. 
\begin{prop}\label{prop:classical_asymptotics}
For any choice of branch $\alpha$, there is a~left-edge $e_\alpha$ of the~Newton polygon with slope $\frac{n_y}{n_x}$, such that
\begin{equation}\label{eq:classicalinitcond}
\lim_{x\rightarrow 0}y^{(\alpha)}(x)x^{\frac{n_x}{n_y}} = C    
\end{equation}
for some non-zero constant $C$. Moreover, the~map $E:\alpha \mapsto e_\alpha$ is a~surjective map onto the~set of left-edges. 
\end{prop}
While the~number of branches $y^{(\alpha)}$ is the~same as the~total $y$-degree (which is the~height of the~Newton polygon), the~number of left-edges is at most the~height of the~Newton polygon. Therefore, $E$ is not injective in general. 
\begin{prop}
For any left-edge $e$, the~number of pre-images of $E$ is $n_y$, the~$y$-height of the~edge $e$. 
\end{prop}
Take any left-edge $e$ with $x$-width $n_x$ and $y$-height $n_y$. Our convention is such that $n_y>0$ but $n_x$ can be negative. When $n_y=1$, it is easy to see that there is a~unique solution with the~initial condition \eqref{eq:classicalinitcond}. 
When $n_y>0$ but $n_x$ is coprime with $n_y$ so that the~edge is non-degenerate (the~endpoints are the~only vertices on this edge), then \eqref{eq:classicalinitcond} represents $n_y$ different initial conditions which differ by multiplication by an~$n_y$-th root of unity. Each of these initial conditions gives a~unique classical solution, and therefore there are $n_y$ number of branches associated to the~edge. 
When $n_x$ is not coprime with $n_y$, the~edge is degenerate (can be broken into smaller edges), and there are some genuine multiplicities associated to the~initial condition \eqref{eq:classicalinitcond}. We will see them in an~example of the~edge of slope $\infty$ for the~knot~$\mathbf{5}_2$. 

Our discussion so far can be summarized in the~following diagram: 
\[
\boxed{
\text{Left-edge }e\text{ of the~Newton polygon}\leftrightarrow n_y\text{ number of branches}
}
\]

\begin{rmk}
If we look at the~behaviour of $y^{(\alpha)}(x)$ near $x=\infty$, their asymptotics are determined by the~slopes of the~right-edge of the~Newton polygon. Due to Weyl symmetry, the~Newton polygon is symmetric under half-rotation, so we basically get the~same set of information. 
\end{rmk}
\begin{rmk}
We can similarly solve $A_K(x,y)=0$ for $x$. With the~role of $x$ and $y$ switched, everything we described in this section holds. 
\end{rmk}
\begin{rmk}
In terms of tropical geometry, what we delineated above can be described using Gr\"obner fans (See e.g.~\cite{Stu}). A choice of an~edge of the~$A$-polynomial corresponds to a~choice of an~equivalence class of weight vectors that induce the~same initial ideal. When the~ideal is neither principal nor homogeneous, there may be no polytope whose normal fan is the~Gr\"obner fan. Therefore, in such cases it is better to think in terms of the~ideal and its Gr\"obner fans. 
\end{rmk}
\begin{rmk}
When the~$y$-height $n_y$ of $e_\alpha$ is $1$, $y^{(\alpha)}(x)$ is a~power series in $x$. In general, however, when $n_y > 1$, $y^{(\alpha)}(x)$ is a~Puiseux series in $x$. 
\end{rmk}

\subsection{\texorpdfstring{$F_K$}{FK} from the~edges}
In the~previous section, we have seen the~correspondence between the~set of classical solutions and the~edges of the~Newton polygon. In this section, we will study the~quantum version of this correspondence. We are interested in solving the~$q$-difference equation
\[\hat{A}_K F_K = 0,\]
where $F_K$ is expanded near $x=0$.\footnote{By the~two remarks in the~previous subsection, we can do the~same for the~expansion near $x=\infty$ or with $y$ instead of $x$. } For each left-edge, there are some natural initial conditions for the~recursion that we can put, which in the~semiclassical limit become the~classical initial conditions that we studied in the~previous section. 
\begin{conj}\label{conj:quantum_initial_condition}
For each left-edge $e$ with slope $\frac{n_y}{n_x}$, there is a~solution to the~$q$-difference equation $\hat{A}_K F_K(x,q)$ of the~form
\[F_K^{(\alpha)}(x,q) = e^{\frac{1}{\hbar}(-\frac{1}{2}\frac{n_x}{n_y}(\log x)^2 + \log C \log x)}(1+\sum_{j\geq 1} f_j(q) x^{\frac{j}{d}}),\]
where $d = \frac{n_y}{m}$ in case $e$ is broken into $m$ non-degenerate edges, $C$ is a~monomial in $q$ determined by the~coefficients of the~vertices of $e$, and $f_j(q)$ are rational functions in $q$, which can be expanded into $q$-series with integer coefficients. 
\end{conj}
\begin{rmk}
When $n_y=1$, this conjecture is a~theorem. This is because we can recursively solve for $F_K^{(\alpha)}$ in a~unique way, analogously to the~way it was done for the~abelian branch in~\cite{EGGKPS}. 
\end{rmk}
\begin{rmk}
In the~$a$-deformed setting, the~exponential prefactor will be of the~form 
\[
\exp(\frac{p(\log x,\log a)}{\hbar}),
\]
where $p$ is a~polynomial of degree at most 2. 
\end{rmk}
Along the~arguments of the~previous section, we see that 
\begin{itemize}
    \item if $n_y = 1$, then there is a unique such solution, 
    \item if the~edge is non-degenerate but $n_y > 1$, then all the~$n_y$ solutions are uniquely determined,
    \item and if the~edge is degenerate, then there are multiple solutions (the~number of solutions is the~same as the~number of branches associated to the~edge).
\end{itemize}
In this way, the~solutions to the~set of initial conditions determined by the~left-edges span the~whole $\deg_y A_K$-dimensional space of wave functions. We claim that these solutions are exactly the~$F_K^{(\alpha)}$ for various branches $\alpha$ we mentioned in the~beginning of this section (possibly up to an~overall factor that is independent of $x$). We can formulate this in the~form of the~following conjecture. 
\begin{conj}\label{conj:FK_other_branches}
Given a~knot $K$, for every branch $y^{(\alpha)}(x)$ of the~$A$-polynomial, there is a~function $F_K^{(\alpha)}(x,q)$ that is the~non-perturbative partition function of the~complex Chern-Simons theory in the~following sense:
\begin{enumerate}
    \item $\hat{A}_K F_K^{(\alpha)}=0$ with the~initial conditions as in Conjecture \ref{conj:quantum_initial_condition}.
    \item It is associated to the~branch $y^{(\alpha)}$ in the~sense that
    \[
    \lim_{q\rightarrow 1}\frac{F_K^{(\alpha)}(qx,q)}{F_K^{(\alpha)}(x,q)} = y^{(\alpha)}(x).
    \]
    \item It agrees with the~perturbative invariant of~\cite{DGLZ} if we set $q=e^{\hbar}$. 
\end{enumerate} 
\end{conj}

\begin{rmk}
There is always an~abelian branch associated to the~vertical edge whose corresponding initial condition gives the~usual $F_K^{(\text{ab})}=F_K$, as studied in~\cite{GM,EGGKPS}. 
\end{rmk}
\begin{rmk}
It is straightforward to generalize everything we discussed in this section to $\mathfrak{sl}_N$ and the~$a$-deformed setup. As we will see later in Section \ref{sec:Bpoly}, we can even consider the~branches of $b$, a~variable that is the~conjugate of $a$. In that context, we consider solutions to $q$-difference equations with respect to the~variable $a$. We will see that the~branches of $b$ are canonically in one-to-one correspondence with the~branches of $y$. 
\end{rmk}

\subsection{Examples}
\subsubsection{Trefoil}
The~quantum $A$-polynomial for the~right-handed trefoil knot is given by
\[
\hat{A}_{\mathbf{3}_1^r}(\hat{x},\hat{y},a,q) = a_0^{\mathbf{3}_1^r}(\hat{x},a,q) + a_1^{\mathbf{3}_1^r}(\hat{x},a,q)\hat{y} + a_2^{\mathbf{3}_1^r}(\hat{x},a,q)\hat{y}^2,
\]
with
\begin{align*}
    a_0^{\mathbf{3}_1^r}(x,a,q) &= -aq(1-x)(1-qax^2),\\
    a_1^{\mathbf{3}_1^r}(x,a,q) &= (1-a x^2)(a^2x^2-q^3ax^2 -qax(1+x-ax(1-x))+q^2(1+a^2x^4)),\\
    a_2^{\mathbf{3}_1^r}(x,a,q) &= qa^2x^3(1-ax)(q-a x^2).
\end{align*}
In the~classical limit, after modding out by the~factor $(1-a x^2)$, the~Newton polygon is illustrated in Figure \ref{fig:NewtPolyA31}. 
\begin{figure}[ht]
    \centering
    \includegraphics[scale=0.5]{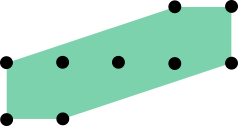}
    \caption{The~Newton polygon of $A_{\mathbf{3}_1^r}$}
    \label{fig:NewtPolyA31}
\end{figure}
While the~abelian branch $F_{\mathbf{3}_1^r}$ is discussed in detail in~\cite{EGGKPS}, let us briefly review how to compute it. 
First, observe that near $x=0$
\[
a_0^{\mathbf{3}_1^r}(x,a,q) = -qa + O(x^1),\qquad
(q^{-1}a) a_1^{\mathbf{3}_1^r}(x,a,q) = qa + O(x^1),\qquad
a_2^{\mathbf{3}_1^r}(x,a,q) = O(x^3).
\]
The~first two non-vanishing $O(x^0)$ terms correspond exactly to the~vertical left-edge of the~Newton polygon. Note that we multiplied the~coefficients by powers of $q^{-1}a$ to make the~sum of the~two $O(x^0)$ terms vanish. This means our initial condition for solving the~recursion is such that
\[
F_{\mathbf{3}_1^r}(x,a,q) = e^{\frac{\log x \log a}{\log q}}x^{-1}\qty(1 + O(x^1) ).
\]
We can recursively solve the~subsequent terms and get
\[
F_{\mathbf{3}_1^r}(x,a,q) = e^{\frac{\log x \log a}{\log q}}x^{-1}\qty(1 + \frac{q-a}{1-q}x +\frac{q^2+(-2q-q^2+q^3)a+(1+q-q^2)a^2}{(1-q)(1-q^2)}x^2 +\cdots),
\]
up to an overall factor independent of $x$. 

For the~non-abelian branch of slope $\frac{1}{3}$, we need to consider the~coefficients of the~quantum $A$-polynomial near $y^{-1}x^3 = 0$. After multiplying appropriate factors, we can make the~sum of the~terms on this left-edge vanish:
\begin{align*}
a_0^{\mathbf{3}_1^r}(x,a,q) &= O(x^0),\\
q^{-\frac{3}{2}1^2}(-q^{\frac{9}{2}}a^{-2}x^{-3}) a_1^{\mathbf{3}_1^r}(x,a,q) &= -q^5a^{-2}x^{-3} + O(x^{-2}),\\
q^{-\frac{3}{2}2^2}(-q^{\frac{9}{2}}a^{-2}x^{-3})^2 a_2^{\mathbf{3}_1^r}(x,a,q) &= q^5a^{-2}x^{-3} + O(x^{-2}).
\end{align*}
The~extra factors we had to multiply by mean that the~initial condition for solving the~recursion is such that
\[
F_{\mathbf{3}_1^r}^{(\frac{1}{3})}(x,a,q) = e^{\frac{-\frac{3}{2}(\log x)^2+\log x \log (-a^{-2})}{\log q}}x^{\frac{9}{2}}\qty(1 + O(x^1)).
\]
We can recursively solve the~subsequent terms and get
\[
F_{\mathbf{3}_1^r}^{(\frac{1}{3})}(x,a,q) = e^{\frac{-\frac{3}{2}(\log x)^2+\log x \log (-a^{-2})}{\log q}}x^{\frac{9}{2}}\qty(1 + \frac{a}{q}x + \frac{a(q-a)}{q(1-q)}x^2 + \cdots),
\]
up to an overall factor independent of $x$. 

\subsubsection{Figure-eight}
The~Newton polygon for the~$A$-polynomial of the~figure-eight knot is shown in Figure \ref{fig:Apoly_Newton_polygon_figure_eight}. Since all the~left-edges are non-degenerate of height 1, everything can be solved term by term in a~unique way, just like for the~trefoil knot. Following the~same procedure, we get the~following series for the~abelian branch~\cite{EGGKPS}: 
\begin{align*}
F_{\mathbf{4}_1}(x,a,q) &= e^{\frac{\log x\log a}{\log q}}x^{-1}\left(1 + \frac{3(q-a)}{1-q}x\right.\\ 
&\qquad\left. + \frac{(q+6q^2+2q^3)-(1+8q+8q^2+q^3)a+(2+6q+q^2)a^2}{(1-q)(1-q^2)}x^2 +\cdots\right),
\end{align*}
up to an overall factor independent of $x$. 

For the~non-abelian branch of slope $-\frac{1}{2}$, we get
\begin{align*}
F_{\mathbf{4}_1}^{(-\frac{1}{2})}(x,a,q) &= e^{\frac{(\log x)^2+\log x \log a}{\log q}}x^{-1}\left(1 + \frac{q-2q^2}{1-q} x \right.\\
&\qquad\left.+\frac{(q^2-2q^3-2q^4+3q^5+q^6)-q(1-q)(1-q^2)a}{(1-q)(1-q^2)}x^2 +\cdots\right),
\end{align*}
up to an overall factor independent of $x$. 

The~other non-abelian branch (with slope $\frac{1}{2}$) is conjugate to this one, and the~corresponding $F_K$ can be obtained easily from this one by inverting $q$, $a$ and $x$ and then using the~Weyl symmetry of~\cite{EGGKPS}.

\subsubsection{\texorpdfstring{$\mathbf{5}_2$}{52} knot}\label{subsubsec:52knot}
The~Newton polygon of $A_{\mathbf{5}_2}$ is given in Figure \ref{fig:NewtPolyA52}. We see that there are 4 branches in total. Two of them have slope $\infty$, one has slope $-\frac{1}{2}$, and the~last one has slope $-\frac{1}{5}$. 
\begin{figure}[ht]
    \centering
    \includegraphics[scale=0.5]{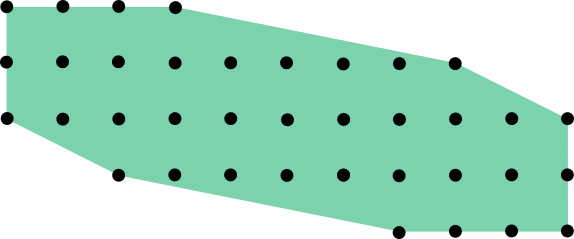}
    \caption{The~Newton polygon of $A_{\mathbf{5}_2}$}
    \label{fig:NewtPolyA52}
\end{figure}
The~branches of slope $-\frac{1}{2}$ and $-\frac{1}{5}$ are non-degenerate with height $1$, and it is straightforward to compute the~corresponding $F_K$ invariants by solving the~quantum $A$-polynomial recursion. We will focus on the~branches of slope $\infty$. 

One of the~branches of slope $\infty$ is the~abelian one, the~other is non-abelian. Let us call them $(\text{ab})$ and $(\infty)$ respectively. Here we explain how to compute $F_{\mathbf{5}_2}^{(\alpha)}$ for the~two branches, with $\mathfrak{sl}_N$ gauge algebra, using \emph{inverted Habiro series} of~\cite{Park3}. In the~next section we will provide an~alternative way to obtain $F_{\mathbf{5}_2}^{(\text{ab})}(x,a,q)$ from the~quiver. 

About 20 years ago, Habiro showed~\cite{Hab, Hab2} that the~coloured Jones polynomials can be decomposed in the~following way. 
\[
J_n(K;q) = \sum_{m\geq 0} a_m(K;q)\prod_{1\leq j\leq m}(x+x^{-1} - q^j-q^{-j}) \bigg\vert_{x = q^n}.
\]
Here $J_n(K;q)$ denotes the~coloured Jones polynomial of $K$ coloured by the~$n$-dimensional irreducible representation of $\mathfrak{sl}_2$, and $\{a_m(K;q)\}_{m\geq 0}$ is a~sequence of Laurent polynomials in $q$ with integer coefficients. 
The~analogue of Habiro series for coloured HOMFLY-PT polynomials $P_r(K;a,q)$ coloured by symmetric representations was studied in~\cite{IMMM, MMM}, and it is given by:
\[
P_r(K;a,q) = \sum_{m\geq 0} a_m(K;a,q) \qbin{r}{m} \prod_{1\leq j\leq m}(a^{\frac{1}{2}}q^{\frac{r+j-1}{2}}-a^{-\frac{1}{2}}q^{-\frac{r+j-1}{2}})(a^{\frac{1}{2}}q^{\frac{j-2}{2}}-a^{-\frac{1}{2}}q^{-\frac{j-2}{2}}).
\]
The~sequences of cyclotomic coefficients $\{a_m(K;q)\}$ and $\{a_m(K;a,q)\}$ are known to be $q$-holonomic~\cite{GS, MM21}, and the~corresponding $q$-difference equation is known as the~\emph{quantum $C$-polynomial}. 

The~idea of inverted Habiro series~\cite{Park3} is to extend the~sequence $\{a_m(K;a,q)\}_{m\geq 0}$ to negative values of $m$. This can be often done by solving the~quantum $C$-polynomial recursion in the~negative direction. In particular, in~\cite{Park3} it was illustrated that $a_m(\mathbf{5}_2;q)$ for $m<0$ can be computed by recursively solving the~$q$-difference equations and that
\[
F_{\mathbf{5}_2}(x,q)=F_{\mathbf{5}_2}^{\mathfrak{sl}_2,(\text{ab})}(x,q) = -\sum_{m<0}\frac{a_m(\mathbf{5}_2;q)}{\prod_{0\leq j\leq -m-1}(x+x^{-1}-q^j-q^{-j})}.
\]
Using the~$a$-deformed quantum $C$-polynomial for $\mathbf{5}_2$ knot given in~\cite{MM21}, we can do the~same for $\mathfrak{sl}_N$. For instance, the~first few coefficients in the~case of $\mathfrak{sl}_3$ are given by
\begin{align*}
a_{-1}^{\mathfrak{sl}_3,(\text{ab})}(\mathbf{5}_2;q) &= -q^{-2}+1-q^3+q^7-q^{12}+q^{18}-\cdots,\\
a_{-2}^{\mathfrak{sl}_3,(\text{ab})}(\mathbf{5}_2;q) &= q^{-3}-q^{-2}-q^{-1}+1+q+q^2-q^3-q^4-q^5-q^6+q^7+\cdots,\\
a_{-3}^{\mathfrak{sl}_3,(\text{ab})}(\mathbf{5}_2;q) &= -1+q^2+q^3+q^4-q^5-q^6-2q^7-q^8+\cdots,
\end{align*}
and so on. Combined with the~following formula
\[
F_{K}^{\mathfrak{sl}_N,(\text{ab})}(x,q) = -\sum_{n\geq 0}a_{-n-1}^{\mathfrak{sl}_N,(\text{ab})}(K;q)\frac{\frac{[-n][-n+1]\cdots[-n+N-3]}{[n]!}}{\prod_{0\leq j\leq n}(x^{\frac{1}{2}}q^{\frac{j}{2}}-x^{-\frac{1}{2}}q^{-\frac{j}{2}})(x^{\frac{1}{2}}q^{\frac{(N-2)-j}{2}}-x^{-\frac{1}{2}}q^{\frac{j-(N-2)}{2}})},
\]
this allows us to compute $F_{\mathbf{5}_2}^{\mathfrak{sl}_N,(\text{ab})}(x,q)$ for any $N$. 

Computation of $F_{\mathbf{5}_2}^{\mathfrak{sl}_N,(\infty)}(x,q)$ is similar. For this, we need to compute the~inverted Habiro coefficients for the~non-abelian branch $(\infty)$. This can be done by setting the~initial condition to be 
\[
a_m^{\mathfrak{sl}_N,(\infty)}(\mathbf{5}_2;q) = 0 \text{ for all } m \geq 0.
\]
We find that, up to an~overall factor,
\[
\sum_{n\geq 0} a_{-n-1}^{\mathfrak{sl}_N,(\infty)}(\mathbf{5}_2;q)E^{n+1} = E + \frac{1+q^3 a^{-1}}{1-q} E^2 + \frac{1 + q^4(1+q)a^{-1} + q^8 a^{-2}}{(1-q)(1-q^2)} E^3 + \cdots \bigg\vert_{a = q^N},
\]
where $E$ is a~formal variable. This implies, for instance, that up to an~overall factor,
\begin{align*}
F_{\mathbf{5}_2}^{\mathfrak{sl}_2,(\infty)}(x,q) &= x + \frac{3-q}{1-q}x^2 + \frac{q^{-1}+7-q-4q^2+q^3}{(1-q)(1-q^2)} x^3 +\cdots,\\
F_{\mathbf{5}_2}^{\mathfrak{sl}_3,(\infty)}(x,q) &= \frac{2}{1-q}x^2 + \frac{5+q-2q^2}{(1-q)^2}x^3 +\frac{q^{-1}+10+7q-3q^2-8q^3-q^4+2q^5}{(1-q)^2(1-q^2)}x^4 +\cdots\\
F_{\mathbf{5}_2}^{\mathfrak{sl}_4,(\infty)}(x,q) &= \frac{3+q}{(1-q)(1-q^2)}x^3 + \frac{7+3q+2q^2-3q^3-q^4}{(1-q)^2(1-q^2)}x^4 + \cdots.
\end{align*}
By computing the~expectation value of the~$\hat{y}$-operator (i.e. $\lim_{q\rightarrow 1}\frac{F_K(q x,q)}{F_K(x,q)}$), it is easy to verify that these solutions are indeed associated to the~non-abelian branch $(\infty)$. 

\begin{rmk}
If we set 
\[
f_0(b,q) = \sum_{N\geq 2}f_0^{\mathfrak{sl}_N}(q) (q)_{N-2}b^{-N},
\]
\[
f_1(b,q) = \sum_{N\geq 2}f_1^{\mathfrak{sl}_N}(q) (1-q)(q)_{N-2}b^{-N},
\]
where $b$ is a~formal variable, and $f_0^{\mathfrak{sl}_N}$ and $f_1^{\mathfrak{sl}_N}$ denote the~first and second coefficients of $F_{\mathbf{5}_2}^{\mathfrak{sl}_N,(\infty)}$ so that $f_0^{\mathfrak{sl}_2}(q) = 1$, $f_0^{\mathfrak{sl}_3}(q) = \frac{2}{1-q}$, and so on, then we find experimentally that they satisfy the~following recurrence relations:
\[
(\hat{b}^2 - 2\hat{b} + 1 - q^{-1}\hat{a}) f_0(b,q) = 1,
\]
\[
(1-q)(\hat{b}-1)f_1(b,q) +(2 -2\hat{b} +q^{-1}(\hat{b}-2)\hat{a})f_0(b,q) = b^{-1}.
\]
Here, $\hat{b}$ and $\hat{a}$ are linear operators characterized by
\[
\hat{b}b^{-N} = b^{-N+1}, \quad \hat{a}b^{-N} = q^N b^{-N},\quad \hat{a}\hat{b} = q^{-1}\hat{b}\hat{a}.
\]
Since the~first two coefficients of $F_{\mathbf{5}_2}^{\mathfrak{sl}_N}$ completely determine the~whole series via recursion, the~above recurrence relations allow one to compute $F_{\mathbf{5}_2}^{\mathfrak{sl}_N,(\infty)}$ up to any desired order. 
\end{rmk}

\section{\texorpdfstring{$F_K$}{FK} invariants and quivers}\label{sec:FK invariants and quivers}

In the~previous section we have constructed $F_K$ invariants for various branches using the~quantum $A$-polynomials. In this section we will focus on the~relation between these newly constructed $F_K$ invariants and quivers.

\subsection{From knot conormal quivers to knot complement quivers}\label{sec:41}

We start by studying how we can obtain $F_K$ invariants for some branches from the~original quivers of~\cite{KRSS1,KRSS2} corresponding to knot conormals. Then we use it to construct quivers corresponding to some knot complements, generalizing~\cite{Kuch}. Finally, we show that a~slight modification of this construction leads to simpler quivers corresponding to the~same $F_K$ invariant. 

The~computation of $F_K$ for abelian branches of left-handed $(2,2p+1)$ torus knots in~\cite{EGGKPS} relies on the~fact that there exists a~simple Fourier transform between coloured HOMFLY-PT polynomials $P_r(a,q)$ and $F_K(x,a,q)$, which is essentially a~substitution $x=q^r$. However, this does not work in general and $F_K$
cannot be obtained directly from the~knowledge of $P_r$. 

One way to deal with this problem is to consider a~knot $K$ with a~suitable framing~$f$, so that after replacing $q^r=x$, the~corresponding coloured HOMFLY-PT polynomials $P_r(a,q)$ become a~power series in $x$. In order to find the~correct framing and to compute the~corresponding power series, the~conormal quivers become crucially important. First of all, the~framing will be the~absolute value of the~minimal entry of the~conormal quiver matrix. Moreover, the~power series in $x$ can be quickly determined and will be given in the~quiver form. Finally, the~obtained power series will be equal to $F_K$ invariant for the~branch corresponding to the~ smallest slope of the~knot $K$. Therefore, for each knot, $F_K$ for one branch can be obtained by this procedure. In particular, this branch will be abelian only when the~corresponding framing, i.e. the~smallest entry of the~conormal quiver matrix, is equal to zero. In all other cases (like figure-eight knot, right-handed trefoil, etc.), the~procedure that we outline below will produce $F_K$ corresponding to a~certain non-abelian branch.

Now let us pass to details of the~connection between conormal quivers, $F_K$ invariants and their quiver forms. Let $K$ be a~knot with a~corresponding quiver $Q$ with $m$ vertices and the~adjacency matrix $C_{ij}$. Also let $\Bf a$ and $\Bf l$ be the~vectors corresponding to the~linear terms and essentially to uncoloured polynomial. Then
\begin{equation*}
    P_r(K;a,q)=\sum_{d_1+\ldots+d_k=r}(-1)^{\sum C_{ii} d_i} a^{\sum a_i d_i}q^{\sum l_i d_i}q^{\frac{1}{2}\sum_{i,j=1}^m C_{ij}d_id_j} \frac{(q)_r}{\prod_{i=1}^k(q)_{d_i}}.
\end{equation*}
Now suppose that 
$$-C_{\textrm{min}}\le C_{ij} \le C_\textrm{max},  \quad i,j=1,\ldots,m,$$
where $C_{\textrm{min}},C_\textrm{max}\ge 0$ (must be since $C_{kk}=0$, for some $k$), and permute rows and columns of~$C$ such that $C_{11}=C_{\textrm{min}}$ and $C_{mm}=C_\textrm{max}$. Note that in the~cases where conormal quivers have been computed, the~largest and smallest entries are on the~diagonal.

Then for knots which allow for a~simple Fourier transform  between $P_r$ and $F_K$ (that boils down to substitution $x=q^r$), we can obtain $F_K$ invariant for $C_{\textrm{min}}$-framed knot $K$ by multiplying each $P_r$ by $q^{\frac{1}{2}C_\textrm{min}(r^2-r)}$:
\begin{eqnarray}
F_{K^{f=C_{\textrm{min}}}}(x,a,q)&=&(-1)^{rC_\textrm{min}} a^{r a_1}q^{r q_1}\sum_{d_2,\cdots,d_m}(-1)^{\sum_{i\ge 2} (C_{ii}+C_\textrm{min}) d_i}a^{\sum_{i\ge 2} (a_i-a_1) d_i} \nonumber
\\
&&\times\,\, q^{\sum_{i\ge 2} (l_i-l_1)d_i} x^{\sum_{i\ge 2}(C_{1i}+C_\textrm{min})d_i}q^{\frac{1}{2}\sum_{i,j\ge 2} (C_{ij}-C_{i1}-C_{1j}+C_{11})d_id_j} 
\label{eq:F_K from knot conormal Q}
\\
&&\times\,\, \frac{(x;q^{-1})_{d_2+\cdots+d_k}}{\prod_{i=2}^k(q)_{d_i}}. \nonumber
\end{eqnarray}
For the~mirror image of $K$, the~quiver and the~change of variables are given by
\[
    C_{m(K)}=-C_K+I_{m\times m}-J_{m\times m},\quad \Bf{a}_{m(K)}=-\Bf{a}_K,\quad \Bf{l}_{m(K)}=-\Bf{l}_K,
\]
where  $J_{m\times m}$ is the~$m\times m$ matrix with all entries equal to $1$. Then the~diagonal entries of $C_{m(K)}$  are bigger than $-C_\textrm{max}$ and smaller than $C_\textrm{min}$, with
$(C_{m(K)})_{11}=C_\textrm{min}$ and $(C_{m(K)})_{mm}=-C_\textrm{max}$. 
In consequence, we can apply the~above procedure for the~$f=C_\textrm{max}$ framing of knot~$m(K)$:
\begin{eqnarray}
F_{(m(K))^{f=C_\textrm{max}}}(x,a,q)&=&\sum_{d_1,\cdots,d_{m-1}}(-1)^{\sum_{i}(C_{ii}+C_\textrm{max}) d_i}a^{\sum_{i} (a_m-a_i) d_i}\nonumber
\\
&&\times\,\, q^{\sum_{i} (l_m-l_i)d_i} x^{\sum_{i}(C_\textrm{max}-C_{im}-1)d_i}q^{-\frac{1}{2}\sum_{i,j} C_{ij}d_id_j}
\label{eq:F_mirrorK from knot conormal Q}
\\
&&\times\,\, q^{\sum_{i} C_{im}d_i\sum_id_i}q^{-\frac{1}{2}C_\textrm{max}(\sum_{i}d_i)^2}q^{-\sum_{i<j}d_id_j}\frac{(x;q^{-1})_{d_1+\cdots+d_{m-1}}}{\prod_{i=1}^{m-1}(q)_{d_i}}.\nonumber
\end{eqnarray}

Equations (\ref{eq:F_K from knot conormal Q}) and (\ref{eq:F_mirrorK from knot conormal Q}) are very close to the~quiver form. The~simplest way of reaching it is an~application of Lemma 4.5 from~\cite{KRSS2}:
\begin{align}
   \frac{(x;q^{-1})_{d_1+\ldots+d_n}}{(q)_{d_{1}}\cdots(q)_{d_{n}}}&=
    \frac{(x \,q^{1-\sum_i d_i};q)_{d_1+\ldots+d_n}}{(q)_{d_{1}}\cdots(q)_{d_{n}}} \nonumber\\
    &=\quad \sum\limits_{\alpha_{1}+\beta_{1}=d_{1}} \cdots \sum\limits_{\alpha_{n}+\beta_{n}=d_{n}} (-q^{1/2})^{\beta_{1}^2+\ldots+\beta_{n}^2+2\sum_{i=1}^{n-1} \beta_{i+1}(d_{1}+\ldots + d_{i})} \label{eq:qpoch-sum-general}  \\
    &\quad\quad\quad\quad\quad \times \frac{\big(x \, q^{1/2-\sum_i \alpha_i-\sum_i \beta_i}\big)^{\beta_{1}+\cdots+\beta_n}}{(q)_{\alpha_{1}}(q)_{\beta_{1}}\cdots (q)_{\alpha_{n}}(q)_{\beta_{n}}}.
    \nonumber
\end{align}
In the~next section we will show in examples that this expansion leads to expressions for knot complement quivers found in~\cite{Kuch}.
On the~other hand, we can use the~following formula:
\begin{equation}\label{eq:new expansion}
(x;q^{-1})_d=(x q^{1-d};q)_d=\frac{(x q^{1-d};q)_{\infty}}{(x q;q)_{\infty}}=\sum_{i,j} (-1)^i x^{i+j} q^{i+j} q^{-di} q^{\frac{1}{2}(i^2-i)}\frac{1}{(q)_i (q)_j}.
\end{equation}
Since it effectively adds two nodes instead of doubling them, resulting quivers are expected to be simpler. In the~next section we will see in examples that it is indeed the~case. We will also see that the~transition between two kinds of quivers can be interpreted in terms of linking and removing a~redundant pair of nodes, defined in Section~\ref{sec:KQ correspondence} (for details see~\cite{EKL2}).

In addition to the~reasoning presented above, quivers corresponding to knot complements can be often obtained directly in the~simpler form by matching quiver adjacency matrix and the~change of variables against order by order expansion of $F_K$ invariants for various branches (which can be obtained from $A$-polynomials, as we saw in Section~\ref{sec:FK different branches}). Many results presented in the~next section were derived in this way.

\subsection{Examples}\label{sec:Examples}

In this section we illustrate the~considerations presented above on the~example of the~figure-eight and trefoil knots, taking into account $F_K$ and $F_{m(K)}$ for various branches. Moreover, the~application of methods from Section \ref{sec:41} to the~results of~\cite{Kuch} enables us to conjecture the~simple quiver form for general $(2,2p+1)$ torus. Analogous results for all knots with 5 or 6 crossings, as well as (3,4) torus knot, are presented in Appendix \ref{sec:Quivers appendix}.

\subsubsection{Figure-eight}

First we shall obtain the~quiver for the~slope $-\frac{1}{2}$ non-abelian branch of the~figure-eight knot, following the~discussion from Section \ref{sec:41}. We start from the~expression for the~coloured HOMFLY-PT polynomial for $\mathbf{4}_1$ obtained in~\cite{KRSS2}:
\begin{equation}P_r(\mathbf{4}_1;a,q)=\sum_{\tilde{d}_1+\cdots+\tilde{d}_5=r}(-1)^{\tilde{d}_3+\tilde{d}_4} a^{\tilde{d}_2-\tilde{d}_5}q^{-\tilde{d}_2-\frac{1}{2}\tilde{d}_3+\frac{1}{2}\tilde{d}_4+\tilde{d}_5}q^{\frac{1}{2}\sum_{i,j=1}^5 C_{ij}\tilde{d}_i\tilde{d}_j}
\frac{{(q;q)_r}}{\prod_{i=1}^5(q;q)_{\tilde{d}_i}}\label{1},
\end{equation}
where 
\[
C=\left(\begin{array}{ccccc}0&0&-1&0&-1\\
0&2&0&1&-1\\
-1&0&-1&0&-2\\
0&1&0&1&-1\\
-1&-1&-2&-1&-2\end{array}\right).
\]
We need to add framing $f=2$, i.e. to multiply by $q^{r(r-1)}$. Since 
\[
C+2\left(\begin{array}{ccc}1&\cdots&1\\
\vdots&&\vdots\\
1&\cdots&1\end{array}\right)
\]
has $0$ at the~bottom right corner and all entries in the~last row are non-negative, we shall replace $\tilde{d}_5=r-\tilde{d}_1-\tilde{d}_2-\tilde{d}_3-\tilde{d}_4$ in (\ref{1}), which leads to
\begin{align*}
    q^{r(r-1)}P_r(\mathbf{4}_1;a,q)&=\sum_{\tilde{d}_1+\cdots+\tilde{d}_4\le r}  (-1)^{\tilde{d}_3+\tilde{d}_4} a^{-r}a^{\tilde{d}_1+2\tilde{d}_2+\tilde{d}_3+\tilde{d}_4}q^{-\tilde{d}_1-2\tilde{d}_2-\frac{3}{2}\tilde{d}_3-\frac{1}{2}\tilde{d}_4} \frac{(q^{r};q^{-1})_{\tilde{d}_1+\tilde{d}_2+\tilde{d}_3+\tilde{d}_4}}{\prod_{i=1}^4(q;q)_{\tilde{d}_i}}\nonumber
    \\
    &\qquad\qquad\times  q^{\frac{1}{2}{\sum_{i,j=1}^4C_{ij}\tilde{d}_i\tilde{d}_j}-(\tilde{d}_1+\tilde{d}_2+2\tilde{d}_3+\tilde{d}_4)(r-\tilde{d}_1-\tilde{d}_2-\tilde{d}_3-\tilde{d}_4)-(r-\tilde{d}_1-\tilde{d}_2-\tilde{d}_3-\tilde{d}_4)^2+r^2 }
    \\
    &=a^{-r}\sum_{\tilde{d}_1,\tilde{d}_2,\tilde{d}_3,\tilde{d}_4\le r} (-1)^{\tilde{d}_3+\tilde{d}_4} a^{\tilde{d}_1+2\tilde{d}_2+\tilde{d}_3+\tilde{d}_4}q^{-\tilde{d}_1-2\tilde{d}_2-\frac{3}{2}\tilde{d}_3-\frac{1}{2}\tilde{d}_4}q^{r(\tilde{d}_1+\tilde{d}_2+\tilde{d}_4)}
    \\
    &\qquad\qquad\times q^{\frac{1}{2}\sum_{i,j=1}^4\tilde{C}_{ij}\tilde{d}_i\tilde{d}_j}\frac{(q^{r};q^{-1})_{\tilde{d}_1+\tilde{d}_2+\tilde{d}_3+\tilde{d}_4}}{\prod_{i=1}^4(q;q)_{\tilde{d}_i}}
\end{align*}
Performing the~substitution $q^{r}\rightarrow x$, we obtain -- up to an~overall prefactor -- the~following expression:
\begin{align*}
    F_{4_1}^{(-\frac{1}{2})}(x,a,q)&=\sum_{\tilde{d}_1,\tilde{d}_2,\tilde{d}_3,\tilde{d}_4\ge 0}(-1)^{\tilde{d}_3+\tilde{d}_4} a^{\tilde{d}_1+2\tilde{d}_2+\tilde{d}_3+\tilde{d}_4}q^{-\tilde{d}_1-2\tilde{d}_2-{\frac{3\tilde{d}_3+\tilde{d}_4}{2}}}q^{{\frac{1}{2}}{\sum_{i,j=1}^4\tilde{C}_{ij}\tilde{d}_i\tilde{d}_j}}
    \\
    &\qquad\qquad\times x^{\tilde{d}_1+\tilde{d}_2+\tilde{d}_4}{\frac{(x;q^{-1})_{\tilde{d}_1+\tilde{d}_2+\tilde{d}_3+\tilde{d}_4}}{\prod_{i=1}^4(q)_{\tilde{d}_i}}},
\end{align*}
where 
$$\tilde{C}=\left(\begin{array}{cccc}0&0&0&0\\
0&2&1&1\\
0&1&1&1\\
0&1&1&1\end{array}\right).$$
In order to obtain $F_K$ in a~quiver form, we need to expand the~$q$-Pochhammer $(x;q^{-1})_{\tilde{d}_1+\tilde{d}_2+\tilde{d}_3+\tilde{d}_4}$. We can do that in either of two ways described in Section \ref{sec:41}.

If we use (\ref{eq:qpoch-sum-general}), we get
\begin{align*}
   \frac{(x;q^{-1})_{\tilde{d}_1+\tilde{d}_2+\tilde{d}_3+\tilde{d}_4}}{\prod_{i=1}^4(q)_{\tilde{d}_{i}}}&= \sum\limits_{\alpha_{1}+\beta_{1}=\tilde{d}_{1}} \cdots \sum\limits_{\alpha_{4}+\beta_{4}=\tilde{d}_{4}} (-q^{1/2})^{\beta_{1}^2+\ldots+\beta_{4}^2+2\sum_{i=1}^{3} \beta_{i+1}(\tilde{d}_{1}+\ldots + \tilde{d}_{i})} \\
   &\qquad\qquad\times\frac{\big(x \, q^{1/2-\sum_i \alpha_i-\sum_i \beta_i}\big)^{\beta_{1}+\cdots+\beta_4}}{(q)_{\alpha_{1}}(q)_{\beta_{1}}\cdots (q)_{\alpha_{4}}(q)_{\beta_{4}}}.
\end{align*}
In such a~way we have eight summation variables $\alpha_1,\ldots,\alpha_4,\beta_1,\ldots,\beta_4$, satisfying $\tilde{d}_i=\alpha_i+\beta_i$, $i=1,\ldots,4$. Focusing on the~powers of $q$ that are quadratic in these summation variables, it can be seen from the~formula above that the~corresponding matrix $C$ is of the~form
$$
C=\left(\begin{array}{c|c}
\tilde{C} & \tilde{C}-I\\ 
\hline
\tilde{C}-I & \tilde{C}-J
\end{array}
\right),
$$
where $I$ is the~identity matrix, and $J$ matrix of all ones.

Therefore, with summation variables $(d_1,\ldots,d_8)=(\alpha_1,\ldots,\alpha_4,\beta_1,\ldots,\beta_4)$, we get
\be
\label{eq:F_4_1-f2}
    F_{\mathbf{4}_1}^{(-\frac{1}{2})}(x,a,q) = \sum_{d_1,\cdots,d_8 \geq 0}(-q^{\frac{1}{2}})^{\sum_{ i,j=1}^8 C_{ij}d_id_j} \prod_{i=1}^{8}\frac{x_i^{d_i}}{(q)_{d_i}},
\ee
with
\begin{equation}
C=\left(\begin{array}{cccccccc}
0 & 0 & 0 & 0 & -1 & 0 & 0 & 0\\
0 & 2 & 1 & 1 & 0 & 1 & 1 &1 \\
0 & 1 & 1 & 1 & 0 & 1 & 0 & 1\\
0 & 1 & 1 & 1 & 0 & 1 & 1 & 0\\
-1 & 0 & 0 & 0 & -1 & -1 & -1 & -1\\
0 & 1 & 1 & 1 & - 1 & 1 & 0 & 0\\
0 & 1 & 0 & 1 & -1 & 0 & 0 & 0\\
0 & 1 & 1 & 0 & -1 & 0 & 0 & 0
\end{array}\right),\quad\begin{pmatrix}x_1 \\ x_2 \\ x_3 \\ x_4 \\ x_5 \\ x_6 \\ x_7 \\ x_8 \end{pmatrix}=\left(\begin{array}{c}
xaq^{-1}\\
xa^2q^{-2}\\
aq^{-\frac{3}{2}}\\
xaq^{-\frac{1}{2}}\\
x^2aq^{-\frac{1}{2}}\\
x^2a^2q^{-\frac{3}{2}}\\
xaq^{-1}\\
x^2a
\end{array}\right).\label{eq:4_1-first expansion}
\end{equation}
Alternatively, if we use the~expansion (\ref{eq:new expansion}), we get
$$(x;q^{-1})_{d_1+\ldots+d_4}=\sum_{i,j\ge 0} (-1)^i x^{i+j} q^{\frac{1}{2}i+j} q^{-i (d_1+d_2+d_3+d_4)} q^{\frac{1}{2}i^2}\frac{1}{(q)_i(q)_j}.  $$
Hence, in this case we get an~analogue of \eqref{eq:F_4_1-f2} with six summation variables $(d_1,\ldots,d_6)=(\tilde d_1,\ldots,\tilde d_4,i,j)$ and with the~quiver matrix $C$ and vector $\Bf{x}$ given by
\begin{equation}
C=\left(\begin{array}{cccccc}
0 & 0 & 0 & 0 & -1 & 0 \\
0 & 2 & 1 & 1 & -1 & 0  \\
0 & 1 & 1 & 1 & -1 & 0 \\
0 & 1 & 1 & 1 & -1 & 0 \\
-1 & -1 & -1 & -1 & 1 & 0 \\
0 & 0 & 0 & 0 & 0 & 0\end{array}\right),\quad\begin{pmatrix}x_1 \\ x_2 \\ x_3 \\ x_4 \\ x_5 \\ x_6  \end{pmatrix}=\left(\begin{array}{c}
xaq^{-1}\\
xa^2q^{-2}\\
aq^{-\frac{3}{2}}\\
xaq^{-\frac{1}{2}}\\
xq^{\frac{1}{2}}\\
xq
\end{array}\right).\label{eq:4_1-new expansion}
\end{equation}

Let us move to the~abelian branch. In that case we cannot apply the~reasoning from Section~\ref{sec:41} since some entries of conormal quiver are negative. In particular, we have $C_{\textrm{min}}=-2$, so framing 2 would be needed, as we saw in the~paragraph above. However, for the~abelian branch we can use a~direct approach, matching quiver adjacency matrix and the~change of variables against order by order expansion of $F_{\mathbf{4}_1}$. This leads to
\begin{equation}\label{piotrfig8ab}
F_{\mathbf{4}_1}(x,a,q) = \sum_{d_1,\cdots,d_6 \geq 0}(-q^{\frac{1}{2}})^{\sum_{ i,j=1}^6 C_{ij}d_id_j} \prod_{i=1}^{6}\frac{x_i^{d_i}}{(q)_{d_i}}
\end{equation}
with
\[
\begin{pmatrix}x_1 \\ x_2 \\ x_3 \\ x_4 \\ x_5 \\ x_6 \end{pmatrix} = \begin{pmatrix} q x \\  q x \\ q x \\  q^{-\frac{1}{2}}a x \\ q^{-\frac{1}{2}}a x \\  q^{-\frac{1}{2}}a x \end{pmatrix},
\]
and $C$ given by any of the~following matrices:
\begin{equation}\label{eq:fig8quiver}
\begin{pmatrix} 0 & 0 & 0 & 0 & 0 & 0 \\ 0 & 0 & -1 & -1 & 0 & 0 \\ 0 & -1 & 0 & 0 & 1 & 0 \\ 0 & -1 & 0 & 1 & 1 & 0 \\ 0 & 0 & 1 & 1 & 1 & 0 \\ 0 & 0 & 0 & 0 & 0 & 1 \end{pmatrix},\;
\begin{pmatrix} 0 & 0 & 0 & 1 & 0 & 0 \\ 0 & 0 & -1 & -1 & 0 & 0 \\ 0 & -1 & 0 & 0 & 0 & 0 \\ 1 & -1 & 0 & 1 & 1 & 0 \\ 0 & 0 & 0 & 1 & 1 & 0 \\ 0 & 0 & 0 & 0 & 0 & 1 \end{pmatrix},\;
\begin{pmatrix} 0 & 0 & 0 & 0 & 0 & 0\\ 0 & 0 & -1 & 0 & 0 & 0 \\ 0 & -1 & 0 & 0 & 1 & -1 \\ 0 & 0 & 0 & 1 & 1 & 0 \\ 0 & 0 & 1 & 1 & 1 & 0 \\ 0 & 0 & -1 & 0 & 0 & 1 \end{pmatrix}.
\end{equation}
Note that all these matrices are equivalent in the~sense of~\cite{Jankowski:2021flt}.

Analogous approach applied to the~non-abelian branch with slope $-\frac{1}{2}$ leads to
\begin{equation}\label{piotrfig8Nab}
F_{\mathbf{4}_1}^{(-\frac{1}{2})}(x,a,q) = \sum_{d_1,\cdots,d_5 \geq 0}(-q^{\frac{1}{2}})^{\sum_{1\leq i,j\leq 5}C_{ij}d_id_j} \prod_{i=1}^{5}\frac{x_i^{d_i}}{(q)_{d_i}}
\end{equation}
with
\begin{equation}\label{piotrfig8Nab-C}
C = \begin{pmatrix} 0 & 1 & 0 & 0 & 0 \\ 1 & 0 & 1 & 0 & 0 \\ 0 & 1 & 1 & 1 & 0 \\ 0 & 0 & 1 & 1 & 0 \\ 0 & 0 & 0 & 0 & 1 \end{pmatrix},\quad \begin{pmatrix}x_1 \\ x_2 \\ x_3 \\ x_4 \\ x_5 \end{pmatrix} = \begin{pmatrix} q x \\  a x \\ q^{\frac{3}{2}} x \\  q^{\frac{3}{2}} x \\  q^{-\frac{1}{2}}a x \end{pmatrix}.
\end{equation}
We expect this quiver along with the quivers \eqref{eq:4_1-first expansion} and \eqref{eq:4_1-new expansion} are all equivalent (up to a factor independent of $x$).

\subsubsection{Trefoil}

In the~case of the~abelian branch of the~left-handed trefoil, we can use the~formula given in~\cite{EGGKPS}:
\be\label{eq:trefoil unexpanded}
F_{\mathbf{3}_1}(x,a,q)=\sum_{k\geq 0} x^k q^k \frac{(x;q^{-1})_k(aq^{-1};q)_k}{(q)_k}.
\ee
Expanding $(aq^{-1};q)_k$ using Lemma 4.5 from~\cite{KRSS2}, we get an~expression analogous to Equation (\ref{eq:F_K from knot conormal Q}):
\[
F_{\mathbf{3}_1}(x,a,q)=\sum_{\tilde{d}_1,\tilde{d}_2\geq 0}(-1)^{\tilde{d}_1}q^{\tilde{d}_2}x^{\tilde{d}_1+\tilde{d}_2}a^{\tilde{d}_1}q^{\frac{1}{2}(\tilde{d}_1^2-\tilde{d}_1)}\frac{(x;q^{-1})_{\tilde{d}_1+\tilde{d}_2}}{(q)_{\tilde{d}_1}(q)_{\tilde{d}_2}}.
\]
As explained in Section \ref{sec:41}, we can expand the~last fraction in two ways. Simpler expansion (\ref{eq:new expansion}) leads to
\be
\label{eq:F_3_1}
    F_{\mathbf{3}_1}(x,a,q) = \sum_{d_1,\cdots,d_4 \geq 0}(-q^{\frac{1}{2}})^{\sum_{ i,j=1}^4 C_{ij}d_id_j} \prod_{i=1}^{4}\frac{x_i^{d_i}}{(q)_{d_i}},
\ee
where
\begin{equation}
C=\left(\begin{array}{cccc}
0 & 0 & 0 & -1\\
0 & 1 & 0 & -1\\
0 & 0 & 0 & 0\\
-1 & -1 & 0 & 1
\end{array}\right),\quad\begin{pmatrix}x_1 \\ x_2 \\ x_3 \\ x_4 \end{pmatrix}=\left(\begin{array}{c}
xq\\
xaq^{-1/2}\\
xq\\
xq^{1/2}
\end{array}\right).\label{eq:new quiver and KQ change of vars}
\end{equation}
On the~other hand, the~expansion (\ref{eq:qpoch-sum-general}) gives \eqref{eq:F_3_1} with the~quiver
\begin{equation}
C=\left(\begin{array}{cccc}
0 & 0 & -1 & -1\\
0 & 1 & 0 & 0\\
-1 & 0 & -1 & -1\\
-1 & 0 & -1 & 0
\end{array}\right),\quad\begin{pmatrix}x_1 \\ x_2 \\ x_3 \\ x_4 \end{pmatrix}=\left(\begin{array}{c}
xq\\
xaq^{-1/2}\\
x^{2}q^{3/2}\\
x^{2}a
\end{array}\right),\label{eq:old quiver and change of vars}
\end{equation}
which reproduces the~result of~\cite{Kuch}.

We can connect \eqref{eq:new quiver and KQ change of vars} with \eqref{eq:old quiver and change of vars} directly, using the~procedures of linking and removing a~redundant pair of nodes, defined in Section \ref{sec:KQ correspondence} (for details see~\cite{EKL2}).
We start from \eqref{eq:new quiver and KQ change of vars} and link
nodes number 2 and 4, which leads to the~quiver
\begin{equation*}
C=\left(\begin{array}{ccccc}
0 & 0 & 0 & -1 & -1\\
0 & 1 & 0 & 0 & 0\\
0 & 0 & 0 & 0 & 0\\
-1 & 0 & 0 & 1 & 0\\
-1 & 0 & 0 & 0 & 0
\end{array}\right),\quad\begin{pmatrix}x_1 \\ x_2 \\ x_3 \\ x_4 \\ x_5 \end{pmatrix}=\left(\begin{array}{c}
xq\\
xaq^{-1/2}\\
xq\\
xq^{1/2}\\
x^{2}a
\end{array}\right).
\end{equation*}
Then linking nodes number 1 and 4 gives
\begin{equation*}
C=\left(\begin{array}{cccccc}
0 & 0 & 0 & 0 & -1 & -1\\
0 & 1 & 0 & 0 & 0 & 0\\
0 & 0 & 0 & 0 & 0 & 0\\
0 & 0 & 0 & 1 & 0 & 0\\
-1 & 0 & 0 & 0 & 0 & -1\\
-1 & 0 & 0 & 0 & -1 & -1
\end{array}\right),\quad\begin{pmatrix}x_1 \\ x_2 \\ x_3 \\ x_4 \\ x_5 \\ x_6 \end{pmatrix}=\left(\begin{array}{c}
xq\\
xaq^{-1/2}\\
xq\\
xq^{1/2}\\
x^{2}a\\
x^{2}q^{3/2}
\end{array}\right).
\end{equation*}
Now we can notice that nodes number 3 and 4 form a~redundant pair.
Removing it leads to
\begin{equation*}
C=\left(\begin{array}{cccc}
0 & 0 & -1 & -1\\
0 & 1 & 0 & 0\\
-1 & 0 & 0 & -1\\
-1 & 0 & -1 & -1
\end{array}\right),\quad\begin{pmatrix}x_1 \\ x_2 \\ x_3 \\ x_4\end{pmatrix}=\left(\begin{array}{c}
xq\\
xaq^{-1/2}\\
x^{2}a\\
x^{2}q^{3/2}
\end{array}\right),
\end{equation*}
which -- after exchanging the~last two nodes -- is equal to quiver
\eqref{eq:old quiver and change of vars}.

For the~mirror image of the~left-handed trefoil (i.e. the~right-handed trefoil) in framing $f=3$, the~prescription from Section \ref{sec:41} leads to the~following formula for the~non-abelian branch with slope $\frac{1}{3}$:
\begin{equation}\label{tref}
F_{{\mathbf{3}^r_1}}^{(\frac{1}{3})}(x,a,q)=\sum_{\tilde{d}_1,\tilde{d}_2 \ge 0} (-1)^{\tilde{d}_1+\tilde{d}_2} a^{\tilde{d}_1+\tilde{d}_2} q^{\frac{1}{2} \tilde{d}_1^2 -\frac{1}{2} \tilde{d}_1 + \frac{1}{2} \tilde{d}_2^2 - \frac{3}{2} \tilde{d}_2} x^{\tilde{d}_1} \frac{ (x;q^{-1})_{\tilde{d}_1+\tilde{d}_2}}{(q)_{\tilde{d}_1}(q)_{\tilde{d}_2}}.
\end{equation}
Then we rewrite the~$q$-Pochhammer $(x;q^{-1})_{\tilde{d}_1+\tilde{d}_2}$ using expansion (\ref{eq:new expansion}):
\begin{equation}\label{pom31}
F_{\mathbf{3}_1^r}^{(\frac{1}{3})}(x,a,q)=\sum_{\tilde{d}_1,\tilde{d}_2,i,j} (-1)^{\tilde{d}_1+\tilde{d}_2+i} a^{\tilde{d}_1+\tilde{d}_2} q^{\frac{1}{2} (\tilde{d}_1^2 - \tilde{d}_1) + \frac{1}{2} (\tilde{d}_2^2 - \tilde{d}_2)} \frac{x^{\tilde{d}_1+i+j} q^{i+j-\tilde{d}_2} q^{-i\tilde{d}_1-i\tilde{d}_2} q^{\frac{1}{2}(i^2-i)}}{(q)_{\tilde{d}_1}(q)_{\tilde{d}_2}(q)_i (q)_j}.
\end{equation}
We can sum over index $\tilde{d}_2$ that does not appear in the~power of $x$:
\be\label{eq:qPochexpansion1}
\sum_{\tilde{d}_2} (-1)^{\tilde{d}_2} a^{\tilde{d}_2}q^{-\tilde{d}_2}q^{-i \tilde{d}_2}q^{\frac{1}{2}(\tilde{d}_2^2-\tilde{d}_2)}\frac{1}{(q)_{\tilde{d}_2}}=(aq^{-1-i};q)_{\infty}=(aq^{-1};q)_{\infty} (aq^{-1-i};q)_i
\ee
and use the~formula
\be\label{eq:qPochexpansion2}
\frac{(aq^{-1-i};q)_i}{(q)_i}=\sum_{\alpha+\beta=i}(-1)^{\alpha}a^{\alpha}q^{-\alpha-i\alpha}q^{\frac{1}{2}(\alpha^2-\alpha)}\frac{1}{(q)_{\alpha}(q)_{\beta}}.
\ee
Substituting (\ref{eq:qPochexpansion1}-\ref{eq:qPochexpansion2}) back in (\ref{pom31}), we obtain
\[
F_{\mathbf{3}_1^r}^{(\frac{1}{3})}(x,a,q)=(aq^{-1};q)_{\infty} \sum_{\tilde{d}_1,\alpha,\beta,j}
(-1)^{\tilde{d}_1+\beta}a^{\tilde{d}_1+\alpha}q^{j+\beta-\alpha}\frac{q^{\frac{1}{2}(\tilde{d}_1^2-\tilde{d}_1)+\frac{1}{2}(\beta^2-\beta)-\tilde{d}_1\alpha-\tilde{d}_1\beta}x^{\tilde{d}_1+\alpha+\beta+j}}{(q)_{\tilde{d}_1}(q)_{\alpha}(q)_{\beta}(q)_j}.
\]
Finally, after dividing by the~overall infinite $q$-Pochhammer independent of $x$, we get the~quiver form
\begin{equation}\label{piotr411}
F_{\mathbf{3}_1^r}^{(\frac{1}{3})}(x,a,q) = \sum_{d_1,d_2,d_3,d_4 \geq 0}(-q^{\frac{1}{2}})^{\sum_{1\leq i,j\leq 4}C_{ij}d_id_j} \prod_{i=1}^{4}\frac{x_i^{d_i}}{(q)_{d_i}}
\end{equation}
with
\begin{equation}\label{piotr411-C}
C = \begin{pmatrix} 0 & 0 & 0 & 0 \\ 0 & 0 & -1 & 0 \\ 0 & -1 & 1 & -1 \\ 0 & 0 & -1 & 1 \end{pmatrix},\quad \begin{pmatrix}x_1 \\ x_2 \\ x_3 \\ x_4\end{pmatrix} = \begin{pmatrix}x q  \\ x a q^{-1} \\ x a q^{-\frac{1}{2}} \\ x q^{\frac{1}{2}} \end{pmatrix}.
\end{equation}

In the~case of the~right-handed trefoil and abelian branch, we cannot apply the~reasoning from Section \ref{sec:41} because it does not correspond to framing $f=C_\textrm{min}$ nor $f=C_\textrm{max}$. However, similarly to the~abelian branch of figure-eight knot, we can match quiver adjacency matrix and the~change of variables against order by order expansion of $F_{\mathbf{3}_1^r}$, which leads to 
\begin{equation}\label{piotr48}
F_{\mathbf{3}_1^r}(x,a,q) = \sum_{d_1,d_2,d_3,d_4 \geq 0}(-q^{\frac{1}{2}})^{\sum_{1\leq i,j\leq 4}C_{ij}d_id_j} \prod_{i=1}^{4}\frac{x_i^{d_i}}{(q)_{d_i}}
\end{equation}
with
\begin{equation}\label{piotr48-C}
C = \begin{pmatrix}0 & 1 & 0 & 0 \\ 1 & 0 & 1 & 0 \\ 0 & 1 & 1 & 0\\ 0 & 0 & 0 & 1 \end{pmatrix},\quad \begin{pmatrix}x_1 \\ x_2 \\ x_3 \\ x_4\end{pmatrix} = \begin{pmatrix} x q  \\ x a~\\ x a q^{-\frac{1}{2}} \\ x a q^{-\frac{1}{2}} \end{pmatrix}.
\end{equation}

\subsubsection{\texorpdfstring{$(2,2p+1)$}{(2,2p+1)} torus knots}

In~\cite{Kuch} it was shown that there is a~recursion relating quivers corresponding to $(2, 2p + 1)$ torus knot complements. Those quivers are obtained by expanding $q$-Pochhammer $(x;q^{-1})_d$ in the~general formula from~\cite{EGGKPS} via (\ref{eq:qpoch-sum-general}). Using expansion~(\ref{eq:new expansion}), we can obtain the~corresponding quivers in an~even simpler form. 
    
Let us start our analysis from a~slightly rearranged form of \eqref{eq:new quiver and KQ change of vars}:
	\begin{equation*}
		F_{\mathbf{3}_1}(x, a, q) 
		= \sum_{d_1, \cdots, d_4 \geq 0} (-q^{1/2})^{\Bf{d}\cdot C \cdot\Bf{d}} \frac{x^{\Bf{n}\cdot \Bf{d}} a^{\Bf{a}\cdot \Bf{d}}q^{\Bf{q}\cdot \Bf{d}-\frac{1}{2}\sum_{i} C_{ii}d_i} }{(q)_{\Bf{d}}},
	\end{equation*}
	where $q_i-\frac{1}{2}C_{ii}=l_i$ and the~vectors $\Bf{n}, \Bf{a}, \Bf{q}$ and the~matrix $C$ are given by
	\[
	C  = \begin{pmatrix}
			1 & 0 & -1 & 0
			\\ 0 & 0 & -1 & 0
			\\ -1 & -1 & 1 & 0
			\\ 0 & 0 & 0 & 0
		\end{pmatrix},
	\qquad
	\begin{array}{l}
	    \Bf{n}   = (1, 1, 1, 1),  \\
	    \Bf{a}   = (1, 0, 0, 0),  \\
	    \Bf{q}   = (0, 1, 1, 1) = 1 - \Bf{a}.
	\end{array}	
	\]
The~key idea is to keep in mind that the~upper-left $2 \times 2$ matrix comes from $q^{\frac{1}{2}d_1^2}$ in \eqref{eq:trefoil unexpanded}, whereas the~last two rows and columns originate from the~expansion of the~$q$-Pochhammer $(x,q^{-1})_{d_1 + d_2}$. For more complicated torus knots, the~upper-left part will grow following the~quadratic $q$ powers in the~analogue of \eqref{eq:trefoil unexpanded}, whereas the~last two rows and columns remain unchanged. Let us see it in the~example of $\mathbf5_1$ and $\mathbf7_1$ knots.

For the~$\mathbf5_1$ knot we can derive
		\begin{align}
			F_{\mathbf{5}_1}(x, a, q) & = \sum_{d_1, \cdots, d_6 \geq 0}  \left(-q^{1/2}\right)^{\Bf{d}\cdot C \cdot\Bf{d}} \frac{x^{\Bf{n}\cdot \Bf{d}} a^{\Bf{a}\cdot \Bf{d}} q^{\Bf{q}\cdot \Bf{d}-\frac{1}{2}\sum_{i} C_{ii}d_i}}{(q)_{\Bf{d}}},\nonumber
			\\ C  & = \begin{pmatrix}
				1 & 0 & 0 & 0 & -1 & 0
				\\ 0 & 0 & -1 & -1 & -1 & 0
				\\ 0 & -1 & -1 & -2 & -1 & 0
				\\ 0 & -1 & -2 & -2 & -1 & 0
				\\ -1 & -1 & -1 & -1 & 1 & 0
				\\ 0 & 0 & 0 & 0 & 0 & 0	\end{pmatrix},\label{eq:51quiver} 
			 \qquad
			 \begin{array}{l}
			      \Bf{n}   = (1, 1, 3, 3, 1, 1),  \\
			      \Bf{a}   = (1, 0, 1, 0, 0, 0), \\
			      \Bf{q}   = (0, 1, 0, 1, 1, 1) = 1 - \Bf{a},
			 \end{array}
		\end{align}
whereas for $\mathbf7_1$ we have
		\begin{align*}
			F_{\mathbf7_1}(x, a, q) & = \sum_{d_1, \cdots, d_8 \geq 0}  \left(-q^{1/2}\right)^{\Bf{d}\cdot C \cdot\Bf{d}} \frac{x^{\Bf{n}\cdot \Bf{d}} a^{\Bf{a}\cdot \Bf{d}} q^{\Bf{q}\cdot \Bf{d}-\frac{1}{2}\sum_{i} C_{ii}d_i}}{(q)_{\Bf{d}}}
			\\ C  & = \begin{pmatrix}
				1 & 0 & 0 & 0 & 0 & 0 & -1 & 0
				\\ 0 & 0 & -1 & -1 & -1 & -1 & -1 & 0
				\\ 0 & -1 & -1 & -2 &  -2 & -2 & -1 & 0
				\\ 0 & -1 & -2 & -2 & -3 & -3 & -1 & 0
				\\ 0 & -1 & -2 & -3 & -3 & -4 & -1 & 0
				\\ 0 & -1 & -2 & -3 & -4 & -4 & -1 & 0
				\\ -1 & -1 & -1 & -1 & -1 & -1 & 1 & 0
				\\ 0 & 0 & 0 & 0 & 0 & 0 & 0 & 0
			\end{pmatrix},
			\qquad
			 \begin{array}{l}
			      \Bf{n}   = (1, 1, 3, 3, 5, 5, 1, 1),  \\
			      \Bf{a}   = (1, 0, 1, 0, 1, 0, 0, 0), \\
			      \Bf{q}   = (0, 1, 0, 1, 0, 1, 1, 1) \\
			      \phantom{\Bf{q}}= 1 - \Bf{a}.
			 \end{array}
		\end{align*}	
Using the~general formula from~\cite{EGGKPS}, we can show that the~pattern continues and
		\begin{align*}
			F_{T_{2,2p+1}}(x, a, q) & = \sum_{d_1, \cdots, d_{2p + 2} \geq 0} \left(-q^{1/2}\right)^{\Bf{d}\cdot C \cdot\Bf{d}} \frac{x^{\Bf{n}\cdot \Bf{d}} a^{\Bf{a}\cdot \Bf{d}} q^{\Bf{q}\cdot \Bf{d}-\frac{1}{2}\sum_{i} C_{ii}d_i}}{(q)_{\Bf{d}}}
			\\ C & = \begin{pmatrix}
				\Bf{I}_{2p} - \Bf{D} & \Bf{-1} & \Bf{0}
				\\ \Bf{-1} & 1 & 0
				\\ \Bf{0} & 0 & 0
			\end{pmatrix},
			\qquad
			 \begin{array}{l}
			      \Bf{n}   = (1, 1, 3, 3, \cdots, 2p - 1, 2p - 1, 1, 1),  \\
			      \Bf{a}   = (1, 0, \cdots, 1, 0, 0, 0), \\
			      \Bf{q}   = (0, 1, \cdots, 0, 1, 1, 1) = 1 - \Bf{a},
			 \end{array}
		\end{align*}
		where $\Bf{-1}, \Bf{0}$ denote constant vectors of appropriate size, $\Bf{I}_{2p}$ is the~identity matrix and $\Bf{D}$ is the~matrix $\Bf{D}_{i j} = \min(i, j) - 1$ with $1 \leq i, j \leq 2p$. Note that we always have a~totally disconnected node which we can remove and replace with a~$q$-Pochhammer prefactor.


\subsection{\texorpdfstring{$F_K$}{FK} invariants and quivers from classical \texorpdfstring{$A$}{A}-polynomials and branches}

In the~previous sections we have made two interesting discoveries. First, we associated $F_K$ invariants to various branches of $A$-polynomials. Second, we found that these $F_K$ invariants can be encoded in quiver generating series, analogously as in the~knots-quivers correspondence~\cite{KRSS1,KRSS2}. Let us now stress that these two ideas are intimately related. On one hand, the~branches $y=y(x)$ of $A$-polynomials can be obviously determined from the~quiver form (or any other form) of $F_K$ invariants. On the~other hand, more interestingly, it turns out that just the~existence of the~quiver form imposes strong constraints, which enable to determine an~underlying quiver simply from the~finite number of terms in the~expansion $y=y(x)$ for a~given branch. In this section we illustrate these relations; note that similar reconstruction methods were discussed in~\cite{PSS,Banerjee:2020dqq}. It would also be desirable to understand better which property of $A$-polynomials asserts the~existence of the~quiver form of the~corresponding $F_K$ invariants; presumably, it might be related to K-theoretic conditions that $A$-polynomials must meet too~\cite{Guk,GS2}.

To proceed, let us again write down the~quiver generating function \eqref{eq:motivic generating series}:
\be
P_Q(x_1,x_2,\ldots,x_m,q) =\sum_{d_1,...,d_m=0}^{\infty}  
(-q^{1/2})^{ \sum_{i,j=1}^m C_{ij}d_i d_j} 
\frac{ x_1^{d_1}\cdots x_m^{d_m}} { (q)_{d_1}\cdots (q)_{d_m} } .       \label{PC}
\ee
To such a~series we can associate a~classical expansion 
\[
y(x_1,\ldots,x_m) = \lim_{q\to 1} \frac{P_Q(qx_1,qx_2,\ldots,qx_m,q)}{P_Q(x_1,x_2,\ldots,x_m,q)} \equiv \sum_{k_1,\ldots,k_m} b_{k_1,\ldots,k_m} x_1^{k_1}\cdots x_m^{k_m}.
\]
In~\cite{PSS} general expressions for coefficients $b_{k_1,\ldots,k_m}$ (as well as for numerical Donaldson-Thomas invariants that arise from factorization of the~above series) in terms of elements of the~quiver matrix $C$ have been found. Currently we are interested in quiver generating series $F(x,a,q)$ and corresponding $y(x)=\sum_i y_i x^i$, which depend on a~single generating parameter~$x$, and arise from specialization $x_i=(-1)^{\#}a^{\#} q^{\#} x^{\#}$ of the~above expressions. In principle, coefficients $y_i$ can be obtained from $b_{k_1,\ldots,k_m}$, however such expressions are quite complicated. It is therefore more convenient to determine coefficients $y_i$ directly, having specialized $x_i$ first. 

Therefore, for the~purpose of our discussion and brevity, consider the~specialization such that each $x_i$ is proportional to a~single power of $x$ (and also fix $t=-1$), so that $x_i = a^{a_i} q^{l_i} x$. In this case $F(x,a,q) = P_Q(\{x_i=a^{a_i} q^{l_i} x\})$ and we are interested in the~corresponding classical series:
\be
y(x) = \lim_{q\to 1} \frac{F(qx,a,q)}{F(x,a,q)} \equiv \sum_{i=0}^{\infty} y_i x^i.   \label{yx-def}
\ee
Note that linear powers of $q$ in $F(x,a,q)$ (that arise from specialization of $x_i$) do not play a~role at the~classical level, so the~coefficients $y_i$ depend only on entries of $C$ and the~vector~$\Bf{a}$. Taking the~above limit explicitly we find coefficients $y_i$:
\begin{equation}
\begin{split}   \label{yx-coeffs}
y_1 &= -\sum_{i=1}^m (-1)^{C_{ii}} a^{a_i}, \\
y_2 & = \sum_{i=1}^m C_{ii} a^{2a_i} + \sum_{i<j} (-1)^{C_{ii} + C_{jj}} a^{a_i + a_j} (1+2C_{ij}), \\
y_3 & = -\frac12 \Big( \sum_{i=1}^m (-1)^{C_{ii}}  a^{3a_i} C_{ii}(3C_{ii}-1)
+ \sum_{i\neq j} (-1)^{C_{jj}} a^{2a_i + a_j}(2C_{ii} + C_{ij})(3C_{ij}+1) \\
  + & \sum_{i<j<k} (-1)^{C_{ii}+C_{jj}+C_{kk}} a^{a_i + a_j + a_k} \big( 2+4C_{jk} + C_{ik}(4+6C_{jk})+ C_{ij}(4+6C_{ik} + 6C_{jk}) \big)
\Big). 
\end{split}
\end{equation}
For a~quiver of size $m=2$, denoting
\[
C_{1,1}=\alpha,\qquad C_{1,2}=C_{2,1}=\beta,\qquad C_{2,2} = \gamma,
\]
and keeping parameters $a_1,a_2$ and $q_1,q_2$, let us also write down one more coefficient explicitly:
\begin{align}
\begin{split}   \label{yx-coeffs-m2}
y_1^{m=2} &= -(-1)^{\alpha} a^{a_1}-(-1)^{\gamma} a^{a_2}, \\
y_2^{m=2} & = \alpha a^{2a_1} + \gamma a^{2a_2} + (-1)^{\alpha+\gamma} a^{a_1 + a_2} (1+2\beta), \\
y_3^{m=2} & = -\frac12 \Big(   
 (-1)^{\alpha}  a^{3a_1} \alpha(3\alpha-1) + (-1)^{\gamma}  a^{3a_2} \gamma(3\gamma-1)  \\
& \qquad  + (-1)^{\gamma} a^{2a_1 + a_2}(2\alpha+ \beta)(3\beta+1) +
(-1)^{\alpha} a^{2a_2 + a_1}(2\gamma+ \beta)(3\beta+1)  \Big), \\
y_4^{m=2} & = -\frac16\Big( 2a^{4a_1} \alpha(2\alpha-1)(4\alpha-1) + 2a^{4a_2} \gamma(2\gamma-1)(4\gamma-1)  \\
&\qquad + (-1)^{\alpha+\gamma} a^{3a_1+a_2} (3\alpha+\beta)(3\alpha+\beta-1)(4\beta+1)  \\
&\qquad + (-1)^{\alpha+\gamma} a^{3a_2+a_1} (3\gamma+\beta)(3\gamma+\beta-1)(4\beta+1)  \\
&\qquad + 6a^{2a_1+2a_2} (\alpha+\beta)(\beta+\gamma)(4\beta+1) \Big).
\end{split}
\end{align}

Clearly, once we know the~quiver matrix $C$ and powers $a_i$, we can immediately determine coefficients $y_i$ in the~classical series $y(x)$. However, we can also consider the~opposite perspective, and reconstruct the~quiver matrix $C$ and coefficients $a_i$, by comparing explicit form of (a finite number of) coefficients $y_i$ with the~above formulas. In particular, this is the~situation we have to deal with when considering various branches of $A$-polynomial. Each branch is associated to one solution $y=y(x)$ of the~($a$-deformed) $A$-polynomial equation $A(x,y,a)=0$. If we assume that for each branch there exists a~$q$-series $F(x,a,q)$ that has a~quiver representation, then (for each branch) we can determine the~corresponding $C$ and~$a_i$. In addition, parameters $q_i$ can be determined by comparison with a~few first terms of the~quantum series (\ref{PC}), if only we can determine them by other means. 


\subsubsection{Examples: trefoil and figure-eight}

Let us illustrate the~above relations on the~example of trefoil and figure-eight knot (still, for simplicity, restricting to $t=-1$ case). For the~abelian branch of right-handed trefoil, the~generating function (\ref{piotr48}) is determined by the~quiver matrix~$C$ of size $m=4$, given  in (\ref{piotr48-C}), and $\Bf{a}=(a_1,a_2,a_3,a_4)=(0,1,1,1)$, so that (also taking appropriate $l_i$ into account):
\[
F(x,a,q) = 1+x\frac{ (a-q)}{q-1} + 
x^2\frac{ \left(a^2 \left(-q^2+q+1\right)+a q \left(q^2-q-2\right)+q^2\right)}{(q-1)^2 (q+1)}+\ldots
\]
The~classical $y(x)$ can be determined either by solving the~$A$-polynomial equation $A(x,y)=0$, or by taking the~limit (\ref{yx-def}), or by substituting $C$ and $a_i$ to the~formulae (\ref{yx-coeffs}). All these methods yield the~same result:
\[
y(x) = 1 + (a-1) x + (a - a^2) x^2 - 2 (a - 3 a^2 + 2 a^3) x^3 +\ldots
\]
Alternatively, starting from this result (obtained e.g. by solving $A(x,y)=0$), and comparing with coefficients in (\ref{yx-coeffs}), we can reconstruct $C$ and $a_i$.

Similarly, for the~non-abelian branch of right-handed trefoil, the~generating series (\ref{piotr411}) is determined by the~quiver matrix (\ref{piotr411-C}), $\Bf{a}=(0,1,1,0)$ (and appropriate $\Bf{q}$), so that
\be  \label{Fxaq-31-r-nab}
F(x,a,q) = 1 + x\frac{a }{q}+x^2\frac{a  (a-q)}{(q-1) q} +\ldots
\ee
Again we can determine the~classical $y(x)$ for this branch either by solving the~$A$-polynomial equation $A(x,y)=0$, or by taking the~limit (\ref{yx-def}), or by substituting $C$ and $a_i$ to the~formulae (\ref{yx-coeffs}). All these methods yield the~same result:
\[
y(x) = 1 + (-2 a + 2 a^2) x^2 + (-3 a^2 + 3 a^3) x^3 +\ldots
\]
Alternatively, starting from this result (obtained e.g. by solving $A(x,y)=0$), and comparing with coefficients in (\ref{yx-coeffs}), we can reconstruct $C$ and $a_i$.


Analogously, the~first (left-most) quiver matrix in (\ref{eq:fig8quiver}), together with $\Bf{a}=(0,0,0,1,1,1)$ (and appropriate $\Bf{l}$), determine the~series (\ref{piotrfig8ab}) for the~abelian branch of figure-eight knot:
\[
F(x,a,q) = 1 + x\frac{3  (a-q)}{q-1}+x^2\frac{ a^2 q^2+6 a^2 q+2 a^2-a q^3-8 a q^2-8 a q-a+2 q^3+6 q^2+q}{(q-1)^2 (q+1)}+\ldots
\]
Substituting $C$ and $a_i$ to (\ref{yx-coeffs}) we also get
\[
y(x) = 1 + 3 (-1 + a) x + (1 - 9 a + 8 a^2) x^2 + (-1 + 6 a~- 27 a^2 + 22 a^3) x^3 +\ldots
\]
or we can reconstruct $C$ and $a_i$ by matching this classical series with coefficients in (\ref{yx-coeffs}).

Finally, for the~non-abelian branch (of slope $-1/2$) of figure-eight knot, the~matrix (\ref{piotrfig8Nab-C}) and $\Bf{a}=(0,1,0,0,1)$ (and appropriate $\Bf{l}$) determine the~series (\ref{piotrfig8Nab})
\[
F(x,a,q) = 1 +x\frac{2 q^2-q}{q-1} +
x^2 \frac{ -a q^4+a q^3+a q^2-a q+q^6+3 q^5-2 q^4-2 q^3+q^2 }{(q-1)^2 (q+1)}+\ldots
\]
and the~corresponding classical series reads
\[
y(x) = 1 + x + 3 x^2 + (9 - 3 a) x^3 +\ldots
\]


\subsubsection{Extremal invariants and the~reconstruction of non-abelian $F_{\mathbf{3}_1^r}$}

Let us now illustrate in a~detailed yet simple example how to determine $F_K(x,a,q)$ from the~classical function~$y(x)$. Among others, interesting and sufficiently simple examples can be constructed in the~extremal limit~\cite{GKS}, defined as follows. For a~given $F_K(x,a,q)=\sum_{n=0}^{\infty} f_n(a,q)x^n$ we take into account only the~term with the~lowest or the~highest power of $a$ in each $f_n(a,q)$, and ignore all other terms in $f_n(a,q)$. More specifically, we call such extremal invariants respectively as minimal or maximal. If $F_K(x,a,q)$ is written in a~quiver form (\ref{PC}), then the~parameters $a_i$ play a~crucial role in determining its extremal version. The~corresponding extremal $A$-polynomial (whose solution encodes the~extremal $y(x)$), can be also defined by taking an~appropriate limit of the~full $A$-polynomial~\cite{GKS}. 

For concreteness, consider the~non-abelian branch for the~right-handed trefoil, whose $F_K$~series is given in (\ref{piotr411}) and (\ref{Fxaq-31-r-nab}). Its quiver matrix (\ref{piotr411-C}) has size $m=4$ and $\Bf{a}=(0,1,1,0)$. The~maximal invariants in this case are encoded in the~submatrix of (\ref{piotr411-C}) given by $C_{ij}$ for $i,j=2,3$. Indeed, since $a_2=a_3=1>a_1=a_4=0$, the~highest powers of $a$ for a~fixed $x$ will arise from contributions from $x_2$ and $x_3$. We thus end up with a~non-trivial example with a~quiver matrix of size $m=2$, which still captures some interesting properties of trefoil knot. Let us relabel indices $(2,3)$ into $(1,2)$, so that the~extremal quiver generating series in this case takes form
\begin{align}
\begin{split}
F_{\mathbf{3}_1^r}(x,a,q) &=\sum_{d_1,d_2=0}^{\infty}  
(-q^{1/2})^{(d_1\ d_2) {\alpha \ \beta \choose \beta \ \gamma} {d_1 \choose d_2}} 
\frac{ a^{a_1 d_1+a_2 d_2} q^{q_1d_1+q_2d_2} x^{d_1+d_2}} { (q)_{d_1} (q)_{d_2} }   =\\
& = 1 + x\frac{a }{q} + x^2\frac{a^2}{(q-1) q} +x^3 \frac{a^3 \left(q^2+1\right)}{(q-1) q^4} + x^4 \frac{a^4 \left(q^5+q^2-1\right) }{(q-1)^2 q^6 (q+1)} + \ldots,
      \label{PC-31-m2}
\end{split}
\end{align}
with ${\alpha \ \beta \choose \beta \ \gamma} = {0 \ -1 \choose -1 \ 1}$, $a_1=a_2=1$, and $(q_1,q_2)=(-1,0)$. The~classical extremal series can be determined either from a~limit of $A$-polynomial or as a~limit
\be
y(x) = \lim_{q\to 1} \frac{F_{\mathbf{3}_1^r}(qx,a,q)}{F_{\mathbf{3}_1^r}(x,a,q)} = 1 + 2 a^2 x^2 + 3 a^3 x^3 + 3 a^4 x^4 + \ldots \label{yi-31r-ex}
\ee

Suppose now that we know only the~above classical series (e.g. from the~analysis of classical $A$-polynomial). We reconstruct (\ref{PC-31-m2}) as follows. First, consider coefficients of  (\ref{yi-31r-ex}):
\[
y_1 = 0, \quad y_2  = 2a^2, \quad y_3  = 3a^3, \quad y_4 = 3a^4, 
\]
and compare them to the~general form (\ref{yx-coeffs-m2}). We observe that:
\begin{itemize}
\item We have $y_1 = 0 \equiv -(-1)^{\alpha} a^{a_1}-(-1)^{\gamma} a^{a_2}$, which means that $\alpha$ and $\gamma$ must be of the~opposite parity and $a_1=a_2$ (so that the~two terms indeed cancel).
\item If we substitute $a_1=a_2$ into the~equation for $y_2$, we get 
\[
y_2=2a^2 \equiv a^{2a_1}(\alpha+\gamma + (-1)^{\alpha+\gamma} (1+2\beta) ) = a^{2a_1}(\alpha+\gamma -1-2\beta ),
\]
where we took into account that $(-1)^{\alpha+\gamma}=-1$ (because $\alpha$ and $\gamma$ have the~opposite parity); we thus conclude that $a_1=a_2=1$ and $\alpha+\gamma-1-2\beta=2$.
\item As $\alpha$ and $\gamma$ have the~opposite parity, assume that $\alpha$ is even and $\gamma$ is odd (which can be justified a~posteriori); then the~expression for the~third coefficient in (\ref{yx-coeffs-m2}) simplifies to
\[
y_3 = 3a^3 \equiv -\frac32 a^3 (\alpha-\gamma)(\alpha+\gamma-1-2\beta) = 3a^3(\gamma-\alpha),   
\]
where in the~last step we used $\alpha+\gamma-1-2\beta=2$ determined above; this last equation implies that $\gamma=\alpha+1$, and using this and $\alpha+\gamma-1-2\beta=2$ we also get $\beta=\alpha-1$.
\item To summarize, so far we have found that $a_1=a_2=1$, $\beta=\alpha-1$, $\gamma=\alpha+1$, and $(-1)^{\alpha+\gamma}=-1$; substituting all this into (\ref{yx-coeffs-m2}), the~formula for $y_4$ reduces to
\[
y_4=3a^4\equiv (8\alpha+3)a^4,
\]
which therefore implies that $\alpha=0$.
\item Ultimately, we have found that
\be
\alpha=0,\quad \beta=\alpha-1=-1,\quad \gamma=\alpha+1 = 1, \quad a_1=a_2=1,    \label{31-ex-C-ai}
\ee
which fixes the~quiver matrix and parameters $a_i$ in (the~first line of) (\ref{PC-31-m2}).
\end{itemize}

So far we have taken advantage of the~classical information. In addition, the~quantum information that fully determines (\ref{PC-31-m2}) is captured by $l_i$, which we can deduce by plugging the~result (\ref{31-ex-C-ai}) into the~first line of (\ref{PC-31-m2}) and comparing initial coefficients with the~known result (assuming we can deduce it independently). In our example it turns out that the~first non-trivial term in the~expansion of (\ref{PC-31-m2}), i.e. $x\frac{a}{q}$, fully fixes $l_i$. Indeed, plugging (\ref{31-ex-C-ai}) into the~first line of (\ref{PC-31-m2}) we get
\[
F_{\mathbf{3}_1^r}(x,a,q) = 1 + ax \frac{q^{l_2}-q^{l_1}}{q-1} + \ldots \equiv 1 + x\frac{a}{q}+\ldots
\]
The~above comparison fixes $l_1=-1$ and $l_2=0$ uniquely. 

To sum up, from the~analysis of $y_2,y_3,y_4$ coefficients in the~classical function $y(x)$, and then the~first non-trivial term in $F_{\mathbf{3}_1^r}(x,a,q)$, we find the~data
\[
{\alpha \ \beta \choose \beta \ \gamma} = {0 \ -1 \choose -1 \ \ \ 1}, \quad a_1=a_2=1, \quad (l_1,l_2)=(-1,0),
\]
which fully determines the~quiver generating series (\ref{PC-31-m2}).


\subsection{Analogues of \texorpdfstring{$d$}{d}-invariants and stabilization}

    When we specialize to $\mathfrak{sl}_2$, $F_K$ has the~form (ignoring prefactors)
    \[
        \sum_{n = 0}^{\infty} f_n(q) x^n.
    \]
    Understanding the~lowest power of $q$ in $f_n (q)$ -- whose $n^2$ coefficient is denoted by $4c$ in~\cite{GM}\footnote{$c = \frac{1}{24}$ for $\mathbf{3}_1^r$, $c = -\frac{1}{24}$ for $\mathbf{3}_1$, and $c = - \frac{1}{16}$ for ${\bf 4_1}$.} -- has several useful applications:

    \begin{itemize}
        \item It tells us when we can apply the~surgery formula in its simplest form~\cite{GM}.
        \item It is an~analogue of $\Delta_a$ which appears as a~prefactor for homological blocks~\cite{GPV, GPPV, GM}:
        \begin{equation*}
            \widehat{Z}(Y; q) \in q^{\Delta_a} \m{Z}[[q]],
        \end{equation*}
        where $a$ is a~choice of spin${}^c$ structure on $Y$.
        \item It tells us by which power of $q$ we should normalize $f_n (q)$ so that the~$q$-expansion starts with $1 + (\ldots)$. This, in turn, is a~key step in exploring stabilization of $f_n (q)$ as $n \to \infty$.
    \end{itemize}
    It turns out that in certain cases it is possible to understand both the~lowest power and stabilisation questions directly from a~quiver. To provide a~concrete example of this, we initially focus on the~reduced $F_K$ invariant for the~figure eight knot.

    We use the~quiver from Equation \eqref{piotrfig8ab}, choosing the~first of the~$3$ options for the~$C$ matrix and setting $t = -1, a =q^2$.
	The~first step is to determine the~growth of the~lowest power of $q$ which we call\footnote{This is closely related to $c$, but $c$ only looks at the~quadratic growth.} $\Delta$. For large $n$, the~$q$ coefficient is dominated by the~quadratic term so we want to solve the~linear programming problem:\footnote{$\Bf{d}$ represents $\frac{\Bf{d}}{n}$ from the~quiver form.}
	\begin{align} \label{eq:linear programing problem}
		& \text{minimise } \Bf{d}\cdot C \cdot \Bf{d} \\ & \text{subject to } 0 \leq \Bf{d}
		\text{ and } \Bf{n}\cdot \Bf{d} = 1. \nonumber
	\end{align}
	Some solution spaces for simple knots are shown in Figure \ref{fig:correction terms structure}.

	\begin{figure}[h]
	    \centering
	    \includegraphics[scale=0.4]{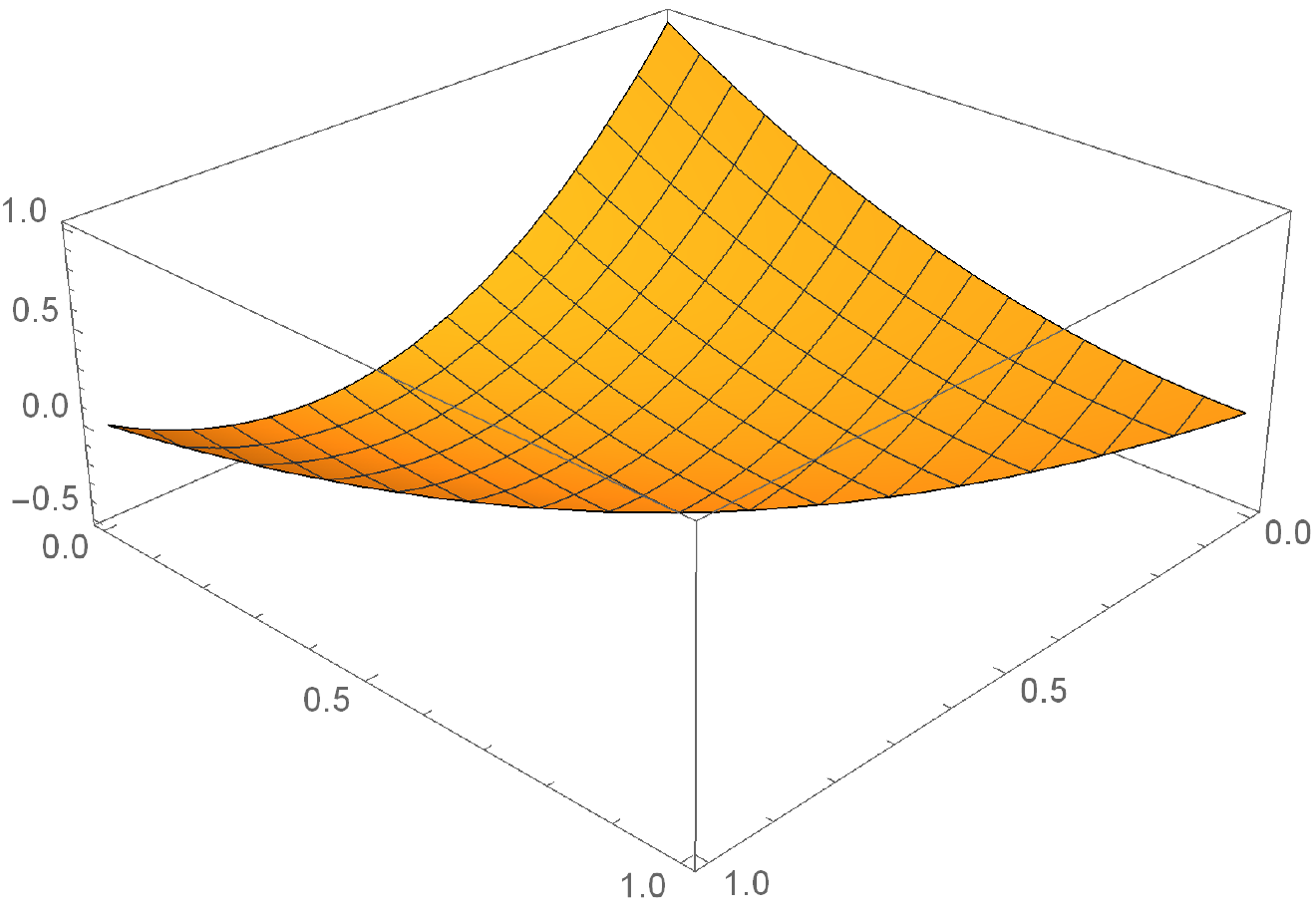}
	    \includegraphics[scale=0.4]{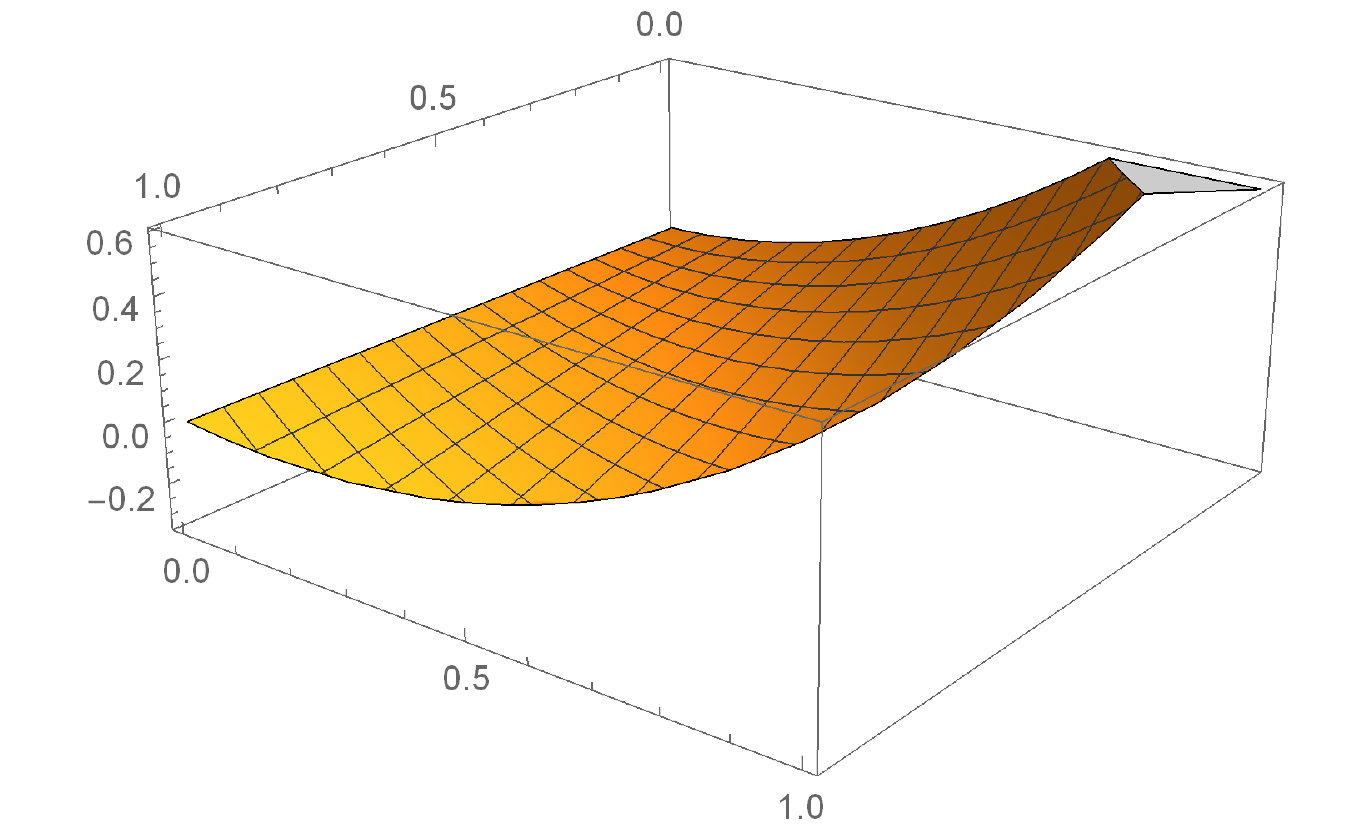}
	    \caption{The~value of the~quadratic form in \eqref{eq:linear programing problem}.\\
	    Left: the~reduced $4_1$ knot with $(d_2, d_3, d_4) \in ([0, 1], [0, 1-d_2], 1 - d_2 - d_3)$.\\
	    Right: the~unreduced $3_1^{r}$ with $(d_1, d_2, d_3) \in ([0, 1], [0, 1-d_1], 1 - d_1 - d_2)$.}
	    \label{fig:correction terms structure}
	\end{figure}
		
	We can often solve this by hand, making use of the~observation that if we can find a~pair $(i, j)$ such that $n_i = n_j$ and $C_{ik} \leq C_{jk}$ for all $k$, then $d_j = 0$. The~idea is simply that if such a~pair exists, it is always advantageous to move weight from $d_j$ to $d_i$.
	
	In the~case of the~$4_1$ knot only $d_2$ and $d_3$ survive and there is a~unique solution
	\begin{equation*}
	    \Bf{d} = \left(0, \frac{1}{2}, \frac{1}{2}, 0, 0, 0
	    \right).
	\end{equation*}
	This splits us into two cases depending on the~parity of $n$. Let us start with the~even parity case, $n = 2i$. Then for large $i$, the~minimal $q$ power of $f_n(q)$ comes from the~vector $\Bf{d} = (0, i, i, 0, 0, 0)$ and is given by
	\[
		\Delta_{even} = \frac{1}{2}\Bf{d}\cdot C \cdot \Bf{d} + \Bf{l}\cdot \Bf{d} = 2i - i^2.
	\]
	Note that the~quadratic term $-i^2$ matches the~result from~\cite{GM} once one accounts for the~differences in notation. Next we want to determine the~even stable series. The~key question we need to answer is for what other vectors $\Bf{d}'$ satisfying $\Bf{n}\cdot (\Bf{d} - \Bf{d}') = 0$ will $\frac{1}{2}\Bf{d}'\cdot C \cdot \Bf{d}' + \Bf{l}\cdot \Bf{d}'$ stay close to $c_{even}$ as $i \to \infty$. This is solved by looking at the~double derivative
	\[
		\frac{d}{d i} \left(\frac{d}{d \Bf{d}} \big(\frac{1}{2}\Bf{d}\cdot C \cdot \Bf{d} + \Bf{l}\cdot \Bf{d}\big)\mid_{\Bf{d} = (0, i, i, 0, 0, 0)}\right) = (0, -1, -1, -1, 1, 0).
	\]
	This means that as long as the~sum $d_2 + d_3 + d_4$ stays constant, $\frac{1}{2}\Bf{d}'\cdot C \cdot \Bf{d}' + \Bf{l}\cdot \Bf{d}'$ will not diverge as we increase $i$. Hence we consider vectors of the~form $(0, i - l, i - m + l, m, 0, 0)$ with $m \geq 0$ and $|l|, |m| << i$. The~quiver term corresponding to a~vector of this form is
	\[
			q^{2i - i^2} \frac{(-1)^m q^{l^2 + \frac{1}{2}(m^2 + m)}}{(q)_{i - l}(q)_{i - m + l}(q)_m} x^i.
	\]
	Taking the~limit as $i$ goes to infinity and summing over $m$ and $l$ gives
	\begin{align*}
		Stable_{even}(q) & = \frac{1}{(q)_\infty^2} \sum_{m = 0}^{\infty} \frac{(-1)^mq^{\frac{1}{2}(m^2 + m)}}{(q)_m} \sum_{l = -\infty}^{\infty} q^{l^2}
		\\ & = \frac{1 + 2q + 2q^4 + 2q^9 + \ldots}{(q)_{\infty}}
		\\ & = 1 + 3 q + 4 q^2 + 7 q^3 + 13 q^4 + 19 q^5 + 29 q^6 + 43 q^7 + 
 62 q^8 + 90 q^9 + \ldots
	\end{align*}
	This can be recognized as the~ratio of Ramanujan theta functions $\psi (q) / f(-q)$ or, equivalently, as a~$q$-series expansion of $q^{1/24} \frac{\eta(q^2)^5}{\eta(q)^3 \eta(q^4)^2}$.
  
	For the~odd case $j = 2i + 1$, the~process is identical. The~optimal vector is either\footnote{In general, there is usually a~difference between between different symmetry breaking choices. See \eqref{eq: trefoil deltas}} $(0, i + 1, i, 0, 0, 0)$ or $(0, i, i + 1, 0, 0, 0)$, which both lead to $c_{odd} = 1 + i - i^2$ and give stable series
	\begin{align*}
		Stable_{odd}(q) & = \frac{1}{(q)_\infty^2} \sum_{m = 0}^{\infty} \frac{(-1)^mq^{\frac{1}{2}(m^2 + m)}}{(q)_m} \sum_{l = -\infty}^{\infty} q^{l^2 - l}
		\\ & = 2\frac{1 + q^2 + q^6 + q^{10} + \cdots}{(q)_{\infty}}
		\\ & = 2(1 + q + 3 q^2 + 4 q^3 + 7 q^4 + 10 q^5 + 17 q^6 + 23 q^7 + 35 q^8 + 48 q^9 + \cdots).
	\end{align*}
    This series can be recognized as a~ratio of Ramanujan theta functions, $\psi (q^2) / f(-q)$ or, equivalently, as a~$q$-series expansion of $q^{-5/24} \frac{\eta(q^4)^2}{\eta(q) \eta(q^2)}$. It also has a~nice enumerative interpretation, as the~generating function of the~number of partitions of $n$ in which each odd part can occur any number of times but each even part is of two kinds and each kind can occur at most once:
    \[
    \frac{1}{\left(q;q^2\right)_{\infty} \left(q^2;q^4\right)_{\infty}^2}
    \; = \;
    \prod_{k=1}^{\infty} \left(1+q^k\right) \left(1+q^{2 k}\right)^2.
    \]
    
    Finally, observe that as both $\Delta_{even}$ and $\Delta_{odd}$ are quadratic in $i$, passing to the~unreduced case -- which, up to an~overall prefactor, corresponds to multiplying by $(1 - q x)$ -- will have no effect on the~stable series and simply apply a~small shift to the~minimal $q$ power. The~reason we make this comment is that in some cases, most notably that of $(2, 2p + 1)$ torus knots, passing to the~unreduced case leads to a~massive simplification in coefficients. For example for the~right-handed trefoil
    \begin{align*}
        F_{\mathbf{3}_1^r}(x,q) & = 1 + \frac{x}{q} + \frac{1 - q}{q^2}x^2 + \frac{1 - q - q^2}{q^3}x^3 + \frac{1 - q - q^2}{q^4}x^4 + \frac{1 - q - q^2 + q^5}{q^4}x^5 + \ldots \\
        F_{\mathbf{3}_1^r}^{unreduced}(x,q) & = 1 - \frac{x^2}{q} - \frac{x^3}{q} + x^5 + q x^6  - q^4x^8 + \ldots
    \end{align*}
    We will return to this example later.
    
    For most purposes, understanding $\Delta$ and the~stable series for the~unreduced case is more useful so we will focus on that case going forward. One immediate issue this presents is that the~quivers we have constructed so far have been all for the~reduced case, as it is not immediately apparent how to convert between an~unreduced and reduced quiver form. The~key observation is that instead of passing to the~full quiver, we can work with an intermediate expressin where
	\begin{equation*}
			F_K^{reduced}(x,q) = \sum_{k, \ldots} (\ldots) (x, q^{-1})_k,
	\end{equation*}
	and then the~unreduced invariant is simply
	\begin{equation*}
			F_K^{unreduced}(x,q) = \sum_{k, \ldots} (\ldots) (qx, q^{-1})_{k + 1},
	\end{equation*}
	and from here we can apply \eqref{eq:new expansion}. Observe that when we apply this method, the~$C$ matrix of the~reduced and unreduced quivers are identical and only the~linear terms change.
	
	Let us study the~left-handed trefoil next. Taking the~$a\to q^2$ limit of the~formula in~\cite{EGGKPS}, we find that the~unreduced $\mathfrak{sl}_2$ invariant is
		\begin{equation}
			F_{3_1}^{unreduced}(x, q)  = \sum_{k = 0}^{\infty} (x q)^k (qx, q^{-1})_{k + 1} = \sum_{i, j, k = 0}^{\infty} \frac{(-1)^i x^{i + j + k} q^{i + 2j + k - ki + \frac{1}{2}(i^2 - i)}}{(q)_i(q)_j}.		\label{eq:unreduced tefoil}
		\end{equation}
		While this is not technically a~quiver form, as there is no $(q)_k$ in the~denominator, this is good enough for our purposes. It can be made into a~quiver form by inserting $1$ and using the~general identity
		\[
			\sum_{\alpha + \beta = k} (-1)^i \frac{q^{\frac{\alpha^2 + \alpha}{2}}}{(q)_{\alpha}(q)_{\beta}} = 1,
		\]
		which gives the~unreduced version of the~quiver in $\eqref{eq:F_3_1}$. Since this does not change any result, we will continue with the~pseudo-quiver form of \eqref{eq:unreduced tefoil}, which can be written as
		\[
		F_{\mathbf{3}_1}^{unreduced}(x, q)  = \sum_{\Bf{d}} \frac{(-q^{\frac{1}{2}})^{\Bf{d}\cdot C \cdot \Bf{d}}q^{\Bf{l}\cdot \Bf{d}} x^{\Bf{n}\cdot \Bf{d}}}{(q)_{d_2}(q)_{d_3}},
		\qquad
		C  = \begin{pmatrix} 0 & -1 & 0 \\ -1 & 1 & 0 \\ 0 & 0 & 0 \end{pmatrix},
		\qquad
		\begin{array}{l}
		  \Bf{n}  = ( 1 , 1 , 1 ),  \\
		  \Bf{l}  = ( 1, \frac{1}{2} , 2 ). 
		\end{array}
		\]
		From here we analyse in an~identical manner. Minimizing the~quadratic term we find that there is a~unique solution $(\frac{2}{3}, \frac{1}{3}, 0)$, so we split the~analysis of $x^j$ into three cases depending on $j \mod 3$ with the~following $\Delta$ values:
		\begin{equation} \label{eq: trefoil deltas}
			\begin{cases}
				j = 3i, \quad \quad \ \ \Bf{d} = (2i, i, 0), & \Delta = \frac{5}{2}i - \frac{3}{2}i^2,
				\\ j = 3i + 1, \quad \Bf{d} = (2i + 1, i, 0),  & \Delta = 1 + \frac{3}{2}i - \frac{3}{2}i^2,
				\\ j = 3i + 2, \quad \Bf{d} = (2i + 1, i + 1, 0), & \Delta = 1 + \frac{1}{2}i - \frac{3}{2}i^2.
			\end{cases}
		\end{equation}
		Note that the~quadratic coefficient $-\frac{3}{2}$ agrees with $c = -\frac{1}{24}$ in~\cite{GM} after aligning notations. Next we compute
		\[
			\frac{d}{d i} \left(\frac{d}{d \Bf{d}} \big(\frac{1}{2}\Bf{d}\cdot C \cdot \Bf{d} + \Bf{l}\cdot \Bf{d}\big)\mid_{\Bf{d} = (2i, i)}\right) = (-1, -1, 0),
		\]
		and so the~three stable series are
		\begin{equation*}
			\begin{cases}
				j = 3i, \quad \quad \ \ \Bf{d} = (2i, i, 0), & \frac{(-1)^{i}}{(q)_{\infty}} \sum_{k = -\infty}^{\infty} (-1)^k q^{\frac{3k^2}{2} + \frac{k}{2}} = (-1)^{i},
				\\ j = 3i + 1, \quad \Bf{d} = (2i + 1, i, 0),  & \frac{(-1)^{i}}{(q)_{\infty}} \sum_{k = -\infty}^{\infty} (-1)^k q^{\frac{3k^2}{2} + \frac{3k}{2}} = 0,
				\\ j = 3i + 2, \quad \Bf{d} = (2i + 1, i + 1, 0), & \frac{(-1)^{i + 1}}{(q)_{\infty}} \sum_{k = -\infty}^{\infty} (-1)^k q^{\frac{3k^2}{2} - \frac{k}{2}} = (-1)^{i + 1}.
			\end{cases}
		\end{equation*}
	The~summation in the~middle is $0$ as $k^2 + k = (-1 - k)^2 + (-1 - k)$. Note that this does not mean that the~coefficients of $x^{3i + 1}$ are $0$. All we can say is that the~minimal $q$ power of $x^{3i + 1}$ grows faster than $-\frac{3}{2}i^2$. Skipping over this for now, we find the~expected structure for large $i$ both in terms of $q$ powers and the~$\pm 1$ stable series.

		It is worth noting that this procedure is highly sensitive to the~initial choice of quiver. For an~example of what can go wrong if one does not start with a~``good'' quiver form, let us study the~right-handed trefoil. We can generate a~quiver form using \eqref{eq:unreduced tefoil}, sending $x \to x^{-1}$ and $q \to q^{-1}$, and then using Weyl symmetry to send $x^{-1} \to q^2x$. From this we get
		\[
			F_{3_1^r}^{unreduced}(x, q) = \sum_{k = 0}^{\infty} (x q)^k (qx, q)_{k + 1} = \sum_{i, j, k = 0}^{\infty} \frac{(-1)^i x^{i + j + k} q^{i + 2j + k + kj + \frac{1}{2}(i^2 - i)}}{(q)_i(q)_j},
		\]
		so the~pseudo-quiver data (with $d_1 = k,\, d_2 = j,\, d_3 = i$) reads
		\[
		F_{\mathbf{3}_1}^{unreduced}(x, q)  = \sum_{\Bf{d}} \frac{(-q^{\frac{1}{2}})^{\Bf{d}\cdot C \cdot \Bf{d}}q^{\Bf{l}\cdot \Bf{d}} x^{\Bf{n}\cdot \Bf{d}}}{(q)_{d_2}(q)_{d_3}},
		\qquad
		C  = \begin{pmatrix} 0 & 1 & 0 \\ 1 & 0 & 0 \\ 0 & 0 & 1 \end{pmatrix},
		\qquad
		\begin{array}{l}
		  \Bf{n}  = ( 1 , 1 , 1 ),  \\
		  \Bf{l}  = ( 1 , 2 , \frac{1}{2} ).
		\end{array}
		\]
		In this case, there are two minima for the~quadratic term, $\Bf{d} = (1, 0, 0)$ and $\Bf{d} = (0, 1, 0)$. We differentiate these by looking at the~linear terms. As $\Bf{l}_1 < \Bf{l}_2$, we conclude that the~the~minimal $q$ power comes from $\Bf{d} = (1, 0, 0)$. Computing the~derivative
		\[
			\frac{d}{d i} \left(\frac{d}{d \Bf{d}} \big(\frac{1}{2}\Bf{d}\cdot C \cdot \Bf{d} + \Bf{l}\cdot \Bf{d}\big)\mid_{\Bf{d} = (i, 0, 0)}\right) = (0, 1, 0),
		\]
		we see that the~stable series will involve summing over $(i - k, 0, k)$ for $k \geq 0$. Substituting this in, we find that the~corresponding quiver term is
		\[
			(-1)^kq^{i + \frac{k^2}{2} - \frac{k}{2}},
		\]
		therefore the~stable series reads
		\[
			\sum_{k = 0}^{\infty} \frac{(-1)^kq^{\frac{k^2}{2} - \frac{k}{2}}}{(q)_k} = (1, q)_{\infty} = 0.
		\]
	    This is to be expected, as for large $i$ it is known that the~$q$ power grows quadratically. It is easy to check that expanding around the~second smallest minimum, $\Bf{d} = (0, 1, 0)$, the~stable series are again $0$. Hence we cannot get an~accurate picture of either the~stable series or large growth in $q$ from this quiver.
	    
	    It is an~interesting question how to find ``good'' quiver forms in general.

\section{Quivers from \texorpdfstring{$R$}{R}-matrices and inverted Habiro series}
\label{sec:Quiver from R-matrices}
In this section we explain how to use $R$-matrices to get quiver expressions of $F_K$ for $\mathfrak{sl}_2$ (or even symmetrically coloured $\mathfrak{sl}_N$) for a~large class of knots. We also explain how to get quiver expressions from the~inverted Habiro series in case of the~trefoil knot and the~figure-eight knot. In a~few explicit examples, we observe that these quiver expressions can be lifted to the~$a$-deformed version. This provides a~yet another method to obtain the~knot complement quivers for the~abelian branch. 

\subsection{Quivers from \texorpdfstring{$R$}{R}-matrices: \texorpdfstring{$\G{sl}_2$}{sl2} case} \label{sec: sl2}
The~\emph{large-colour $R$-matrix} studied in~\cite{Park2, Park3} is the~large-colour limit of the~$R$-matrix for the~$n$-coloured Jones polynomials. Explicitly, they are given by\footnote{Here, $\qbin{n}{k}_q := \frac{[n]_q!}{[k]_q![n-k]_q!}$ where $[n]_q := \frac{1-q^n}{1-q}$.}
\begin{align}
\check{R}(x)_{i,j}^{i',j'} &= \delta_{i+j,i'+j'} q^{\frac{j+j'+1}{2}}x^{-\frac{j+j'+1}{2}} q^{jj'} \qbin{i}{j'}_q \prod_{1\leq l\leq i-j'}(1-q^{j+l}x^{-1}), \label{eq: R matrix sl2}\\
\check{R}^{-1}(x)_{i,j}^{i',j'} &= \delta_{i+j,i'+j'} q^{-\frac{i+i'+1}{2}}x^{\frac{i+i'+1}{2}}q^{-ii'}\qbin{j}{i'}_{q^{-1}} \prod_{1\leq l\leq j-i'}(1-q^{-i-l}x).
\nonumber
\end{align}
The~$R$-matrix and its inverse correspond to the~positive and the~negative crossing as in the~figure below. 
\begin{figure}[ht]
    \centering
    \includegraphics[scale=0.6]{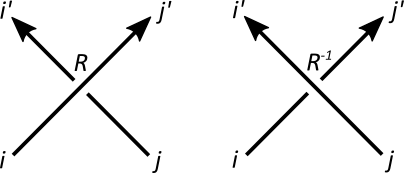}
    \caption{Positive and negative crossings}
    \label{fig:RandRinv}
\end{figure}
Throughout this section, as in Figure \ref{fig:RandRinv}, we use the~indices $i,j,i',j'$ to denote the~segment on the~bottom left, bottom right, top left, and the~top right of the~crossing, respectively. 

Since we can expand $q$-Pochhammer symbols using the~$q$-binomial theorem, it is easy to see that the~$R$-matrix can be written as a~sum of (monomial times) $q$-multinomial coefficients. This observation is very useful in finding a~quiver expression for $F_K$. 
In fact, we have the~following theorem:
\begin{thm} \label{thm: Pos Braid Knot}
For any positive braid knot $K$, there is an~algorithm to produce a~quiver form of $F_K(x,q)$ from the~$R$-matrix state sum. 
\end{thm}
\begin{proof}
We assume that the~reader is familiar with the~$R$-matrix state sum described in~\cite{Park2}. We set the~weight of the~open strand to be $0$. Let $c$ be the~number of crossings. The~number of internal segments is $\frac{4c-2}{2} = 2c-1$. The~number of conditions `$i+j=i'+j'$' is $c$, but only $c-1$ of them are independent. Therefore, there are (at most) $c$ free parameters in choosing the~weights of a~state. Our first step is to find a~nice set of free parameters. 

Let us focus on the~segments that correspond to `$j'$' for some crossing (i.e. the~top right segment of a~crossing). Let $X$ be the~set of such segments. It is in one-to-one correspondence with the~set of crossings. Obviously, $|X|=c$ and $X$ is naturally partitioned into those that are connected to form a~diagonal over-strand (chain of segments) from bottom left to top right. Let us write $X = \sqcup_{\alpha \in I} X_\alpha$ with some index set $I$. In each group $\alpha$, let us say $X_{\alpha} = \{s_1^{\alpha},s_2^{\alpha},\ldots,s_{l_\alpha}^{\alpha}\}$ where $s_1^{\alpha}$ is the~top right segment, $s_2^{\alpha}$ is the~one right below it, etc., and $l_\alpha$ is the~length of the~chain of segments in $X_{\alpha}$; That is, $s_2^{\alpha}$ is `$i$' of a~crossing for which $s_1^{\alpha}$ is `$j'$', and so on. 
Let $n_{s_i^{\alpha}}$ denote the~weight associated to the~segment $s_i^{\alpha}$. (By our definition, each $n_{s_i^{\alpha}}$ is `$j'$' for some crossing.) Clearly $0\leq n_{s_1^{\alpha}} \leq n_{s_2^{\alpha}} \leq \ldots \leq n_{s_{l_\alpha}^{\alpha}}$. The~state sum can be expressed as a~summation over $n_{s}$ for all $s\in X$ (or more precisely, for all $s$ except those whose has to be $0$ because of the~weight $0$ on the~open strand). This is because all the~right-most segments are in $X$, and from there, column by column, all the~other weights are uniquely fixed. Let $m_{s_k^{\alpha}} = n_{s_k^{\alpha}} - n_{s_{k-1}^{\alpha}}$ for $k=2,\ldots,l_{\alpha}$ and $m_{s_1^{\alpha}} = n_{s_1^{\alpha}}$. (See Figure \ref{fig:quivers_from_Rmatrices_proof1} for an~illustration.)
\begin{figure}[ht]
    \centering
    \includegraphics[scale = 0.7]{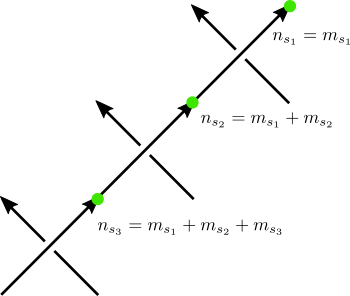}
    \caption{The~preferred parameters $m_{s_k}$ for $\alpha\in I$}
    \label{fig:quivers_from_Rmatrices_proof1}
\end{figure}
With this change of variables, the~state sum can now be expressed as a~summation over $m_s$ with $m_s \geq 0$. It is easy to see that all the~internal weights which are either `$j$' or `$j'$' for some crossing, are $\mathbb{N}$-linear combinations of $\{m_s\}_{s\in X}$. In particular, in the~state sum the~degree of $x$ is bounded above by $-\sum_{s\in X} a_s m_s$ for some $a_s >0$, which in particular implies the~convergence of the~state sum. 

With this description of the~state sum at hand, it is easy to turn it into a~quiver form. Recall that
\begin{align*}
\check{R}(x)_{i,j}^{i',j'} &= \delta_{i+j,i'+j'} q^{\frac{j+j'+1}{2}}x^{-\frac{j+j'+1}{2}} q^{jj'} \qbin{i}{j'}_q\sum_{k=0}^{i-j'}q^{\frac{k(k+1)}{2}}\qbin{i-j'}{k}_q(-q^jx^{-1})^k\\
&= \delta_{i+j,i'+j'} q^{\frac{j+j'+1}{2}}x^{-\frac{j+j'+1}{2}} q^{jj'} \qbin{i}{j'}_q\sum_{k\geq 0}(-1)^k q^{\frac{k(k+1)}{2}+jk}x^{-k}\frac{(q^{i-j'-k+1})_{k}}{(q)_k}.
\end{align*}
The~$q$-binomial coefficient involving $k$ is already in a~good form where we can use Lemma~4.5 of~\cite{KRSS2}, so let us focus on the~other $q$-binomial coefficients, $\qbin{i}{j'}_q$. For each $\alpha\in I$, the~product of the~$q$-binomial coefficients looks like
\[
\qbin{m_{s_1^{\alpha}}+m_{s_2^{\alpha}}}{m_{s_1^{\alpha}}}_q\qbin{m_{s_1^{\alpha}}+m_{s_2^{\alpha}}+m_{s_3^{\alpha}}}{m_{s_1^{\alpha}}+m_{s_2^{\alpha}}}_q \cdots \qbin{i}{m_{s_1^{\alpha}}+\cdots+m_{s_{l_\alpha}^{\alpha}}}_q,
\]
where $i$ denotes the~`$i$' of the~crossing whose `$j'$' is $n_{s_{l_\alpha}^\alpha}$. This product of $q$-binomials can be simplified as
\[
\frac{(q^{i-m_{s_1}-\cdots-m_{s_{l}}+1})_{m_{s_1}+\cdots+m_{s_{l}}}}{(q)_{m_{s_1}}\cdots (q)_{m_{s_{l}}}},
\]
where we have dropped $\alpha$ for simplicity of notation. In this form, the~application of the~Lemma~4.5 finishes the~task.
All in all, we get a~quiver form of $F_K(x,q)$ with at most $4c$ nodes. In practice, it will be smaller than that, because there can be many segments whose weight has to be $0$ because of the~weight $0$ on the~open strand. 
\end{proof}

This theorem can be extended to homogeneous braid knots as well, using the~inverted state sum~\cite{Park3}. 

\begin{thm} \label{thm: Homogenous Braid Knot}
For any homogeneous braid knot $K$, there is an~algorithm to produce a~quiver form of $F_K(x,q)$ from the~inverted state sum. 
\end{thm}
\begin{proof}
We proceed similarly to the~proof of the~previous theorem. Given a~homogeneous braid, we focus on the~segments that correspond to `$j'$' for some positive crossing and the~segments that correspond to `$j$' for some negative crossing. Let $X$ be the~set of such segments. Again, $X$ is naturally partitioned into those that are connected to form a~diagonal over-strand. Let us write $X=\bigsqcup_{\alpha\in I}X_\alpha \sqcup \bigsqcup_{\beta\in J}X_\beta$, where $I$ is an~index set for the~diagonal over-strands from bottom left to top right, and $J$ is an~index set for the~diagonal over-strands from bottom right to top left. If we use the~same notation for the~weights as before, with $s_1$ denoting the~right-most segment, $s_2$ the~second right-most segment and so on,
\[
0\leq n_{s_1^\alpha} \leq \cdots \leq n_{s_{l_\alpha}^\alpha}\quad\text{for any }\alpha\in I
\quad
\text{and}
\quad
0> n_{s_1^\beta} \geq \cdots \geq n_{s_{l_\beta}^\beta} \quad\text{for any }\beta\in J.
\]
Dividing the~braid along vertical lines, we see by a~simple `conservation of weight' argument that the~overall $x^{-1}$-degree of a~given configuration is (ignoring the~constant part)
\[
\sum_{\alpha\in I}\sum_{1\leq k\leq l_\alpha}n_{s_k^\alpha} - \sum_{\beta\in J}\sum_{1\leq k\leq l_\beta}n_{s_k^\beta}.
\]
This shows first that the~inverted state sum converges absolutely, and also that we have a~set of elementary generators
\[
m_{s_k^\alpha} = n_{s_k^\alpha} - n_{s_{k-1}^\alpha},\quad \alpha\in I
\quad
\text{and}
\quad
m_{s_k^\beta} = -(n_{s_k^\beta} - n_{s_{k-1}^\beta}),\quad \beta\in J,
\]
where, by convention, we set $n_{s_0^\alpha} = 0$ and $n_{s_0^\beta} = -1$ (see Figure \ref{fig:quivers_from_Rmatrices_proof2} for an~illustration).
\begin{figure}[ht]
    \centering
    \includegraphics[scale = 0.7]{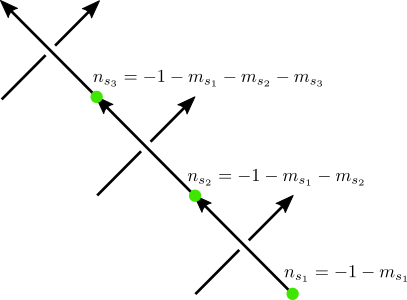}
    \caption{The~preferred parameters $m_{s_k}$ for $\beta\in J$}
    \label{fig:quivers_from_Rmatrices_proof2}
\end{figure}
Now recall that
\begin{align*}
\check{R}(x)_{i,j}^{i',j'} &= \begin{cases}\delta_{i+j,i'+j'} q^{\frac{j+j'+1}{2}}x^{-\frac{j+j'+1}{2}} q^{jj'} \qbin{i}{i-j'}_q \prod_{1\leq l\leq i-j'}(1-q^{j+l}x^{-1}) &
\begin{array}{l}
     \text{if } i\geq j'\geq 0  \\
     \text{or } 0>i\geq j', 
\end{array}
\\ \delta_{i+j,i'+j'} q^{\frac{j+j'+1}{2}}x^{-\frac{j+j'+1}{2}} q^{jj'} \qbin{i}{j'}_q \frac{1}{\prod_{0\leq l\leq j'-i-1}(1-q^{j-l}x^{-1})} &\text{if }j'\geq 0>i, \\ 0 &\text{otherwise},\end{cases}
\end{align*}
and
\begin{align*}
\check{R}^{-1}(x)_{i,j}^{i',j'} &= \begin{cases}\delta_{i+j,i'+j'} q^{-\frac{i+i'+1}{2}}x^{\frac{i+i'+1}{2}} q^{-ii'} \qbin{j}{j-i'}_{q^{-1}} \prod_{1\leq l\leq j-i'}(1-q^{-i-l}x) &
\begin{array}{l}
     \text{if } j\geq i'\geq 0  \\
     \text{or } 0>j\geq i', 
\end{array}
\\ \delta_{i+j,i'+j'} q^{-\frac{i+i'+1}{2}}x^{\frac{i+i'+1}{2}} q^{-ii'} \qbin{j}{i'}_{q^{-1}} \frac{1}{\prod_{0\leq l\leq i'-j-1}(1-q^{-i+l}x)} &\text{if }i'\geq 0>j, \\ 0 &\text{otherwise.}\end{cases}
\end{align*}
We have already analyzed positive crossings with $i\geq j'\geq 0$, so let us take a~look at $j'\geq 0 > i$. In that case
\begin{align*}
\check{R}(x)_{i,j}^{i',j'} &= \delta_{i+j,i'+j'} q^{\frac{j+j'+1}{2}}x^{-\frac{j+j'+1}{2}} q^{jj'} \qbin{i}{j'}_q \sum_{k\geq 0}\qbin{j'-i+k-1}{k}_{q^{-1}}(q^j x^{-1})^k\\
&= \delta_{i+j,i'+j'} q^{\frac{j+j'+1}{2}}x^{-\frac{j+j'+1}{2}} q^{jj'} \qbin{i}{j'}_q \sum_{k\geq 0}q^{(i'+1)k}x^{-k}\frac{(q^{j'-i})_k}{(q)_k}.
\end{align*}
The~$q$-binomial coefficient involving $k$ is already in a~good form where we can use Lemma 4.5. The~other $q$-binomial coefficient $\qbin{i}{j'}_q$ can be handled in the~same way we did in the~proof of the~previous theorem. 

Now that we are done with positive crossings, let us take a~look at negative crossings. When $0>j\geq i'$, we have
\begin{align*}
&\check{R}^{-1}(x)_{i,j}^{i',j'} = \delta_{i+j,i'+j'} q^{-\frac{i+i'+1}{2}}x^{\frac{i+i'+1}{2}} q^{-i'j'} \qbin{j}{j-i'}_{q} \sum_{k=0}^{j-i'}q^{-\frac{k(k+1)}{2}}\qbin{j-i'}{k}_{q^{-1}}(-q^{-i}x)^k\\
&\qquad = \delta_{i+j,i'+j'} q^{-\frac{i+i'+1}{2}}x^{\frac{i+i'+1}{2}} q^{-i'j'} \qbin{j}{j-i'}_{q} \sum_{k=0}^{j-i'}q^{-\frac{(j-i'-k)(j-i'-k+1)}{2}}\qbin{j-i'}{k}_{q^{-1}}(-q^{-i}x)^{j-i'-k}\\
&\qquad= \delta_{i+j,i'+j'} q^{-\frac{i+i'+1}{2}}x^{\frac{j+j'+1}{2}} q^{-i'j'} \qbin{j}{j-i'}_{q} \sum_{k\geq 0}(-1)^{j-i'-k}q^{-\frac{(j-i'-k)(i+j'+k+1)}{2}}x^{-k}\frac{(q^{j-i'-k+1})_k}{(q)_k}.
\end{align*}
The~$q$-binomial coefficient involving $k$ is already in a~good form, so we can focus on the~other one, $\qbin{j}{j-i'}_q$. For each $\beta\in J$, the~product of the~$q$-binomial coefficients look like
\[
\qbin{-1-m_{s_1^\beta}}{m_{s_2^\beta}}_q\qbin{-1-m_{s_1^\beta}-m_{s_2^\beta}}{m_{s_3^\beta}}_q\cdots \qbin{-1-m_{s_1^\beta}-\cdots - m_{s_{l_\beta}^\beta}}{i'}_q,
\]
where $i'$ denotes the~`$i'$' of the~crossing whose `$j$' is $n_{s_{l_\beta}^\beta}$. 
Using the~fact that 
\[
\qbin{-1-n}{k}_q = (-1)^k q^{-nk-\frac{k(k+1)}{2}}\qbin{n+k}{n}_q,
\]
we see that the~product of $q$-binomials can be simplified as (up to an~overall sign and $q$~power)
\[
\qbin{m_{s_1}+m_{s_2}}{m_{s_1}}_q \qbin{m_{s_1}+m_{s_2}+m_{s_3}}{m_{s_1}+m_{s_2}}_q\cdots \qbin{m_{s_1}+\cdots+m_{s_l}+i'}{m_{s_1}+\cdots+m_{s_l}}\\
= \frac{(q^{i'+1})_{m_{s_1}+\cdots+m_{s_l}}}{(q)_{m_{s_1}}\cdots (q)_{m_{s_l}}},
\]
where we have dropped $\beta$ for simplicity of notation. In this form we can apply Lemma 4.5. 

The~last case we need to consider is the~negative crossing with $i'\geq 0 > j$, for which we have
\begin{align*}
& \check{R}^{-1}(x)_{i,j}^{i',j'} = \delta_{i+j,i'+j'} q^{-\frac{i+i'+1}{2}}x^{\frac{i+i'+1}{2}} q^{-i'j'} \qbin{j}{i'}_{q} \frac{(-1)^{i'-j}q^{\frac{(i'-j)(i+j'+1)}{2}}x^{-(i'-j)}}{\prod_{0\leq l\leq i'-j-1}(1-q^{i-l}x^{-1})}\\
&\quad= \delta_{i+j,i'+j'} q^{-\frac{i+i'+1}{2}}x^{\frac{j+j'+1}{2}} q^{-i'j'} \qbin{j}{i'}_{q} \sum_{k\geq 0}(-1)^{i'-j}q^{\frac{(i'-j)(i+j'+1)}{2}} \qbin{i'-j+k-1}{k}_{q^{-1}}(-q^ix^{-1})^k \\
&\quad= \delta_{i+j,i'+j'} q^{-\frac{i+i'+1}{2}}x^{\frac{j+j'+1}{2}} q^{-i'j'} \qbin{j}{i'}_{q} \sum_{k\geq 0}(-1)^{i'-j+k}q^{\frac{(i'-j)(i+j'+1)}{2}+k(j'+1)} x^{-k}\frac{(q^{i'-j})_k}{(q)_k}.
\end{align*}
The~$q$-binomial coefficient involving $k$ is in a~good form where Lemma 4.5 can be applied. The~other $q$-binomial coefficient $\qbin{j}{i'}$ was already considered above. 

All in all, we have shown that for any homogeneous braid knot $K$, there is a~quiver form of $F_K(x,q)$, with the~size of the~quiver being at most $4c$. 
\end{proof}

\subsubsection{Right-handed trefoil} \label{Ex: Trefoil Quiver}
Let us take a~look at the~right-handed trefoil knot as our first example in this section. Take the~braid $\sigma_1^3$. The~reduced $F_{\mathbf{3}_1^r}(x,q)$ can be expressed as a~power series in $x^{-1}$ in the~following way:\footnote{If one wants to work with power series in $x$, one simply needs to replace $x^{-1}$ by $x$, thanks to Weyl symmetry.}
\begin{align*}
F_{\mathbf{3}_1^r}(x,q) &= \sum_{m\geq 0}x^{\frac{1}{2}}q^{-\frac{1}{2}-m}\check{R}(x)_{0,m}^{m,0}\check{R}(x)_{m,0}^{m,0}\check{R}(x)_{m,0}^{0,m}\\
&= x^{-1}q\sum_{m\geq 0}\sum_{k=0}^{m} (-1)^k x^{-m-k}q^{\frac{k(k+1)}{2}}\frac{(q)_m}{(q)_k(q)_{m-k}}\\
&= x^{-1}q\sum_{d_1,d_2\geq 0}(-1)^{d_1}x^{-2d_1-d_2}q^{\frac{d_1(d_1+1)}{2}}\frac{(q)_{d_1+d_2}}{(q)_{d_1}(q)_{d_2}}.
\end{align*}
Using Lemma 4.5 again, we can write
\begin{align*}
F_{\mathbf{3}_1^r}(x,q)&= x^{-1}q \sum_{\alpha_1,\beta_1,\alpha_2,\beta_2\geq 0}(-1)^{\alpha_1+\beta_1}x^{-2(\alpha_1+\beta_1)-(\alpha_2+\beta_2)}q^{\frac{(\alpha_1+\beta_1)(\alpha_1+\beta_1+1)}{2}}\\
&\qquad\qquad\times \frac{(-1)^{\alpha_1+\alpha_2}q^{\frac{\alpha_1(\alpha_1+1)}{2}+\frac{\alpha_2(\alpha_2+1)}{2}+\alpha_2(\alpha_1+\beta_1)}}{(q)_{\alpha_1}(q)_{\beta_1}(q)_{\alpha_2}(q)_{\beta_2}}\\
&= x^{-1}q \sum_{\Bf{d}}(-q^{\frac{1}{2}})^{\Bf{d}\cdot C\cdot \Bf{d}} \frac{x^{\Bf{n}\cdot\Bf{d}}q^{\Bf{l}^{\mathfrak{sl}_2}\cdot \Bf{d}}}{(q)_{\Bf{d}}},
\end{align*}
where
\[
C = 
\begin{pmatrix}
2 & 1 & 1 & 0 \\
1 & 1 & 1 & 0 \\
1 & 1 & 1 & 0 \\
0 & 0 & 0 & 0
\end{pmatrix}, \qquad\qquad \Bf{n} = (-2,-2,-1,-1), \qquad \qquad
\Bf{l}^{\mathfrak{sl}_2} = \left(1,\frac{1}{2},\frac{1}{2},0\right).
\]
Note that we could have used Lemma 4.5 in a~slightly different way, in which case we would have $\alpha_1(\alpha_2+\beta_2)$ instead of $\alpha_2(\alpha_1+\beta_1)$, and the~corresponding quiver matrix would be
\[
C' = 
\begin{pmatrix}
2 & 1 & 1 & 1 \\
1 & 1 & 0 & 0 \\
1 & 0 & 1 & 0 \\
1 & 0 & 0 & 0
\end{pmatrix}.
\]
In general, there are many ways to apply the~Lemma 4.5 which lead to different but equivalent quivers -- for a~thorough study of this phenomenon see~\cite{Jankowski:2021flt}. 

Looking at the~first few coefficients, we can easily upgrade this to the~$a$-deformed version. Up to a~prefactor, we have
\[
F_{\mathbf{3}_1^r}(x,a,q) = \sum_{\Bf{d}}(-q^{\frac{1}{2}})^{\Bf{d}\cdot C\cdot \Bf{d}} \frac{x^{\Bf{n}\cdot\Bf{d}}a^{\Bf{a}\cdot \Bf{d}}q^{\Bf{l}\cdot \Bf{d}}}{(q)_{\Bf{d}}},
\]
where
\[
\Bf{n} = (2,2,1,1),\qquad\qquad
\Bf{a} = (2,1,1,0),\qquad\qquad
\Bf{l} = \left(-3,-\frac{3}{2},-\frac{3}{2},0\right).
\]

\subsubsection{Figure-eight knot}
As our next example, let us take a~look at the~figure-eight knot, given as the~closure of the~braid $\sigma_1\sigma_2^{-1}\sigma_1\sigma_2^{-1}$. 
In terms of the~power series in $x^{-1}$ we have
\begin{align*}
F_{\mathbf{4}_1}(x,q) &= -\sum_{\substack{m\geq 0\\ k<0}}x q^{-1-m-k}\check{R}(x)_{0,m}^{m,0}\check{R}(x)_{m,0}^{0,m}\check{R}^{-1}(x)_{0,k}^{0,k}\check{R}^{-1}(x)_{m,k}^{m,k}\\
&= -x^{-1}\sum_{m,k,a,b\geq 0}x^{-m-2k-a-b}q^{-k^2-ak-bk}\frac{(q)_{m+k+b}}{(q)_{m}(q)_{k}(q)_{b}}\frac{(q)_{k+a}}{(q)_k(q)_a}\\
&\stackrel{\text{lemma 4.5}}{=} -x^{-1}\sum_{\alpha_1,\beta_1,\alpha_2,\beta_2,\alpha_3,\beta_3,\alpha_4,\beta_4\geq 0}x^{-(\alpha_1+\beta_1)-2(\alpha_2+\beta_2)-(\alpha_4+\beta_4)-(\alpha_3+\beta_3)}\\
&\qquad\qquad\times q^{-(\alpha_2+\beta_2)^2-(\alpha_4+\beta_4)(\alpha_2+\beta_2)-(\alpha_3+\beta_3)(\alpha_2+\beta_2)}\\
&\qquad\qquad\times \frac{(-1)^{\alpha_1+\alpha_2+\alpha_3}q^{\frac{\alpha_1(\alpha_1+1)}{2}+\frac{\alpha_2(\alpha_2+1)}{2}+\frac{\alpha_3(\alpha_3+1)}{2}+\alpha_2(\alpha_1+\beta_1)+\alpha_3(\alpha_1+\beta_1+\alpha_2+\beta_2)}}{(q)_{\alpha_1}(q)_{\beta_1}(q)_{\alpha_2}(q)_{\beta_2}(q)_{\alpha_3}(q)_{\beta_3}}\\
&\qquad\qquad\times \frac{(-1)^{\alpha_4}q^{(\alpha_2+\beta_2+1)\alpha_4}q^{\frac{\alpha_4(\alpha_4-1)}{2}}}{(q)_{\alpha_4}(q)_{\beta_4}}\\
&= -x^{-1}\sum_{\Bf{d}}(-q^{\frac{1}{2}})^{\Bf{d}\cdot C\cdot \Bf{d}} \frac{x^{\Bf{n}\cdot\Bf{d}}q^{\Bf{l}^{\mathfrak{sl}_2}\cdot \Bf{d}}}{(q)_{\Bf{d}}},
\end{align*}
where
\[
C = 
\begin{pmatrix}
1 & 0 & 1 & 0 & 1 & 0 & 0 & 0 \\
0 & 0 & 1 & 0 & 1 & 0 & 0 & 0 \\
1 & 1 & -1 & -2 & 0 & -1 & 0 & -1 \\
0 & 0 & -2 & -2 & 0 & -1 & 0 & -1 \\
1 & 1 & 0 & 0 & 1 & 0 & 0 & 0 \\
0 & 0 & -1 & -1 & 0 & 0 & 0 & 0 \\
0 & 0 & 0 & 0 & 0 & 0 & 1 & 0 \\
0 & 0 & -1 & -1 & 0 & 0 & 0 & 0
\end{pmatrix},
\qquad
\begin{array}{l}
     \Bf{n} = (-1,-1,-2,-2,-1,-1,-1,-1),  \\
     \\
     \Bf{l}^{\mathfrak{sl}_2} = \left(\frac{1}{2},0,\frac{1}{2},0,\frac{1}{2},0,\frac{1}{2},0\right). 
\end{array}
\]
Again, there are many different ways to apply Lemma 4.5 which produce slightly different but equivalent quivers. 

Looking at the~first few coefficients, we can easily upgrade this to the~$a$-deformed version. Up to a~prefactor, we have
\[
F_{\mathbf{4}_1}(x,a,q) = \sum_{\Bf{d}}(-q^{\frac{1}{2}})^{\Bf{d}\cdot C\cdot \Bf{d}} \frac{x^{\Bf{n}\cdot\Bf{d}} a^{\Bf{a}\cdot \Bf{d}}q^{\Bf{l}\cdot \Bf{d}}}{(q)_{\Bf{d}}},
\]
where
\begin{align*}
\Bf{n} &= (1,1,2,2,1,1,1,1),\qquad
\Bf{a} = (1,0,1,0,1,0,1,0),\qquad
\Bf{l} = \left(-\frac{3}{2},0,-\frac{3}{2},0,-\frac{3}{2},0,-\frac{3}{2},0\right).
\end{align*}

\subsubsection{\texorpdfstring{$\mathbf{6}_2$}{62} knot}
Using the~braid presentation $\sigma_1\sigma_2^{-1}\sigma_1\sigma_2^{-3}$, we obtain
\begin{align*}
F_{\mathbf{6}_2}(x,q) &= -\sum_{\substack{m\geq 0 \\ k, k_1, k_2 < 0}}x q^{-1-m-k} \check{R}(x)_{0,m}^{m,0}\check{R}(x)_{m,0}^{0,m}\check{R}^{-1}(x)_{0,k}^{0,k}\\
&\qquad\qquad \times \check{R}^{-1}(x)_{m,k}^{m+k-k_1,k_1}\check{R}^{-1}(x)_{m+k-k_1,k_1}^{m+k-k_2,k_2}\check{R}^{-1}(x)_{m+k-k_2,k_2}^{m,k}
\\
&= -q^{-1}x^{-2}\sum_{m,k,k_1,k_2,a,b,c,d\geq 0}(-1)^{k_1+k_2}x^{-m-2k-k_1-k_2-a-b-c-d}\\
&\qquad\qquad \times q^{\frac{k_1(k_1-1)}{2}+\frac{k_2(k_2-1)}{2}-k k_1-k k_2-k_1 k_2-a k_1-b k_2-c k - d k}\\
&\qquad\qquad \times \frac{(q)_{m+k_1+a}(q^{m-k+k_1+1})_{k_2+b}(q^{m-k+k_2+1})_{k+c}(q^{k+1})_{d}}{(q)_{m}(q)_{k}(q)_{k_1}(q)_{k_2}(q)_a(q)_b(q)_c(q)_d}\\
&\stackrel{\text{lemma 4.5}}{=} -q^{-1}x^{-2}\sum_{\Bf{d}}(-q^{\frac{1}{2}})^{\Bf{d}\cdot C\cdot \Bf{d}} \frac{x^{\Bf{n}\cdot\Bf{d}}q^{\Bf{l}^{\mathfrak{sl}_2}\cdot \Bf{d}}}{(q)_{\Bf{d}}},
\end{align*}
where
\[
C = 
\begin{pmatrix}
 1 & 0 & 1 & 0 & 1 & 0 & 1 & 0 & 1 & 0 & 1 & 0 & 1 & 0 & 0 & 0 \\
 0 & 0 & 1 & 0 & 1 & 0 & 1 & 0 & 1 & 0 & 1 & 0 & 1 & 0 & 0 & 0 \\
 1 & 1 & -1 & -1 & -1 & -1 & -1 & 0 & 0 & 0 & -1 & 0 & -1 & -1 & 0 & -1 \\
 0 & 0 & -1 & 0 & -1 & -1 & -2 & -1 & 0 & 0 & -1 & 0 & -1 & -1 & 0 & -1 \\
 1 & 1 & -1 & -1 & 2 & 1 & 0 & -1 & 0 & -1 & 1 & 0 & 0 & 0 & 0 & 0 \\
 0 & 0 & -1 & -1 & 1 & 1 & 0 & -1 & 0 & -1 & 1 & 0 & 0 & 0 & 0 & 0 \\
 1 & 1 & -1 & -2 & 0 & 0 & 2 & 1 & 0 & 0 & 0 & -1 & 1 & 0 & 0 & 0 \\
 0 & 0 & 0 & -1 & -1 & -1 & 1 & 1 & 0 & 0 & 0 & -1 & 1 & 0 & 0 & 0 \\
 1 & 1 & 0 & 0 & 0 & 0 & 0 & 0 & 1 & 0 & 0 & 0 & 0 & 0 & 0 & 0 \\
 0 & 0 & 0 & 0 & -1 & -1 & 0 & 0 & 0 & 0 & 0 & 0 & 0 & 0 & 0 & 0 \\
 1 & 1 & -1 & -1 & 1 & 1 & 0 & 0 & 0 & 0 & 1 & 0 & 0 & 0 & 0 & 0 \\
 0 & 0 & 0 & 0 & 0 & 0 & -1 & -1 & 0 & 0 & 0 & 0 & 0 & 0 & 0 & 0 \\
 1 & 1 & -1 & -1 & 0 & 0 & 1 & 1 & 0 & 0 & 0 & 0 & 1 & 0 & 0 & 0 \\
 0 & 0 & -1 & -1 & 0 & 0 & 0 & 0 & 0 & 0 & 0 & 0 & 0 & 0 & 0 & 0 \\
 0 & 0 & 0 & 0 & 0 & 0 & 0 & 0 & 0 & 0 & 0 & 0 & 0 & 0 & 1 & 0 \\
 0 & 0 & -1 & -1 & 0 & 0 & 0 & 0 & 0 & 0 & 0 & 0 & 0 & 0 & 0 & 0
\end{pmatrix},
\]
and
\begin{align*}
\Bf{n} &= (-1,-1,-2,-2,-1,-1,-1,-1,-1,-1,-1,-1,-1,-1,-1,-1),\\
\Bf{l}^{\mathfrak{sl}_2} &= \left(\frac{1}{2}, 0, \frac{1}{2}, 0, 0, -\frac{1}{2}, 0, -\frac{1}{2}, \frac{1}{2}, 0, \frac{1}{2}, 0, \frac{1}{2}, 0, \frac{1}{2}, 0\right).
\end{align*}
Upgrading this to the~$a$-deformed version, 
\begin{align*}
\Bf{a} &= (1,0,0,1,1,0,0,0,1,0,1,0,0,0,1,0),\\
\Bf{l} &= \left(-\frac{3}{2}, 0, \frac{1}{2}, -2, -2, -\frac{1}{2}, 0, -\frac{1}{2}, -\frac{3}{2}, 0, -\frac{3}{2}, 0, \frac{1}{2}, 0, -\frac{3}{2}, 0\right).
\end{align*}

\subsection{Quivers from \texorpdfstring{$R$}{R}-matrices: \texorpdfstring{$\G{sl}_N$}{slN} case} \label{sec: slN}

In order to extend~\cite{Park2, Park3} and the~discussion in Section \ref{sec: sl2} to symmetrically coloured $\G{sl}_N$, we need to find an~expression for the~$R$-matrix of the~$k$-th symmetric representation of $U_q(\G{sl}_N)$ similar to \eqref{eq: R matrix sl2}. In order to do this, we first briefly study $U_q(\G{sl}_N)$ and its symmetric representations. 

Recall that $\G{sl}_N$ is the~Lie algebra corresponding to the~$A_{N - 1}$ root system. Using the~usual dictionary order on $\m{C}^N$, the~simple roots of $A_{N - 1}$ are elements $e_i - e_{i + 1} \in \m{C}^N$ where $e_i$ is the~the~vector with $i$-th entry $1$ and all other entries $0$. The~Cartan matrix is
\[    A_{ij} = \begin{cases}
        2 & i = j,
        \\ -1 & |i - j| = 1,
        \\ 0 & |i - j| = 1,
    \end{cases}
\]
and both $\G{sl}_N$ and $U_q(\G{sl}_N)$ are constructed from this via the~Chevalley-Serre relations and their $q$-analogue. For each simple root $\alpha_i = e_i - e_{i + 1}$, let $\{X_i^+, X_i^-, K_i = q^{\frac{H_i}{2}}\}$ denote the~quantum $\G{sl}_2$ triple in $U_q(\G{sl}_N)$.

The~$k$-th symmetric representation has a~basis labelled\footnote{The~variables $a_0$ and $a_N$ are introduced for convenience.} by $\textbf{a} = (a_1, \cdots, a_{N - 1})$ with    $k = a_0 \geq a_1 \geq \cdots \geq a_{N-1} \geq a_N = 0$,
where the~generators act as
\[
	X_i^+ \ket{\textbf{a}}  = [a_i - a_{i+1}]_q \ket{\textbf{a} - e_i},
	\quad X_i^- \ket{\textbf{a}}  = [a_{i - 1} - a_i]_q \ket{\textbf{a} + e_i},
	\quad K_i \ket{\textbf{a}}  = q^{a_{i - 1} + a_{i + 1} - 2 a_i} \ket{\textbf{a}}.
\]
We leave out the~proof that this is actually a~representation as this is an~easy but tedious computation.
A general form for $U_q(\G{sl}_N)$ $R$ matrices is given in~\cite{Bur} and we can find the~corresponding specialisation to the~$k$-th symmetric representation using the~definitions of the~representation above.

The~$R$-matrix is an~infinite summation over $\textbf{r} = r_i^j$ with $1 \leq i \leq j \leq N - 1$ and each $r_i^j$ is a~non-negative integer. Define $\textbf{r}_{(i, j)}^{(k, l)}$ to be the~vector $(r_i^k, \cdots, r_i^l, r_{i + 1}^k, \cdots, r_j^l)$ and denote
	\[
		\textbf{r}^j  = \textbf{r}_{(1, j)}^{(j, j)} = (r_1^j, \cdots, r_j^j),
		\qquad\qquad \textbf{r}_j  = \textbf{r}_{(j, j)}^{(j, N - 1)} = (r_j^j, \cdots, r_j^{N - 1}).
	\]
	Finally, let $|\cdot |$ denote the~$l^1$ norm which, as our vectors are non-negative, is simply the~sum of the~entries and replace all appearances of $q^k$ by $x$. Then the~$R$-matrix is (up to a~constant prefactor):
	\begin{align*}
		R_{\G{sl}_N} \ket{\textbf{a}, \textbf{b}} & = \sum_{\textbf{r} > 0} \frac{(-1)^{|\textbf{r}|}q^{C_N} x^{-\frac{1}{2}(a_1 + b_1 + |\textbf{r}_1|)}(x q^{-b_1}; q^{-1})_{|\textbf{r}_1|}(q^{a_1 - a_2 + |\textbf{r}_2|}; q^{-1})_{|\textbf{r}^1|}}{(q)_{\textbf{r}}}
		\\ & \quad \quad \times \prod_{j = 2}^{N - 1} (q^{b_{j - 1} - b_j}; q^{-1})_{|\textbf{r}_j|}(q^{a_j - a_{j + 1} + |\textbf{r}_{j + 1}|}; q^{-1})_{|\textbf{r}^j|} \ket{\textbf{b}', \textbf{a}'},
	\end{align*}
	where
	\begin{align*}	
		C_N & = \frac{1}{2}\textbf{r} \cdot \textbf{r} + \textbf{a}\cdot M \cdot\textbf{b} +
		 \sum_{j = 1}^{N - 1} \frac{1}{4}|\textbf{r}^j|(a_{j + 1} + b_{j + 1} - 2) - \frac{1}{4} \left(\sum_{i = 2}^j r_i^j(a_{i - 1} + b_{i - 1})\right)
		\\ & \qquad + \sum_{i = 1}^j r_i^j\left(|\textbf{r}_{(i+1, j)}^{(j, N-1)}| + \frac{3}{4}(a_i - a_j) - \frac{1}{4}(b_i - b_j)\right),
		\\ M_{ij} & = \begin{cases}
			1 & i = j
			\\ -\frac{1}{2} & |i - j| = 1
			\\ 0 & \text{else.}
		\end{cases}, 
		\qquad a_i'  = a_i - |\textbf{r}_{(1, i)}^{(i, N - 1)}|,
		\qquad b_i'  = b_i + |\textbf{r}_{(1, i)}^{(i, N - 1)}|.
	\end{align*}
	While this formula is complicated, it is manageable for small values of $N$ and in certain cases can be studied with $N$ kept generic. It will produce a~series in $q$ and $x^{-1}$, as in order for the~summand to be nonzero, we require $a_1 > |\textbf{r}_1|$.
	The~key point is that all $q$ powers are quadratic, so for any positive braid knot this will give a~quiver form in an~identical manner to Theorem \ref{thm: Pos Braid Knot}. The~main difference is that all our variables and relations are $N - 1$ dimensional vectors and we have a~collection of extra free variables coming from the~summations over $\textbf{r}$'s. There is also an~analogue of Theorem \ref{thm: Homogenous Braid Knot} and an~extension of the~inverse state sum method to $\G{sl}_N$ but we do not discuss it here.
	
	Finally, in order to match the~above description up with \eqref{eq: R matrix sl2}, we define $R_{\textbf{a}, \textbf{b}}^{\textbf{b}', \textbf{a}'}$ by
\begin{equation}
		R_{\textbf{a}, \textbf{b}}^{\textbf{b}', \textbf{a}'} = \bra{\textbf{b}', \textbf{a}'} R \ket{\textbf{a}, \textbf{b}}.
\label{Rmatelement}
\end{equation}
	
	\subsubsection{The~Trefoil}

	The~trefoil knot is given by the~closure of the~braid $\sigma_1^3$. We colour the~top and bottom of the~leftmost strand by $\textbf{a} = \textbf{0}$ and close off the~right hand strand. This means that we need to take the~quantum trace over all possible labels $\textbf{b}$. Analysing possibilities, we actually have no other free variables, as labels must always decrease going left to right over a~positive crossing. Hence there are three $R$-matrices we need to consider. Going from bottom to top they are
	\[
		R_{\textbf{0}, \textbf{b}}^{\textbf{b}, \textbf{0}}, \qquad\qquad R_{\textbf{b}, \textbf{0}}^{\textbf{b}, \textbf{0}}, \qquad\qquad R_{\textbf{b}, \textbf{0}}^{\textbf{0}, \textbf{b}}.
	\]
	It turns out that for each of these cases, the~summation over $\textbf{r}$ collapses to a~single value. For $R_{\textbf{0}, \textbf{b}}^{\textbf{b}, \textbf{0}}$ and $R_{\textbf{b}, \textbf{0}}^{\textbf{0}, \textbf{b}}$, the~only nonzero terms corresponds to $\textbf{r} = 0$ and so
	\[
		R_{\textbf{0}, \textbf{b}}^{\textbf{b}, \textbf{0}} = x^{\frac{-b_1}{2}},
		\qquad\qquad
		R_{\textbf{b}, \textbf{0}}^{\textbf{0}, \textbf{b}}  = x^{\frac{-b_1}{2}}.
	\]
	Similarly, for $R_{\textbf{b}, \textbf{0}}^{\textbf{b}, \textbf{0}}$, we have a~collection of $q$-Pochhammers of the~form $(1; q^{-1})_{|\textbf{r}_j|}$ for $j > 1$. These are $0$ unless $|\textbf{r}_j| = 0$ so the~only $r_i^j$'s which can be non zero are the~$r_1^j$'s. Then the~relation between $\textbf{b}$ and $\textbf{b}'$ forces $r_1^j = b_j - b_{j+1}$. After fixing this, we find that most of the~$q$-Pochhammers are either $1$ or cancel, and thus
	\[
	R_{\textbf{b}, \textbf{0}}^{\textbf{b}, \textbf{0}} = (-1)^{b_1} q^{C_N} x^{-b_1}(x; q^{-1})_{b_1},
	\]
	where
	\begin{align*}
		C_N & = \sum_{j = 1}^{N - 1} \frac{1}{2} (b_j - b_{j + 1})^2 + \frac{1}{4}(b_j - b_{j + 1})(b_{j + 1} - 2) + \frac{3}{4}(b_j - b_{j + 1})(b_1 - b_j)
		\\ & = \sum_{j = 1}^{N - 1} \frac{1}{4}\big(3b_1(b_j - b_{j + 1}) - 2(b_j - b_{j + 1}) - (b_j^2 - b_{j + 1}^2)\big) = \frac{b_1^2 - b_1}{2}.
	\end{align*}
Finally, we need to deal with the~quantum trace over $\textbf{b}$. In the~general $U_q(\G{sl}_N)$ case, we find that the~trace factor will be $q^{-|\textbf{b}| - \frac{N k}{2}}$. We drop the~term $-\frac{N k}{2}$ in the~exponent (as it would correspond to a~simple prefactor), which leaves us with the~general formula for the~trefoil:
	\[
		F^{\mathfrak{sl}_N}_{3_1}(x,q) = \sum_{\textbf{b}} (-1)^{b_1} q^{\frac{1}{2}(b_1 - 1)b_1} q^{-|\textbf{b}|} x^{-2b_1} (x; q^{-1})_{b_1}.
	\]
We can switch from the~series in $x^{-1}$ to the~one in $x$ using Weyl symmetry $x^{-1} \to a x$:
	\[
		F^{\mathfrak{sl}_N}_{3_1}(x,q) = \sum_{\textbf{b}} (-1)^{b_1} q^{\frac{1}{2}(b_1 - 1)b_1} q^{-|\textbf{b}|} a^{2b_1} x^{2b_1} (a^{-1} x^{-1}; q^{-1})_{b_1}.
	\]
In order to transform this back into the~usual quiver form, we first perform the~summations over the~$b_i$ variables for $i > 1$. Recalling that $b_1 \geq b_2 \geq \cdots \geq b_{N - 1} \geq 0$, we find that
	\[
		\sum_{\textbf{b}} q^{-|\textbf{b}|} = \sum_{b_1 = 0}^{\infty} q^{b_1} a^{-b_1} \frac{(q^{b_1})_{N-2}}{(q)_{N-2}},
	\]
	and so
	\[
		F^{\mathfrak{sl}_N}_{3_1}(x,q) = \sum_{b_1 = 0}^{\infty} (-1)^{b_1} q^{\frac{1}{2}(b_1 + 1)b_1} a^{b_1} x^{2b_1} \frac{(q^{b_1})_{N-2}(a^{-1} x^{-1}; q^{-1})_{b_1}}{(q)_{N-2}}.
	\]
	Next we can expand the~$q$-Pochhammer $(a^{-1} x^{-1}; q^{-1})_{b_1}$ as
	\[
		(a^{-1} x^{-1}; q^{-1})_{b_1} = (a^{-1} x^{-1} q^{1 - b_1}; q)_{b_1} = \sum_{i = 0}^{b_1} (-1)^i q^{\frac{i(i - 1)}{2}} \frac{(q)_b}{(q)_i(q)_{b - i}} x^{-i} a^{-i} q^{i - b_1 i}.
	\]
	Substituting this in and letting $b_1 = i + j$, we get
	\[
		F^{\mathfrak{sl}_N}_{3_1}(x,q) = \sum_{i, j}^{\infty} (-1)^{j} q^{\frac{1}{2}(2i + j + j^2)} a^j x^{i + 2j} \frac{(q)_{i + j + N - 2}}{(q)_{N-2}(q)_{i}(q)_{j}}.
	\]
	To proceed, we make use of the~general identity
	\[
		\sum_{\alpha + \beta = d} (-1)^i \frac{q^{\frac{\alpha^2 + \alpha}{2}}}{(q)_{\alpha}(q)_{\beta}} = 1.
	\]
	This allows us to use Lemma 4.5 from~\cite{KRSS2} to get
	\begin{align*}
		F^{\mathfrak{sl}_N}_{3_1}(x,q) & = \sum_{\alpha_1, \alpha_2, \beta_1, \beta_2} \frac{(-1)^{\alpha_2 + \beta_2} q^{\frac{1}{2}(2(\alpha_1 + \beta_1) + (\alpha_2 + \beta_2) + (\alpha_2 + \beta_2)^2)} a^{\alpha_2 + \beta_2} x^{(\alpha_1 + \beta_1) + 2(\alpha_2 + \beta_2)}}{(q)_{\alpha_1}(q)_{\alpha_2}(q)_{\beta_1}(q)_{\beta_2}}
		\\ & \qquad \qquad \times (-q)^{\alpha_1 + \alpha_2} q^{\frac{1}{2}(\alpha_1^2 + \alpha_2^2)} q^{(N - 2)\alpha_1 + \alpha_2(N - 2 + \alpha_1 + \beta_1)} q^{-\frac{1}{2}(\alpha_1 + \alpha_2)}
		\\ & = \sum_{\Bf{d}=(d_1, d_2, d_3, d_4)} \frac{(-1)^{d_1 + d_4} q^{\frac{1}{2}\Bf{d}\cdot C\cdot\Bf{d} + \frac{1}{2}(-d_1 - 2d_1 + 2d_3 + d_4)} a^{d_1 + 2d_2 + d_4} x^{(d_1 + d_3) + 2(d_2 + d_4)}}{(q)_{d_1}(q)_{d_2}(q)_{d_3}(q)_{d_4}}
		\end{align*}
		with
		\[
		C  = \begin{pmatrix}
			1 & 1 & 0 & 0 \\
			1 & 2 & 1 & 1 \\
			0 & 1 & 0 & 0 \\
			0 & 1 & 0 & 1
		\end{pmatrix}.
	    \]
	Up to a~conventional choice (which can be fixed by $x \to q x$) and reordering variables, this is exactly the~$a$-deformed quiver from Section \ref{Ex: Trefoil Quiver}.

	In principle, this can be applied to other positive braid knots, though the~analysis will be more complicated for generic $N$. Fixing a~small value of $N$ though this will easily produce $F^{\G{sl}_3}_K$, $F^{\G{sl}_4}_K$, $F^{\G{sl}_5}_K,\ldots$ invariants and quiver forms for this class of knots.

\subsection{\texorpdfstring{$R$}{R}-matrices and braiding from brane configurations}

The~matrix elements \eqref{eq: R matrix sl2} and \eqref{Rmatelement} have a~natural physical interpretation in the~brane setup of Section \ref{sec:Brane constructions}. If a~knot (link) $K$ is represented as the~closure of a~braid $\beta$, each elementary transformation (crossing) represented by a~generator of the~braid group corresponds to an~elementary interface (domain wall) $\mathcal{D}_i$ in the~direction along which which the~braid is stretched. Following~\cite{CGR,GNSSS}, we parametrize this direction by $\sigma \in S^1$. When both the~rank and the~dimensions of representations colouring the~strands are finite-dimensional, the~interfaces can be described by matrix factorizations in topological Landau-Ginzburg models (see Figure~23 in~\cite{CGR} or Figure~9 in~\cite{GNSSS}).

Our setup in this paper is similar, except that strands carry infinite-dimensional representations, namely the~Verma modules of $U_q(\G{sl}_N)$. In the~brane configuration, this corresponds to replacing M2-branes -- that produce strands coloured by finite-dimensional representations -- with M5-branes. The~two are related: specializing $x = q^n$ corresponds to the~partial Higgsing implemented in the~brane geometry by a~version of the~Hanany-Witten effect that creates $n$ M2-branes stretched between M5 and M5$'$ branes (see loc. cit. and references therein). Since in this paper we are mostly interested in infinite-dimensional representations with arbitrary complex weights $\log x_i$, the~relevant brane configuration \eqref{Mdeformed} involves two sets of fivebranes: $N$ M5-branes that produce $SL(N, \mathbb{C})$ Chern-Simons theory on $S^3$, and additional M5$'$-branes supported on $L_K$ that produce the~strands of $K = L_K \cap S^3$ coloured by infinite-dimensional representations with highest weights $x_i \in \mathbb{C}^*$, $i=1, \ldots, \rho$.\footnote{In a~sense, the~most generic case is when $\rho = N$. In this case, the~two sets of fivebranes in \eqref{Mdeformed} can be reconnected in single smooth configuration supported on the~Lagrangian submanifold
$$
(S^3 \setminus \nu K) \cup_{T^2} (L_K \setminus \nu K)
$$
where, as above, $K = L_K \cap S^3$.}

In order to see how the~$R$-matrices and braiding are realized from the~perspective of 3d-3d correspondence, we need to focus on the~$\mathbb{R}^2 \times S^1$ part of the~spacetime in \eqref{Mdeformed}. This is the~three-dimensional space where 3d $\mathcal{N}=2$ theory $T[M_3]$ ``lives.'' Many aspects of this 3d $\mathcal{N}=2$ theory are understood quite well. It contains $T[S^3]$, which is a~relatively simple theory. The~additional degrees of freedom come a~codimension-two BPS defect in the~6d $(0,2)$ theory on M5-branes. According to \eqref{Mdeformed}, the~support of the~6d $(0,2)$ theory is $\mathbb{R}^2 \times S^1 \times S^3$, whereas the~support of a~codimension-two BPS defect is $\mathbb{R}^2 \times S^1 \times K$.\footnote{Note that $L_K$ is non-compact. If it was compact, 3d $\mathcal{N}=2$ theory would also contain a~non-trivial sector associated with the~modes of 6d $(0,2)$ theory compactified on $L_K$.} As before, we consider $K$ represented by a~closure of a~braid stretched, say, along the~great circle of $S^3$ (see Figure~23 in~\cite{CGR} or Figure~9 in~\cite{GNSSS}). Then the~two sets of the~fivebranes in \eqref{Mdeformed} share a~total of four dimensions: the~three dimensions of $\mathbb{R}^2 \times S^1$ and one dimension of $S^1 \subset S^3$, parametrized by $\sigma$.

\begin{figure}[ht]
	\centering
	\includegraphics[width=3.0in]{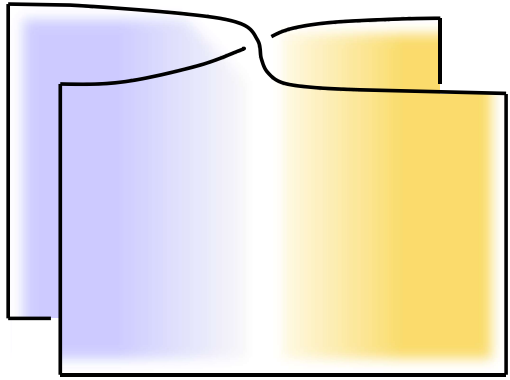}
	\caption{A half-BPS interface interpolating between two vacua of the~theory on a~stack of branes that represents the~$R$-matrix / braiding. Evaluating the~partition function of the~combined system in the~Omega-background along $\mathbb{R}^2_q$ gives the~matrix elements \eqref{eq: R matrix sl2} and \eqref{Rmatelement}.}
	\label{fig:Rinterface}
\end{figure}

The~degrees of freedom of this 4-dimensional theory are precisely the~degrees of freedom of the~codimension-two defect in 6d $(0,2)$ theory. The~3d $\mathcal{N}=2$ theory $T[M_K]$ is the~result of compactifying this 4d theory on the~great circle $S^1 \subset S^3$. This compactification is not entirely trivial, though, because the~strands of $K$ are braided as $\sigma \in S^1$ is traversed. One can imagine that the~4d theory is piecewise constant and undergoes sharp transitions where nontrivial braiding occurs (see Figure~\ref{fig:Rinterface}). These are precisely the~codimension-one interfaces representing the~$R$-matrices, and the~theory $T[M_K]$ is basically a~composition of these interfaces.

The~individual interfaces representing braiding and $R$-matrices preserve 4 real supersymmetries, just like the~3d $\mathcal{N}=2$ theory $T[M_K]$ itself. Along the~$\mathbb{R}^2$ directions of $\mathbb{R}^2 \times S^1$ there is Omega-background with the~equivariant parameter $q$. Therefore, the~matrix elements \eqref{eq: R matrix sl2} and \eqref{Rmatelement} have a~physical interpretation in terms of $\mathbb{R}^2_q \times S^1 \times \mathbb{R}$ partition function of the~4-dimensional defect in 6d $(0,2)$ theory with two different vacua at $+\infty$ and $-\infty$ along the~$\mathbb{R}$-direction. One can reduce this problem to the~study of interfaces in a~3d QFT by compactifying on the~$S^1$ or a~circle in $\mathbb{R}^2_q$, such that $\mathbb{R}^2_q / U(1) \cong \mathbb{R}_+$.

\subsection{Quivers from inverted Habiro series}

There is a~slightly different way to get quivers for the~trefoil and the~figure-eight knot. That is to use the~\emph{inverted Habiro expansion} studied in~\cite{Park3}. 
As briefly explained in Section \ref{subsubsec:52knot}, inverted Habiro expansion of $F_K$ is the~series 
\[
F_K(x,q) = -\sum_{m=1}^{\infty}\frac{a_{-m}(K)}{\prod_{j=0}^{m-1}(x+x^{-1}-q^j-q^{-j})},
\]
where $a_{-m}(K)$ denotes a~natural extension of Habiro's cyclotomic coefficients to negative direction. When the~knot is either the~trefoil or the~figure-eight knot, $a_m(\mathbf{4}_1) = 1$ and $a_m(\mathbf{3}_1^r) = (-1)^m q^{\frac{m(m+3)}{2}}$, so they can be extended to negative $m$ in a~straightforward manner. 
We will use these expressions to find the~corresponding quiver forms. 

\subsubsection{Right-handed trefoil}
As a~power series in $x$, 
\begin{align*}
F_{\mathbf{3}_1^r}(x,q) &= q\sum_{n\geq 0}\frac{(-1)^n q^{-\frac{n(n-1)}{2}}}{\prod_{j=0}^{n}(x+x^{-1}-q^j-q^{-j})} = \frac{q x}{1-x}\sum_{i,k\geq 0}(-1)^i q^{-\frac{i(i-1)}{2}}q^{-k i}\qbin{2i+k}{k}_q x^{i+k}\\
&= \frac{q x}{1-x}\sum_{i,k_1,k_2\geq 0} \frac{(-1)^{i+k_1} q^{-\frac{i(i-1)}{2}+\frac{k_1(k_1+1)}{2}+(k_1-k_2)i}}{(q)_{k_1}(q)_{k_2}} x^{i+k_1+k_2} \\
&= qx \sum_{i_1,i_2,j_1,j_2,k_1,k_2\geq 0}\frac{(-1)^{i_2+j_1+k_1}q^{-\frac{i_2(i_2-1)}{2}-i_1i_2+i_1+\frac{j_1(j_1+1)}{2}+\frac{k_1(k_1+1)}{2}+(i_1+i_2)(k_1-k_2)}}{(q)_{i_1}(q)_{i_2}(q)_{j_1}(q)_{j_2}(q)_{k_1}(q)_{k_2}}\\
&\qquad\qquad\qquad\qquad\times x^{i_1+i_2+j_1+j_2+k_1+k_2}. 
\end{align*}
This can be written as
\[
F_{\mathbf{3}_1^r}(x,q) = qx \sum_{\Bf{d}}(-q^{\frac{1}{2}})^{\Bf{d}\cdot C\cdot \Bf{d}} \frac{x^{\Bf{n}\cdot\Bf{d}}q^{\Bf{l}^{\mathfrak{sl}_2}\cdot \Bf{d}}}{(q)_{\Bf{d}}},
\]
where
\[
C = 
\begin{pmatrix}
0 & -1 & 0 & 0 & 1 & -1 \\ 
-1 & -1 & 0 & 0 & 1 & -1 \\
0 & 0 & 1 & 0 & 0 & 0 \\
0 & 0 & 0 & 0 & 0 & 0 \\
1 & 1 & 0 & 0 & 1 & 0 \\
-1 & -1 & 0 & 0 & 0 & 0
\end{pmatrix},
\qquad
\begin{array}{l}
     \Bf{n} = (1,1,1,1,1,1),  \\
     \\
     \Bf{l}^{\mathfrak{sl}_2} = \left(-1,\frac{1}{2},-\frac{3}{2},0,-\frac{3}{2},0\right).
\end{array}
\]
Upgrading this to the~$a$-deformed version, 
\[
\Bf{a} = (1,0,1,0,1,0), \qquad \Bf{l} = \left(-1,\frac{1}{2},-\frac{3}{2},0,-\frac{3}{2},0\right). 
\]

\subsubsection{Figure-eight knot}
As a~power series in $x$, 
\begin{align*}
F_{\mathbf{4}_1}(x,q) &= -\sum_{n\geq 0}\frac{1}{\prod_{j=0}^{n}(x+x^{-1}-q^j-q^{-j})} = -\frac{x}{1-x}\sum_{i,k\geq 0}q^{-k i}\qbin{2i+k}{k}_q x^{i+k}\\
&= -\frac{x}{1-x}\sum_{i,k_1,k_2\geq 0}  \frac{(-1)^{k_1} q^{\frac{k_1(k_1+1)}{2}+(k_1-k_2)i}}{(q)_{k_1}(q)_{k_2}} x^{i+k_1+k_2}\\
&= -x\sum_{i_1,i_2,j_1,j_2,k_1,k_2\geq 0}\frac{(-1)^{i_1+j_1+k_1}q^{\frac{i_1(i_1+1)}{2}+\frac{j_1(j_1+1)}{2}+\frac{k_1(k_1+1)}{2}+(i_1+i_2)(k_1-k_2)}}{(q)_{i_1}(q)_{i_2}(q)_{j_1}(q)_{j_2}(q)_{k_1}(q)_{k_2}}\\
& \qquad\qquad\qquad\qquad\times x^{i_1+i_2+j_1+j_2+k_1+k_2}\\
&= -x\sum_{\Bf{d}}(-q^{\frac{1}{2}})^{\Bf{d}\cdot C\cdot \Bf{d}} \frac{x^{\Bf{n}\cdot\Bf{d}}q^{\Bf{l}^{\mathfrak{sl}_2}\cdot \Bf{d}}}{(q)_{\Bf{d}}},
\end{align*}
where
\[
C = 
\begin{pmatrix}
1 & 0 & 0 & 0 & 1 & -1 \\ 
0 & 0 & 0 & 0 & 1 & -1 \\
0 & 0 & 1 & 0 & 0 & 0 \\
0 & 0 & 0 & 0 & 0 & 0 \\
1 & 1 & 0 & 0 & 1 & 0 \\
-1 & -1 & 0 & 0 & 0 & 0
\end{pmatrix},
\qquad
\begin{array}{l}
     \Bf{n} = (1,1,1,1,1,1),  \\
     \\
     \Bf{l}^{\mathfrak{sl}_2} = \left(\frac{1}{2},0,\frac{1}{2},0,\frac{1}{2},0\right).
\end{array}
\]
Rearranging the~nodes, we can write
\[
C = 
\begin{pmatrix}
0 & 0 & 0 & 0 & 0 & 0 \\ 
0 & 0 & -1 & -1 & 0 & 0 \\
0 & -1 & 0 & 0 & 1 & 0 \\
0 & -1 & 0 & 1 & 1 & 0 \\
0 & 0 & 1 & 1 & 1 & 0 \\
0 & 0 & 0 & 0 & 0 & 1
\end{pmatrix},
\qquad
\begin{array}{l}
     \Bf{n} = (1,1,1,1,1,1),  \\
     \\
     \Bf{l}^{\mathfrak{sl}_2} = \left(0,0,0,\frac{1}{2},\frac{1}{2},\frac{1}{2}\right),
\end{array}
\]
which is the~same as the~quiver \eqref{eq:fig8quiver} that was found empirically. We have just derived this quiver! Upgrading this to the~$a$-deformed version, 
\[
\Bf{a} = (0,0,0,1,1,1),\qquad
\Bf{l} = \left(0,0,0,-\frac{3}{2},-\frac{3}{2},-\frac{3}{2}\right). 
\]

\section{Quantized Coulomb branch of the~3d-5d system}\label{sec:Bpoly}

In the~physical realization of the~HOMFLY-PT polynomials, the~variable $a$ is identified with the~K\"ahler parameter of the~resolved conifold geometry $X$ or, equivalently, with the~Coulomb branch parameter of the~5d $\CN=2$ gauge theory ``engineered''~\cite{Katz:1996fh} by compactification on the~resolved conifold. In the~standard string terminology, it is a~{\it closed string} modulus because it describes the~geometry of the~background with no branes.

Similarly, in the~M-theory setup of~\cite{OV}, the~variable $x$ is identified with the~open string modulus or, more precisely, with the~modulus of the~fivebranes supported on a~Lagrangian submanifold in $X$. In five spacetime dimensions complementary to $X$, these fivebranes span a~three-dimensional subspace. In other words, they represent a~codimension-two defect (a.k.a. a~surface operator) in the~5d $\CN=2$ gauge theory. Following~\cite{MR2459305}, this half-BPS surface operator can be conveniently described by coupling a~3d $\CN=2$ theory -- in fact, precisely the~$T[S^3 \setminus K]$ theory --  to the~5d $\CN=2$ theory~\cite{DGH,Gukov:2014gja}. Then the~variable $x$ can be understood as the~Coulomb branch parameter of the~3d $\CN=2$ theory, much like $a$ is the~Coulomb branch parameter of the~5d $\CN=2$ theory.

Both variables $x$ and $a$ have their corresponding ``conjugates'', $y$ and $a^D$, respectively, such that turning on the~Omega-background in the~coupled 3d-5d system has the~effect of quantizing the~holomorphic symplectic phase space parametrized by $(a,a^D)$ in the~case of 5d theory and parametrized by $(x,y)$ in the~case of 3d theory.\footnote{Note, both ingredients can exist independently. For example, one can consider a~trivial 5d theory, in which case the~coupled 3d-5d system is nothing but the~3d standalone theory. And, similarly, one can consider a~trivial 3d theory, in which case there is no surface operator in 5d theory.} In particular, in the~presence of surface operators, the~Nekrasov partition function~\cite{Nekrasov:2002qd,Nakajima:2003uh} computes the~K-theoretic instanton-vortex partition function of the~coupled system, with the~asymptotic behaviour
\[
Z_{3d/5d} \simeq \exp \left( - \frac{1}{\hbar^2} \mathcal{F} (a) + \frac{1}{\hbar} \widetilde W (a,x) + \ldots \right).
\]
If we identify $\hbar = g_s$, it has the~familiar form of the~``closed + open'' topological string partition function, with the~small but important caveat that all variables are $\C^*$-valued, associated with the~K-theoretic lift of the~instanton / vortex counting. In particular, the~holomorphic symplectic forms relevant to 3d and 5d Coulomb branches are, respectively,
\[
d \log x \wedge d \log y
\,,\qquad\qquad
d \log a \wedge d \log a^D.
\]
Therefore, quantization with respect to these holomorphic symplectic forms on the~space $(\C^*)^4$ parametrized by $(x,y,a,a^D)$ replaces the~algebra of functions by the~algebra of operators that obey the~following $q$-commutation relations:
\be
\hat{y}\hat{x}  =q\hat{x}\hat{y}, \qquad\qquad \hat{b}\hat{a}=q  \hat{a}\hat{b},
\label{xyabcomm}
\ee
where we introduced a~new notation $b \equiv a^D$ to make the~parallel with $A$-polynomial more manifest. In 3d/5d coupled system, these operators can be understood as line operators. Their various combinations which represent Ward identities are certain $q$-difference operators that annihilate the~K-theoretic vortex-instanton partition function, see  e.g.~\cite{Mironov:2012xk,Aminov:2014wra,Awata:2017lqa,Jeong:2021bbh}.

\subsection{The~holomorphic Lagrangian subvariety}\label{subsec:holLag}
From the~physical discussion at the~start of this section, it is clear that classically the~Coulomb branch of the~3d-5d coupled system is a~holomorphic Lagrangian $\Gamma_K$ in $(\mathbb{C}^*)^4$ endowed with the~holomorphic symplectic form\footnote{One can see this also from \eqref{eq:abxyrelation}.}
\begin{equation}\label{eq:Omega}
\boxed{
\Omega := d\log x \wedge d\log y + d\log a \wedge d\log b.
}
\end{equation}
Naturally, the~projection of $\Gamma_K$ on $(\mathbb{C}^*)^3_{x,y,a}$ is the~zero locus of the~$a$-deformed $A$-polynomial. In this sense, $\Gamma_K$ can be thought of as the~holomorphic Lagrangian lift of the~$A$-polynomial. Note that such holomorphic Lagrangian is uniquely determined, as we can solve for $b$ as a~function of $a$ and $x$ by 
\[
b = \exp\qty(\int \frac{\partial \log y}{\partial \log a}d \log x).
\]
Away from the~discriminant locus, different branches of $y$ will lead to different branches of $b$, and they can be understood as the~expectation values of the~corresponding operators ($\hat{y}$ and $\hat{b}$) acting on $F_K$ for those branches that we discussed in Section \ref{sec:FK different branches}. 

A novel feature of this holomorphic Lagrangian is that instead of projecting it to $(\mathbb{C}^*)^3_{x,y,a}$, we can project it down to other hyperplanes such as $(\mathbb{C}^*)^3_{a,b,x}$, and study the~polynomial defining the~locus. This particular polynomial (the~one we get upon projection to $(\mathbb{C}^*)^3_{a,b,x}$) will be called the~\emph{$B$-polynomial} of the~knot $K$ and is the~subject of the~next subsection. To recap, the~zero sets of the~$A$- and $B$- polynomials are simply the~projections of $\Gamma_K$, so we have the~following diagram. 
\[
\begin{tikzcd}
 & \arrow[dl,swap,"\pi_b"] \Gamma_K = Z(AB_K) \arrow[rd,"\pi_y"]& \\
Z(A_K)\arrow[rd,"\pi_y"] & & \arrow[dl,swap,"\pi_b"] Z(B_K)\\
& (\mathbb{C}^*)^2_{x,a} & 
\end{tikzcd}
\]
We call the~ideal defining $\Gamma_K$ the~$AB$-ideal and denote it by $AB_K$ for an~obvious reason; it unifies $A$- and $B$-polynomials. 

Another notable feature of $\Gamma_K$ is that it enjoys Weyl symmetry. This is because $\Gamma_K$ describes the~semiclassical behaviour of $F_K$ which enjoys Weyl symmetry. Recall that on the~$A$-polynomial (or on $F_K$), the~Weyl symmetry acts on the~variables $x$, $y$ and $a$ as 
\[
x\mapsto a^{-1}x^{-1},\qquad y\mapsto y^{-1},\qquad a\mapsto a.
\]
The~action of the~Weyl symmetry on the~variable $b$ (dual to $a$) can be deduced as follows.
\begin{prop}
Under the~Weyl symmetry, $b$ transforms in the~following way: 
\[
b \mapsto y^{-1}b.
\]
\end{prop}
\begin{proof}
This can be most easily seen by making the~following change of variables:
\[
u_1 = a x,\qquad u_2 = x.
\]
Let $v_1, v_2$ be the~variables dual to $u_1, u_2$, respectively. That is, $v_1 = b$ and $v_2 = b^{-1}y$. 
Then the~Weyl symmetry acts by
\[
u_1 \mapsto u_2^{-1},\qquad u_2 \mapsto u_1^{-1}
\]
on $u_1$ and $u_2$. Since $v_1$ and $v_2$ are variables dual to $u_1$ and $u_2$, the~Weyl symmetry should act on them by
\[
\quad v_1 \mapsto v_2^{-1},\qquad v_2 \mapsto v_1^{-1}.
\]
Therefore, if we denote the~Weyl symmetry map by $W$, we have
\[
W(b) = W(v_1) = v_2^{-1} = y^{-1}b.
\]
\end{proof}
Consequently, the~holomorphic Lagrangian is preserved under the~action of the~Weyl symmetry $W$:
\[
W(\Gamma_K) = \Gamma_K.
\]

Just like the $A$-polynomial can be quantized into a~$q$-difference operator annihilating $F_K$, the~holomorphic Lagrangian itself can be quantized. The~quantum $AB$-ideal (quantization of the~$AB$-ideal) is a~left ideal of $q$-difference operators in $\hat{x}, \hat{y}, \hat{a}, \hat{b}$, which act on $F_K(x,a,q)$ by
\begin{align}\label{eq:xyabactiononFK}
\hat{x}F_K(x,a,q) &= x F_K(x, a, q), & \hat{a}F_K(x,a,q) &= a~F_K(x, a, q),\\
\hat{y}F_K(x,a,q) &= F_K(qx, a, q), &
\hat{b}F_K(x,a,q) &= F_K(x, qa, q). \nonumber
\end{align}
This quantum ideal annihilates $F_K$ regardless of the~choice of branch. We will see some explicit examples in Section \ref{subsec:AB ideals}. 

\begin{rmk}
While we primarily work with $F_K$, we could have chosen knot conormal instead of knot complement as the~Lagrangian filling. Then the~wave function we get would be the~coloured HOMFLY-PT generating function, written as a~function of $y,a,q$. The~holomorphic Lagrangian $\Gamma_K$, however, does not change, since it comes from a~theory at infinity, which specializes to the~curve count of knot contact homology that gives the~$A$-polynomial, and that is independent of the~choice of filling. 
\end{rmk}

\subsubsection{Asymptotic behavior of the holomorphic Lagrangian near the boundary}
Let $(\mathbb{C}^*)^4$ be the space parametrized by $x,y,a,b\in \mathbb{C}^*$ and equipped with the holomorphic symplectic form
\[
\Omega = d\log x \wedge d\log y + d\log a \wedge d\log b.
\]
Let $\Gamma_K\subset (\mathbb{C}^*)^4$ be the holomorphic Lagrangian associated to a knot $K$. 

In this subsubsection, we are interested in studying the asymptotic behavior of $\Gamma_K$ near the boundary of $\mathbb{C}^*$ where we can talk about various branches of $\Gamma_K$ (and therefore the corresponding wave functions $F_K$) without ambiguity. 

Asymptotically near the boundary, $\Gamma_K$ should look like a 2-dimensional linear subspace $V$ of $\mathbb{C}^4$ (parametrized by logarithmic variables) that is holomorphic Lagrangian. Suppose that this linear subspace is described by
\[
x^{k_i}y^{l_i}a^{m_i}b^{n_i} = 1,\quad i = 1,2.
\]
We are assuming that the two vectors $\mathbf{k}_i := (k_i,l_i,m_i,n_i)$ are independent. By eliminating some coordinates, we may choose to work with two vectors such that $n_1 = 0$ and $l_2 = 0$. Then the first equation determines the slope of the branch of $A_K$, and the second equation determines that of $B_K$. 

The Lagrangian condition of this linear subspace can be described in a simple way. 
\begin{prop}
The above linear subspace $V$ is Lagrangian iff
\[
\det \begin{psmallmatrix} k_1 & l_1 \\ k_2 & l_2 \end{psmallmatrix} + \det \begin{psmallmatrix} m_1 & n_1 \\ m_2 & n_2 \end{psmallmatrix} = 0.
\]
\end{prop}
\begin{proof}
The symplectic form $\Omega$ can be expressed as the following matrix:
\[
S:=\begin{psmallmatrix} 0 & 1 & 0 & 0 \\ -1 & 0 & 0 & 0 \\ 0 & 0 & 0 & 1 \\ 0 & 0 & -1 & 0 \end{psmallmatrix}.
\]
By definition, $V$ is Lagrangian iff $v_1^T S v_2 = 0$ for any choice of basis $\{v_1,v_2\} \subset V$. Moreover, a vector $v$ is in $V$ iff $\mathbf{k}_i^T v = 0$ for $i=1,2$. From this, it is easy to see that $V$ is Lagrangian iff $\mathbf{k}_1^T S \mathbf{k}_2 = 0$, and this is exactly the equation written above. 
\end{proof}
A direct application of this simple fact is the following. Suppose that $n_1=0$ and $l_2=0$ so that 
\[
x^{k_1}y^{l_1}a^{m_1} = 1,\quad x^{k_2}a^{m_2}b^{n_2} = 1. 
\]
Then the Lagrangian condition is simply $\frac{l_1}{m_1} = \frac{n_2}{k_2}$. That is, for any branch, the slope of the Newton polygon of $A_K$ as a polynomial in $y$ and $a$ equals that of $B_K$ as a polynomial in $b$ and $x$. 
Below, we check this in a few explicit examples. 

\begin{eg}[Trefoil, right-handed]
We write $A_K$ and $B_K$ in matrix form. For $A_K$ (resp. $B_K$), horizontal direction corresponds to $a$-degree (resp. $x$-degree) and the vertical direction corresponds to $y$-degree (resp. $b$-degree). The degrees get bigger going up and right. 
\[
A_{\mathbf{3}_1} = \begin{psmallmatrix}
 0 & 0 & -x^3 & x^4 \\
 -1 & 2 x^2+x & -x^4+x^3-2 x^2 & 0 \\
 0 & 1-x & 0 & 0 \\
\end{psmallmatrix}
\]
\[
B_{\mathbf{3}_1} = \begin{psmallmatrix}
 0 & 0 & a^2-a^3 & a^4-a^3 \\
 -1 & a & a-2 a^2 & 0 \\
 0 & 1 & 0 & 0 \\
\end{psmallmatrix}
\]
\end{eg}

\begin{eg}[Figure-eight]
\[
A_{\mathbf{4}_1} = \begin{psmallmatrix}
 0 & -x^2 & 2 x^3 & -x^4 & 0 \\
 1 & -3 x & 2 x^4+2 x^2 & -3 x^5 & x^6 \\
 0 & -2 x^2+3 x-1 & 0 & -x^6+3 x^5-2 x^4 & 0 \\
 0 & 0 & x^4-2 x^3+x^2 & 0 & 0 \\
\end{psmallmatrix}
\]
\[
B_{\mathbf{4}_1} = \begin{psmallmatrix}
 0 & -a^3+2 a^2-a & 2 a^4-4 a^3+2 a^2 & -a^5+2 a^4-a^3 & 0 \\
 1-a & 3 a^2-3 a & -2 a^3+4 a^2-2 a & 3 a^2-3 a^3 & a^4-a^3 \\
 0 & 2-a & -3 a & 2 a^2-a & 0 \\
 0 & 0 & 1 & 0 & 0 \\
\end{psmallmatrix}
\]
\end{eg}\

\begin{eg}[$\mathbf{5}_2$]
\begin{align*}
&A_{\mathbf{5}_2}=\\ 
&\begin{psmallmatrix}
 -1 & 3 x & -3 x^2 & x^3 & 0 & 0 & 0 & 0 \\
 0 & 3 x^2-4 x+2 & x^4-4 x^3+5 x^2-4 x & -4 x^5-x^4+2 x^3+2 x^2 & 6 x^6+2 x^5-3 x^4 & -4 x^7 & x^8 & 0 \\
 0 & 0 & -3 x^4+8 x^3-8 x^2+4 x-1 & -3 x^6+2 x^5-x^4+x^3+x & 8 x^7-x^6-4 x^5+x^4-4 x^3 & -8 x^8+x^7+x^6+6 x^5 & 4 x^9-4 x^7 & x^9-x^{10} \\
 0 & 0 & 0 & x^6-4 x^5+6 x^4-4 x^3+x^2 & 3 x^8-4 x^7-x^6+2 x^5 & -4 x^9+5 x^8+2 x^7-3 x^6 & 2 x^{10}-4 x^9+2 x^8 & 0 \\
 0 & 0 & 0 & 0 & 0 & -x^{10}+3 x^9-3 x^8+x^7 & 0 & 0 \\
\end{psmallmatrix}
\end{align*}
\begin{align*}
&B_{\mathbf{5}_2} =\\ 
&\begin{psmallmatrix}
 -a^3+3 a^2-3 a+1 & 3 a^4-9 a^3+9 a^2-3 a & -3 a^5+9 a^4-9 a^3+3 a^2 & a^6-3 a^5+3 a^4-a^3 & 0 & 0 & 0 & 0 \\
 0 & -3 a^2+6 a-3 & 7 a^3-15 a^2+9 a-1 & -5 a^4+14 a^3-13 a^2+4 a & a^5-8 a^4+13 a^3-6 a^2 & 4 a^5-8 a^4+4 a^3 & -a^6+2 a^5-a^4 & 0 \\
 0 & 0 & a^2-4 a+3 & -a^3+5 a^2-7 a+3 & -3 a^3+11 a^2-8 a & 2 a^4-10 a^3+8 a^2 & 4 a^4-4 a^3 & a^4-a^5 \\
 0 & 0 & 0 & a-1 & a-3 & 4 a-a^2 & -2 a^2 & 0 \\
 0 & 0 & 0 & 0 & 0 & 1 & 0 & 0 \\
\end{psmallmatrix}
\end{align*}
\end{eg}

In all the examples above, not just the slopes but the whole Newton polygons match up. We conjecture that this is true in general:
\begin{conj}
The Newton polygon of $A_K$ as a polynomial in $y$ and $a$ agrees with that of $B_K$ as a polynomial in $b$ and $x$. 
\end{conj}

\subsection{\texorpdfstring{$B$-polynomial}{B-polynomial}}

\subsubsection{Quantum $B$-polynomial}
In~\cite{MM21}, it was observed that the~cyclotomic coefficients of the~coloured HOMFLY-PT polynomials satisfy not only $q$-difference equations in the~colour variable but also $q$-difference equations in the~variable $a=q^N$. 
The~$q$-holonomicity of the~cyclotomic coefficients implies that the~$\mathfrak{sl}_N$ coloured Jones polynomials are $q$-holonomic in the~variable $a=q^N$. 

Let us recall that the~operators $\hat{a},\hat{b}$ and $\hat{x},\hat{y}$ act on $\mathfrak{sl}_N$ coloured Jones polynomials by
\begin{align*}
    \hat{x}J^{\mathfrak{sl}_N}_r(K;q) & =q^{r}J^{\mathfrak{sl}_N}_r(K;q), & \hat{a}J^{\mathfrak{sl}_N}_r(K;q)= & q^{N}J^{\mathfrak{sl}_N}_r(K;q),\\
    \hat{y}J^{\mathfrak{sl}_N}_r(K;q) & =J^{\mathfrak{sl}_N}_{r+1}(K;q), & \hat{b}J^{\mathfrak{sl}_N}_r(K;q)= & J^{\mathfrak{sl}(N+1)}_r(K;q).\nonumber
\end{align*} 
We can see that $(\hat{a},\hat{b})$ interact with the~group rank $N$ in complete analogy with the~action of $(\hat{x},\hat{y})$ on the~representation $r$; the~commutation relations are as in \eqref{xyabcomm}.

As pointed out above, the~$\mathfrak{sl}_N$ coloured Jones polynomials satisfy a~recurrence relation in variable $N$. We will call the~corresponding $q$-difference operator the~\emph{quantum $B$-polynomial} and denote it by $\hat{B}_K(\hat{a},\hat{b},x,q)$. In other words, the~recurrence relation in $N$ is given by
\begin{equation*}
    \hat{B}_K(\hat{a},\hat{b},x,q)J^{\mathfrak{sl}_N}_r(K;q)=0.
\end{equation*}
Note, this is in analogy with the~quantum $A$-polynomials which annihilate $\mathfrak{sl}_N$ coloured Jones polynomials acting with $\hat{x}$ and $\hat{y}$. 
Likewise, the~quantum $B$-polynomial annihilates $F_K(x,a,q)$ associated to any branch: 
\[
\hat{B}_K(\hat{a},\hat{b},x,q)F_K^{(\alpha)}(x,a,q) = 0.
\]
The~action of $\hat{a}$ and $\hat{b}$ on $F_K$ were described in \eqref{eq:xyabactiononFK}.

\subsubsection{Brief comment on normalisations}
Before presenting some $B$-polynomials, we should briefly pause to comment on how the~different normalisations of $F_K$ and $J^{\mathfrak{sl}_N}_r$ discussed in Section~\ref{sec: conventions} affect the~$B$-polynomial. Observing that
\begin{align*}
    \hat{b}\, \frac{(xq; q)_{\infty}}{(xa; q)_{\infty}} & = (1 - xa)\frac{(xq; q)_{\infty}}{(xa; q)_{\infty}}\hat{b},
    \\ \hat{b} \, e^{\frac{-\log(x)\log(a)}{2\hbar}}x^{\frac{1}{2}}\frac{(a; q)_{\infty}(xq; q)_{\infty}}{(xa; q)_{\infty}(q; q)_{\infty}} & = \frac{1 - xa}{x^{\frac{1}{2}}(1 - a)}e^{\frac{-\log(x)\log(a)}{2\hbar}}x^{\frac{1}{2}}\frac{(a; q)_{\infty}(xq; q)_{\infty}}{(xa; q)_{\infty}(q; q)_{\infty}} \hat{b},
\end{align*}
we see that given a~recursion relation $\hat{B}(\hat{b}, \hat{a}, x, q)$ for a~reduced invariant, the~corresponding relations for the~unreduced and fully unreduced invariants are 
\begin{align*}
    \hat{B}^{unreduced}(\hat{a}, \hat{b}, x, q) & = \hat{B}\left( \hat{a}, \frac{1}{1 - x\hat{a}}\hat{b},x, q\right),
    \\ \hat{B}^{fully\,unreduced}(\hat{a}, \hat{b}, x, q) & = \hat{B}\left( \hat{a}, \frac{x^{\frac{1}{2}}(1 - \hat{a})}{1 - x\hat{a}}\hat{b}, x, q\right).
\end{align*}
When making these substitutions for the~$\hat{b}$ operator, we also usually rescale by left multiplication so that the~coefficients of $\hat{b}^n$ are polynomials in $\hat{a}, x, q$ with no common factors. 

Looking at the~unknot, the~reduced $B$-polynomial is easily seen to be equal to
\begin{equation*}
    B_{\mathbf{0}_1}(\hat{a}, \hat{b}, x, q) = \hat{b} - 1,
\end{equation*}
so the~corresponding unreduced and fully unreduced $B$-polynomials are given by
\begin{align}
    B^{unreduced}_{\mathbf{0}_1}(\hat{a}, \hat{b}, x, q) & = \hat{b} - (1 - x\hat{a}) \nonumber
    \\ B^{fully \,unreduced}_{\mathbf{0}_1}(\hat{a}, \hat{b}, x, q) & = x^{\frac{1}{2}}(1 - \hat{a})\hat{b} - (1 - x\hat{a}). \label{eq:quantum B 01 f unred}
\end{align}

Similarly, if we compute the~$B$-polynomial using quiver forms and want to include a~prefactor as in Equation \eqref{eq: Full Prefac FK}, then $\hat{b}$ acts on this prefactor as
\begin{equation*}
    \hat{b} \exp(\frac{p(\log x,\log a)}{\hbar}) = \exp(\frac{p(\log x, \hbar + \log a)}{\hbar})\hat{b}.
\end{equation*}
Defining
\[
    p'(\log(x), \log(a)) := p(\log x,\log a) - p(\log x, \hbar + \log a),
\]
we see that including the~prefactor modifies the~$B$ polynomial by 
\begin{equation} \label{eq: add prefactor}
    \hat{B}(\hat{b}, \hat{a}, x, q) \mapsto \hat{B}\left(\exp(\frac{p'(\log x,\log a)}{\hbar}) \hat{b}, \hat{a}, x, q\right).
\end{equation}
Overall, we see that with a~little care it is easy to adjust the~$B$-polynomial to various conventions and normalisations.

Quantum $B$-polynomials for simple knots in the~reduced normalisation are given in Table~\ref{tab:Bpolynomials}.
\begin{table}[ht]
    \centering
    \begin{tabular}{|c|c|}
        \hline
        $K$ & $\hat{B}_K(\hat{a},\hat{b},x,q)$\\ \hline\hline
         $0_1$ &  $1 - \hat{b}$\\ \hline 
         $3_1$ &  $q x^2 - x(1 + q - (1 + qx)\hat{a} + qx^2 \hat{a}^2)\hat{b} + (1 - \hat{a})(1 - qx\hat{a})\hat{b}^2$ \\ \hline 
         $4_1$ & 
         $\begin{aligned}
             & q^2x^2\hat{a}^2  + qx\hat{a}(1 + q -(1 + 3qx + q^2x^2)\hat{a} + q x^2(1 + q)\hat{a}^2)\hat{b}
             \\ & + (1 - \hat{a})(1 - q x \hat{a})(1 - 2qx(1 + qx)\hat{a} + q^3x^3\hat{a}^2)\hat{b}^2
             \\ & - x(1 - \hat{a})(1 - q\hat{a})(1 - qx \hat{a})(1 - q^2x\hat{a})\hat{b}^3
         \end{aligned}$ \\ \hline 
         $5_1$ & 
         $\begin{aligned}
             & - qx^4\big( 1 + q + q^2 - (1 + q)(1 + qx)\hat{a} + qx(1 + x + qx)\hat{a}^2 - qx^2(1 + qx)\hat{a}^3 + q^2 x^4\hat{a}^4\big)\hat{b}
             \\ & + x^2(1 - \hat{a})(1 - qx\hat{a})(1 + q + q^2 - q(1 + qx)\hat{a} + q^2x^2(1 + q)\hat{a}^2)\hat{b}^2
             \\ & -(1 - \hat{a})(1 - q\hat{a})(1 - q x \hat{a})(1- q^2x\hat{a})\hat{b}^3
             +q^3x^6 
         \end{aligned}$ \\ \hline
    \end{tabular}
    \caption{Quantum $B$-polynomials for some simple knots.}
    \label{tab:Bpolynomials}
\end{table}

\subsubsection{Classical $B$-polynomial}

Similarly to the~case of $A$-polynomial, the~$q\to 1$ limit of
the~quantum $B$-polynomial can be obtained directly from the~effective twisted superpotential:
\begin{equation}\label{eq:BfromW}
    \lim_{q\to 1}\hat{B}_K(\hat{a},\hat{b},x,q) = B_K(a,b,x)= 0 \quad \Leftrightarrow \quad \log b  =\frac{\partial\widetilde{\mathcal{W}}_{T[M_{K}]}^{\textrm{eff}}(x,a)}{\partial\log a}.
\end{equation}
Recall that $\widetilde{\mathcal{W}}_{T[M_{K}]}^{\textrm{eff}}(x,a)$ comes from integrating out the~dynamical fields in the~twisted superpotential (see \eqref{eq:Weff}), which can be read from \eqref{eq:Semiclassical limit F_K} or double-scaling limit \eqref{eq:W from P} with $N\to\infty,\;q^{N}=a$:
\begin{equation*}
F_{K}(x,a,q)\underset{\hbar\rightarrow0}{\longrightarrow}
\int\prod_{i}\frac{dz_{i}}{z_{i}}\exp\left(\frac{1}{\hbar}\widetilde{\mathcal{W}}_{T[M_{K}]}(z_{i},x,a)+\mathcal{O}(\hbar^{0})\right)
\stackrel[\hbar\rightarrow0]{r,N\rightarrow\infty}{\longleftarrow}J^{\mathfrak{sl}_N}_r(K;q).
\end{equation*}

Note that the~construction of $A$- and $B$-polynomials via the~effective twisted superpotential (\ref{eq:AfromW}, \ref{eq:BfromW}) immediately leads to the~constraint
\begin{equation}\label{eq:abxyrelation}
\frac{\partial\log y}{\partial\log a}=\frac{\partial^{2}\widetilde{\mathcal{W}}_{T[M_{K}]}^{\textrm{eff}}(x,a)}{\partial\log x\,\partial\log a}=\frac{\partial\log b}{\partial\log x},
\end{equation}
which is equivalent to the~holomorphic Lagrangian condition:
\[
\Omega = d\log x \wedge d\log y + d\log a \wedge d\log b = \frac{\partial \log y}{\partial \log a}d\log x \wedge d\log a + \frac{\partial \log b}{\partial \log x}d\log a\wedge d\log x = 0. \]

On the~other hand, it allows us to derive $A(x,y,a)$ from $B(a,b,x)$ up to a~function $f(x)$. Namely, we can solve $B(a,b,x)=0$ for $b(a,x)$, integrate over $\log a$
\begin{equation*}
    \widetilde{\mathcal{W}}_{T[M_{K}]}^{\textrm{eff}}(x,a)=\int \log b(a,x)\; d(\log a) + f(x),
\end{equation*}
and differentiate with respect to $\log x$:
\begin{equation*}
    A_K(x,y,a)= 0 \quad \Leftrightarrow \quad \log y=\frac{\partial}{\partial \log x} \left( \int \log b(a,x)\; d(\log a) + f(x) \right).
\end{equation*}
We can apply the~same reasoning to derive $B(a,b,x)$ from $A(x,y,a)$ up to a~function $f(a)$.

Let us present the~relations discussed above on the~example of the~unknot. In the~fully unreduced normalisation we have
\begin{equation*}
J^{\mathfrak{sl}_N,fully \,unreduced}_r(0_1;q)
\stackrel[\hbar\rightarrow0]{r,N\rightarrow\infty}{\longrightarrow}
\exp\left[\frac{1}{\hbar}\left(-\frac{\log x\log a}{2}+\textrm{Li}_{2}(x)-\textrm{Li}_{2}(xa)+\textrm{Li}_{2}(a)-\frac{\pi^{2}}{6}\right)\right],
\end{equation*}
so the~twisted superpotential is given by
\begin{equation*}
\mathcal{\widetilde{W}}_{T[M_{\mathbf{0}_1}]}(x,a)=-\frac{\log x\log a}{2}+\textrm{Li}_{2}(x)-\textrm{Li}_{2}(xa)+\textrm{Li}_{2}(a)-\frac{\pi^{2}}{6}.
\end{equation*}
This is consistent with~\cite{FGS} as we do not have dynamical fields and $\mathcal{\widetilde{W}}_{T[M_{\mathbf{0}_1}]}(x,a)=\mathcal{\widetilde{W}}^{\textrm{eff}}_{T[M_{\mathbf{0}_1}]}(x,a)$. Plugging it in equation \eqref{eq:BfromW}, we obtain
\begin{equation}\label{eq:logb of unknot}
\log b=-\frac{\log x}{2}+\log(1-xa)-\log(1-a).
\end{equation}
It leads to
\begin{equation}\label{eq:B of unknot}
B^{fully \,unreduced}_{\mathbf{0}_{1}}(a,b,x)=(1-a)b-x^{-1/2}(1-xa)=0,
\end{equation}
which is in line with the~classical limit of \eqref{eq:quantum B 01 f unred}. One can also check that $A_{\mathbf{0}_1}(x,y,a)=(1-x)y-a^{1/2}(1-ax)$ can be obtained from (\ref{eq:logb of unknot}-\ref{eq:B of unknot}) with $f(x)=\textrm{Li}_{2}(x)-\frac{\pi^{2}}{6}$.

Classical $B$-polynomials for simple knots in the~reduced normalisation are given in Table \ref{tab:ClassicalBpolynomials}.
\begin{table}[ht]
    \centering
    \begin{tabular}{|c|c|}
        \hline
        $K$ & $B_K(a,b,x)$\\ \hline\hline
         $0_1$ &  $1 - b$\\ \hline 
         $3_1$ &  $x^2 - x(2 - (1 + x)a + x^2 a^2)b + (1 - a)(1 - xa)b^2$ \\ \hline 
         $4_1$ & 
         $\begin{aligned}
             x^2a^2 & + ax(2 -(1 + 3x + x^2)a + 2x^2a^2)b
             \\ & + (1 - a)(1 - x a)(1 - 2x(1 + x)a + x^3a^2)b^2
             \\ & - x(1 - a)(1 - a)(1 - xa)(1 - xa)b^3
         \end{aligned}$ \\ \hline 
         $5_1$ & 
         $\begin{aligned}
             x^6 & - x^4\big(3 - 2(1 + x)a + x(1 + 2x)a^2 - x^2(1 + x)a^3 + x^4a^4\big)b
             \\ & + x^2(1 - a)(1 - xa)(3 - (1 + x)a + 2x^2a^2)b^2
             \\ & -(1 - a)(1 - a)(1 - x a)(1 - xa)b^3
         \end{aligned}$ \\ \hline
         $5_2$ & 
         $\begin{aligned}
             1 &-x^{-2} \left(2 a^2 x^3+a^2 x^2-4 a x^2-a x-a+3 x+1\right)b\\ 
             &-x^{-3}(a-1) (a x-1) \left(a^3 x^4-3 a^2 x^3-2 a^2 x^2+5 a x^2+a x+a-3 x-3\right)b^2\\
             &-x^{-4}(a-1)^2 (a x-1)^2 \left(a^2 x^3-2 a x^2-a x+x+3\right)b^3\\
             &+x^{-5}(a-1)^3 (a x-1)^3 b^4
         \end{aligned}$ \\ \hline
    \end{tabular}
    \caption{Classical $B$-polynomials for some simple knots.}
    \label{tab:ClassicalBpolynomials}
\end{table}
\begin{rmk}
If we set $x=1$, we see that $B_K(a,b,x=1)$ always has a~factor of $b-1$. This is analogous to the presence of the~factor $y-1$ in $A_K(x,y,a=1)$. This condition lets us fix the integration constants coming from equation \eqref{eq:abxyrelation}. It similarly allows us to reconstruct $\Gamma_K$ from either $A_K$ or $B_K$.
\end{rmk}

\begin{rmk}
Thanks to the~relation $\frac{\partial \log b}{\partial \log x} = \frac{\partial \log y}{\partial \log a}$, we can solve for $y$ and vice versa for any branch of $b$. It follows that there is a~canonical one-to-one correspondence between the~branches of $b$ and the~branches of $y$, as functions of $x$ and $a$. 
This correspondence will be used later in Section \ref{subsec: Computing FK on different branches} to determine $F_K$ for various branches, by solving the~recurrence relation with respect to quantum $A$- and $B$-polynomials at the~same time. 
\end{rmk}

\begin{rmk}
Practically, there are many ways to compute the~classical $B$-polynomial. One way is to start from a~quiver expression, eliminate variables from the~quiver $A$-polynomials to get the~ideal defining the~holomorphic Lagrangian $\Gamma_K$, and then eliminate variable $y$ to project it down to get the~$B$-polynomial.
Another way is to start from $F_K$ associated to various branches and compute the~expectation values of the~$\hat{b}$ operator. One gets $b^{(\alpha)}(x,a)$ in a~power series form. The~elementary symmetric functions of $b^{(\alpha)}$'s are Laurent polynomials in $x$ and $a$, and they are the~coefficients of the~$B$-polynomial. 
\end{rmk}

\subsubsection{Relations to Alexander and HOMFLY-PT polynomials}
Recall from~\cite{DE,EGGKPS} that in case of the~abelian branch
\begin{equation*}
    \lim_{q\to1}\frac{F_{K}(x,qa,q)}{F_{K}(x,a,q)}=b(x,a)=\exp\left(\frac{\partial\mathcal{\widetilde{W}}(x,a)}{\partial\log a}\right)=\exp\left(\frac{\partial U_{K}(x,a)}{\partial\log a}\right)
\end{equation*}
is an~$a$-deformation of the~inverse of the~Alexander polynomial, where $U_K(x,a)$ denotes the~Gromov-Witten disk potential for the~knot complement. 
In particular, when $a=1$, $b = \frac{1}{\Delta_K(x)}$.
This also means that
\begin{equation*}
B_K(a=1,\frac{1}{\Delta_K(x)},x) = 0.
\end{equation*}

In case of non-abelian branches, it turns out $b(x,a)$ has a~pole at $a=1$. 
As a~result, the~$a=1$ specialization of the~classical $B$-polynomial is simply
\begin{equation*}
\boxed{
B_K(a=1,b,x) = 1-\Delta_K(x)b
}
\end{equation*}
up to an~overall multiplication by a~monomial. 

The~main theorem of~\cite{DE} states that
\[
\langle \hat{b}\rangle = \exp(\int \frac{\partial \log y}{\partial \log a}d\log x) =  \exp(-\int \frac{\partial_{\log a}A_K}{\partial_{\log y}A_K} d\log x)
\]
becomes $\frac{1}{\Delta_K(x)}$ when $(a,y)=(1,1)$. 
Inspired by this, we can consider the~``$B$-polynomial analogue" of this theorem:
\[
\langle \hat{y}\rangle = \exp(\int \frac{\partial \log b}{\partial \log x}d\log a) =  \exp(-\int \frac{\partial_{\log x}B_K}{\partial_{\log b}B_K} d\log a),
\]
and when $(x,b)=(1,1)$, it should give us the~``$B$-polynomial analogue of $\frac{1}{\Delta_K(x)}$". 
\begin{table}
    \centering
    \begin{tabular}{c|c}
        Knot & $\langle \hat{y}\rangle\vert_{(x,b)=(1,1)}$ \\
        \hline
        $\mathbf{3}_1^l$ & $2a-a^2$\\
        $\mathbf{3}_1^r$ & $-a^{-2}+2a^{-1}$\\
        $\mathbf{4}_1$ & $a^{-1}-1+a$\\
        $\mathbf{5}_2$ & $a+a^2-a^3$\\
        $\mathbf{6}_1$ & $a^{-2}-a^{-1}+a$\\
        $\mathbf{7}_2$ & $a+a^3-a^4$\\
        $\mathbf{8}_1$ & $a^{-3}-a^{-2}+a$\\
    \end{tabular}
    \caption{Expectation value of $\hat{y}$ at $(x,b)=(1,1)$}
    \label{tab:yOperatorExpValue}
\end{table}
These can also be directly obtained by solving the~equation $A_K(x=1,y,a)=0$ for $y$. 
It turns out that these are just HOMFLY-PT polynomials $P(K;a,q=1)$ specialized to $q=1$. In fact, it is a~well-known property of (reduced) coloured HOMFLY-PT polynomials that $P_n(K;a,q=1) = P_1(K;a,q=1)^n$, so it is consistent with what we have just observed!

\subsubsection{Geometric properties of the~\texorpdfstring{$B$}{B}-polynomial} \label{sec: B polynomial geometry}

We discuss geometric aspects of the~$B$-polynomial ranging from the~known to the~more conjectural. Consider the~Legendrian conormal torus of a~knot in the~unit cotangent bundle of $S^{3}$ filled by the~conormal or complement Lagrangian in the~resolved conifold. Let $U_{K}(x,a)$ denote its disk potential. Consider a~braid representation of $K$ and take the~limit where $K$ collapses onto a~multiple of the~unknot. In the~resolved conifold the~corresponding multicover of the~toric brane at the~vertex arises as the~limit of the~complement and conormal Lagrangians, and the~$x$-cycle and $y$-cycle are restored when this Lagrangian is shifted along different legs. Holomorphic curves converge in the~corresponding limit to basic curves on the~corresponding toric brane and we find that for the~knot complement such curves are combinations of basic disks in homology classes $x$ and $xa$. It follows that holomorphic curves on the~knot complement have homology class $x^{k}(xa)^{l}$, were $k\ge  0$ and $l\ge 0$. In particular, the~disk potential can be expressed as
\[ 
U_{K}(x,a)=\sum_{k,l\ge 0} c_{k,l}x^{k+l}a^{l}.
\]      
We first consider the~geometric interpretation of $\partial_{\log x}\partial_{\log a}U_{K}$. The~action of the $\partial_{\log x}$-operator can be understood in the following way: fix a~cycle in $M_{K}$ filling the~longitude and add a~boundary marked point on the~disks at intersection points with this cycle. On the~other hand, the~action of $\partial_{\log a}$-operator means: fix a~fiber in the~resolved conifold dual to $\mathbb{CP}^1$ and add an~interior marked point at intersections with this fiber.

Consider now adding the~fiber near the~other vertex in the~toric diagram. Inserting a~toric brane $L$ there, it is easy to see that we can open up any intersection to a~small boundary on this toric brane. This means that $\partial_{\log a}\partial_{\log x} U_{K}(x,a)$ counts holomorphic annuli stretching between $M_{K}$ and $L$ with a~boundary marked point on both boundary components. To state this more precisely, let $\log \xi$ and $\log\eta$ be a~homology basis for the~torus at the~boundary of $L$ and write $V_{K}(x,\xi,a)$ for the~count of annuli stretching between $L$ and $M_{K}$. Then
\[ 
\partial_{\log x}\partial_{\log a} U_{K}(x,a) = \partial_{\log x}\partial_{\log \xi} V_{K}(x,\xi=a,a),
\]  
which shows that the~variable $\eta$ dual to $\xi$ is related to $b$ dual to $a$, but generally not equal to it, see Figure \ref{fig:diskannulus}.

\begin{figure}[htp]
	\centering
	\begin{tikzpicture}
	    \draw (-1, 1) -- (0, 0) -- (2, 0) -- (3, 1);
	    \draw (-1, -1) -- (0, 0);
	    \draw (2, 0) -- (3, -1);
	    \node at (0, 0) [circle,fill,inner sep=1.5pt]{};
	    
	    \draw (0.25, 0) -- (0.25, 0.5);
	    \node[scale=.7] at (0.4, 0.6) {$L$};
	    
	    \draw (2.25, 0.25) -- (2.25, 1);
	    \node[scale=.7] at (2.6, 1) {$M_K$};

	    \node at (0, -3) [circle,fill,inner sep=1.5pt]{};
	    \draw (0, -3) to [out=90, in=190] (2, -2.25);
	    \draw (0, -3) to [out=-90, in=170] (2, -3.75);

	    \draw[dashed] (2, -2.25) to [out=-30,in=30] (2, -3.75);
	    \draw (2, -2.25) to [out=-150,in=150] (2, -3.75);
	    \node[scale=.7] at (2.25, -2) {$M_K$};
	    
	    \node[scale=.7] at (4, -3) {$\partial_{\log x}\partial_{\log a} U_K$};
	    
	    \draw [->,
                    line join=round,
                    decorate, decoration={
                        zigzag,
                        segment length=4,
                        amplitude=.9,post=lineto,
                        post length=2pt
                    }]  (1, -4) -- (1,-5);
	    
	    \draw (0, -5.75) to [out=60, in=190] (2, -5.25);
	    \draw (0, -6.25) to [out=-60, in=170] (2, -6.75);
	    \draw (0, -5.75) to [out=-30,in=30] (0, -6.25);
	    \draw[dashed] (0, -5.75) to [out=-150,in=150] (0, -6.25);

	    \draw[dashed] (2, -5.25) to [out=-30,in=30] (2, -6.75);
	    \draw (2, -5.25) to [out=-150,in=150] (2, -6.75);
	    \node[scale=.7] at (2.25, -5) {$M_K$};
	    
	    \node[scale=.7] at (-0.1, -5.5) {$L$};
	    
	    \node[scale=.7] at (4, -6) {$\partial_{\log x}\partial_{\log\xi} V_K$};
	    
    \end{tikzpicture}
	\caption{Opening up a~disk intersection to annulus}
 	\label{fig:diskannulus}
\end{figure}
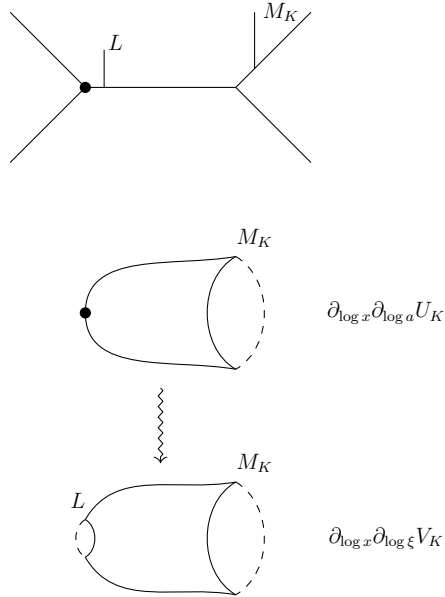


Note next that from the~knot theory perspective, the~boundary of the~toric brane $L$ is the~conormal Legendrian of a~braid axis of $K$. This means that the~annulus count $V_{K}(x,a)$ can be computed from (the~partial information about) the~dg-algebra of the~link $K\cup S$, where $S$ is the~braid axis. More precisely, using one-dimensional curves with a~positive puncture at the~degree $1$ Reeb chord of $L$, we find the~equation
\begin{equation}\label{eq:withtriangles} 
\Delta_{K} + \partial_{\log\eta} A_{L}\cdot \partial_{\log x}\partial_{\log \xi} V_{K}(x,a) = 0,
\end{equation} 
where $\Delta_{K}$ counts holomorphic triangles filled by disks with two positive punctures, see Figure~\ref{fig:triangleandannulus}.
\begin{figure}[htp]
	\centering
	\begin{tikzpicture}
	   
	    \node[scale=1.5] at (-0.5, 0) {$\partial$};
	    \draw (0.5, 2) to [out=-120,in=120] (0.5, -2);
	   
	    \draw (2.25, 1.5) to (0.75, 1.5);
	    \draw (2.25, 1.5) to [out=-135, in=90] (2.25, -1);
		\draw (2.25, -1) arc(0:-180:.75);
		\draw (0.75, 1.5) to [out=-45, in=90] (0.75, -1);
		
		\draw (1.5, -0.9) ellipse (0.2cm and 0.4cm);
	   
	    \draw (2.5, 2) to [out=-60,in=60] (2.5, -2);
	   
	    \node at (1.5, -2.5) {$\dim = 1$};
	   
	    \node[scale=1] at (3.5, 0) {$=$};
	   
	    \draw (5.5, 1.5) to (4.5, 1.5);
	    \draw (5.5, 1.5) to [out=-115, in=90] (5.75, -0.5);
		\draw (5.75, -1.1) arc(0:-180:.75);
		\draw (5.25, -1.1) arc(0:-180:.25);
		\draw (5.25, -1.1) -- (5.25, -0.85) -- (5.75, -0.85) -- (5.75, -1.1);
		\draw (4.25, -1.1) -- (4.25, -0.85) -- (4.75, -0.85) -- (4.75, -1.1);
		
		\draw (5.25, -0.5) arc(0:180:.25);
		\draw (5.25, -0.5) -- (5.25, -0.75) -- (5.75, -0.75) -- (5.75, -0.5);
		\draw (4.25, -0.5) -- (4.25, -0.75) -- (4.75, -0.75) -- (4.75, -0.5);
		
		\node at (4.75, -0.5) [circle,fill,inner sep=1pt]{};
		\node[scale=0.5] at (5, -0.5) {$M_K$};
		
		\draw (4.5, 1.5) to [out=-65, in=90] (4.25, -0.5);
		
		\draw[-{Latex[length=1mm,width=1mm]}] (4.9, -2.3) to [out=90,in=-90] (5.4, -1.8);
		\node at (5, -2.5) {$\dim = 0$};
		
		\node at (5.75, 1) {$L$};
		\node at (9, 1.25) {$L$};
		\node at (9.75, -1.25) {$L$};
		
		\node[scale=1.5] at (6.75, 0.25) {$\cup$};
		
		\node[scale=0.75] at (4, 2) {at $\infty$};
		\draw[-{Latex[length=1mm,width=1mm]}] (4, 1.8) to [out=-90,in=150] (4.5, 1);
		
	    \draw (8.75, 1.5) to (8.25, 1.5);
	    \draw (8.75, 1.5) to [out=-115, in=90] (9, 0);
	    \draw (9, 0) arc(0:-180:.5);
		\draw (8.25, 1.5) to [out=-65, in=90] (8, 0);
		
		\draw[dotted] (8.5, -0.55) -- (8.5, -0.825);
		
		\draw (8.5, -1.25) ellipse (1cm and 0.4cm);
		\draw (8.5, -1.25) ellipse (0.5cm and 0.2cm);
		
		\node at (8, -1.25) [circle,fill,inner sep=1pt]{};
		\node[scale=0.5] at (8.5, -1.25) {$M_K$};
		
		\node[scale=0.75] at (7.5, 2) {at $\infty$};
		\draw[-{Latex[length=1mm,width=1mm]}] (7.5, 1.8) to [out=-90,in=150] (8.2, 1);
	
	    \draw[-{Latex[length=1mm,width=1mm]}] (8.2, -2.3) to [out=90,in=-90] (8.6, -1.8);
	    \node at (8.5, -2.5) {$\dim = 0$};
	    
    \end{tikzpicture}
	\caption{The~boundary of the~1-dimensional moduli space of annuli stretching between $L$ and $M_K$ with positive puncture at one degree one Reeb chord of $L$}
	
	\label{fig:triangleandannulus}

\end{figure}
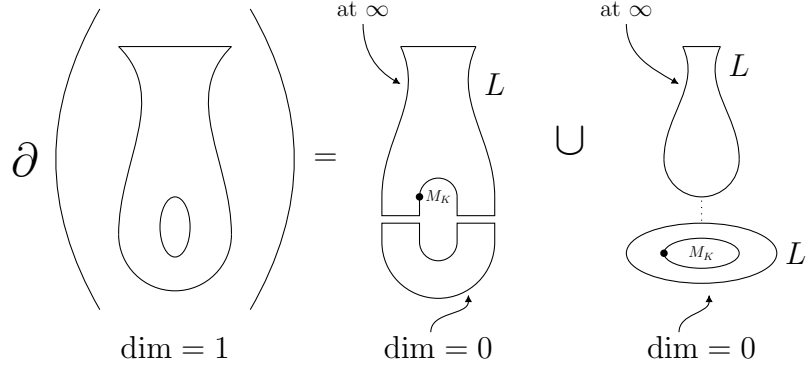

Note also that $L$ is simply the~conormal of an~unknot and therefore its augmentation polynomial $A_{L}$ is well-known. When $K$ is the~unknot, this calculation is a~calculation for the~Hopf link that was carried out in~\cite{EN}. In this case, because of the~simplicity of the~curves involved, $V_{U}(x,\xi,a)$ is independent of $a$ and the~$B$-polynomial of the~unknot can be computed directly from the~above equation. 

In general, by the~definition of $b$ and differentiation, the~$B$-polynomial satisfies a~similar equation:
\begin{equation}\label{eq:withouttriangles}
\partial_{\log x}B_{K} + \partial_{\log\eta} B_{K}\cdot \partial_{\log x}\partial_{\log \xi} U_{K}(x,a) = 0.
\end{equation}    
Geometrically, the~difference between \eqref{eq:withtriangles} and \eqref{eq:withouttriangles} is that in the~former there are no annuli at infinity -- all annuli come from disks with marked points. 

We next speculate what kind of theory would allow us to identify $\xi=a$ and $\eta=b$ and give a~direct geometric interpretation of $b$. Consider the~configuration of $M_{K}$ and $L$ in the~resolved conifold as above. 
As $a\to \infty$, any holomorphic disk on $L_{K}$ with $k$ marked points at the~point where $L$ is inserted contains in the~limit a~disk component where all the~marked points collide, with a~sphere with $k$ marked points on it attached. Consider now $S^{1}$-equivariant curves over the~split off $\mathbb{CP}^{1}$. Then any sphere in class $a^{m}$ has $m$-fold branch points at the~vertices. From the~point of view of the~brane $L$, the~split curve looks like a~disk with a~curve with some $x$-charge split off at the~center. It is natural to conjecture that the~splitting gives rise to a~new contact form at infinity (on $S^3\times S^2$ with a~defect determined by the~Legendrian conormal of $K$) and that for this contact manifold the~open string theory for $L$ is related to the~string theory of $L_{K}$ simply by $\xi=a$ and $\eta=b$, and the~$B$-polynomial would arise simply as the~augmentation polynomial of the~dg-algebra of the~Legendrian torus at infinity.

In the~previous section we observed two properties of the~polynomials $B_{K}(a,b,x=1)$: first, the~coefficient of the~top power of $b$ is divisible by $(1-a)$ and second, the~polynomial contains the~factor $(1-b)$. Assuming that the~theory discussed above exists, these properties of the~$B$-polynomial follow from usual arguments. The~$x=1$ limit corresponds to no shifting, no initial annuli and therefore no disk potential, and $b=1$ is a~solution for which the~$\eta$-cycle in $L$ contracts without correction (the~flux through the~cycle is zero). 

To see that the~coefficient of the~top $b$-degree term must contain a~factor $a-1$, we first recall how a~similar property of the~$A$-polynomial is derived. First, the~coefficient of the~top power of $y$ is divisible by $(1-x)$ and second, at $a=1$ the~polynomial contains the~factor $(1-x)(1-y)$. To see this, we use the~geometric interpretation of the~equation
\[ 
A_{K}(x(y),y,a)=0
\]  
as the~count of ends of the~moduli space of generalized holomorphic disks with boundary on $L_{K}$ and one positive puncture. Here the~boundary corresponds to disks on the~Legendrian torus at infinity, where the~solution corresponds to inserting rigid disks at intersections with bounding chains at infinity, whereas the~polynomial $A_{K}(x,y,a)$ counts augmented disks in the~$\R$-invariant region. Consider a~change of variables of the~degree zero Reeb chords such that the~boundary of all augmentation disks have homology class a~non-negative multiple of $\log y$. After a~similar change of variables for the~degree one Reeb chords, we arrange that the~minimum degree coefficient of augmented disks at infinity also equals zero. Consider now following a~degree zero disk when it moves into the~filling. It cannot split off any disk since any disk without positive puncture has boundary which is a~positive multiple of $\log y$ by positivity of area of holomorphic disks. Also, if the~degree zero disk picks up a~rigid disk, its total homology class turns positive. It follows that the~coefficient of the~constant term in the~augmentation polynomial must have a~solution corresponding to a~cycle that shrinks without splitting, which means that it contains a~factor $(1-x)$ or $(1-ax)$, depending on the~choice of capping disk for $x$.

Consider now $L$ in the~resolved conifold with defect, sitting at the~vertex. Expressing the~theory in terms of $\xi$ or $\eta$ corresponds to shifting $L$ along different legs in the~toric diagram. In the~shift where we expand in $\eta=b$, the~above argument for $A_{K}$ gives the~observed property for $B_{K}$.

\subsection{\texorpdfstring{The~Quantum $AB$-ideal}{AB ideals}}\label{subsec:AB ideals}

\subsubsection{The~quantum AB-ideal}
    Given a~knot $K$, consider the~set of recursion relations
    \begin{equation*}
        \widehat{AB}_K = \{\hat{T} \in \m{C}[\hat{x}^{\pm 1},  \hat{a}^{\pm 1}, q^{\pm 1}, \hat{y}^{\pm 1}, \hat{b}^{\pm 1}] \mid \hat{T} F_k(x, a, q) = 0\}.
    \end{equation*}
    It is clear that $\widehat{AB}_K$ is a~left ideal and that
    \[
        \widehat{AB}_K \cap \m{C}[\hat{x}^{\pm 1},  \hat{a}^{\pm 1}, q^{\pm 1}, \hat{y}^{\pm 1}]  = \langle \hat{A}_K \rangle,
        \qquad \widehat{AB}_K \cap \m{C}[\hat{x}^{\pm 1},  \hat{a}^{\pm 1}, q^{\pm 1}, \hat{b}^{\pm 1}]  = \langle \hat{B}_K \rangle.
    \]
    Hence $\widehat{AB}_K$ is a~generalisation of both the~$\hat{A}$ and $\hat{B}$-polynomials and indeed it is exactly the~\emph{quantum $AB$-ideal} we mentioned earlier. There is a~natural classical limit of this picture via the~ring homomorphism which sends $q \to 1$ and maps the~quantum $AB$-ideal to the~classical $AB$-ideal.
    A question we can immediately ask is how much information is lost when moving to the~classical ideal. The~most we can hope for is captured in the~following conjecture.
    
    \begin{conj} \label{conj: Quantization AB Ideal}
        Given a~set of polynomials $T_1, \cdots, T_n$ which generate $AB_K$, there exist quantizations $\hat{T}_1, \cdots, \hat{T}_n$ which generate $\widehat{AB}_K$.
    \end{conj}
    
    The~main obstacle in studying the~$AB$-ideal compared to either the~$A$- or $B$-polynomials is that it is much more difficult to compute. In most cases all we can say is that we have found a~sub-ideal which we suspect is the~whole ideal. With that being said, we are able to show that in general this ideal is richer than simply the~$A$- and $B$-polynomials combined.
    
    We call $AB_K$ \emph{simple} if $AB_K = \langle A_K, B_K \rangle$, and for the~collection of knots we studied $AB_K$ is simple only for the~unknot. This is a~slightly surprising result as recently a~similar phenomenon was studied in~\cite{MM21}, where authors worked with the~equivalent of the~$A$- and $B$-polynomials coming from coefficients of the~cyclotomic expansion of the~coloured Jones polynomial. In that case they claim that their version of the~$AB$-ideal is always simple. Further work needs to be done to understand the~origin of the~difference between these results.

\subsubsection{Unknot}
    For the~unknot we work with the~unreduced normalisation
    \[
        F^{unreduced}_{\mathbf{0}_1}(x,a,q) = \frac{(xq;q)_{\infty}}{(xa;q)_{\infty}}
    \]
    as the~ideal is clearly uninteresting when working with the~reduced normalisation\footnote{As $F_{\mathbf{0}_1}(x, a, q) = 1$ in the~reduced normalisation.}. In this case we can easily describe the~action of the~$\hat{y}$ and $\hat{b}$ operators:
    \begin{align*}
        \hat{y}F^{unreduced}_{\mathbf{0}_1}(x, a, q) & = \frac{1 - a x}{1 - q x}F^{unreduced}_{\mathbf{0}_1}(x, a, q),
        \\ \hat{b}F^{unreduced}_{\mathbf{0}_1}(x, a, q) & = (1 - ax)F^{unreduced}_{\mathbf{0}_1}(x, a, q).
    \end{align*}
    and this instantly gives us three linear elements which clearly generate the~ideal:
    \[
        \hat{b} - (1 - \hat{a}\hat{x}), \qquad \qquad (1 - q\hat{x})\hat{y} - (1 - \hat{a}\hat{x}), \qquad \qquad (1 - q\hat{x})\hat{y} - \hat{b}.
    \]
    In this case the~third element is the~difference of the~first two, so $AB_{\mathbf{0}_1} = \langle A_{\mathbf{0}_1}, B_{\mathbf{0}_1}\rangle$ and we conclude that the~$AB$-ideal of the~unknot is simple.
\subsubsection{Trefoil} \label{sec: trefoil AB}
    For knots more complicated than the~unknot, it is not easy to simply compute $I_K$ from observation. Instead, we use the~quiver forms introduced earlier. Given a~knot $K$, let $Q$ be an~associated quiver with identifications $x_i \mapsto x^{n_i}a^{a_i}q^{l_i}$. Each node $i$ gives rise to a~classical quiver $A$-polynomial as in equation \eqref{eq:Classical Quiver A Polynomial}. Combining these polynomials together, we obtain the~classical quiver ideal $I_{Q} \subset \m{C}[\Bf{x}^{\pm 1}, \Bf{y}]$. The~identifications
    \begin{equation*}
        y = \prod_i y_i^{n_i}, \qquad\qquad b = \prod_i y_i^{a_i}
    \end{equation*}
    embed $\m{C}[\Bf{x}^{\pm 1}, y, b]$ as a~subring inside $\m{C}[\Bf{x}, \Bf{y}]$ and there is a~map
    \begin{align*}
        j_K: \m{C}[\Bf{x}^{\pm 1}, y, b] & \to \m{C}[x^{\pm 1}, a^{\pm 1}, y, b],
        \\ x_i &\mapsto x^{n_i}a^{a_i}. \nonumber
    \end{align*}
    Hence we can map $I_{Q_K}$ to an~ideal in $\m{C}[x^{\pm 1}, a^{\pm 1}, y, b]$ by intersecting with $\m{C}[\Bf{x}^{\pm 1}, y, b]$ and applying $j_K$. This is the~classical $AB$-ideal\footnote{In principle, the~intersection can leave behind extra components. In this case the~$AB$-ideal is the~unique component whose projections to $\m{C}_{x, y, a}$ and $\m{C}_{a, b, x}$ are the~$A$- and $B$-ideals.}:
    \begin{equation*}
        AB_K = j_K\big(I_{Q_K} \cap \m{C}[\Bf{x}, y, b]\big).
    \end{equation*}
    Given a~small quiver, this computation can be performed using Gr\"obner basis, so in principle we can compute classical $AB$-ideals. However, it is still highly nontrivial  to re-quantize $AB_K$ to construct the~quantum ideal. Currently the~only approach is essentially trial and error. Additionally, it is not clear how to prove Conjecture \ref{conj: Quantization AB Ideal}, so the~resultant ideal may not be the~entire quantum ideal.
    
    For a~more concrete example, let us work with the~trefoil. Recall from Section \ref{sec:Examples} that it has a~quiver form with
	\[
	    C  = \begin{pmatrix}
			0 & -1 & 0 & -1 \\
			-1 & -1 & 0 & -1 \\
			0 & 0 & 1 & 0 \\
			-1 & -1 & 0 & 0 \\
		\end{pmatrix},
		\qquad
		\begin{array}{l}
		     \Bf{n}  = (2, 2, 1, 1),  \\
		     \Bf{a}  = (1, 0, 1, 0), \\
		     \Bf{l}  = \frac{1}{2}(0, 3, -1, 1).
		\end{array}
	\]
	From this setup we can write down the~four classical quiver $A$-polynomials:
	\begin{align*}
		A_1(\Bf{x}, \Bf{y}) &= 1 - y_1 + x_1y_2^{-1}y_4^{-1}, & & A_2(\Bf{x}, \Bf{y}) = 1 - y_2 + x_2y_1^{-1}y_2^{-1}y_4^{-1},
		\\ A_3(\Bf{x}, \Bf{y}) &= 1 - y_3 + x_3y_3, & & A_4(\Bf{x}, \Bf{y}) = 1 - y_4 + x_4y_1^{-1}y_2^{-1},
	\end{align*}
	and see that our $y$ and $b$ operators are given by
	\[
		y = y_1^2y_2^2y_3y_4, \qquad \qquad b = y_1y_3.
	\]
	We can use Mathematica to compute the~intersection and then manually remove the~spurious irreducible components. Passing to the~ideal, we shift $b$ and $y$ to incorporate the~the~prefactor of $x^{-1} e^{\frac{\log(x)\log(a)}{\log(q)}}$ and we are left with
	\begin{equation*}
		\langle B_{\mathbf{3}_1}, 1 - x^{-1}(1 - a)(1 - a x^2)b - y\rangle.
	\end{equation*}
    For simplicity, we call the~other polynomial $D_{\mathbf{3}_1}$.
    \begin{prop}
        We have a~strict containment of ideals:
        \begin{gather*}
            \langle A_{\mathbf{3}_1}, B_{\mathbf{3}_1}\rangle \subsetneq \langle B_{\mathbf{3}_1}, D_{\mathbf{3}_1}\rangle.
        \end{gather*}
    \end{prop}
    \begin{proof}
        Containment follows from the~equality
        \begin{align} 
	        A_{\mathbf{3}_1} & = \frac{(1 - a)(1 - ax^2)^2}{a^2}B_{\mathbf{3}_1}- \frac{1}{a^2}\Big((1 - a~- 2ax^2 + 2a^2x^2 + a^2x^3 - a^3x^4)\nonumber
		    \\ & \qquad\qquad\qquad\qquad\qquad\qquad\qquad - x^{-1}(1 - a)(1 - ax)(1 - ax^2)b + (1 - ax)y\Big)D_{\mathbf{3}_1}. \label{eq:A poly in ideal}
	    \end{align}
	    Proving that this containment is strict follows from a~easy commutative algebra exercise.
    \end{proof}

	There are two obvious question to ask here. First, can we quantize $D_{\mathbf{3}_1}$, and second, can we promote Equation \eqref{eq:A poly in ideal} to a~quantum version which relates $\hat{A}_{\mathbf{3}_1}, \hat{B}_{\mathbf{3}_1}$, and $\hat{D}_{\mathbf{3}_1}$. Through trial and error, we find that the~answer to both questions is positive. One possible pair of quantizations is
	\begin{align*}
	    \hat{D}_{\mathbf{3}_1} & = 1 - x^{-1}(1 - q^{-1}\hat{a})(1 - q \hat{a} \hat{x}^2)\hat{b} - \hat{y}
	    \\ 
	    \hat{A}_{\mathbf{3}_1} & =  \frac{q(1 - q^{-1}\hat{a})(1 - q\hat{a}\hat{x}^2)(1 - q^2\hat{a}\hat{x}^2)(1 - q^3\hat{a}\hat{x}^2)}{\hat{a}^2}\hat{B}_{\mathbf{3}_1}
		\\ & \quad - \frac{1}{\hat{a}^2}  \Big(q^2(1 - q^{-1}\hat{a} - (q + q^2)\hat{a}\hat{x}^2 + (1 + q)\hat{a}^2\hat{x}^2 + q^2\hat{a}^2\hat{x}^3 - q^2\hat{a}^3\hat{x}^4)(1 - q^3\hat{a}\hat{x}^2) \nonumber
		\\ & \quad \quad - qx^{-1}(1 - q^{-1}\hat{a})(1 - q\hat{a}\hat{x})(1 - q\hat{a}\hat{x}^2)(1 - q^3\hat{a}\hat{x}^2)\hat{b}  + q^2(1 - q\hat{a}\hat{x})(1 - q\hat{a}\hat{x}^2)\hat{y}\Big)\hat{D}_{\mathbf{3}_1}. \nonumber
	\end{align*}
	Note that there is an~interesting uniqueness question here as there are many ways to quantize both Equation \eqref{eq:A poly in ideal} and $D_{\mathbf{3}_1}$. In particular, as both $\hat{B}_{\mathbf{3}_1}$ and $\hat{D}_{\mathbf{3}_1}$ annihilate $F_{\mathbf{3}_1}$, the~multiplication of Equation \eqref{eq:A poly in ideal} by any $q$ factors will produce an~element of the~quantum ideal whose classical limit is the~classical $A$-polynomial. This gives rise to a~large family of elements whose classical limits all correspond to $A_{\mathbf{3}_1}$. In this particular case, we can define~$\hat{A}_{\mathbf{3}_1}$ as the~minimal element of the~family which lies in $\m{C}[\hat{y}^{\pm 1}, \hat{x}^{\pm 1}, a^{\pm 1}, q^{\pm 1}]$, but in general there will not be a~uniquely defined quantization. For a~simple example of this phenomenon, consider the~element
	\[
		x^{-1}(1 - a)(1 - x)b^2 - a^{-1}y - (1 - x)b + a^{-1}x^{-1}by,
	\]
	which lies in the~classical ideal. This is quantized by the~family
	\[
		a^{-1}x^{-2}(1 - a)(1 - q^i + a q^{1 + i}x - aq^2x^2)\hat{b}^2 - q^{i + 1}a^{-1}\hat{y} - a^{-1}x^{-1}(1 - q^i + a q^{i + 1}x- a q^{2 + i}x^2)\hat{b} +  a^{-1}x^{-1}\hat{b}\hat{y}
	\]
	for any integer $i$. Additionally, note that while we have only shown
	\[
		\widehat{AB}_{\mathbf{3}_1} \supset \langle \hat{B}_{\mathbf{3}_1}, \hat{D}_{\mathbf{3}_1}\rangle,
	\]
	this does imply that $\widehat{AB}_{\mathbf{3}_1}$ is larger than $\langle \hat{A}_{\mathbf{3}_1}, \hat{B}_{\mathbf{3}_1}\rangle$, which means that the~$AB$-ideal for the~trefoil is not simple.

\subsubsection{Other Torus Knots}
    For knots whose associated quivers are larger, passing through the~quantum quiver $A$-polynomials runs into to computational problems. It is still possible though to find interesting elements of the~$AB$-ideal and show that the~ideal is larger than the~one generated by the~$A$- and $B$-polynomials. As we show here, in certain cases we can upgrade known recursion relations to elements of this ideal. In~\cite{Hik}, Hikami shows the~existence of a~non-homogeneous recursion relation for the~coloured Jones polynomial of the~torus knot $T_{2, 2p + 1}$:
    \begin{equation*}
        J_{r}(T_{2, 2p + 1}; q) = q^{p r}\frac{1 - q^{2r + 1} }{1 - q^{r + 1}} - q^{(2p + 1)r + p + 1}\frac{1 - q^{r}}{1 - q^{r + 1}}J_{r - 1}(T_{2, 2p + 1}; q).
    \end{equation*}
    Promoting $q^r = x$ and replacing the~Jones polynomial by $F_p = F_{T_{2, 2p+1}}$, this becomes
    \begin{equation}
        (1 - qx)F_p(x, q) = x^p(1 - q x^2) - q^{p + 1}x^{2p + 1}(1 - x)F_p(q^{-1}x, q).
    \label{eq:Torus knots F_K}
    \end{equation}
    Next, recall that for the~$a$-deformed $F_K$ we have the~following pair of relations:
    \[
        F_K(x, q^2, q)  = F_K(x, q),
        \qquad\qquad F_K(x, q, q)  = 1.
    \]
    Incorporating this into \eqref{eq:Torus knots F_K}, we find that
    \begin{equation*}
        (1 - qx)F_p(x, q^2, q) = x^p(1 - q x^2) F_p(x, q, q) - q^{p + 1}x^{2p + 1}(1 - x)F_p(q^{-1}x, q^2, q).
    \end{equation*}
    This looks exactly like the~$N = 2,\; a =q^2$ limit of a~more general recursion relation. Through a~little trial and error, we can restore the~$a$ dependence to get
    \begin{equation*}
        (1 - q^{-1} a x)F_p(x, a, q) = x^p(1 - q^{-1} a x^2)F_p(x, q^{-1} a, q) - q^{- p - 1} a^{p + 1} x^{2p + 1}(1 - x)F_p(q^{-1}x, a, q),
    \end{equation*}
    so we have a~general recursion relation for $F_p$:
    \begin{equation*}
        (1 - q^{-1} \hat{a}\hat{x}) - x^p(1 - q^{-1} \hat{a}\hat{x}^{2})\hat{b}^{-1} + q^{- p - 1} a^{p + 1}\hat{x}^{2p + 1}(1 - x)\hat{y}^{-1}.
    \end{equation*}
    It is easy to show that the~classical limit of this polynomial does not lie in $\langle A_p, B_p \rangle$ for any torus knot $T_{2, 2p + 1}$ so that $AB$-ideal is not simple in all these cases.

\subsection{Refinement}

Almost everything we have discussed so far can be generalized to the refined (i.e.\ $t$-deformed) setting. 
Solving the recursion given by the quantum super-$A$-polynomial $\widehat{A}_K(\hat{x},\hat{y},a,q,t)$, we obtain wave functions $F_K^{(\alpha)}(x,a,q,t)$, one for each branch $\alpha$ of the super-$A$-polynomial, which are $t$-deformations of the wave functions $F_K^{(\alpha)}(x,a,q)$ we have discussed in previous sections. 

Just like the unrefined wave functions, these wave functions admit quiver expressions. In fact, in all the examples we have considered, the quiver for the refined wave function is the same as the quiver for the unrefined one; refinement only affects the knots-quivers change of variables which now involves the variable $t$. 

Once we have a quiver expression for a wave function $F_K^{(\alpha)}(x,a,q,t)$, the left-ideal of $q$-difference operators in $\hat{x},\hat{y},\hat{a},\hat{b},\hat{t},\hat{u}$ (where $u$ is a variable conjugate to $t$) that annihilates the wave function can be obtained from the quantum quiver $A$-polynomials for the quiver expression, through non-commutative elimination theory. This quantum ideal is the quantisation of the ideal defining the complex $3$-dimensional holomorphic Lagrangian subvariety $\Gamma_K$ in Conjecture \ref{conj:refholLag}. 

The classical holomorphic Lagrangian subvariety $\Gamma_K \subset (\mathbb{C}^*)^6$ can be described and computed more explicitly, by applying elimination theory (e.g.\ Gr\"obner basis) to the ideal defined by the quiver $A$-polynomials. 

Empirically, we find that the Weyl symmetry of the (unrefined) holomorphic Lagrangian subvariety described in Section \ref{subsec:holLag} can be lifted to the Weyl symmetry of the refined holomorphic Lagrangian, which is stated in part (1) of Conjecture \ref{conj:refholLag}. When projected down to $(\mathbb{C}^*)^4$ parametrised by variables $x,y,a,t$, this version of Weyl symmetry was already noted in equation (8.5) of \cite{GNSSS}. 

There are some interesting limits regarding the new variable $u$ conjugate to $t$. 
For example, the expectation value of the $\hat{u}$-operator on $F_K^{(\text{ab})}(x,a,q,t)$, when $a=1, t=-1$ is
\[
\langle \hat{u}\rangle \big\vert_{a=1,t=-1} = \frac{1}{\Delta_K(x)}.
\]
One way to see this is by setting $t=-q^N$ in the equation (108) of \cite{EGGKPS}. 
Another way to see this is from $\mathfrak{sl}_1$ pairs. Recall that the Alexander polynomial arises from a count of annuli and that to get the expression we are interested in $\exp(\frac{\partial U_K}{2\partial \log a})$ at $a=1$. 
We should think of a deformed (stretched) situation where all annuli are generalised annuli consisting of disks with a 4-chain intersection. Assume now that there is some quiver description where the annuli we count are the basic disks with 4-chain intersection. Here a quiver node with `charges'
$a^r q^s t^l x^k$
would contribute with 
$\frac{1}{2}r (-1)^l x^k$
annuli. 
Here $r$ comes from the number of 4-chain intersection, and $(-1)^l$ comes from the spin structure orientation sign. The interpretation of $t$ is the number of twists in the trivialization along the boundary, and $x^k$ is just the homology class of the boundary. 
Assume now that nodes come in $\mathfrak{sl}_{1}$ pairs corresponding to factors $(1+tq^{-2}a^2)$ in the super polynomial. 
We would then have two nodes $a^r t^l x^k$  and $a^{r+2}t^{l+1}x^k$ (suppressing $q$ powers which play no role). 
Taking derivative with respect to $\log a^2$ and setting $a=1$, $t=-1$ or taking derivative with respect to $\log t$ and setting $a=1$, $t=-1$ gives the same total contribution to annuli:
$(-1)^{l+1}x^k$. After an overall change of framing, also the contributions to the $\log a^2$ and $\log t$ derivatives from the surviving $\mathfrak{sl}_1$ node agree.

\subsection{Applications to computing \texorpdfstring{$F_K$}{FK} on different branches} \label{subsec: Computing FK on different branches}

In this section we show how to use the~quantum $B$-polynomial to further analyse $F_K$ on different branches as introduced in Section \ref{sec:FK different branches}.

\subsubsection{General argument}

In Section \ref{sec:FK different branches}, sometimes we computed
$F_K$ recursively, using the~quantum $A$-polynomial. In general, this approach gives solutions up to multiplication by a~function of $a$ and $q$.
While it is not possible to completely eliminate this uncertainty, using the~quantum $B$-polynomial we can fully determine the~$a$-dependence. Then in the~event that we know $F_K$ for any specialization $a = q^N$, we can determine the~$q$-dependence as well. This method relies on the~observation that Proposition \ref{prop:classical_asymptotics} and Conjecture \ref{conj:quantum_initial_condition} apply equally well to the~$B$-polynomial, as well as the~following conjecture. 
\begin{conj}\label{conj:AB_Solution_Match}
    Let $e_A$ be a~left edge of the~Newton polygon of the~$A$-polynomial with slope $\frac{n_x}{n_y}$. Then there exists a~left edge $e_B$ of the~Newton polygon of the~$B$-polynomial with slope $\frac{n_a}{n_b}$ and a Puiseux series in $x, a, q$ satisfying
    \begin{align}
        \hat{A}_KF_K^{(e_A, e_B)} = 0\quad \text{and} \quad \hat{B}_KF_K^{(e_A, e_B)} = 0.
    \end{align}
    This series has the~form
    \[
        F_K^{(e_A, e_B)}(x, a, q) = 
            \exp(\frac{p(\log x,\log a)}{\hbar}) \left(1 + \sum_{i, j \geq 1} f_{i, j}(q) x^{\frac{i}{n_y}} a^{\frac{j}{n_b}}\right),
    \]
    where $p$ is a~polynomial of degree at most 2 determined by $e_A$ and $e_B$, and each coefficient $f_{i, j}$ is itself a~Laurent series in $q$ with integer coefficients.
\end{conj}
Note that we could have started off with left edges of the~Newton polygon of the~$B$-polynomial. Since, after dealing with degeneracies, we get a~bijection between the~two sets of left edges, the~end result would be the~same. To illustrate this further, we will focus on the~example of the~$\mathbf4_1$ knot. 

\subsubsection{\texorpdfstring{$4_1$}{41} knot}
The~Newton polygon for the~$A$-polynomial was shown in Figure \ref{fig:Apoly_Newton_polygon_figure_eight}, whereas the~Newton polygon for the~$B$-polynomial is presented in Figure \ref{fig:Bpoly_Newton_polygon_figure_eight}. Looking at them, we see that both have three left edges when counted with multiplicity. 
\begin{figure}[ht]
    \centering
    \includegraphics[scale=0.5]{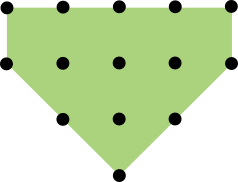}
    \caption{The~Newton polygon for $B_{\mathbf{4}_1}$}
    \label{fig:Bpoly_Newton_polygon_figure_eight}
\end{figure}

The~prefactors in Conjecture \ref{conj:quantum_initial_condition} arise from the~slope of the~edge on the~Newton polygon and through the~requirement that the~$F_K$ is computed in the~form $1 + O(x)$ or $1 + O(a)$.  For the~$A$ polynomial, the~prefactors are\footnote{With $r = \frac{\log(x)}{\hbar}$ and $N = \frac{\log(a)}{\hbar}$.}
\[
    q^{r^2 + rN + r},\ q^{rN - r},\ q^{-r^2 - rN + R},
\]
whereas for the~$B$-polynomial they are given by
\[
    q^{-rN},\ e^{N \pi i}q^{\frac{N^2}{2} + rN - \frac{N}{2}}.
\]
Note that one of the~edges for the~$B$-polynomial is degenerate, so the~second prefactor appears with multiplicity $2$. Next, we compute the~$F_K$ on various branches recursively. Starting with the~$A$-polynomial, we find the~three branches to be
\begin{align*}
    F_{\mathbf{4}_1}^{A, (ab)} & = q^{rN - r}\left(1 - \frac{3(a-q)}{1-q}x - \frac{(a - q)(1 - 2a + 6q - 6aq - 2q^2 + aq^2)}{(1 - q)(1 - q^2)} x^2 + O(x^3)\right), \nonumber
    \\ F_{\mathbf{4}_1}^{A,(\frac{1}{2})} & = q^{-r^2 - rN + r}\left(1 + \frac{a(2-q)}{q(1-q)}x + \left(\frac{(1 + 3q - 2q^2 - 2q^3 + q^4)a^2}{q^3(1 - q)(1 - q^2)} - \frac{a}{q}\right) x^2 + O(x^3)\right),
    \\ F_{\mathbf{4}_1}^{A, (-\frac{1}{2})} & = q^{r^2 + rN + r}\left(1 + \frac{q(1 - 2q)}{1-q}x + \left(\frac{q^2(1 - 2q - 2q^2 + 3q^3 + q^4)}{(1 - q)(1 - q^2)} - aq\right) x^2 + O(x^3)\right). \nonumber
\end{align*}
Similarly, the~branches of $B$-polynomial are given by
\begin{align*}
    F_{\mathbf{4}_1}^{B, (\infty)} & = q^{-rN}\Bigg(1 - \frac{(1 - x)(1 - x + qx)}{q(1 - q)}a
    \\ &  + \frac{(1 - x)(q^2 - q (2 - q^2) x + (1 - q^2)(1 + 2 q - q^2) x^2 - (1 - q)(1 - q^2)x^3}{q^3(1 - q)(1 - q^2)} a^2 + O(a^3) \Bigg), \nonumber
    \\ F_{\mathbf{4}_1}^{B, (-1)} & = e^{N \pi i}q^{\frac{N^2}{2} + rN - \frac{N}{2}} \Bigg(1 + c_1(x,q) a
    \\ &   + \left(\frac{x \left(-3 + (1 - q)^2x - q^2(2 - q)x^2\right)}{(1 - q)(1 - q)^2} - \frac{1 + 3qx + q^2(1 - q)x^2}{1 - q^2}c_1\right) a^2 + O(a^3)\Bigg).
\end{align*}

Note that $c_1$ is an~arbitrary function of $x, q$ and appears due to the~degeneracy of the~$-1$ slope edge. It remains to match up the~$F_K$'s coming from the~$A$- and $B$-polynomials. In this case, this can be done entirely from analysis of prefactors. As the~$A$-recursion is defined up to factors of $a, q$ and the~$B$-recursion up to factors of $x, q$, we can pair up the~branches and work out the~overall prefactors. We obtain
\begin{alignat*}{2}
    F_{\mathbf{4}_1}^{A, (ab)} & \sim F_{\mathbf{4}_1}^{B, (-1)} \quad && \to \quad e^{N \pi i}q^{\frac{N^2}{2} + rN - r - \frac{N}{2}},
    \\ F_{\mathbf{4}_1}^{A, (\frac{1}{2})} & \sim F_{\mathbf{4}_1}^{B, (\infty)} \quad && \to \quad q^{-r^2 - rN + r},
    \\ F_{\mathbf{4}_1}^{A, (-\frac{1}{2})} & \sim F_{\mathbf{4}_1}^{B, (-1)} \quad && \to \quad e^{N \pi i}q^{\frac{N^2}{2} + r^2 + rN + r  - \frac{N}{2}}.
\end{alignat*}
Once we have matched the~series up, we can determine the~combined series as we expect (ignoring prefactors)
\[
        F^{\text{combined}}(x, a, q) = F^{A}(x, a, q)F^{B}(0, a, q) = F^{A}(x, 0, q)F^{B}(x, a, q).
\]
It turns out that in several cases computing the~combined $F_K$ is even easier. In particular, after correcting the~prefactor $F_{\mathbf{4}_1}^{A, (ab)}$ is annihilated by the~$B$-polynomial and $F_{\mathbf{4}_1}^{B, (\infty)}$ is annihilated by the~$A$ polynomial. This leaves the~final pairing of $F_{\mathbf{4}_1}^{A, (-\frac{1}{2})}$ and $F_{\mathbf{4}_1}^{B, (-1)}$, which is more difficult as we have the~unknown constant $c_1$. Through some trial and error we find that the~right value for $c_1$ is $1 + O(x)$ and
\[
    F_{\mathbf{4}_1}^{B, (\infty)}(0, a, q) = \left.F_{\mathbf{4}_1}^{B, (-1)}(c_1, 0, a, q)\right|_{c_1 = \frac{-1}{q(1 - q)}} = \frac{1}{(aq^{-2}, q^{-1})_{\infty}} = (q^{-1}a; q)_{\infty}.
\]
Hence the~three branches of $F_{\mathbf{4}_1}$ are  given by (up to an~overall function of $q$)
\begin{align*}
    F_{\mathbf{4}_1}^{(ab)} & = e^{N \pi i}q^{\frac{N^2}{2} + rN - r - \frac{N}{2}}
    \\ & \quad \quad \quad \times \left(1 - \frac{3(a-q)}{1-q}x - \frac{(a - q)(1 - 2a + 6q - 6aq - 2q^2 + aq^2)}{(1 - q)(1 - q^2)} x^2 + O(x^3)\right), \nonumber
    \\ F_{\mathbf{4}_1}^{(\frac{1}{2})} & = q^{-R^2 - RN + R}(q^{-1}a; q)_{\infty}
    \\ & \quad \quad \quad \times \left(1 + \frac{a(2-q)}{q(1-q)}x + \left(\frac{(1 + 3q - 2q^2 - 2q^3 + q^4)a^2}{q^3(1 - q)(1 - q^2)} - \frac{a}{q}\right) x^2 + O(x^3)\right), \nonumber
    \\ F_{\mathbf{4}_1}^{(-\frac{1}{2})} & = e^{N \pi i}q^{\frac{N^2}{2} + r^2 + rN + r  - \frac{N}{2}}(q^{-1}a; q)_{\infty}
    \\ & \quad \quad \quad \times \left(1 + \frac{q(1 - 2q)}{1-q}x + \left(\frac{q^2(1 - 2q - 2q^2 + 3q^3 + q^4)}{(1 - q)(1 - q^2)} - aq\right) x^2 + O(x^3)\right). \nonumber
\end{align*}
Observe that both non abelian branches which correspond to the~slopes $\frac{1}{2}$ and $-\frac{1}{2}$ have poles in the~$q \to 1$ limit and disappear when $a = q$. On the~other hand, the~abelian branch is well-defined as $q \to 1$ and in the~$a = q$ limit it is equal to $1$.

\subsection{Closed sector recursions}

When we consider the~generating functions of open topological strings, we usually mode out the~whole closed sector. However, using $B$-polynomials and $AB$-ideals we can include it in the~recursions.


\subsubsection{Closed partition function}

The~closed sector partition function reads
\begin{equation*}
\phi(a,q)=\exp\left[\sum_{d}\frac{1}{d}\frac{a^{d}}{(1-q^{d})^{2}}\right].
\end{equation*}

Since
\begin{equation*}
\frac{\phi(a,q)}{\phi(aq,q)}=\psi(a,q)=\sum_{d}\frac{a^{d}}{(q)_{d}},
\qquad\qquad
\left(1-\hat{a}-\hat{b}\right)\psi(a,q)=0,
\end{equation*}
we can write
\begin{equation*}
\left(1-\hat{a}\right)\frac{\phi(a,q)}{\phi(aq,q)}-\frac{\phi(aq,q)}{\phi(aq^{2},q)}=0.
\end{equation*}
Therefore we have
\begin{equation*}
\left(1-\hat{a}\right)\phi(a,q)\hat{b}^{2}\phi(a,q)-\left(\hat{b}\phi(a,q)\right)^{2}=0.
\end{equation*}

Let us discuss the~geometry underlying these formulas. Consider first the~toric brane $L$ in~$\C^3$. The~open string partition function 
\[
\psi(x,q)=\exp(\sum_d \frac{1}{d}\frac{x^d}{(1-q^d)})
\]
can be interpreted, after SFT-stretching around $L$, as the~Gromov-Witten curve count. The~coefficient of $x^d$ counts the~contribution from connected curves that are asymptotic to the~multiplicity $d$ Reeb orbit over the~unique index zero geodesic in $L$. Consider now shifting $x$ to $qx$. This can be realized geometrically by shifting the~brane. We can then compute the~new partition function $\psi(xq,q)$ by stretching around both the~original $L_0$ and the~shifted $L_1$ branes. Denoting the~$d$-fold Reeb orbits of $L_j$ as $\gamma_j^{(\pm d)}$, there are basic cylinders with two positive punctures inside the~neighbourhood of $L_0$ and a~basic cylinder stretching between $L_0$ and $L_1$. Then the~corresponding $d$-fold covers contribute respectively as
\[
\frac{1}{d}\gamma_0^{(d)}\gamma_0^{(-d)}\quad\text{ and }\quad \frac{q^d}{d}\gamma_0^{(-d)}\gamma_1^{(d)}
\]
to the~count of connected curves, where $dg_s=d\log(q)$ is the~area of the~cylinder stretching between the~Lagrangians. (Here the~two positive puncture cylinder has area $0$ since it lies in the~negative end. Its actual non-zero area is visible only after rescaling and taking any non-zero area to infinity.) Gluing these curves we find that
\[
\psi(xq,q)=\exp(\sum_d \frac{d^2}{d^3} x^dq^d ) = \exp(\sum_d \frac{1}{d} x^dq^d ),
\]
as expected. Here the~$d^2$-factor comes from gluing along $d$-fold Reeb orbits twice. We point out that this gives a~curve counting proof of the~recursion relation
\[
(1-\hat x-\hat y)\psi(x,q)=0.
\]

We will next give a~similar curve counting proof of the~closed string recursion. Consider local $\mathbb{CP}^1$. We first note that we can compute the~closed string partition function by inserting a~toric brane over the~equator and applying SFT stretching. Connected curves glue over the~area zero cylinders in the~negative end and we find that
\[
\phi(a^2) = \exp(\sum_d \frac{d^2}{d^3} \left(\frac{a^d}{(1-q^d)}\right)^2)
= \exp(\sum_d \frac{1}{d} 
\frac{a^{2d}}{(1-q^d)^2}).
\]
In order to find the~effect of $a\to aq$, we argue as in the~open case (or simply give area $g_s$ to the~cylinder in the~negative end). The~result is
\[
\phi(qa^2) = \exp(\sum_d \frac{d^4}{d^5} q^d\left(\frac{a^d}{(1-q^d)}\right)^2),
\]
where the~holomorphic building consists of two outside curves, two cylinders in the~negative ends and one cylinder stretching between Lagrangians. The~curves are glued over multiple $d$ Reeb orbits in four places. We conclude that  
\[
\phi(qa^2)\phi(a^2)=\exp(\sum_d\frac{1}{d}(1-q^d)\left(\frac{a^d}{(1-q^d)}\right)^2)=\psi(a^2),
\]
where the~factor of $(1-q^d)$ comes from the~area difference between a~zero area cylinder in the~negative end and a~cylinder which is a~$d$-fold cover of the~of the~basic area $g_s$ cylinder that increases the~area of the~$\mathbb{CP}^1$.

\subsubsection{Combination with the~unknot}

In the~next step we combine the~closed sector with the~open one for the~unknot (in the~unreduced normalisation). Multiplying $\phi$ by $F^{unreduced}_{\mathbf{0}_{1}}$, we obtain
\begin{equation*}
\Phi(x,a,q)=\phi(a,q)F^{unreduced}_{\mathbf{0}_{1}}(x,a,q)=\exp\left[\sum_{d}\frac{1}{d}\frac{a^{d}}{(1-q^{d})^{2}}\right](xq;q)_{\frac{\log a}{\hbar}-1}.
\end{equation*}
Since
\begin{equation*}
\hat{b}F^{unreduced}_{\mathbf{0}_{1}}(x,a,q)=(1-x\hat{a})F^{unreduced}_{\mathbf{0}_{1}}(x,a,q),
\end{equation*}
we have 
\begin{equation*}
\frac{\Phi(x,a,q)}{\Phi(x,aq,q)}=\frac{\phi(a,q)}{\phi(aq,q)}\frac{F^{unreduced}_{\mathbf{0}_{1}}(x,a,q)}{F^{unreduced}_{\mathbf{0}_{1}}(x,aq,q)}=\frac{\psi(a,q)}{(1-xa)}.
\end{equation*}

Therefore 
\begin{equation*}
\left(1-\hat{a}-\hat{b}\right)(1-xa)\frac{\Phi(x,a,q)}{\Phi(x,aq,q)}=0,
\end{equation*}
which leads to
\begin{equation*}
(1-a)(1-xa)\Phi(x,a,q)\hat{b}^{2}\Phi(x,a,q)-(1-xaq)\left(\hat{b}\Phi(x,a,q)\right)^{2}=0.
\end{equation*}

\subsubsection{Non-linear recursions}\label{sec:Non-linear}

    One may wonder what happens if we consider $\Phi_K = \phi F_K$ for more complicated knots $K$. It turns out that similar relations exist for all knots $K$ and can be computed directly from the~$B$-polynomial.

	Observe that the~quadratic relation for $\phi$ actually follows from a~$B$-polynomial whose coefficients involve the~infinite $q$-Pochhammers. Namely, $\phi$ satisfies the~linear relation
	\[
		\frac{\phi(q a, q)}{(a;q)_{\infty}} - \phi(a, q) = 0,
	\]
	which means it is annihilated by $B$-polynomial $\frac{1}{(\hat{a};q)_{\infty}}\hat{b} - 1 = \hat{b}' - 1$, where $\hat{b}' = \frac{1}{(\hat{a};q)_{\infty}}\hat{b}$.

	\begin{prop}
		Given a~knot $K$ with $B$-polynomial
		\[
			B_K(\hat{b}, \hat{a}, x, q) = \sum_{i = 0}^n c_i(\hat{a}, x, q) \hat{b}^n
		\]
		annihilating $F_K$, $\Phi_K$ is annihilated by $B_K(\hat{b}', \hat{a}, x, q)$, where $\hat{b}'$ is defined above.
	\end{prop}
	The~proof of this is immediate from the~observation that $\hat{b}'\Phi_K = \phi \hat{b}F_k$. The~remaining question is how to pass from this to the~non-linear relations which do not involve infinite $q$-Pochhammers. To make notation a~little easier, let $\Phi_{K, n} = \hat{b}^n \Phi_K \hat{b}^{-n} = \Phi_K(q^n a, x, q)$ and observe that
	\[
	    (\hat{b}')^n = \left(\prod_{i = 1}^{n} \frac{(a;q)_{i-1}}{(a;q)_{\infty}}\right)\hat{b}^n = \frac{\prod_{j = 1}^{n} (a;q)_{j-1}}{(a;q)^n_{\infty}}\hat{b}^n.
	\]
	Then the~proposition above is equivalent to the~statement
	\[
		S[0] := \sum_{i = 0}^n \Bigg(\Big(c_i(a, x, q)\left(\prod_{j = 1}^{i} (a;q)_{j-1}\right)\Phi_{K, i}\Big) \frac{1}{(a;q)_{\infty}^i}\Bigg)  = 0.
	\]
	From this we can define $S[j]$ by
	\[
		S[k] := \frac{1}{((a;q)_{k})^n} \hat{b}^k S[0] = \sum_{i = 0}^n \Bigg(\left(\frac{c_i(q^k a, x, q)\left(\prod_{j = 1}^{i} (q^ka;q)_{j-1}\right)}{((a;q)_{k})^{n - i}} \Phi_{K, j + i}\right) \frac{1}{(a;q)_{\infty}^i}\Bigg) = 0.
	\]
	Now consider the~set of $n + 1$ equations $S[0], \cdots, S[n]$ as linear equations in the~$n + 1$ variables $1, \frac{1}{(a;q)_{\infty}}, \cdots, \frac{1}{(a;q)_{\infty}^n}$. Elementary linear algebra tells us that we can eliminate the~the~variables $\frac{1}{(a;q)_{\infty}^i}$ from the~set of equations $S[0], \cdots, S[n]$ which will leave us with an~equation purely in terms of $\Phi_{K, i}, a, x, q$. On top of this, the~final equation will be a~homogeneous polynomial of order $n$ in the~$\Phi_{K, i}'s$ -- exactly the~nonlinear equation we are looking for.

	The~basic closed sector case corresponds to $F_K = 1,\ B_K = 1 - b$, and if we plug this in, we directly recover
	\[
		(1 - a)\phi(a, q) \phi(q^2 a, q) - \phi(q a, q)^2 = 0.
	\]
	If we set $K = \mathbf{0}_{1}$ to be the~unknot, we similarly recover the~earlier formula
	\[
		(1 - a)(1 - xa)\Phi_{\mathbf{0}_{1}, 0}\Phi_{\mathbf{0}_{1}, 2} - (1 - q x a)\Phi_{\mathbf{0}_{1}, 1}^2 = 0.
	\]
	Finally, setting $K = \mathbf{3}_1$ to be the~trefoil, we compute
	\begin{align*}
		0 & = (a-1) \left(a q^2-1\right)^2 (a q-1)^2 \left(a q^3 x-1\right) \left(a^2 q^3 x^2-a q^2 x-a q+q+1\right)\Phi_{\mathbf{3}_1, 0}\Phi_{\mathbf{3}_1, 2}\Phi_{\mathbf{3}_1, 4} 
		\\ &  \quad + (a-1) (a q-1)^3 \left(a q^2 x-1\right) \left(a^2 q^5 x^2-a q^3 x-a q^2+q+1\right)\Phi_{\mathbf{3}_1, 0}\Phi_{\mathbf{3}_1, 3}^2
		\\ &  \quad + (a q-1)^2 \left(a q^2-1\right)^2 \left(a q^3 x-1\right) \left(a^2 q x^2-a q x-a+q+1\right)\Phi_{\mathbf{3}_1, 1}^2 \Phi_{\mathbf{3}_1, 4}
		\\ &  \quad - (a q-1) \Big(a^4 q^6 x^3-a^3 q^6 x^3-a^3 q^5 x^2+a^2 q^5 x^2+a^4 q^4 x^3-2 a^3 q^4 x^2+a^2 q^4 x^2-a^3 q^3 x^3
		\\ & \quad \quad \quad + a^2 q^3 x^2-a^3 q^2 x^2+a^2 q^4 x-2 a^3 q^3 x+3 a^2 q^3 x+3 a^2 q^2 x+a^2 q^2+a^2 q x^2
		\\ & \quad \quad \quad + a^2 q x+a^2 q-2 a q^3 x-2 a q^2 x-2 a q^2-2 a q x-2 a q-2 a+2 q+2\Big)\Phi_{\mathbf{3}_1, 1} \Phi_{\mathbf{3}_1, 2} \Phi_{\mathbf{3}_1, 3}
		\\ &  \quad - (a-1) \Phi_{\mathbf{3}_1, 2}^3 (a q x-1) \left(a^2 q^3 x^2-a q^2 x-a q+q+1\right).
	\end{align*}

Combining above discussion with results from Section \ref{subsec:AB ideals}, we can define open-closed $\widehat{AB}$-ideals by a~redefinition of operator $\hat{b}$:
\begin{equation*}
    \hat{b}\to \hat{b}'=\hat{b}(\hat{a};q)_{\infty}^{-1}.
\end{equation*}
For the~unknot in the~unreduced normalisation, it gives
\begin{align*}
\hat{A}_{\mathbf{0}_1}(\hat{x},\hat{y},a,q) & =\left(1-\hat{x}q\right)\hat{y}-\left(1-\hat{x}a\right) & &\to & \hat{A}^{open-closed}_{\mathbf{0}_1}(\hat{x},\hat{y},a,q) & =\left(1-\hat{x}q\right)\hat{y}-\left(1-\hat{x}a\right),
\\
\hat{B}_{\mathbf{0}_1}(\hat{a},\hat{b},x,q) & =\hat{b}-\left(1-x\hat{a}\right) & & \to & \hat{B}_{\mathbf{0}_1}^{open-closed}(\hat{a},\hat{b},x,q) & =\hat{b}(\hat{a};q)_{\infty}^{-1}-\left(1-x\hat{a}\right), \nonumber
\end{align*}
and we have $\widehat{AB}_{\mathbf{0}_1}^{open-closed} = \langle \hat{A}_{\mathbf{0}_1}^{open-closed}, \hat{B}_{\mathbf{0}_1}^{open-closed} \rangle$.

\section{Closed 3-manifolds and log-CFT structures} \label{sec-3-logsft}


Authors of~\cite{GM} proposed the~following surgery formula connecting $F_K$ with $\widehat Z$ invariant for 3-manifold obtained by Dehn surgery in which we glue the~complement of $K$ and $-1/r$ solid torus:
\be
\label{eq:surgery}
\widehat Z \big( S^3_{-1/r} (K) \big)
\; = \; \sum_{n=1}^{\infty}  q^{\frac{(nr-1)^2}{4r}} (q^n-1) f_n (q),
\ee
where
\[
\sum_n f_n (q) \cdot x^{\frac{n}{2}}
\; = \; - x (x^{\frac{1}{2}} - x^{- \frac{1}{2}}) F_K (q^{-1} x,q),
\]
which comes from matching the~conventions used in this work and in~\cite{GM}.

If we write $F_K (x,q) \; = \; \sum_{n \ge 0} \tilde{f}_n (q) \cdot x^n,$ then \eqref{eq:surgery} leads to
\[
\widehat Z \big( S^3_{-1/r} (K) \big)
\; = \; - \sum_{n \ge 0}
\tilde{f}_n (q) \Big[ q - q^{2r+2nr} - q^{2+2n} + q^{3+2n+2r+2nr} \Big]
q^{r n^2 + r n - 2n + \frac{r}{4} + \frac{1}{4r} - \frac{3}{2}}.
\]
For each $F_K$ that can be expressed in a~quiver form\footnote{Note that for convenience in this section we moved $(-1)^{C_{ii}d_i^2}=(-1)^{t_id_i}$ into $x_i$.}
\be
\label{eq:FK quiver form}
F_K(x,q)=\sum_{\boldsymbol{d}\geq0}q^{\frac{1}{2}\boldsymbol{d}\cdot\boldsymbol{C}\cdot\boldsymbol{d}}\frac{\boldsymbol{x^{d}}}{(q)_{\boldsymbol{d}}}
\ee
with $x_i\propto x$, we have $n=\sum_i d_i$ and, therefore, $\widehat Z \big( S^3_{-1/r} (K) \big)$ is a~linear combination of ``characters''
\[
-\sum_{\boldsymbol{d}\geq 0} \frac{1}{(q)_{\boldsymbol{d}}}\,q^{r (\sum_i d_i)^2 + \frac{1}{2} \boldsymbol{d} \cdot C \cdot \boldsymbol{d}+ (\text{terms linear in } \boldsymbol{d})}
=
-\sum_{\boldsymbol{d}\geq 0} \frac{1}{(q)_{\boldsymbol{d}}}\,q^{ \frac{1}{2} \boldsymbol{d} \cdot C^{(2r)} \cdot \boldsymbol{d}+ (\text{terms linear in } \boldsymbol{d})}.
\]
Effectively, this means that the~modification of matrix $C$ is the~same as framing by $2r$ (see~\cite{KRSS2}). For example, for the~right-hand trefoil $\widehat Z \big( S^3_{-1/r} (K) \big)$ is a~linear combination of Nahm sums with the~matrix
\be
\label{Crtrefoil}
C^{(2r)}=
\begin{pmatrix}
	2r & 2r+1 & 2r & 2r \\
	2r+1 & 2r & 2r+1 & 2r \\
	2r & 2r+1 & 2r+1 & 2r \\
	2r & 2r & 2r & 2r+1
\end{pmatrix},
\ee
whereas for the~figure-eight we have the~matrix
\[
C^{(2r)}=
\begin{pmatrix}
2r & 2r-1 & 2r & 2r & 2r-1 & 2r \\
2r-1 & 2r & 2r & 2r & 2r & 2r+1 \\
2r & 2r & 2r & 2r & 2r & 2r \\
2r & 2r & 2r & 2r+1 & 2r & 2r \\
2r-1 & 2r & 2r & 2r & 2r+1 & 2r+1 \\
2r & 2r+1 & 2r & 2r & 2r+1 & 2r+1
\end{pmatrix}.
\]
We expect the~corresponding element of the~Bloch group and the~value of $c_{\text{eff}}$ to be qualitatively different for $r=1$ compared to $r>1$.

\subsection{Anomalies}

Recalling the~analysis from Section~\ref{sec:KQ for knot complements}, we can derive the~field content and interactions of the~3d $\CN=2$ theory $T[Q]$ and dual theory $T[M_K]$. 
The~matrix $C$ is the~matrix of Chern-Simons coefficients for $U(1)^m$ gauge theory. Indeed, we write
\[
q = e^{\hbar} \,, \qquad y_i = q^{d_i} = e^{\hbar d_i}
\]
and, as usual, take the~double-scaling limit $\hbar \to 0$ while keeping $\hbar d_i$ fixed. In this limit, using $d_i = \frac{1}{\hbar} \log y_i$, from the~quadratic term $\frac{1}{2} \boldsymbol{d}\cdot C \cdot \boldsymbol{d}$ we get
\be
\widetilde{\mathcal{W}}_K = \frac{1}{2} \sum_{ij} C_{ij} \log y_i \log y_j + \ldots
\label{WCS}
\ee
Note, the~powers of $q$ linear in 
$\boldsymbol{d}$ do not contribute to the~twisted superpotential. We denote it by $\widetilde{\mathcal{W}}_K$ to stress that we assume the knots-quivers identification $x_ i = (-1)^{t_i} q^{l_i} a^{a_i} x^{n_i}$, which implies that the moduli space of supersymmetric vacua of theories $T[Q]$ and $T[M_K]$ is the same, by construction.

Similarly, $x_i^{d_i}$ contributes $\log x_i \log y_i$ into $\widetilde{\mathcal{W}}_K$, and using\footnote{$\text{Li}_2 (x) = \sum_{n=1}^{\infty} \frac{x^n}{n^2}$, so that $\frac{d}{d \log x} \text{Li}_2 (x) = - \log (1 - x)$ and $\text{Li}_2 (1) = \frac{\pi^2}{6}$.}
\[
(x;q)_n \; = \; \prod_{i=0}^{n-1} (1 - x q^i) \; \sim \;
e^{\frac{1}{\hbar} \big( \text{Li}_2 (x) - \text{Li}_2 (x q^n) \big)}
\]
we conclude that the~$q$-Pochhammer symbols in the~denominator contribute to the~twisted superpotential a~term:
\be
\widetilde{\mathcal{W}}_K = \sum_i 
\big( \text{Li}_2 (y_i) - \text{Li}_2 (1) \big) + \ldots
\label{WLi}
\ee

For example, using these (by now standard) rules, for the~right-handed trefoil we get
\begin{multline}
\widetilde{\mathcal{W}}_{\mathbf{3}^r_1} = \log y_1 \log y_2 + \log y_2 \log y_3 + \frac{1}{2} (\log y_3)^2 + \frac{1}{2} (\log y_4)^2
+ \\
+ \log y_1 \log x
+ \log y_2 \log ax
+ (\log y_3 + \log y_4) \log (-ax)
+ \sum_{i=1}^4
\big( \text{Li}_2 (y_i) - \text{Li}_2 (1) \big).
\label{W31}
\end{multline}
The~field content and interactions of the~corresponding 3d $\CN=2$ theory can be conveniently summarized in a~quiver diagram illustrated in Figure~\ref{fig:TM3quiver}, where, as usual, a~circle denotes a~gauge node and a~square represents a~global (flavour) symmetry.

\begin{figure}[ht]
	\centering
	\includegraphics[width=4.0in]{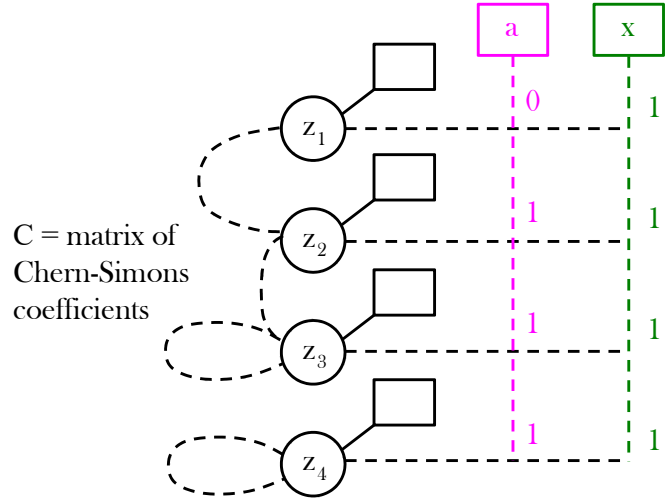}
	\caption{The~field content and interactions of 3d $\CN=2$ theory corresponding to the trefoil knot can be conveniently encoded in a~quiver form. Solid lines represent charged matter fields, whereas dashed lines represent Chern-Simons couplings.}
	\label{fig:TM3quiver}
\end{figure}

The~space of vacua in this theory has three ``branches'' (in terminology of the~$a=1$ specialization):
\be\label{branches}
\text{spurious} : x-1, \qquad\qquad
\text{abelian}  :  y-1, \qquad\qquad
\text{non-abelian} : y x^3 + 1,
\ee
where the~first, ``spurious'', branch comes with multiplicity 2.
When $a \ne 1$, one of the~spurious branches recombines with the~abelian and non-abelian branches into a~single irreducible component. The~other spurious branch becomes $ax-1=0$.

From the~quiver of 3d $\CN=2$ theory or from the~corresponding twisted superpotential we can easily read off the~anomaly coefficients, i.e. the~effective Chern-Simons couplings.
The~symmetry $U(1)_a$ is a~symmetry of 5d sector, and to account for the~anomalies of this symmetry, we need to incorporate into considerations the~``closed string'' 5d sector that so far has been ignored.
On the~other hand, the~symmetry $U(1)_x$ is a~global symmetry of the~``open string'' sector associated with the~Lagrangian brane, i.e. with a~3d surface operator in 5d bulk theory. All fields charged under this symmetry are part of the~3d $\CN=2$ theory $T[M_K]$, and therefore we focus on the~anomaly of $U(1)_x$ symmetry first.
Since the~symmetry $U(1)_x$ is preserved by the~2d $(0,2)$ boundary conditions that are used either in cutting and gluing or in computation of $\widehat Z$-invariants, we expect that the~effective Chern-Simons coupling for the~$U(1)_x$ should vanish:
\be
k_{xx} = 0.
\label{kxx}
\ee
Below we check that this anomaly vanishes on all branches.

In general, suppose that for a~knot $K$ the~corresponding $F_K$ can be written in the~quiver form \eqref{eq:FK quiver form}.
This means that 3d $\CN=2$ theory has $U(1)^m$ gauge symmetry with one chiral multiplet per each $U(1)$ factor and lots of Chern-Simons couplings for the~dynamical gauge fields, as well as for the~global symmetries $U(1)_a$ and $U(1)_x$. We are especially interested in the~latter, which can be computed as follows (see~\cite{GGP2} for details). First, from the~data of the~matrix $C$ we construct the~``dual'' matrix $G$ by writing a~quadratic form\footnote{This anomaly polynomial summarizes Chern-Simons couplings, including the~effective Chern-Simons coupling $- \frac{1}{2}$ for each dynamical $U(1)$ gauge factor that comes from integrating out a~charged chiral multiplet~\cite{GGP2}.}
\be
\sum_{i,j=1}^m C_{ij} u_i u_j - \frac{1}{2} \sum_{i=1}^m u_i^2 + 2 \sum_{i=1}^m u_i \kappa_i
\label{quadrform}
\ee
and completing the~squares for variables $u_i$. In other words, extremizing this quadratic form with respect to $u_i$ gives $m$ linear equations, from which we determine $u_i$ and substitute the~solution back into \eqref{quadrform}. This gives a~quadratic (Gaussian) expression for the~global symmetries,
\be
\sum_{ij} G_{ij} \kappa_i \kappa_j,
\label{Gauss}
\ee
where $\kappa_i$ should be understood as the~fugacity $\log x_i$. In other words,
\bea
\kappa_i & = & n_i \log x + a_i \log a + \pi i t_i  \notag \\
& = & n_i \kappa + a_i \alpha + \pi i t_i,
\label{mat}
\eea
where there is no summation over label $i$ and for convenience we introduced $(\kappa, \alpha)$ in place of $(\log x, \log a)$. Substituting \eqref{mat} into \eqref{Gauss}, we get the~desired quadratic polynomial in $\kappa$ and $\alpha$:
\[
k_{xx} \kappa^2 + k_{ax} \alpha \kappa + k_{aa} \alpha^2 + \ldots,
\]
whose coefficients are effective Chern-Simons terms for $U(1)_a$ and $U(1)_x$ global symmetries. In particular, according to \eqref{kxx}, we expect that $k_{xx} = 0$.

For example, for the~abelian and non-abelian branches of the~trefoil we get\footnote{Note the~symmetry between these two matrices. This is probably an~accident and does not hold for general knots; in particular, we have many more branches for general knots. On the~other hand, anomaly vanishing (or matching) discussed here are expected to hold for all knots.}
\[
\mathbf{3}_1: \quad
G_{\text{ab}} \; = \;
\begin{pmatrix}
	10 & 4 & -8 & 0 \\
	4 & 2 & -4 & 0 \\
	-8 & -4 & 6 & 0 \\
	0 & 0 & 0 & -2
\end{pmatrix}
,\quad  \quad
G_{\text{nab}} \; = \;
\begin{pmatrix}
	2 & 0 & 0 & 0 \\
	0 & -6 & 4 & 8 \\
	0 & 4 & -2 & -4 \\
	0 & 8 & -4 & -10
\end{pmatrix},
\]
which, together with the~data of the~vectors $\boldsymbol{a}$, $\boldsymbol{t}$ and $\boldsymbol{n}$, lead to the~anomalies
\[
0 \cdot \kappa^2 - 12 \alpha \kappa - 2 \alpha^2 - 16 \pi i \kappa - 4 \pi^2
\]
for the~abelian branch, and
\[
0 \cdot \kappa^2 + 8 \alpha \kappa + 0 \cdot \alpha^2 - 16 \pi i \kappa + 12 \pi i \alpha + 20 \pi^2
\]
for the~non-abelian branch. In both cases \eqref{kxx} holds, as expected.\footnote{Note, the~coefficients of terms linear in $\kappa$ also match on the~two branches. This, however, has no significance since these terms describe mod 2 anomaly, which automatically vanishes in both cases. In other words, we do have a~matching of this mod 2 discrete anomaly on both branches, but it is less impressive. Also, note that all anomalies listed here vanish in the~$\mathfrak{sl}_2$ specialization, i.e. when $a=1$ or $\alpha = 0$.}

Note, the~term $k_{xx}$ on which we focused here is a~very concrete combination of the~data $C$, $\boldsymbol{a}$, $\boldsymbol{t}$ and $\boldsymbol{n}$:
\[
k_{xx} = \sum_{ij} G_{ij} n_i n_j.
\]
Moreover, from the~physical point of view, extremization with respect to $u_i$ variables in \eqref{quadrform} comes from integrating over dynamical $U(1)^m$ gauge degrees of freedom. Therefore, the~dual matrix $G$ is basically the~inverse matrix to $C - \frac{1}{2} {\bf 1}$, so the~anomaly vanishing condition \eqref{kxx} can be stated as
\[
\boldsymbol{n} \cdot \frac{1}{C - \tfrac{1}{2} {\bf 1}} \cdot \boldsymbol{n} \; = \; 0.
\]
Note, this anomaly vanishing condition imposes a~constraint on $\boldsymbol{n}$ when only $C$ is known {\it a~priori}. It says that $\boldsymbol{n}$ is a~null vector of the~inverse matrix to $C - \frac{1}{2} {\bf 1}$ and should hold for any branch of the~theory and any choice of $C$.
Similarly, expressing $\widetilde{\mathcal{W}}_K$ in terms of $C$, $\boldsymbol{a}$, $\boldsymbol{t}$, $\boldsymbol{n}$ as above, and identifying $y = \exp \frac{d \widetilde{\mathcal{W}}_K}{d \log x}$ on different branches gives a~set of ``invariants,'' i.e. combinations of $C$, $\boldsymbol{a}$, $\boldsymbol{t}$, $\boldsymbol{n}$ that have the~same value on different branches of the~$A$-polynomial curve.

\subsection{Runaway vacua}

In \eqref{branches} we glossed over one important phenomenon: while the~assignment of branches makes sense only for $a=1$ specialization, actually it is not true that $a=1$ specialization of the~superpotential \eqref{W31} has all three branches as its critical points. In the~limit $a \to 1$ the~entire branch runs off to infinity.

Indeed, the~critical points of \eqref{W31} are solutions to the~following algebraic equations:
\begin{align}
xy_2 & =  1 - y_1,  & ay_3 (y_1 - 1) & =  1- y_3, \nonumber \\
ax y_1 y_3 & =  1 - y_2, &
-a x y_4 & =  1 - y_4. \label{BAExa}
\end{align}
The~first equation is linear in $y_2$; it has a~unique and simple solution for $y_2$ in terms of $y_1$. Similarly, the~third equation is linear in $y_3$ and also gives $y_3$ as a~function of $y_1$. The~fourth equation also has a~unique solution. Substituting the~resulting expressions for $y_2$, $y_3$ and $y_4$ into the~second equation we get a~quadratic equation for $y_1$, when $a \ne 1$ is generic.

However, when $a=1$, it becomes a~linear equation. It is easy to see that the~second solution of this equation runs off to infinity (in $\C^*$ where all our variables are valued) when $a \to 1$:
\be
y_1 \; \simeq \;
\frac{1-x}{x^2-x+1} (a-1) + O\left((a-1)^2\right).
\label{z1approx}
\ee
At the~same time, $y_3$ also goes off to infinity, so that the~product $y_1 y_3$ remains finite
\be
y_1 y_3 \; \simeq \; \frac{x-1}{x^2}+ \ldots
\label{z1z3}
\ee
As a~result, the~whole non-abelian branch is completely missing in the~theory \eqref{W31} when $a=1$. In order to keep this branch as part of the~vacua, we need to keep $a$ close to unity, but $a \ne 1$.

Note, in the~superpotential \eqref{W31} some of the~terms that involve $y_1$ and $y_3$ can be manifestly grouped into terms that depend only on the~product $y_1 y_3$ and, therefore, are non-singular in the~limit $a \to 1$. The~remaining terms that can be potentially singular are the~following:
\be
\frac{1}{2} (\log y_3)^2 
+ \log y_3 \log (-a)
+ \text{Li}_2 (y_1)
+ \text{Li}_2 (y_3).
\label{Wsingular}
\ee
Writing $z = y_3$, $y_1 = \frac{A}{z}$ and using the~dilogarithm reflection property
\[
\text{Li}_2 \left( \frac{1}{z} \right)
\; = \; - \text{Li}_2 \left( z \right)
- \frac{\pi^2}{6} - \frac{1}{2} \log^2 (-z),
\]
we can write the~potentially singular terms as
\begin{multline*}
\frac{1}{2} \log^2 (-z) + \log z \log a~- \frac{1}{2} \log^2 (-1) + \text{Li}_2 \left( \frac{A}{z} \right)
+ \text{Li}_2 (z) = \\
= - \frac{\pi^2}{6} + \frac{\pi^2}{2} + \log z \log a
+ \text{Li}_2 \left( \frac{A}{z} \right)
- \text{Li}_2 \left( \frac{1}{z} \right).
\end{multline*}
Since $z \sim (1-a)^{-1}$, the~term $\log z \log a$ goes to zero in the~limit $a \to 1$. Therefore, we only need to estimate $\text{Li}_2 \left( \frac{A}{z} \right)
- \text{Li}_2 \left( \frac{1}{z} \right)$ as $z \to \infty$ or, equivalently, $\text{Li}_2 \left( Aw \right)
- \text{Li}_2 \left( w \right)$ as $w \to 0$. Differentiating with respect to $w$, we see that
\[
\lim_{w \to 0} \, \frac{1}{w} \log \frac{1-w}{1-Aw} \; = \; A-1.
\]
Therefore, it follows that $\lim\limits_{z \to \infty} \Big[ \text{Li}_2 \left( \frac{A}{z} \right)
- \text{Li}_2 \left( \frac{1}{z} \right) \Big] \to \text{const}$ and it is then easy to check (say, numerically) that the~constant is zero.
We conclude that all potentially singular terms in the~superpotential \eqref{Wsingular} cancel each other up to a~finite contribution $\frac{\pi^2}{3}$.

\subsection{Closed 3-manifolds}

For a~surgery on a~knot $K$, we have the~relation
\be
\widetilde{\mathcal{W}} \left( S^3_{-p/r} (K) \right) = \frac{r}{p} (\log y)^2 + \log y \log x + \widetilde{\mathcal{W}}_K (x),
\label{Wsurgery}
\ee
where $y$ is a~new variable that came from $q^d$ in the~surgery formula for $\widehat Z$-invariants.
This is consistent with the~surgery formula for the~twisted superpotential in~\cite{GGP2,GMP}.

We wish to apply a~surgery formula to a~general expression of the~form
\be
F_K (x,a,q) \; = \; \sum_{n \ge 0} f_n (q) \cdot x^n
\; = \; \sum_{\boldsymbol{d}} q^{\frac{1}{2} \boldsymbol{d} C \boldsymbol{d} + \boldsymbol{d} \cdot \boldsymbol{B} + A}
\, \prod_{i} \frac{x_i^{d_i}}{(q)_{d_i}}
\label{Fnew}
\ee
with
$x_ i = (-1)^{t_i} q^{l_i} a^{a_i} x^{n_i}$.
Specializing to $a=q^2$ and then replacing $x \to q^{-1}x$, we get
\[
F_K (q^{-1} x,q) \; = \; \sum_{n \ge 0} f_n (q) \cdot (q^{-1} x)^n
= \sum_{\boldsymbol{d}} q^{\frac{1}{2} \boldsymbol{d} C \boldsymbol{d} + A }
\, \prod_{i} \frac{ (-1)^{t_i d_i} q^{(2a_i + l_i + B_i - n_i)d_i} x^{n_i d_i} }{(q)_{d_i}},
\]
where, of course, $n \; = \; \sum_i n_i d_i$.

Now we are ready to plug this expression in the~above surgery formula with $n = \sum_i n_i d_i$:
\begin{multline}
\widehat Z \big( S^3_{-1/r} (K) \big)
= \sum_{\boldsymbol{d}}
q^{\frac{1}{2} \boldsymbol{d} C \boldsymbol{d} + A + r n^2 + r n - 2n + \frac{r}{4} + \frac{1}{4r} - \frac{3}{2}}
\Big[ q - q^{2r+2nr} - q^{2+2n} + q^{3+2n+2r+2nr} \Big]
\\
\times \prod_{i} \frac{ (-1)^{t_i d_i} q^{(2a_i + l_i + B_i)d_i} }{(q)_{d_i}}.
\label{smallZhat}
\end{multline}

For example, for the~right-handed trefoil we have
\[
C \; = \;
\begin{pmatrix}
	0 & 1 & 0 & 0 \\
	1 & 0 & 1 & 0 \\
	0 & 1 & 1 & 0 \\
	0 & 0 & 0 & 1
\end{pmatrix},
\qquad
\begin{array}{rcl}
	A & = & 0, \\
	\boldsymbol{B} & = & (0,0,0,0), \\
	\boldsymbol{t} & = & (0,0,1,1), \\
	\boldsymbol{l} & = & (1,0,-\frac{1}{2},-\frac{1}{2}), \\
	\boldsymbol{a} & = & (0,1,1,1), \\
	\boldsymbol{n} & = & (1,1,1,1),
\end{array}
\]
and the~above surgery formula gives
\bea
\widehat Z \big( S^3_{-1} (\mathbf{3}_1^r) \big)  & = & 1-q-q^5+q^{10}-q^{11}+q^{18}+q^{30}-q^{41}+q^{43}-q^{56}+\ldots, \notag \\
\widehat Z \big( S^3_{-1/2} (\mathbf{3}_1^r) \big) & = & q^{1/8} \left( 1-q-q^{11}+q^{16}-q^{23}+q^{30}+q^{60}-q^{71}+q^{85}+\ldots \right) , \notag \\
\widehat Z \big( S^3_{-1/3} (\mathbf{3}_1^r) \big)  & = & q^{1/3} \left( 1-q-q^{17}+q^{22}-q^{35}+q^{42}+q^{90}-q^{101}+q^{127}+\ldots \right),  \notag \\
  & \vdots &  \notag
\eea
that agree with the~earlier calculations up to overall powers of $q$.

\subsection{Nahm sums and MTC\texorpdfstring{$[M_3]$}{[M3]}}

Using the~familiar rules \eqref{WCS}--\eqref{WLi}, we can quickly read off the~twisted superpotential from the~explicit expression \eqref{smallZhat}. For the~trefoil we get
\be
\begin{split}
\widetilde{\mathcal{W}} \left( S^3_{-1/r} (\mathbf{3}_1^r) \right) &= \log y_1 \log y_2 + \log y_2 \log y_3 + \frac{1}{2} (\log y_3)^2 + \frac{1}{2} (\log y_4)^2 \\
+ & r \left(\sum \log y_i\right)^2
+ (\log y_3 + \log y_4) \log (-1)
+ \sum_{i=1}^4
\big( \text{Li}_2 (y_i) - \text{Li}_2 (1) \big)
\label{Wsmall}
\end{split}
\ee
Note, the~quiver of the~corresponding 3d $\CN=2$ theory $T[S^3_{-1/r} (\mathbf{3}_1^r)]$ is very similar to the~one shown in Figure~\ref{fig:TM3quiver}. Namely, it is a~$U(1)^4$ gauge theory with lots of Chern-Simons couplings and one charged chiral per each $U(1)$ gauge factor.

Also note that \eqref{Wsmall} here is in agreement with the~general surgery formula \eqref{Wsurgery}.
Indeed, $\log y$ appears only quadratically in \eqref{Wsurgery}. Extremizing with respect to $\log y$, we generate a~quadratic term for $\log x$, which, in turn, appears in $\widetilde{\mathcal{W}}_{\mathbf{3}_1}$ only linearly,\footnote{Specializing $a \to q^2 \to 1$ in \eqref{W31} gives
$$
\log y_1 \log x
+ \log y_2 \log ax
+ (\log y_3 + \log y_4) \log (-ax) \to \log x (\sum \log y_i) + (\log y_3 + \log y_4) \log (-1)
$$} multiplying $\sum_i \log y_i$. Then integrating out $\log x$ gives $r (\sum \log y_i)^2 = r(\hbar m)^2$.

Now let us consider a~topological Landau-Ginzburg (LG) model~\cite{Vafa,MP} with the~(twisted) superpotential \eqref{Wsmall}. Its chiral ring and all partition / correlation functions are controlled by the~critical points,
\[
y_i^{(\lambda)}:\qquad
\exp \left( \frac{\partial \widetilde{\mathcal{W}}}{\partial \log y_i} \right) = 1 \qquad \forall i.
\]
These are the~celebrated Bethe~ansatz equations (BAEs) that in the~context of gauge/Bethe~and 3d-3d correspondence take the~familiar exponentiated form (as opposed to the~``2d version'' $\frac{\partial \widetilde{\mathcal{W}}}{\partial y_i} = 0$). For a~general expression of the~form \eqref{smallZhat} these BAEs take the~form
\be
1-y_i = (-1)^{t_i} \prod_{j} y_j^{C^{(2r)}_{ij}} \qquad (\forall i~\text{fixed}),
\label{BAEz}
\ee
where we used $\frac{d}{d \log x} \text{Li}_2 (x) = - \log (1-x)$ and $C^{(2r)}$ is the~matrix introduced earlier that combines all terms in $\widetilde{\mathcal{W}} \left( S^3_{-1/r} (K) \right)$ quadratic in $\log y_i$. Note that these equations generalize the~so-called Nahm equations~\cite{Nahm}:
\[
1-y_i = \prod_{j} y_j^{C_{ij}} \qquad (\forall i~\text{fixed}),
\]
and reduce to them when all $t_i = 0$. Already in the~simple example of the~trefoil knot, we have $\boldsymbol{t} = (0,0,1,1)$, so we have to consider a~more general version \eqref{BAEz}. In the~trefoil case, we get the~explicit form of the~BAEs using \eqref{Crtrefoil}:
\begin{align*}
1 - y_3^{-1} & =  1- y_1,   &   y_1 y_3 (1 - y_4^{-1}) & =  1 - y_2,  \\
y_2 (1 - y_4^{-1}) & =  1 - y_1,  &
y_1^{2r} y_2^{2r+1} y_3^{2r} y_4^{2r} & =  1 - y_1 .
\end{align*}
These equations, however, have no solutions $y_i \in \C^*$.
We have already seen this phenomenon in \eqref{BAExa}. One way to go around it is to work with $a \approx 1$ and then take the~limit $a \to 1$ at the~end of the~calculation.

The~relevant equations are
\[
y = y_1 y_2 y_3 y_4
, \qquad
x = y^{2r}
, \qquad
(1 - y)(1 + x^3 y)=0,
\]
and
\begin{align*}
xy_2 & =  1 - y_1 , & - ax y_2 y_3 & =  1- y_3 , \\
ax y_1 y_3 & =  1 - y_2,  &
-a x y_4 & =  1 - y_4 .
\end{align*}
From the~previous discussion we know that $y_2$, $y_3$, and $y_4$ are uniquely determined in terms of~$y_1$ which, in turn, obeys a~quadratic equation. The~non-abelian branch of the~$A$-polynomial corresponds to a~choice of solution \eqref{z1approx} that escapes ``to infinity'' in the~limit $a \to 1$:
\begin{align*}
y_1 & \simeq  \frac{1-x}{1-x+x^2} (a-1) + \ldots &  y_3 & \simeq  \frac{(x-1-x^2)}{x^2 (a-1)} , \\
y_2 & \simeq \frac{1}{x},  &
y_4 & = \frac{1}{1-ax} .
\end{align*}
Substituting these into \eqref{Wsmall}, we find
\begin{multline*}
\widetilde{\mathcal{W}} \left( S^3_{-1/r} (\mathbf{3}_1^r) \right) = \text{Li}_2\left(\frac{1}{1-x}\right)+\text{Li}_2\left(\frac{1}{x}\right)+r \log ^2\left(-\frac{1}{x^3}\right)+\log \left(\frac{x-1}{x^2}\right) \log \left(\frac{1}{x}\right) \\
+\frac{1}{2} \log \left(\frac{1}{1-x}\right) \left(\log \left(\frac{1}{1-x}\right)+2 \pi i \right)-\frac{\pi ^2}{3},
\end{multline*}
where we took advantage of \eqref{z1z3} and the~analysis below it. This twisted superpotential can be brought to the~form
\begin{align*}
\widetilde{\mathcal{W}} \left( S^3_{-1/r} (\mathbf{3}_1^r) \right) &=
r \log \left( \frac{-1}{x^3} \right)^2 + \log (x)^2 + \frac{1}{2} \log (-x)^2 - \pi i \log (-x) - \frac{\pi^2}{3} \\
&= r \log (y)^2 + \log (x)^2 + \frac{1}{2} \log (-x)^2 - \pi i \log (-x) - \frac{\pi^2}{3}
\end{align*}
from which classical Chern-Simons values can be seen at critical points of $\widetilde{\mathcal{W}}$.

This identification is a~part of much richer structure. Namely, it was proposed in~\cite{GPV} that to each 3d $\CN=2$ theory one can associate a~braided tensor category of line operators, which in many cases is a~modular tensor category.\footnote{See~\cite{FG,DGNPY,CGK,CCFGH} for further discussion and applications.}
For 3d $\CN=2$ theories $T[M_3]$, it is denoted $\text{MTC} [M_3]$.
The~Grothendieck ring of this category is the~Jacobi ring of the~Landau-Ginzburg model with the~superpotential $\widetilde{\mathcal{W}}$. In particular, critical points of $\widetilde{\mathcal{W}}$ correspond to simple objects of $\text{MTC} [M_3]$, and various data -- such as matrix elements of $S$ and $T$ matrices, conformal dimensions, effective central charge $c_{\text{eff}}$,  etc. -- are determined by this effective Landau-Ginzburg model:
\be
T_{\lambda \mu} = \delta_{\lambda \mu} e^{\widetilde{\mathcal{W}} (\lambda)}, \label{MTCM3} \qquad\qquad
S_{0 \lambda}  =  e^{- U} \det \text{Hess} \, \widetilde{\mathcal{W}} \Big|_{\lambda} .
\ee
As a~result, the~twisted partition function on a~genus-$g$ surface $\Sigma_g$ can be written in the~standard form~\cite{Vafa,MP}:
\[
Z (S^1 \times \Sigma_g)
\; = \; \sum_{\lambda} (S_{0 \lambda})^{2-2g}
\; =
\sum_{\text{crit. pts. of } \widetilde{\mathcal{W}}}
\left( e^{- U} \det \text{Hess} \, \widetilde{\mathcal{W}} \right)^{g-1},
\]
with $S_{0 \lambda}$ given by \eqref{MTCM3}.

The~effective central charge $c_{\text{eff}}$ can be read off from the~asymptotic growth of the~coefficients~$a_n$ in the~expansion of $\widehat Z$-invariants, which are identified with the~characters of logarithmic vertex operator algebras (log-VOAs):
\[
\chi_b (q) \; = \; \widehat Z_b (q) \; = \; q^{\Delta_b} \sum_n a_n q^n,
\]
where it is important that we include in $\widehat Z_a (q)$ the~$(q)_{\infty}$ denominator associated with the~center of mass chiral multiplet of the~$T[M_3]$ theory~\cite{CCFGH}. Specifically, as $n \to \infty$,
\be
a_n \sim \exp \, 2\pi \sqrt{\frac{1}{6} c_{\text{eff}} \, n}
\label{ceffan}
\ee
For example, for small surgeries on the~right-handed trefoil knot we have
$$
\widehat{Z}_0 \big( S^3_{-1/r} (\mathbf{3}_1^r) \big) \; = \;
\frac{1}{\eta(q)}
\left(
\widetilde \Psi_{36r+6}^{(6r-5)}
- \widetilde \Psi_{36r+6}^{(6r+7)}
- \widetilde \Psi_{36r+6}^{(30r-1)}
+ \widetilde \Psi_{36r+6}^{(30r+11)}
\right)
$$
and, therefore, from \eqref{MTCM3} and \eqref{ceffan} we read off
$$
c_{\text{eff}} = c - 24 h_{\text{min}} = 1
\qquad \text{and} \qquad
c = 1 - \frac{(36r+5)^2}{6r+1}.
$$

Since in this case we actually know the~precise log-VOA whose characters can be identified with $\widehat Z$-invariants, we can check that these indeed match with the~VOA central charges for all values of $r \in \Z_+$. 

The~situation is a~little more interesting in cases where the~precise log-VOA has not been identified yet. In those cases, we can use \eqref{MTCM3} and \eqref{ceffan} to obtain its modular data, central charges and conformal dimensions $h_{\lambda}$ related to $T_{\lambda \lambda} = e^{2\pi i (h_{\lambda} - \frac{c}{24})}$.
The~simplest class of such examples is a~family of small surgeries on the~figure-eight knot $S^3_{-1/r} (\mathbf {4}_1)$. For example, for $r=2$ we have
\begin{multline}
\widehat Z (S^3_{-1/2} ({\bf 4_1}))
\; = \; \frac{q^{-\frac{1}{2}}}{\eta (q)}
\big( 1-q+2 q^3-2 q^6+q^9+3 q^{10}+q^{11}-q^{14}-3 q^{15}
+ \ldots \\
\ldots -15040 q^{500} + \ldots \big).
\label{Z41half}
\end{multline}
Note, without the~Dedekind $\eta$-function in the~denominator, this would be a~$q$-series with oscillation coefficients of a~growing amplitude. However, with $\eta (q)$ in the~denominator, all coefficients $a_n$ have the~same sign, as expected for a~character of a~(logarithmic) VOA. Moreover, their rate of growth \eqref{ceffan}, illustrated in Figure~\ref{fig:half41}, is consistent with
\[
c_{\text{eff}} \approx 1
\]
which suggests that $c_{\text{eff}} = 1$, as in the~case of small surgeries on the~trefoil.

\begin{figure}[ht]
	\centering
	\includegraphics[width=3.5in]{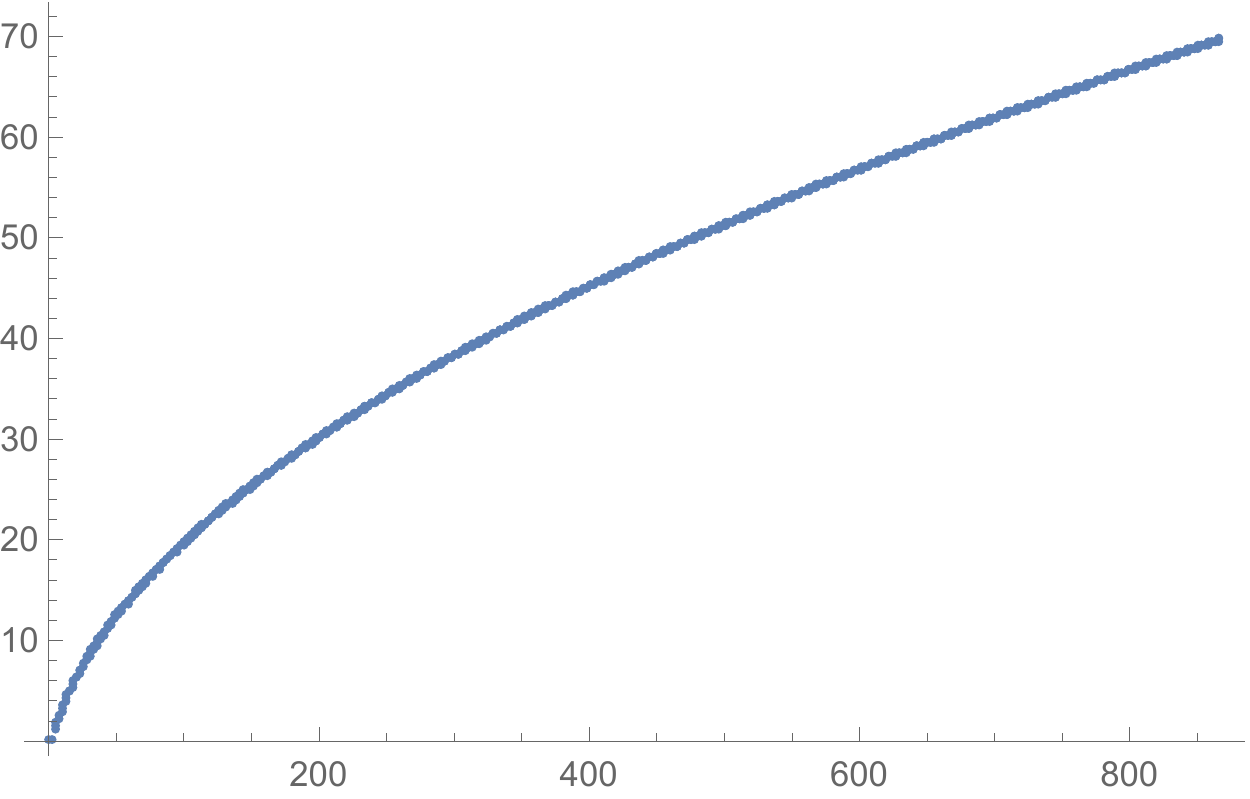}
	\caption{Plot of $\log a_n$ as a~function of $n$ for $- \frac{1}{2}$ surgery on the~figure-eight knot.}
	\label{fig:half41}
\end{figure}

\section{Future directions} \label{sec-future}

In this section, we provide a~summary of interesting open problems that emerged during our research:

\begin{itemize}
    \item Is there a~way to unify the~wave functions $F_K^{(\alpha)}(x,a,q)$ associated to different branches? Recall that the~wave functions $F_K^{(\alpha)}(x,a,q)$ are associated to branches on the~boundary $(x,a)=(0,0)$ (or via Weyl symmetry, $(x,a)=(\infty,\infty)$). In order to unify them, we probably need to understand better what happens in the~middle of the~holomorphic Lagrangian. 
    
    An interesting observation in this direction is that for $\mathfrak{sl}_2$ the~non-abelian branch $F_K$ invariants seem to appear as a~certain limit of 3d index from \cite{DGG}. 
    \begin{conj}
    Let $K$ be a~knot and let $\alpha$ denote a~non-degenerate, non-abelian branch of the~$A$-polynomial of slope $-\frac{1}{p}$. That is, $F_K^{\alpha}$ has a~prefactor of the~form $e^{\frac{p(\log x)^2}{2\hbar}}$.
    Let's write (unreduced) $F_K^{\alpha}(x,q)$ as
    \[
    F_K^{\alpha}(x,q) = e^{\frac{\frac{p}{2}(\log x)^2 + \epsilon \log(-1)\log x}{\hbar}}x^d\sum_{j=0}^{\infty} f_j^{\alpha}(q)x^j
    \]
    for some $\epsilon \in \{0,1\}$ and $d\in \frac{1}{2}\mathbb{Z}$. 
    Then $f_j^{\alpha}(q)$ can be obtained from $\mathcal{I}_K(m,e)$ by
    \[
    f_j^{\alpha}(q) = \lim_{m\rightarrow \infty} (-1)^{(\epsilon+1)m} q^{-d(j+1) m}\mathcal{I}_K(m,p\,m+j).
    \]
    \end{conj}
    For instance, for the~figure-eight knot, with $p=2$, we empirically checked that
    \begin{align*}
    \lim_{m\rightarrow \infty}(-1)^m q^{-\frac{1}{2}m} \mathcal{I}_{\mathbf{4}_1}(m,2m) &= 1, \\
    \lim_{m\rightarrow \infty}(-1)^m q^{-\frac{2}{2}m} \mathcal{I}_{\mathbf{4}_1}(m,2m+1) &= -q-q^2-q^3-q^4-q^5-q^6-\cdots\\ 
    &= \frac{1-2q}{1-q}-1,\\
    \lim_{m\rightarrow \infty}(-1)^m q^{-\frac{3}{2}m} \mathcal{I}_{\mathbf{4}_1}(m,2m+2) &= -q-q^2+q^5+q^6+2q^7+\cdots\\ 
    &= \frac{1-3q-q^2+4q^3}{(1-q)(1-q^2)} - \frac{1-2q}{1-q},
    \end{align*}
    so they match up perfectly. 
    If true, this conjecture can be seen as a~way to unify non-abelian branch $F_K$ invariants in $\mathfrak{sl}_2$ case. 
    
    In recent works of \cite{GGM2,GGM1}, authors studied resurgence in complex Chern-Simons theory and identified the~entries of the~Stokes matrix between non-abelian flat connections with the~3d index from \cite{DGG}. Moreover, the~abelian branch $F_K$ is related to the~first row of the~Stokes matrix. 
    This is a strong encouragement to study resurgence in the~$a$-deformed setting, in particular looking for a~generalization of 3d index that unifies $F_K$ from all the~branches. 
    
    \item How are quivers associated to different branches related to each other? So far we only understand how they are related in a~rather indirect way; e.g. they give rise to the~same holomorphic Lagrangian after eliminating variables from the~quiver $A$-polynomials. It would be nice to understand how to get from a~quiver from one branch to a~quiver from another branch. 
    Probably, understanding of the~previous question (i.e. unification of $F_K$'s for various branches) will lead to an~answer to this question and vice versa. 
    
     \item It would be interesting to understand the~categorification of $F_K$ invariants in the~context of quivers. Recent results for knot conormals~\cite{EKL3} suggest that the~presence of quiver nodes with $x_i \sim y^{-n}, \ n>1$ imply a~different structure of the~refined generating function. Another subtlety lies in the~fact that the~$t$-deformation inherited from the~superpolynomials and super-$A$-polynomials is not consistent with the~one guided by the~number of loops. Further complication comes from the~fact that different forms of knot complement quivers would lead to different $t$-deformations. 
    
    \item It would be desirable to find explicit formulas for $F_K$ invariants in quiver forms for knots and branches not covered in Section \ref{sec:Examples} and Appendix \ref{sec:Quivers appendix}.  The case of non-abelian branch of slope $\infty$ for knot $5_2$, partly discussed in Appendix \ref{sec:Quivers appendix}, seems to be especially interesting.
\end{itemize}

\section*{Acknowledgements}

P.K. was supported by the~Polish Ministry of Education and Science through its programme Mobility Plus (decision number~1667/MOB/V/2017/0) and by NWO vidi grant (number~016.Vidi.189.182). S.P. was partially supported by junior fellowship at Institut Mittag-Leffler and by Kwanjeong Educational Foundation. The work of M.S. was supported by the Portuguese Funda\c c\~ao para a Ci\^encia e a Tecnologia (FCT) Exploratory Grant IF/0998/2015, and by the Ministry of Education, Science and Technological Development of the Republic of Serbia through Mathematical Institute SANU.  The~work of P.S. was supported by the~TEAM programme of
the~Foundation for Polish Science co-financed by the~European Union under the~European
Regional Development Fund (POIR.04.04.00-00-5C55/17-00).

\newpage

\appendix

\section{Quivers for complements of various knots}\label{sec:Quivers appendix}

In this appendix we present formulas for $F_K$ invariants for all knots with 5 or 6 crossings, as well as (3,4) torus knot, obtained using the~method presented in Section \ref{sec:41}. For simplicity, we will keep $q$-Pochhammers $(x;q^{-1})_{(...)}$ in the~concise form. In case the~quiver form of $F_K$ is needed, they can be easily expanded in two ways, as discussed in Section \ref{sec:41}.

\subsubsection*{\texorpdfstring{$\mathbf{5}_1$}{51} {\bf knot}}

The~$F_K$ invariant for the~abelian branch of $5_1$ is given by

\begin{eqnarray*}
F_{\mathbf{5}_1}(x,a,q)&=&\sum_{\tilde{d}_1,\tilde{d}_2,\tilde{d}_3,\tilde{d}_4}(-1)^{\tilde{d}_2+\tilde{d}_4}a^{\tilde{d}_2+\tilde{d}_4}q^{\frac{1}{2} (2\tilde{d}_1-\tilde{d}_2+4\tilde{d}_3+\tilde{d}_4)} \\
&&\times\,\, x^{\tilde{d}_1+\tilde{d}_2+3\tilde{d}_3+3\tilde{d}_4}q^{\frac{1}{2}\sum_{i,j} \tilde{C}_{ij}\tilde{d}_i\tilde{d}_j} \times \frac{(x;q^{-1})_{\tilde{d}_1+\cdots+\tilde{d}_4}}{\prod_{i=1}^4(q)_{\tilde{d}_i}},
\end{eqnarray*}

where
$$\tilde{C}=\left(\begin{array}{cccc}
 0 & 0 & -1 & -1\\
 0 & 1 & 0 & 0\\
 -1 & 0 & -2 & -2\\
 -1 & 0 & -2 & -1
\end{array}\right).
$$
Expanding the~$q$-Pochhammer $(x;q^{-1})_{\tilde{d}_1+\cdots+\tilde{d}_4}$ one can check that this formula is consistent with Equation \eqref{eq:51quiver}.

The~quiver form of the~$F_K$ for the~mirror image $m(\mathbf{5}_1)=\mathbf{5}_1^r$ and its non-abelian branch with slope $-\frac{1}{5}$ (corresponding to framing $f=5$) is given by
\begin{eqnarray*}
F_{\mathbf{5}_1^r}^{(-\frac{1}{5})}(x,a,q)&=&\sum_{\tilde{d}_1,\cdots,\tilde{d}_4}(-1)^{\tilde{d}_1+\tilde{d}_2+\tilde{d}_4}a^{\tilde{d}_1+\tilde{d}_2+\tilde{d}_4}\\
&&\times\,\, q^{\frac{1}{2}(\tilde{d}_1-\tilde{d}_2+2\tilde{d}_3-3\tilde{d}_4)} x^{\tilde{d}_1+\tilde{d}_2}q^{-\frac{1}{2}\sum_{i,j} \tilde{C}_{ij}\tilde{d}_i\tilde{d}_j}\\
&&\times\,\, q^{(3\tilde{d}_1+3\tilde{d}_2+4\tilde{d}_3+4\tilde{d}_4)\sum_i\tilde{d}_i}q^{-\frac{5}{2}(\sum_{i}\tilde{d}_i)^2}q^{-\sum_{i<j}\tilde{d}_i\tilde{d}_j}\frac{(x;q^{-1})_{\tilde{d}_1+\cdots+\tilde{d}_4}}{\prod_{i=1}^4(q)_{\tilde{d}_i}},
\end{eqnarray*}
where 
$$\tilde{C}=\left(\begin{array}{cccc}
0 & 1 & 1 & 3 \\
1 & 2 & 2 & 3 \\
1 & 2 & 3 & 4 \\
3 & 3 & 4 & 4 \\
\end{array}\right).
$$

\subsubsection*{\texorpdfstring{$\mathbf{5}_2$}{52} {\bf knot}}
The~$F_K$ invariant for $\mathbf{5}_2$ corresponding to the~abelian branch reads:
\begin{eqnarray*}
F_{5_2}(x,a,q)&=&\sum_{\tilde{d}_1,\ldots,\tilde{d}_6}(-1)^{\tilde{d}_2+\tilde{d}_3+\tilde{d}_6}a^{\tilde{d}_2+\tilde{d}_4+\tilde{d}_5+2\tilde{d}_6}q^{\frac{1}{2} (2\tilde{d}_1-\tilde{d}_2+\tilde{d}_3-2\tilde{d}_5-3\tilde{d}_6)}  \\
&&\times\,\, x^{\tilde{d}_1+\tilde{d}_2+2\tilde{d}_4+\tilde{d}_6}q^{\frac{1}{2}\sum_{i,j} \tilde{C}_{ij}\tilde{d}_i\tilde{d}_j} \frac{(x;q^{-1})_{\tilde{d}_1+\cdots+\tilde{d}_6}}{\prod_{i=1}^6(q)_{\tilde{d}_i}},
\end{eqnarray*}
where 
$$\tilde{C}=\left(\begin{array}{cccccc}0&0&0&-1&0&0\\
0&1&0&0&1&1\\
0&0&1&0&1&1\\
-1&0&0&0&1&1\\
0&1&1&1&2&2\\
0&1&1&1&2&3
\end{array}\right).$$

For the~mirror image $m(\mathbf{5}_2)=\mathbf{5}_2^r$ we need to consider it in framing $f=5$ and we get $F_K$ for the~non-abelian branch corresponding to slope $-\frac{1}{5}$:
\begin{eqnarray*}
F_{\mathbf{5}_2^r}^{(-\frac{1}{5})}(x,a,q)&=&\sum_{\tilde{d}_1,\cdots,\tilde{d}_{6}}(-1)^{\tilde{d}_1+\tilde{d}_2+\tilde{d}_5+\tilde{d}_6}a^{2\tilde{d}_1+2\tilde{d}_2+\tilde{d}_3+2\tilde{d}_4+\tilde{d}_5+\tilde{d}_6} \\
&&\times\,\, q^{\frac{1}{2}(-5\tilde{d}_1-3\tilde{d}_2-2\tilde{d}_3-4\tilde{d}_4-3\tilde{d}_5-\tilde{d}_6)} x^{2\tilde{d}_1+3\tilde{d}_2+\tilde{d}_3+2\tilde{d}_4+\tilde{d}_6}q^{-\frac{1}{2}\sum_{i,j} \tilde{C}_{ij}\tilde{d}_i\tilde{d}_j}\\
&&\times\,\, q^{(2\tilde{d}_1+\tilde{d}_2+3\tilde{d}_3+2\tilde{d}_4+4\tilde{d}_5+3\tilde{d}_6)\sum_i\tilde{d}_i}q^{-\frac{5}{2}(\sum_{i}\tilde{d}_i)^2}q^{-\sum_{i<j}\tilde{d}_i\tilde{d}_j}\frac{(x;q^{-1})_{\tilde{d}_1+\cdots+\tilde{d}_{6}}}{\prod_{i=1}^{6}(q)_{\tilde{d}_i}},
\end{eqnarray*}
where
\begin{equation*}
\tilde{C}=\left(\begin{array}{cccccc}
2 & 1 & 2 & 1 & 2 & 1 \\
1 & 0 & 1 & 0 & 2 & 0 \\
2 & 1 & 3 & 1 & 3 & 2 \\
1 & 0 & 1 & 1 & 2 & 1 \\
2 & 2 & 3 & 2 & 4 & 3 \\
1 & 0 & 2 & 1 & 3 & 2 \end{array}\right).
\end{equation*}

For knots for which $f_0 \ne 1$, some of the~coefficients of $\tilde{d}_i$ in the~exponent of $x$ are equal to~$0$. For example, let us take $5_2$ for infinite slope branch, and try to find quivers.
We have 
$$f_0^{\mathfrak{sl}_N}=\frac{1}{(q)_{N-2}} \sum_{i=0}^{N-2} \left[ \begin{array}{c} N-2\\ i \end{array} \right].$$

As we know, we can rewrite this in quiver form:

$$f_0^{\mathfrak{sl}_N}=\sum_{i=0}^{N-2} \frac{1}{(q)_i (q)_{N-2-i}}=\sum_{i=0}^{\infty} \frac{(q^{N-1-i};q)_{\infty}}{(q)_i (q)_{\infty}}=\sum_{i=0}^{\infty}\sum_{j=0}^{\infty}\sum_{k=0}^{\infty} \frac{(-1)^j q^{Nj} q^{-j-i j} q^{\frac{1}{2}(j^2-j)}q^k}{(q)_i (q)_j (q)_k},$$
and so
$$f_0(a,q)=\sum_{\tilde{d}_1,\tilde{d}_2,\tilde{d}_3} \frac{(-1)^{\tilde{d}_2} a^{\tilde{d}_2} q^{-\tilde{d}_2-\tilde{d}_1 \tilde{d}_2} q^{\frac{1}{2}(\tilde{d}_2^2-\tilde{d}_2)}q^{\tilde{d}_3}}{(q)_{\tilde{d}_1} (q)_{\tilde{d}_2} (q)_{\tilde{d}_3}}.$$

This fixes the~first $3\times3$ block of the~quiver matrix and the~first 3 entries of the~rows of~$\Bf{a,l}$.
For further ones, let us denote 
$$a_N:=f_0^{\mathfrak{sl}_N}(=\sum_{i,j,k} \frac{(-1)^j q^{Nj} q^{-j-i j} q^{\frac{1}{2}(j^2-j)}q^k}{(q)_i (q)_j (q)_k}),\qquad\qquad
b_N:=f_1^{\mathfrak{sl}_N}.$$
Then we have the~recursion
$$(1-q^N)a_{N+2}=2a_{N+1}-a_N, \qquad (1-q) b_N=(1-q^{N-1})(a_N+a_{N+1}).$$

This suggests that we should have (at least) 4 nodes whose $\tilde{d}_i$'s have coefficient 1 in the~exponent of $x$ in $F_K$. The~following formula (that is relevant for the~off-diagonal entries in the~quiver for the~expressions in $f_1(a,q)$ and higher ones) helps in simplifying the~expressions:

\begin{lem}\label{lemma}
For integers $\alpha,\beta,\gamma\ge 0$ we have:
$$
\sum_{i,j,k} q^{\alpha i+ \beta j+\gamma k}\frac{(-1)^j q^{Nj} q^{-j-i j} q^{\frac{1}{2}(j^2-j)}q^k}{(q)_i (q)_j (q)_k}=(q)_{\gamma} \sum_{i=0}^{N-2+\beta}\frac{q^{\alpha i}}{(q)_i(q)_{N-2+\beta-i}}.
$$
In addition, if we denote
$$
P_{\alpha,\beta}(N)=\sum_{i=0}^{N-2+\beta}\frac{q^{\alpha i}}{(q)_i(q)_{N-2+\beta-i}},
$$
then it can be obtained recursively by
$$P_{0,\beta}(N)=a_{N+\beta},$$
and
$$P_{\alpha,\beta}(N)=a_{N+\beta}-\sum_{\delta=0}^{\alpha-1} q^{\delta} P_{\delta,\beta-1},\quad \alpha>0.$$

In particular:
\begin{equation}\label{gl}
\sum_{i,j,k} q^{ \beta j+\gamma k}\frac{(-1)^j q^{Nj} q^{-j-i j} q^{\frac{1}{2}(j^2-j)}q^k}{(q)_i (q)_j (q)_k}=(q)_{\gamma} a_{N-2+\beta}.
\end{equation}

\end{lem}

In consequence, the~first guess for the~first $7\times7$ block can be:

$$C=\left(
\begin{array}{ccccccc}
0&-1&0&0&0&0&0\\
-1&1&0&0&1&0&1\\
0&0&0&0&0&0&0\\
0&0&0&0&0&0&0\\
0&1&0&0&0&0&0\\
0&0&0&0&0&1&0\\
0&1&0&0&0&0&1\\
\end{array}
\right),
\qquad
\begin{array}{l}
\Bf{a}=(0,1,0,0,0,1,1),\\
\Bf{l}=0,-1,1,0,0,-1,-1),\\
\Bf{n}=(0,0,0,1,1,1,1).
\end{array}
$$
By using Lemma \ref{lemma} (especially (\ref{gl})), one can check that we get the~correct answer for $f_1^{\mathfrak{sl}_N}$. 
However, one can see that the~$f_2^{\mathfrak{sl}_N}$ is not quite correct, but it is close. Therefore either we should either change $C$ (choose different off-diagonal entries), or add more nodes/generators. They can also be added with different flavours: some of them proportional to $x$, and maybe some of them proportional to $x^2$.

\subsubsection*{\texorpdfstring{$\mathbf{6}_1$}{61} {\bf knot}}


For the~knot $\mathbf{6}_1$, the~ $F_K$ invariant corresponding to the~non-abelian branch with slope $-\frac{1}{2}$ is given by
\begin{eqnarray*}
F_{\mathbf{6}_1}^{(-\frac{1}{2})}(x,a,q)&=&\sum_{\tilde{d}_1,\ldots,\tilde{d}_8}(-1)^{\tilde{d}_3+\tilde{d}_4+\tilde{d}_6+\tilde{d}_7}a^{\tilde{d}_1+2\tilde{d}_2+\tilde{d}_3+\tilde{d}_4+3\tilde{d}_5+2\tilde{d}_6+2\tilde{d}_7+\tilde{d}_8} \\
&&\times\,\,
q^{\frac{1}{2}(-2\tilde{d}_1-4\tilde{d}_2-3\tilde{d}_3-\tilde{d}_4-6\tilde{d}_5-5\tilde{d}_6-3\tilde{d}_7-2\tilde{d}_8)}x^{\tilde{d}_1+\tilde{d}_2+\tilde{d}_4+2\tilde{d}_5+\tilde{d}_6+2\tilde{d}_7}  \\
&&\times\,\,
q^{\frac{1}{2}\sum_{i,j} \tilde{C}_{ij}\tilde{d}_i\tilde{d}_j+(\tilde{d}_1+\tilde{d}_2+2\tilde{d}_3+\tilde{d}_4+\tilde{d}_6+2\tilde{d}_8)\sum_{i}\tilde{d}_i-(\sum_i \tilde{d}_i)^2}\frac{(x;q^{-1})_{\tilde{d}_1+\cdots+\tilde{d}_8}}{\prod_{i=1}^8(q)_{\tilde{d}_i}},\end{eqnarray*}
where
\begin{equation*}
\tilde{C}=\left(\begin{array}{ccccccccc}
0 & 0 & -1 & 0  & 0 & -1 & 0 & -1\\
0 & 2 & 0 & 1  & 2 & 0 & 1 & -1\\
-1 & 0 & -1 & 0  & 1 & -1 & 0 & -2\\
0 & 1 & 0 & 1 & 2 & 1 & 1 & -1\\
0 & 2 & 1 & 2  & 4 & 2 & 3 & 1\\
-1 & 0 & -1 & 1  & 2 & 1 & 2 & 0\\
0 & 1 & 0 & 1  & 3 & 2 & 3 & 1\\
-1 & -1 & -2 & -1 & 1 & 0 & 1 & 0
\end{array}\right).         
\end{equation*}
%

For the~mirror image $m(\mathbf6_1)=\mathbf6_1^r$ and for non-abelian branch of slope $-\frac{1}{4}$ corresponding to framing $f=4$, the~$F_K$ invariant is given by
\begin{eqnarray*}
F_{\mathbf6_1^r}^{(-\frac{1}{4})}(x,a,q)&=&\sum_{\tilde{d}_1,\cdots,\tilde{d}_{8}}(-1)^{\tilde{d}_3+\tilde{d}_4+\tilde{d}_6+\tilde{d}_7}a^{2\tilde{d}_1+\tilde{d}_2+2\tilde{d}_3+2\tilde{d}_4+3\tilde{d}_5+\tilde{d}_6+\tilde{d}_7+2\tilde{d}_8} \\
&&\times\,\, q^{\frac{1}{2}(-4\tilde{d}_1-2\tilde{d}_2-3\tilde{d}_3-5\tilde{d}_4-6\tilde{d}_5-\tilde{d}_6-3\tilde{d}_7-4\tilde{d}_8)} x^{3\tilde{d}_1+\tilde{d}_2+2\tilde{d}_3+\tilde{d}_4+3\tilde{d}_5+\tilde{d}_6+2\tilde{d}_8}\\
&&\times q^{-\frac{1}{2}\sum_{i,j} \tilde{C}_{ij}\tilde{d}_i\tilde{d}_j}\\
&&\times\,\, q^{(2\tilde{d}_2+\tilde{d}_3+2\tilde{d}_4+2\tilde{d}_6+3\tilde{d}_7+\tilde{d}_8)\sum_i\tilde{d}_i}q^{-2(\sum_{i}\tilde{d}_i)^2}q^{-\sum_{i<j}\tilde{d}_i\tilde{d}_j}\frac{(x;q^{-1})_{\tilde{d}_1+\cdots+\tilde{d}_{8}}}{\prod_{i=1}^{8}(q)_{\tilde{d}_i}},
\end{eqnarray*}
where
\begin{equation*}
\tilde{C}=\left(\begin{array}{cccccccc}
0 & 0 & -1 & 0 & -1  & -1 & 0 & -1\\
0 & 2 & 0 & 1 & -1  & 0 & 1 & -1\\
-1 & 0 & -1 & 0 & -2 & -1 & 0 & -2\\
0 & 1 & 0 & 1 & -1  & 1 & 1 & -1\\
-1 & -1 & -2 & -1 & -2 & -1 & 0 & -2\\
-1 & 0 & -1 & 1 & -1  & 1 & 2 & 0\\
0 & 1 & 0 & 1 & 0 & 2 & 3 & 1\\
-1 & -1 & -2 & -1 & -2  & 0 & 1 & 0
\end{array}\right).       
\end{equation*}

\subsubsection*{\texorpdfstring{$\mathbf6_2$}{62} {\bf knot}}

For the~knot $\mathbf6_2$ in framing $f=2$, the~$F_K$ invariant corresponding to a~non-abelian branch with slope $-\frac{1}{2}$, is given by
\begin{eqnarray*}
F_{\mathbf6_2}^{(-\frac{1}{2})}(x,a,q)&=&\sum_{\tilde{d}_1,\ldots,\tilde{d}_{10}}(-1)^{\tilde{d}_1+\tilde{d}_4+\tilde{d}_5+\tilde{d}_8+\tilde{d}_9}a^{\tilde{d}_1+\tilde{d}_2+\tilde{d}_4+\tilde{d}_5+\tilde{d}_6+\tilde{d}_9+2\tilde{d}_7+2\tilde{d}_8+2\tilde{d}_{10}}\\
&&\times\,\,
q^{\frac{1}{2}(-3\tilde{d}_1-2\tilde{d}_2+2\tilde{d}_3-\tilde{d}_4-\tilde{d}_5-4\tilde{d}_7-3\tilde{d}_8+\tilde{d}_9-2\tilde{d}_{10})}x^{\tilde{d}_2+\tilde{d}_3+\tilde{d}_4+\tilde{d}_5+2d+6+\tilde{d}_7+3\tilde{d}_8+4\tilde{d}_9+3\tilde{d}_{10}} \\
&&\times\,\,
q^{\frac{1}{2}\sum_{i,j} \tilde{C}_{ij}\tilde{d}_i\tilde{d}_j+(2\tilde{d}_1+\tilde{d}_2+\tilde{d}_3+\tilde{d}_4+\tilde{d}_5+\tilde{d}_7-\tilde{d}_8-\tilde{d}_9-\tilde{d}_{10})\sum_{i}\tilde{d}_i-(\sum_i \tilde{d}_i)^2}\frac{(x;q^{-1})_{\tilde{d}_1+\cdots+\tilde{d}_{10}}}{\prod_{i=1}^{10}(q)_{\tilde{d}_i}},\end{eqnarray*}
where
\[
\tilde{C}=\left(
\begin{array}{cccccccccc}
 -1\ & -1\ & 0\ & 0\ & 0\ & 1\ & 0\ & 1\ & 2\ & 2 \\
-1\ & 0\ & 1\ & 0\ & 0\ & 1\ & 0\ & 1\ & 2\ & 2 \\
 0\ & 1\ & 0\ & 0\ & 0\ & 1\ & 0\ & 2\ & 1\ & 1 \\
0\ & 0\ & 0\ & 1\ & 1\ & 1\ & 1\ & 2\ & 2\ & 2 \\
 0\ & 0\ & 0\ & 1\ & 1\ & 1\ & 1\ & 2\ & 2\ & 2 \\
 1\ & 1\ & 1\ & 1\ & 1\ & 2\ & 1\ & 2\ & 2\ & 2 \\
 0\ & 0\ & 0\ & 1\ & 1\ & 1\ & 2\ & 2\ & 3\ & 3 \\
  1\ & 1\ & 2\ & 2\ & 2\ & 2\ & 2\ & 3\ & 3\ & 3 \\
  2\ & 2\ & 1\ & 2\ & 2\ & 2\ & 3\ & 3\ & 3\ & 3 \\
  2\ & 2\ & 1\ & 2\ & 2\ & 2\ & 3\ & 3\ & 3\ & 4 \\
\end{array}
\right).
\]

For the~mirror image $m(\mathbf6_2)=\mathbf6_2^r$ in framing $f=4$, we get the~following $F_K$ invariant corresponding to non-abelian branch of slope $-\frac{1}{4}$:
\begin{eqnarray*}
F_{\mathbf6_2^r}^{(-\frac{1}{4})}(x,a,q)&=&\sum_{\tilde{d}_1,\cdots,\tilde{d}_{10}}(-1)^{\tilde{d}_2+\tilde{d}_5+\tilde{d}_6+\tilde{d}_9+\tilde{d}_{10}}a^{2\tilde{d}_1+\tilde{d}_2+\tilde{d}_3+2\tilde{d}_4+\tilde{d}_5+\tilde{d}_6+\tilde{d}_7+\tilde{d}_{10}} \\
&&\times\,\, q^{\frac{1}{2}(-2\tilde{d}_1+\tilde{d}_2-4\tilde{d}_4-\tilde{d}_5-\tilde{d}_6-2\tilde{d}_7+2\tilde{d}_8+\tilde{d}_9-3\tilde{d}_{10})} x^{2\tilde{d}_1+\tilde{d}_2+\tilde{d}_3+2\tilde{d}_4+\tilde{d}_5+\tilde{d}_6+\tilde{d}_7}\\
&&\times\,\, q^{-\frac{1}{2}\sum_{i,j} \tilde{C}_{ij}\tilde{d}_i\tilde{d}_j}q^{(\tilde{d}_1+2\tilde{d}_2+2\tilde{d}_3+\tilde{d}_4+2\tilde{d}_5+2\tilde{d}_6+2\tilde{d}_7+3\tilde{d}_8+3\tilde{d}_9+3\tilde{d}_{10})\sum_i\tilde{d}_i}q^{-2(\sum_{i}\tilde{d}_i)^2}\\
&&\times\,\,q^{-\sum_{i<j}\tilde{d}_i\tilde{d}_j}\frac{(x;q^{-1})_{\tilde{d}_1+\cdots+\tilde{d}_{10}}}{\prod_{i=1}^{10}(q)_{\tilde{d}_i}},
\end{eqnarray*}
where
\[
\tilde{C}=\left(
\begin{array}{cccccccccc}
 -2\ & -2\ & -1\ & -1\ & -1\ & -1\ & 0\ & -1\ & 1\ & 1 \\
 -2\ & -1\ & -1\ & 0\ & 0\ & 0\ & 1\ & 0\ & 1\ & 2 \\
 -1\ & -1\ & 0\ & 1\ & 0\ & 0\ & 1\ & 0\ & 1\ & 2 \\
 -1\ & 0\ & 1\ & 0\ & 0\ & 0\ & 1\ & 0\ & 2\ & 1 \\
 -1\ & 0\ & 0\ & 0\ & 1\ & 1\ & 1\ & 1\ & 2\ & 2 \\
 -1\ & 0\ & 0\ & 0\ & 1\ & 1\ & 1\ & 1\ & 2\ & 2 \\
 0\ & 1\ & 1\ & 1\ & 1\ & 1\ & 2\ & 1\ & 2\ & 2 \\
 -1\ & 0\ & 0\ & 0\ & 1\ & 1\ & 1\ & 2\ & 2\ & 3 \\
 1\ & 1\ & 1\ & 2\ & 2\ & 2\ & 2\ & 2\ & 3\ & 3 \\
 1\ & 2\ & 2\ & 1\ & 2\ & 2\ & 2\ & 3\ & 3\ & 3 
 \end{array}
\right).
\]

\subsubsection*{\texorpdfstring{$\mathbf6_3$}{63} {\bf knot}}


We note that knot $\mathbf6_3$ is amphichiral. 
For $\mathbf6_3$ in framing $f=3$, we obtain the~following $F_K$ invariant corresponding to a~non-abelian branch for slope $-\frac{1}{3}$:
\begin{eqnarray*}
F_{\mathbf6_3}^{(-\frac{1}{3})}(x,a,q)&=&\sum_{\tilde{d}_1,\ldots,\tilde{d}_{12}}(-1)^{\tilde{d}_1+\tilde{d}_3+\tilde{d}_4+\tilde{d}_5+\tilde{d}_8+\tilde{d}_{10}+\tilde{d}_{11}}a^{\tilde{d}_1+2\tilde{d}_2+\tilde{d}_3+\tilde{d}_4+2\tilde{d}_5+\tilde{d}_6+\tilde{d}_7+2\tilde{d}_9+\tilde{d}_{10}+\tilde{d}_{11}} \\
&&\times\,\,
q^{\frac{1}{2}(-\tilde{d}_1-4\tilde{d}_2-\tilde{d}_3-3\tilde{d}_4-3\tilde{d}_5-2\tilde{d}_7+\tilde{d}_8-2\tilde{d}_9+\tilde{d}_{10}-\tilde{d}_{11}+2\tilde{d}_{12})}\\
&&\times\,\,
x^{
2\tilde{d}_1+\tilde{d}_2+\tilde{d}_3+2\tilde{d}_5+2\tilde{d}_6+\tilde{d}_7+2\tilde{d}_9+2\tilde{d}_{10}+\tilde{d}_{11}+\tilde{d}_{12}} \\
&&\times\,\,
q^{\frac{1}{2}\sum_{i,j} \tilde{C}_{ij}\tilde{d}_i\tilde{d}_j+(\tilde{d}_1+2\tilde{d}_2+2\tilde{d}_3+3\tilde{d}_4+\tilde{d}_5+\tilde{d}_6+2\tilde{d}_7+3\tilde{d}_8+\tilde{d}_9+\tilde{d}_{10}+2\tilde{d}_{11}+2\tilde{d}_{12})\sum_{i}\tilde{d}_i-\frac{3}{2}(\sum_i \tilde{d}_i)^2}\\
&&\times\,\,
\frac{(x;q^{-1})_{\tilde{d}_1+\cdots+\tilde{d}_{12}}}{\prod_{i=1}^{12}(q)_{\tilde{d}_i}},\end{eqnarray*}
where
\[
\tilde{C}=\left(
\begin{array}{cccccccccccc}
 0\ & 0\ & 0 & -1\ & 0\ & 0\ & -1\ & -1\ & 0\ & 0\ & -1\ & -1 \\
 0\ & 1\ & 0 & -1\ & 1\ & 0\ & -1\ & -2\ & 1\ & 1\ & 0\ & -1 \\
 0\ & 0\ & 0 & -1\ & 1\ & 0\ & 0\ & -2\ & 1\ & 1\ & 0\ & 0 \\
 -1\ & -1\ & -1\ & -2\ & 0\ & -1\ & -2\ & -3\ & -1\ & 0\ & -2\ & -2 \\
 0\ & 1\ & 1\ & 0 \ & 2\ & 1\ & 0\ & -1\ & 2\ & 1\ & 1\ & -1 \\
 0\ & 0\ & 0\ & -1\ & 1\ & 1\ & 0\ & -1\ & 2\ & 1\ & 1\ & 0 \\
 -1\ & -1\ & 0\ & -2\ & 0\ & 0\ & -1\ & -2\ & 0\ & 0\ & -1\ & -2 \\
 -1\ & -2\ & -2\ & -3 \ & -1\ & -1\ & -2\ & -2\ & 0\ & -1\ & -1\ & -2 \\
 0\ & 1 & 1\ & -1\ & 2\ & 2 & 0\ & 0\ & 3\ & 2\ & 1\ & 0 \\
 0\ & 1 & 1\ & 0\ & 1\ & 1 & 0\ & -1\ & 2\ & 2\ & 1\ & 0 \\
 -1\ & 0 & 0\ & -2\ & 1\ & 1 & -1\ & -1\ & 1\ & 1\ & 0\ & -1 \\
 -1\ & -1 & 0\ & -2\ & -1\ & 0 & -2\ & -2\ & 0\ & 0\ & -1\ & -1 \\
\end{array}
\right).
\]

\subsubsection*{\texorpdfstring{$\mathbf8_{19}$}{819} {\bf knot}}


For $\mathbf8_{19}$, which is $(3,4)$-torus knot, the~corresponding $F_K$ invariant for abelian branch is given by
\begin{eqnarray*} F_{\mathbf8_{19}}(x,a,q)&=&\sum_{\tilde{d}_1,\cdots,\tilde{d}_{10}}(-1)^{\tilde{d}_5+\tilde{d}_6+\tilde{d}_7+\tilde{d}_8+\tilde{d}_9}a^{\tilde{d}_5+\tilde{d}_6+\tilde{d}_7+\tilde{d}_8+\tilde{d}_9+2\tilde{d}_{10}}\\
&&\times\,\, q^{\frac{1}{2}(2\tilde{d}_1+2\tilde{d}_2+4\tilde{d}_3+6\tilde{d}_4-\tilde{d}_5-\tilde{d}_6+\tilde{d}_7+\tilde{d}_8+3\tilde{d}_9-2\tilde{d}_{10})} x^{\tilde{d}_1+2\tilde{d}_2+3\tilde{d}_3+5\tilde{d}_4+\tilde{d}_5+2\tilde{d}_6+3\tilde{d}_7+4\tilde{d}_8+5\tilde{d}_9+4\tilde{d}_{10}}
\\
&&\times\,\,q^{\frac{1}{2}\sum_{i,j} \tilde{C}_{ij}\tilde{d}_i\tilde{d}_j}q^{(-\tilde{d}_1-2\tilde{d}_2-3\tilde{d}_3-5\tilde{d}_4-\tilde{d}_5-2\tilde{d}_6-3\tilde{d}_7-4\tilde{d}_8-5\tilde{d}_9-4\tilde{d}_{10})\sum_i \tilde{d}_i}\\
&&\times\,\, \frac{(x;q^{-1})_{\tilde{d}_1+\cdots+\tilde{d}_{10}}}{\prod_{i=1}^{10}(q)_{\tilde{d}_i}},
\end{eqnarray*}
where
\[
\tilde{C}=\left(\begin{array}{ccccccccccc}
2&3&3&5&2&3&3&5&5&5\\
3&4&4&5&3&4&4&5&5&5\\
3&4&4&5&4&4&4&6&5&6\\
5&5&5&6&6&5&6&6&6&6\\
2&3&4&6&3&3&4&5&6&5\\
3&4&4&5&3&5&4&6&5&6\\
3&4&4&6&4&4&5&6&6&6\\
5&5&6&6&5&6&6&7&7&7\\
5&5&5&6&6&5&6&7&7&7\\
5&5&6&6&5&6&6&7&7&8
\end{array}\right).
\]

For the~mirror image $m(\mathbf8_{19})=\mathbf8_{19}^r$, we obtain $F_K$ invariant corresponding to the~non-abelian branch of slope  $-\frac{1}{6}$: 

\begin{eqnarray*}
F_{\mathbf8_{19}^r}^{f=8}(x,a,q)&=&
\sum_{\tilde{d}_1,\cdots,\tilde{d}_{10}}(-1)^{\tilde{d}_5+\tilde{d}_6+\tilde{d}_7+\tilde{d}_8+\tilde{d}_9}a^{2\tilde{d}_1+2\tilde{d}_2+2\tilde{d}_3+2\tilde{d}_4+2\tilde{d}_5+\tilde{d}_6+\tilde{d}_7+\tilde{d}_8+\tilde{d}_9+\tilde{d}_{10}}\\
&&\times\,\, q^{\frac{1}{2}(-2\tilde{d}_1-4\tilde{d}_2-4\tilde{d}_3-6\tilde{d}_4-8\tilde{d}_5-\tilde{d}_6-\tilde{d}_7-3\tilde{d}_8-3\tilde{d}_9-5\tilde{d}_{10})} x^{3\tilde{d}_1+2\tilde{d}_2+2\tilde{d}_3+\tilde{d}_4+\tilde{d}_5+2\tilde{d}_6+\tilde{d}_7+\tilde{d}_8}\\
&&\times\,\,q^{-\frac{1}{2}\sum_{i,j} \tilde{C}_{ij}\tilde{d}_i\tilde{d}_j} q^{(-2\tilde{d}_1-\tilde{d}_2-\tilde{d}_3-\tilde{d}_6+\tilde{d}_9+\tilde{d}_{10})\sum_i\tilde{d}_i}\\
&&\times\,\,q^{-\sum_{i<j}\tilde{d}_i\tilde{d}_j}\frac{(x;q^{-1})_{\tilde{d}_1+\cdots+\tilde{d}_{10}}}{\prod_{i=1}^{10}(q)_{\tilde{d}_i}},
\end{eqnarray*}
where
\[
\tilde{C}=\left(\begin{array}{cccccccccc}
0&1&2&3&5&1&2&3&4&5\\
1&2&3&3&5&2&3&3&5&5\\
2&3&4&4&5&3&4&4&5&5\\
3&3&4&4&5&4&4&4&6&5\\
5&5&5&5&6&6&5&6&6&6\\
1&2&3&4&6&3&3&4&5&6\\
2&3&4&4&5&3&5&4&6&5\\
3&3&4&4&6&4&4&5&6&6\\
4&5&5&6&6&5&6&6&7&7\\
5&5&5&5&6&6&5&6&7&7
\end{array}\right).
\]

\newpage
\bibliography{ref}
\bibliographystyle{alpha}

\end{document}